\documentclass[12pt]{JHEP3}

\usepackage{amscd,amsmath,amssymb,amsfonts,xspace,mathrsfs,amsthm}
\usepackage{color, mathtools}
\usepackage[dvipsnames]{xcolor}

\usepackage{url}
\let\normalcolor\relax

\usepackage{latexsym}
\usepackage{graphicx}
\usepackage{dsfont}
\usepackage{longtable}
\usepackage{enumitem}

\usepackage[utf8]{inputenc}
\usepackage{multirow,float}
\usepackage{pbox}
\usepackage{rotating}
\usepackage{placeins}

\newtheorem{theorem}{Theorem}
\newtheorem{proposition}{Proposition}

\newtheorem{lemma}{Lemma}
\newtheorem{conj}{Conjecture}

\newtheorem{definition}{Definition}

\numberwithin{equation}{section}

\def\bea{\begin{eqnarray}}
\def\eea{\end{eqnarray}}
\def\be{\begin{equation}}
\def\ee{\end{equation}}
\def\ba{\begin{align}}
\def\ea{\end{align}}
\def\bse{\begin{subequations}}
\def\ese{\end{subequations}}

\newcommand{\nn}{\nonumber}

\def\det{\,{\rm det}\, }

\def\Im{\,{\rm Im}\,}

\DeclareMathOperator{\ch}{ch}

\DeclareMathOperator{\Coh}{Coh}

\DeclareMathOperator{\Pic}{Pic}
\DeclareMathOperator{\Li}{Li}
\DeclareMathOperator{\Aut}{Aut}

\newcommand{\rank}{\mbox{rank}}

\def\({\left(}
\def\){\right)}
\def\[{\left[}
\def\]{\right]}
\def\<{\left\langle}
\def\>{\right\rangle}

\newcommand{\eps}{\epsilon}

\newcommand{\de}{\mathrm{d}}

\newcommand{\I}{\mathrm{i}}

\newcommand{\cA}{\mathcal{A}}
\newcommand{\cB}{\mathcal{B}}
\newcommand{\cC}{\mathcal{C}}
\newcommand{\cD}{\mathcal{D}}
\newcommand{\cE}{\mathcal{E}}
\newcommand{\cF}{\mathcal{F}}
\newcommand{\cG}{\mathcal{G}}
\newcommand{\cH}{\mathcal{H}}

\newcommand{\cJ}{\mathcal{J}}
\newcommand{\cK}{\mathcal{K}}
\newcommand{\cL}{\mathcal{L}}
\newcommand{\cM}{\mathcal{M}}

\newcommand{\cO}{\mathcal{O}}

\newcommand{\cR}{\mathcal{R}}
\newcommand{\cS}{\mathcal{S}}

\newcommand{\ccF}{\check{\mathcal{F}}}

\newcommand{\IR}{\mathds{R}}
\newcommand{\IC}{\mathds{C}}
\newcommand{\IZ}{\mathds{Z}}

\newcommand{\IP}{\mathds{P}}
\newcommand{\IF}{\mathds{F}}

\def\whh{\widehat{h}}
\def\whE{\widehat{E}}

\def\bOm{\overline{\Omega}}

\def\gmax{g_{\rm max}}

\newcommand{\q}{\mbox{q}}

\newcommand\PT{\operatorname{PT}}

\def\GW{{\rm GW}}
\def\GV{{\rm GV}}

\def\BM{\begin{matrix}}
\def\EM{\end{matrix}}

\def\hatS{\widehat{S}}
\def\hatf{\widehat{f}}
\def\hatkappa{\widehat{\kappa}}
\def\hatc{\widehat{c}}
\def\tildeS{\widetilde{S}}
\def\tildef{\widetilde{f}}

\newcommand{\AESZ}[1]{\ensuremath{{}^{\rm A\!E}_{\rm SZ}} #1}

\title{
Revisiting the Quantum Geometry of Torus-fibered Calabi-Yau Threefolds
}

\author{Boris Pioline$^1$ and Thorsten Schimannek$^2$,
\\
$^1$ \textit{
Laboratoire de Physique Th\'eorique et Hautes Energies, 
CNRS \\ and Sorbonne Universit\'e, 
Campus Pierre et Marie Curie, 4 place Jussieu, F-75005 Paris, France} \\

$^2$  \textit{Institute for Theoretical Physics \& Department of Mathematics, \\ Utrecht University, 3584 CC Utrecht, The Netherlands}

\vspace*{2mm} {\tt e-mail:
\email{pioline@lpthe.jussieu.fr},
\email{thorsten.schimannek@gmail.com}
}

\vspace*{-3mm}

}

\abstract{About ten years ago, Katz, Klemm and Huang conjectured that topological string amplitudes on compact, elliptically fibered Calabi-Yau threefolds at fixed base degree could be expressed in terms of meromorphic Jacobi forms for $SL(2,\mathbb{Z})$, giving access to Gromov-Witten invariants at arbitrary genus. This was later generalized to torus-fibered CY threefolds with $N$-sections, where topological string amplitudes are conjecturally governed by meromorphic Jacobi forms under the congruence subgroup $\Gamma_1(N)$. In this work, we show that these modularity properties follow from (and are equivalent to) the wave-function property of the topological string partition function $Z_{\rm top}$ under a relative conifold monodromy, implementing a particular Fourier-Mukai transformation on the derived category of coherent sheaves. In particular, we introduce a variant of $Z_{\rm top}$ which is both holomorphic and modular covariant.
Under the same relative conifold monodromy, the generating series of genus 0 Gopakumar-Vafa invariants at fixed base degree is mapped to the generating series of rank 0 Donaldson-Thomas indices counting D4-D2-D0-brane bound states wrapped on the torus fiber. We show that the quasimodularity of the generating series of GV invariants matches the expected mock-modular behavior of the generating series of D4-D2-D0 indices, despite having different multi-cover contributions. We analyze and tabulate a large number of CY threefolds fibered over del Pezzo surfaces, with an $N$-section for $N\leq 5$, including several new examples beyond the realm of toric geometry.}

\begin{document}
\setlength{\parskip}{0.2cm}

\section{Introduction}
Type II strings compactified on a Calabi-Yau threefold $X$ provide a tractable yet extremely rich arena to investigate non-perturbative aspects of string theory, with profound connections to many topics in mathematics including algebraic and symplectic geometry. Topological string theory, obtained by a suitable  twist (of A or B-type) of the superstring worldsheet theory \cite{Witten:1988xj,Bershadsky:1993ta}, encodes a particular set of protected couplings in the low energy effective action \cite{Antoniadis:1993ze}. Remarkably, it also determines (at least in principle) the full spectrum of BPS states \cite{Katz:1999xq,gw-dt,Feyzbakhsh:2021nds}. 
Mathematically, the topological A-model counts holomorphic curves of arbitrary genus in $X$ (more precisely, computes their Gromov-Witten invariants), while 
the topological B-model does not yet have a first principle mathematical formulation beyond genus 0 and 1 (where it reduces to variation of Hodge structure and analytic Ray-Singer torsion, respectively; 
see \cite{Costello:2012cy,Caldararu:2024muy} for attempts to define the B-model at all genera).  
Both models are related by mirror symmetry, and computable by localization methods when $X$ is toric (hence non-compact) \cite{Aganagic:2003db,Fang:2016svw}.

When $X$ is a smooth, compact CY threefold, the only currently available method for computing topological string amplitudes at higher genus, that does not rely on the existence of a fibration structure, is to  exploit the holomorphic anomaly equations~\cite{Bershadsky:1993ta}.
These have to be supplemented by certain boundary conditions in order to fix the holomorphic ambiguity arising at each genus. Since the number of ambiguities grows faster than the currently known boundary conditions (including 
Castelnuovo bounds at large volume and gap vanishing conditions at conifold points), this puts a bound on the maximal computable genus, e.g. 53 for the benchmark case of the quintic threefold $\IP^4[5]$\cite{Huang:2006hq}. Recently, this upper bound was pushed up further for the quintic threefold and other one-modulus hypergeometric 
models~\cite{Alexandrov:2023zjb,Alexandrov:2023ltz}, by exploiting relations between Gromov-Witten (GW) invariants and rank 0 Donaldson-Thomas invariants associated to the derived category of coherent sheaves $\cC=D^b\Coh(X)$~\cite{gw-dt,Feyzbakhsh:2021nds}, and enforcing the (mock) modularity of the corresponding generating series of  D4-D2-D0 BPS 
indices predicted by S-duality~\cite{Alexandrov:2019rth,Alexandrov:2018lgp,Alexandrov:2022pgd}. Unfortunately, the infinite set of additional constraints gained in this way still grows slower than the number of holomorphic ambiguities. Moreover, despite overwhelming evidence the modularity properties are still conjectural (see however~\cite{Sheshmani} for some recent progress). 

In the case of CY threefolds with K3 or genus one fibrations,  one may obtain further boundary conditions by exploiting the modularity associated with the fiber. For example, for CY threefolds fibered by lattice-polarized K3 surfaces, the generating series of vertical (i.e. zero base degree) Gromov-Witten invariants are known to transform as vector-valued modular forms, being closely related to the Noether-Lefschetz invariants of the fibration, and can be determined to arbitrary genus \cite{maulik2007gromov,klemm2010noether}.  Similar modularity constraints apply at non-zero base degree, although they are currently only understood at genus 0 and become increasingly more complicated when the 
degree increases~\cite{Henningson:1996jf,Berglund:1997eb,Doran:2024kcb}. Moreover, vertical D4-D2-D0
invariants are determined by the same Noether-Lefschetz invariants \cite{Bouchard:2016lfg}, and hence also transform as vector-valued modular forms, establishing the S-duality predictions in this case.

In this work, we study the case of smooth CY threefolds fibered by genus one curves\footnote{We reserve the terminology \textit{elliptic fibration} for a genus one fibration with a section. A general genus-one fibered CY threefold need only admit a multisection of degree $N\geq 1$. We will say that a geometry exhibits a genus one fibration with an $N$-section if there is no $N'$-section with $N'<N$.}. In this case, it was observed in \cite{Candelas:1994hw,Alim:2012ss,Klemm:2012sx} that generating series of GW invariants with fixed base degree and genus transform as quasi-holomorphic modular forms under $SL(2,\IZ)$ (in the case of elliptic fibrations), and satisfy holomorphic (or modular) anomaly equations similar to \cite{Bershadsky:1993ta}. Based on this structure, and the duality with F-theory~\cite{Klemm:1996hh,Haghighat:2013gba,Haghighat:2014vxa}, it was conjectured in \cite{Huang:2015sta} that the topological string partition function at fixed base degree (but arbitrary genus) can be expressed in terms of meromorphic Jacobi forms, where the topological string coupling plays the role of the elliptic parameter. Since the ring of Jacobi forms is finitely generated, this opens the way to determine the topological amplitude at fixed base degree and arbitrary genus, provided the order of the poles in 
the elliptic parameter can be controlled. While the elliptic transformation property has by now been rigorously proven for reduced divisor classes in general elliptic fibrations \cite{Oberdieck:2016nvt}, and the holomorphic anomaly equations proven in some special cases \cite{Oberdieck:2017pqm}, the modular transformation property remains conjectural. In \cite{Cota:2019cjx} these conjectures were extended to the case of genus one fibrations with an $N$-section, where the modular group is restricted to the congruence subgroup $\Gamma_1(N)\subset SL(2,\IZ)$ defined in \eqref{eq:defg1N}.
It was further observed in~\cite{Knapp:2021vkm} that the topological string partition functions of two smooth genus one fibered Calabi-Yau threefolds with a $5$-section that share the same Jacobian fibration transform into each other under an Atkin-Lehner involution.
More generally, it was proposed that the topological string partition functions of different torus fibrations that share the same Jacobian fibration transform as vector valued Jacobi forms under the full modular group, after taking into account the presence of flat but topologically non-trivial B-fields that non-commutatively resolve certain singularities~\cite{Schimannek:2021pau}, see also~\cite{Katz:2022lyl,Katz:2023zan,Schimannek:2025cok,Knapp:2025hnf}.
This structure has been further clarified and generalized in~\cite{Duque:2025kaa}.
Physically, the vector valued modularity can be understood after relating the topological string partition function to the twisted-twined elliptic genera of non-critical strings in the six-dimensional F-theory compactification that is associated to the fibration~\cite{Schimannek:2021pau,Dierigl:2022zll}.

Our first main result is to derive the holomorphic anomaly equations and modular transformation rules for smooth genus one fibrations $X\stackrel{\pi}\to B$ over a generalized del Pezzo surface $B$ with $h^{1,1}(X)=h^{1,1}(B)+1$ (in particular no fibral divisors), from the wave function property of the topological string partition function \cite{Witten:1993ed,Verlinde:2004ck,Aganagic:2006wq,Gunaydin:2006bz,Alexandrov:2010ca}. Recall that the holomorphic anomaly equations of \cite{Bershadsky:1993ta} can be interpreted \cite{Witten:1993ed} as the statement that 
the full, non-holomorphic topological string partition function $\Psi_{\rm BCOV}(t^a,\bar t^a,\lambda,x^a)$ can be viewed as the overlap of a background-independent wave function $| \Psi_{\rm top} \rangle$, living in a finite-dimensional Hilbert space $\cH\sim L^2(\IC^{b_2(X)+1})$ associated to the quantization of the even cohomology of $X$ (or middle cohomology of the mirror $\widehat{X}$, depending whether one is interested in the $A$ or $B$-model) against a family of coherent states $_{t,\bar t}\langle\lambda, x^a |$ 
parametrized by the background moduli $t^a$ specifying the complexified K\"ahler or complex structure, respectively. Under a monodromy around the discriminant locus in moduli space, both the state $| \Psi_{\rm top} \rangle$ and the coherent states $_{t,\bar t}\langle\lambda, x^a |$ transform according to the metaplectic representation of $Sp(2b_2(X)+2,\IZ)=Sp(b_3(\widehat{X}),\IZ)$, such that their overlap is invariant. In the limit $\bar t^a\to-\I\infty$ keeping $t^a$ fixed, $\Psi_{\rm BCOV}(t^a,\bar t^a,\lambda,x^a)$ reduces to the usual topological string partition function $Z_{\rm top}(t^a,\lambda)$.
This interpretation was used in \cite{Aganagic:2006wq} to elucidate the modular properties of topological string amplitudes on local CY manifolds.

In the case of a smooth, compact genus-one fibered CY threefold $X$ with an $N$-section, the homology class of the generic fiber is $N$-divisible.
It is therefore natural to introduce a complexified K\"ahler parameter $T$ such that $NT$ is the complexified volume of the torus fiber.
There is a natural monodromy $U$, the so called relative conifold monodromy, which acts on this parameter as $T\mapsto T/(1+NT)$, and on the derived category $D^b\Coh(X)$ by a Fourier-Mukai transformation with kernel given by the ideal sheaf of the relative diagonal~\cite{Seidel2001,HernndezRuiprez2009,Schimannek:2019ijf,Cota:2019cjx}.\footnote{The usual conifold monodromy, around the locus where the central charge of the 6-brane vanishes, corresponds to the Fourier-Mukai transformation with kernel given by the ideal sheaf of the diagonal $\Delta\subset X\times X$~\cite{Kontsevich:1994dn}. The \textit{relative} conifold mondromy arises around the component of the discriminant (in the stringy K\"ahler moduli space of the Calabi-Yau) where the central charge of a 2-brane wrapping the torus fiber vanishes and the Fourier-Mukai kernel is the ideal sheaf of the relative diagonal $\Delta_B\subset X\times_B X$. For a gentle introduction we refer to~\cite[Section 3.2-3.3]{Cota:2019cjx}.}
For $N=1$, this can be combined with large volume monodromies $T\mapsto T+1$ to obtain another monodromy $S$ acting as $T\mapsto -1/T$, corresponding to a double T-duality with respect to the genus-one fiber~\cite{Andreas:2000sj,Andreas:2001ve,Bena:2006qm}, but in general $S$ is not a symmetry, rather it maps $X$ to another element in the Tate-Shafarevitch group~\cite{Caldararu2002,Schimannek:2021pau,Duque:2025kaa}. Assuming that $Z_{\rm top}(t^a,\lambda)$ transforms under $U$ according to the metaplectic representation, we shall derive the Jacobi transformation property of the topological string partition function at fixed base degree under $T\mapsto T/(1+NT)$, recovering the predictions of \cite{Huang:2015sta,Cota:2019cjx}. Conversely, this may be taken as evidence that the topological string partition function (an object whose mathematical definition remains obscure) does transform according to the metaplectic representation under monodromies. A key step in the derivation of the Jacobi properties is a new representation of the topological wave function $Z_{\rm mod}(t^a,\lambda)$ which is completely free of modular anomalies, and encodes the `depth zero' part of the quasimodular generating series of invariants.

Our second main result is to  compute D4-D2-D0 indices (or rank 0 DT invariants) supported on  divisors $D=\pi^*(\check{D})$ pulled back from the base (i.e. D4-branes wrapped both along a basis divisor $\check D\subset B$ and the generic genus-one fiber) by applying the same monodromy $U$ to vertical D2-D0 branes. This generalizes the work \cite{Klemm:2012sx} to arbitrary genus one fibrations without fibral divisors. In particular, we show that the holomorphic anomaly equations satisfied by generating series of genus 0 GW invariants at fixed base degree agree with the modular anomaly equations satisfied by rank 0 DT invariants, even though multicover effects are different on both sides. 

The rest of this work is organized as follows. In \S\ref{sec:preliminaries} we set up notations and review basic facts and conjectures about various enumerative invariants of genus one fibered CY threefolds. In \S\ref{sec:basedegreezero}, we focus on base degree zero Gromov-Witten invariants, which are entirely determined by the Euler numbers of the 
threefold $X$ and the base $B$ and by the multiplicities $N_k$ of fibral curves that intersect the $N$-section $k$ times. We express the resulting generating series as linear combinations of Eisenstein series of $\Gamma_1(N)$, and (in the genus 0 and genus 1 case) holomorphic Eichler integrals thereof. 
These results rely on Proposition~\ref{prop:eisenstein} as well as the Lemmas~\ref{lem:g2g},~\ref{lem:g01} and Theorem~\ref{thm:eichler}, the proofs of which are relegated to Appendix~\ref{sec:modularity}.
In \S\ref{sec_Ztop}, after reviewing the wave function property of the topological string partition function, we compute the
monodromy matrix implementing the Fourier-Mukai transformation with respect to the ideal sheaf of the relative diagonal, and spell out the resulting transformation properties of the 
generating series of Gromov-Witten invariants at fixed base degree. For vanishing base degree, we recover the anomalous transformation properties following from the Eichler integral representations obtained in \S\ref{sec:basedegreezero}. We introduce a new `modular polarization' where modular anomalies are absent, and use it to establish the Jacobi properties of the normalized topological string partition function at fixed base degree. In \S\ref{sec_DTmodEll}, we study the implications of the invariance under the relative conifold monodromy on Donaldson-Thomas invariants (assuming the absence of wall-crossing). In this way we obtain (heuristic) derivations of the elliptic property of the generating series of PT invariants, of the periodicity of base degree zero GV invariants, and of the S-duality property of generating series of MSW invariants counting D4-D2-D0 branes wrapped on the genus fiber (times a fixed divisor on the base). In \S\ref{sec_discussion} we summarize and discuss a few open questions. The remainder of this long paper consists of a suite of appendices collecting background material, technical computations and lots of examples.
In \S\ref{sec:modularity}, we discuss Eisenstein series for $\Gamma_1(N)$ and prove new results for their transformation under Fricke involutions, expressions in terms of polylogarithms and their holomorphic Eichler integrals.
Further details on the evaluation of the action of the relative conifold monodromy on the Chern classes of branes are provided in \S\ref{app_relcon}.
In \S\ref{sec:genusoneconstruction}, we discuss the construction of generic genus one fibrations over generalized del Pezzo surfaces with $N$-sections for $N\le 5$ and provide general expressions for their topological invariants (summarized in Table~\ref{tab:fibrationsGenericData}).
A large set of examples of generic genus one fibered CY threefolds over $\mathbb{P}^2$ with $h^{1,1}=2$ is collected/constructed in \S\ref{sec:examplesP2}, while some examples over other del Pezzo surfaces are discussed in \S\ref{app_higher_rank}.\footnote{The topological invariants of 72 genus one fibrations over $\mathbb{P}^2$ are summarized in Tables~\ref{tab:gp22sec},~\ref{tab:gp23sec},~\ref{tab:gp24sec} and~\ref{tab:gp25sec}, and are collected in a \texttt{Mathematica} worksheet available with the source code on arXiv.}
In both cases we also provide tables of modular generating series of GW invariants at base degree 1 and 2.

\medskip

\noindent {\bf Acknowledgements:} We thank Sergey Alexandrov, Cesar Fierro Cota, Markus Dierigl, Emanuel Scheidegger, Amir-Kian Kashani-Poor, Jan Manschot, Paul Oehlmann and Stefan Vandoren for valuable discussions. We are particularly grateful to Sergey Alexandrov for collaboration at the initial stage of this project, and detailed comments on the draft. We also thank the referee for their careful reading of the draft and suggesting various improvements. The research of BP is supported by Agence Nationale de la Recherche under contract number ANR-21-CE31-0021.
 \textit{For the purpose of Open Access, a CC-BY public copyright licence has been applied by the authors to the present document and will be applied to all subsequent versions up to the Author Accepted Manuscript arising from this submission.}

\section{Preliminaries}
\label{sec:preliminaries}
In this section, we first set up notations and review basic properties of enumerative invariants and topological string amplitudes
on a general smooth projective CY threefold $X$. We then review some of the special properties that have been conjectured to arise in the genus one fibered case.

\subsection{Prepotential, topological free energies and Gromov-Witten invariants}
\label{sec:ZtopGW}

Let us denote by $\{C^a\}$, $a=1\dots b_2(X)=h^{1,1}(X)$ a basis of effective curves 
in $H_2(X,\IZ)$, $\{H_a\}$ the dual basis in $H_4(X,\IZ)$ 
 such that  $H_a \cap C^b=\delta_a^b$, by $\kappa_{abc}=H_a\cap H_b\cap H_c$ the triple intersection numbers, and by $c_{2,a}=\int_{H_a} c_2(TX)$ the intersections with the
second Chern class. 

At two-derivative order, the low energy effective action describing type IIA string compactified on $X$ in the vector multiplet sector is determined by the tree-level prepotential $F^{(0)}(X^\Lambda)=(X^0)^2 \cF(t^a)$, a holomorphic homogeneous function of degree 2 in local coordinates $X^\Lambda, \Lambda=0,1,\dots, b_2(X)$  where $t^a=X^a/X^0$, $a=1\dots b_2(X)$ are the complexified K\"ahler parameters, such that the complexified K\"ahler form reads $\omega=t^a H_a$
with $\Im t^a>0$ (here and below, we use the same symbol for an element in $H_p(X,\IZ)$ and
its Poincar\'e dual in  $H^{6-p}(X,\IZ)$).~\footnote{If the K\"ahler cone is non-simplicial, then $H_a$ have to be chosen such that they generate $H^2(X,\IZ)$ and $\Im t^a>0$ is a simplicial sub-cone of the K\"ahler cone.}
In the large volume limit $t^a\to\I\infty$, one has~\footnote{Here we use a non-integral basis, 
referred to as \textit{primed} basis in \cite{Alexandrov:2010ca}, that differs slightly from the integral basis typically used in the topological string literature, but which allows to get rid of  quadratic and linear terms in $t^a$. \label{foo:primed}}
\be
\cF^{(0)}(t^a) = \frac16 (2\pi\I)^3 \kappa_{abc}t^a t^b t^c- \frac{1}{2} \zeta(3) \chi_X
+\sum_{\beta_a> 0} \GW^{(0)}_{\beta_a} e^{2\pi\I \beta_a t^a}
\ee
where $\GW^{(0)}_{\beta_a}$ are the genus 0 Gromov-Witten invariants counting rational curves in $X$ of homology class $\beta_a C^a$ (with $\beta_a>0$ meaning that 
$\beta_a\geq 0$ for all $a=1\dots b_2(X)$, not all vanishing at once). The homogeneous prepotential $F^{(0)}$ can be understood as the generating function of a Lagrangian subspace $F_\Lambda=\partial_{X^\Lambda}F^{(0)}$ inside the complex vector space $\IC^{2b_2(X)+2} \sim H^{\rm even}(X,\IC)$ with coordinates $V=(F_\Lambda,X^\Lambda)^t$ equipped with the symplectic form $\omega=\de X^\Lambda \wedge \de F_\Lambda$. As a result, under a monodromy 
\be
\label{SpABCDV}
V=\begin{pmatrix} F_\Lambda \\ X^\Lambda   \end{pmatrix} \mapsto 
V'=\begin{pmatrix} A & B \\ C & D \end{pmatrix} \begin{pmatrix}  F_\Lambda \\ X^\Lambda  \end{pmatrix}\,,
\ee 
where $A,B,C,D$ are square matrices of size $b_2(X)$ such that
\be
\label{SpABCD}
A^T D - C^T B = 1, \quad A^T C = C^T A, \quad B D^T = D^T B \,,
\ee
the prepotential transforms as~\cite{Ceresole:1995jg} 
\be
\label{Ftrans}
F(X^\Lambda) \mapsto F'(X'^\Lambda) = F(X^\Lambda) - S(X^\Lambda, X'^\Lambda)
\ee
where, assuming that the matrix $C$ is invertible,
\be
\label{SABCDX}
S(X^\Lambda,X'^\Lambda)=-\frac12 X^\Lambda (C^{-1} D)_{\Lambda\Sigma} X^\Sigma 
+ X^\Lambda C^{-1}_{\Lambda\Sigma} X'^{\Sigma}
-\frac12 X'^{\Lambda} (AC^{-1})_{\Lambda\Sigma}  X'^{\Sigma}
\ee
satisfies $\partial_{X^\Lambda} S=F_\Lambda, \partial_{X'^\Lambda}S=-F'_\Lambda$, such that the r.h.s. of \eqref{Ftrans} is extremized with respect to $X^\Lambda$.

 At higher order in the derivative expansion, special protected couplings are governed by the topological free energies $\cF^{(g)}(t^a,\bar t^a)$. These observables are not holomorphic, but satisfy holomorphic anomaly equations which will be reviewed in \S\ref{sec_Ztop}. 
The holomorphic free energies $F^{(g)}:=(X^0)^{2-2g} \cF^{(g)}(t^a)$
are defined as the holomorphic limit $\bar t^a\to-\I\infty$ keeping $t^a$ fixed.
Near the large volume point, the $\cF^{(g)}$'s are determined by the Gromov-Witten invariants $\GW^{(g)}_{\beta_a}$ counting stable maps from a curve of genus $g$ into $X$ with image in the class $\beta_a C^a$, up to a classical, linear term at genus 1, and a constant map contribution at genus $g\geq 2$,
\be
\begin{split}
\cF^{(1)}(t^a)  =& -\frac{2\pi \I}{24} c_{2,a} t^a 
+\sum\limits_{\beta_a> 0}^\infty\GW^{(1)}_{\beta_a} 
e^{2\pi \I \beta_a t^a}\,,\\
\cF^{(g\geq 2)}(t^a)  =&\frac{(-1)^{g-1} B_{2g} B_{2g-2} }{(2g)(2g-2) (2g-2)!}  \chi_X
+
\sum\limits_{\beta_a> 0}^\infty\GW^{(g)}_{\beta_a} 
e^{2\pi \I \beta_a t^a}\,.
\end{split}
\ee
The topological string partition function on $X$ is defined as a formal asymptotic series 
\be
\label{defZtop}
Z_{\rm top}(t^a, \lambda) := \lambda^{\frac{\chi_X}{24}-1} 
e^{\sum_{g\geq 0} \lambda^{2g-2} \cF^{(g)}(t^a)}\,,
\ee
where $\lambda$ is the topological string coupling. We note that the $\chi_X$-dependent terms are all captured by a power of the MacMahon function\footnote{The asymptotic expansion of the (logarithm of the) MacMahon function was first worked out in \cite[(E.32)]{Dabholkar:2005dt}, and later rederived in a much simpler way in 
\cite[(4.37)]{Pioline:2006ni}. The two expressions agree after correcting 
$\zeta'(1)$ into $\zeta'(-1)$ in \cite[(4.37)]{Pioline:2006ni}, and using the identity 
$\frac12 \log(2\pi) - \frac{1}{2\pi^2} \zeta'(2) + \frac{1}{12} \gamma_E= -\zeta'(-1)+\frac{1}{12}$. There remains a discrepancy by the additive constant $\frac{1}{12}$, which was apparently missed in \cite{Dabholkar:2005dt}. }
\be
\begin{split}
M(\lambda):=& \prod_{k\geq 1} (1-e^{\I k\lambda})^{-k}
\\
=& 
\exp\left[
-\frac{\zeta(3)}{\lambda^2} + \frac{1}{12} \log(-\I\lambda) +\zeta'(-1) + \sum_{g\geq 2} 
(-1)^{g-1}\tfrac{B_{2g} B_{2g-2}  \lambda^{2g-2} }{(2g)(2g-2) (2g-2)!} 
\right]\,.
\end{split}
\ee
In \S\ref{sec_Ztop} we shall explain that
$Z_{\rm top}$ transforms as a wave-function under monodromies in K\"ahler moduli space,
hence specifying the transformation properties of the respective holomorphic free energies $F^{(g)}$.

GW invariants 
are in general rational numbers, but can be expressed in terms of the integer-valued
Gopakumar-Vafa (GV) invariants $\GV^{(g)}_{\beta_a}$ by the multi-cover formula \cite{Gopakumar:1998ii,Gopakumar:1998jq,IonelParker:2013}
\bea
\label{GWtoGV}
\sum\limits_{g=0}^\infty
\sum\limits_{\beta_a>0} \GW^{(g)}_{\beta_a} 
e^{2\pi \I \beta_a t^a} \lambda^{2g-2}
&=&
\sum\limits_{g=0}^\infty
\sum\limits_{k=1}^\infty
\sum\limits_{\beta_a > 0}
\frac{\GV^{(g)}_{\beta_a}}{k}
\left[2\sin\left(\frac{k\lambda}{2}\right)\right]^{2g-2}e^{2\pi \I k \beta_a t^a}\,.
\eea
At genus 0, this reduces to the standard multicover formula for rational curves 
$ \GW^{(0)}_{\beta_a}=\sum_{d|\beta_a} \frac{1}{d^3} \GV^{(0)}_{\beta_a/d} $. 
While GW invariants are supported on the Mori cone $\beta_a>0$, GV invariants are typically supported on a smaller cone, sometimes called the 
`infinity cone' \cite{Gendler:2022ztv}, up to a few 
so-called nilpotent rays, which by definition support a finite number of non-vanishing genus 0 GV invariants.
The infinity cone is given by the intersection of the Mori cones of all 
CY threefolds birationally equivalent to $X$, obtained by flopping the curves associated to the nilpotent rays, see \cite{Gendler:2022ztv} for a more precise statement. 
Moreover, for fixed class $\beta_a$, $\GV^{(g)}_{\beta_a}$ vanishes when $g>g_{\rm max}(\beta_a)$ is large enough \cite{Doan:2021}. 

\subsection{BPS states and Donaldson-Thomas invariants}
\label{sec:BPS}

BPS states correspond to stable objects $E$ in the bounded derived category of coherent sheaves $\cC=D^b{\rm Coh}(X)$. This category is graded by the Chern character $\ch(E)$, related to the electromagnetic charge $\gamma=(p^0,p^a;q_a,q_0)$ (in the `primed basis', see footnote \ref{foo:primed})
\be
p^0 = \ch_0, \quad
p^a = \int_{C^a} \ch_1, \quad
q_a = - \int_{H_a} (\ch_2 + \tfrac{c_2(TX)}{24} \ch_0), \quad
q_0 = \int_{X} (\ch_3 + \tfrac{c_2(TX)}{24} \ch_1),
\ee
such that the following quantization conditions are satisfied:
\be
\label{qcond}
p^0, p^a\in \IZ, \quad 
q_a \in \IZ - \frac{c_{2,a}}{24}  p^0 - \frac12 \kappa_{abc} p^b p^c, \quad 
q_0\in \IZ - \frac{c_{2,a}}{24}p^a 
\ee
We denote by $\Gamma\subset\mathbb{Q}^{2b_2(X)+2}$ the lattice specified by these quantization conditions, and by $\gamma=(p^0,p^a;q_a,q_0)$ a generic lattice vector. 
$\Gamma$ is equipped with an antisymmetric integer pairing, known as the Dirac-Schwinger-Zwanziger product, or Euler form,
\be
\langle\gamma,\gamma' \rangle = q_0 p'^0 + q_a p'^a - q'_a p^a -  q'_0 p^0
\ee
Donaldson-Thomas invariants depend on a choice of stability condition $\sigma=(Z,\cA)$, where
$Z$ is a central charge function and $\cA$ an Abelian subcategory of $\cC$, determined locally by $Z$, satisfying various axioms. Physical stability conditions are those where the central charge
 is determined by the tree-level prepotential, 
\be
Z_t (\gamma) = q_0 + q_a t^a- p^a \partial_{t^a} \cF^{(0)} - p^0 ( 2 \cF^{(0)}-t^a \partial_{t^a} \cF^{(0)} )
\ee
such that $|Z_t(\gamma)|$ measures the mass of a BPS state with electromagnetic charge $\gamma$.
In the  large volume limit $t^a\to\I\infty$, $Z_t$ is determined by the Chern character of the object $E$,
\be
Z_t(E) \sim \int_X  e^{- t^a H_a}  \ch(E) \left(1+\frac{c_2(TX)}{24} \right) 
\ee
up to terms  proportional to $\zeta(3) \chi_X$,
arising from $\cO(\alpha')^3$ corrections in string theory, 
and up to exponentially suppressed terms as $t^a\to\I\infty$, 
corresponding to worldsheet instanton corrections. We choose
 an Abelian subcategory $\cA_t$ compatible\footnote{Such an Abelian subcategory, or more precisely heart of $t$-structure, exists provided $X$ satisfies the Bayer-Macr\`i-Toda inequality, which remains conjectural for the CY threefolds of interest in this paper. See~\cite{Alexandrov:2023zjb} for references and introduction to stability conditions on CY threefolds aimed at physicists.} with $Z_t$, and  denote by $\Omega_t(\gamma)$ the Donaldson-Thomas invariant 
counting semi-stable objects of charge $\gamma$ for the stability condition $\sigma_t=(Z_t,\cA_t)$. 
When $\sigma_t$ is generic, this is an integer number, corresponding physically to the 
index counting (with signs) BPS states of charge $\gamma\in\Gamma$.
We further define the rational DT invariant $\bar\Omega_t(\gamma)=\sum_{d|\gamma}d^{-2} \Omega_t(\gamma/d)$, which has simpler behavior under wall-crossing~\cite{Manschot:2010qz}. Importantly, $\Omega_t(\gamma)$ and its rational counterpart are invariant under complex structure deformations, and under monodromies in K\"ahler moduli space, in the sense that 
\be
\label{monOm}
\Omega_\sigma(\gamma) = \Omega_{g\cdot \sigma} ( \gamma \cdot g)
\ee
where $g\cdot \sigma$ is the action of the auto-equivalence $g\in{\rm Aut}(\cC)$ implementing the monodromy (acting conventionally from the left on $\sigma$), while $\gamma\cdot g$ is its action on the charge row vector (acting conventionally from the right).
In particular, $\Omega_t(\gamma)$ is invariant under the large volume monodromy 
\bea
\label{lvm}
p^0\mapsto p^0, \quad
p^a\mapsto p^a+p^0 \eps^a, \quad q_a\mapsto q_a - \kappa_{abc} p^b \eps^c
-\frac12 p^0 \kappa_{abc} \eps^b \eps^c, \quad \nn\\
q_0 \mapsto q_0 - q_a \eps^a  + \frac12 \kappa_{abc} p^a \eps^b \eps^c
+\frac16 p^0\kappa_{abc} \eps^b \eps^c \eps^c
\eea
provided the K\"ahler moduli (and the corresponding stability condition) are shifted as 
$t^a\mapsto t^a+\epsilon^a$ (for any $\eps^a\in\IZ$).
DT invariants may however jump on certain walls of marginal stability where the phase of $Z_t(E)$ aligns with the phase of one of the subobjects of $E$. In particular, the r.h.s. of \eqref{monOm} need
not be equal to $\Omega_\sigma(\gamma \cdot h)$.

For $p^0=p^a=0$, the invariant $\Omega_{t^a\to\I\infty} (0,0;\beta_a=q_a,q_0)$ counting semi-stable sheaves supported on the effective curve $C=\beta_a C^a$ is expected to be independent of $q_0$, and to coincide with the genus 0 Gopakumar-Vafa invariant $\GV^{(0)}_{\beta_a}$ \cite{Katz:2006gn}. 
Note however that its rational counterpart $\bar\Omega_{t^a\to\I\infty} (0,0;q_a,q_0)$ is  \textit{not} independent of the D0-brane charge (or Euler number of the sheaf) $q_0$, and differs from the genus 0 GW invariant due to different multi-cover effects. 

For $p^0=-1, p^a=0$,
the  invariant $\Omega_{t\to e^{\I \pi/3} \infty} (-1,0;\beta_a+\frac{c_{2,a}}{24},-m)$ in a suitable large volume, large B-field limit coincides with the invariant $\PT(\beta_a,m)$ counting stable pairs 
 $E: \cO_X \stackrel{s}{\to} F$ where $F$ is a pure one-dimensional sheaf with 
$\ch_2(F)=\beta_a C^a$ and $\chi(F)=m$, and $s$ is a section of $F$ with zero-dimensional kernel
\cite{Pandharipande:2007kc}. PT invariants vanish for $m\ll 0$, and 
are related to GV invariants via \cite{gw-dt,gw-dt2} 
\be
\label{eqGVPT}
\begin{split}
\sum_{\beta_a>0,m} \PT(\beta_a,m) \, e^{2\pi \I \beta_a t^a} \q^m = & 
\prod_{\beta_a>0}\prod_{k>0} 
\left(1-(-\q)^k e^{2\pi \I \beta_a t^a}  \right)^{k \GV^{(0)}_{\beta_a}}
\\
\times \prod_{\beta_a>0}\prod_{g=1}^{\gmax(\beta)}& 
\prod_{\ell=0}^{2g-2}
\left(1- (-\q)^{g-\ell-1} e^{2\pi \I \beta_a t^a} 
\right)^{(-1)^{g+\ell} {\scriptsize{\begin{pmatrix} 2g-2 \\ \ell \end{pmatrix}}}
\GV^{(g)}_{\beta_a}}\, .
\end{split}
\ee
Upon identifying $\q=-e^{\I\lambda}$, 
the right-hand side is recognized as the topological
string partition function, up to a power of the MacMahon function and a polynomial term coming from classical contributions at genus zero and one, 
\be
\label{ZtopPT}
Z_{\rm top} (t^a,\lambda) =  M(\lambda)^{\frac{\chi_X}{2}}\, 
e^{F_{\rm pol}(t^a,\lambda) }\,
 \sum_{\beta_a,m} \PT(\beta_a,m) \, e^{2\pi \I \beta_a t^a +\I m\lambda}\,.
\ee

\subsection{Modularity of D4-D2-D0 indices}
\label{sec_DTreview}

In general, the generating series of rational D4-D2-D0 indices with fixed D4-brane charge $p^a$ associated to an ample divisor class $\cD_p=p^a H_a$
and residual D2-brane charge  $\mu_a$, evaluated in the so-called large volume attractor chamber, 
\be
\label{defhpmu}
h_{p;\mu}(\tau) = \sum_{\hat q_0 \leq \frac{\chi(\cD_p)}{24}} \bOm_{p;\mu}(\hat q_0) \, e^{-2\pi\I\tau \hat q_0}
\ee
transforms as a vector valued mock modular form of weight $-1-\frac{b_2(X)}{2}$ under $SL(2,\IZ)$ \cite{Alexandrov:2016tnf,Alexandrov:2017qhn,Alexandrov:2018lgp,Alexandrov:2019rth,
Alexandrov:2024jnu} (see 
\cite{Manschot:2010sxc,Manschot:2010xp} for early work on this topic, and \cite{Alexandrov:2025sig} for a recent review). 
Here $\hat q_0$ is the \textit{invariant D0-brane charge}
\be
\label{defhq0}
\hat q_0 = q_0 - \frac12 \kappa^{ab} q_a q_b
\ee
where $\kappa^{ab}$ is the inverse\footnote{When $\cD_p$ is ample, 
the matrix $\kappa_{ab}$ has signature $(1,b_2(X)-1)$ and is invertible by the Hodge index theorem. \label{fooHodge}} 
matrix to $\kappa_{ab}:=\kappa_{abc} p^c$, $\mu_a$ labels the coset in $\Lambda^*/\Lambda$ (of cardinality $|\det(\kappa_{ab})|$), with $\Lambda=H^2(X,\mathbb{Z})$, such that the actual D2-brane charge is 
\be
\label{qtomu}
q_a = \mu_a + \kappa_{ab} \eps^b + \frac12 \kappa_{abc} p^b p^c
\ee
 and $\chi(\cD_p)=\kappa_{abc} p^a p^b p^c+ c_{2,a} p^a$ is the Euler number of a divisor $\cD_p$ with class $p$.
 More precisely, there exists a canonical completion 
 \be
 \label{defwh}
 \whh_{p;\mu}(\tau) := h_{p;\mu}(\tau)
 +\sum_{n\geq 2} \sum_{\check\gamma=\sum_{i=1}^n \check\gamma_i}
R_n\left(\{ \check{\gamma}_i\}, \tau_2\right)\, e^{\I \pi\tau\, Q_n(\{ \check{\gamma}_i\})} \prod_{i=1}^n
h_{p_i,\mu_i}(\tau)
\ee
where $\check\gamma=(p^a,q_a)$ denotes the vector of D4 and D2 charges, $Q_n$ is a quadratic form of signature $((n-1) (b_2(X)-1),n-1)$,
\be
Q_n(\{ \check{\gamma}_i\}) = \kappa^{ab} q_a q_b  - \sum_{i=1}^n \kappa_i^{ab} q_{i,a} q_{i,,b}
\ee
and $R_n\left(\{ \check{\gamma}_i\}, \tau_2\right)$ a sum of products of generalized error functions and derivatives thereof~\cite{Alexandrov:2016enp,Nazaroglu:2016lmr,Pioline:2025xgf}, such that 
$\whh_{p;\mu}$
 transforms as a  vector valued modular form
 with  weight $-1-\frac{b_2(X)}{2}$ in the Weil representation of $\Lambda^*/\Lambda$. In particular, under the $S:\tau\mapsto -1/\tau$ and $T:\tau\mapsto \tau+1$ transformations, it transforms by the matrices
\cite[Eq.(2.10)]{Alexandrov:2019rth} (see also \cite{Gaiotto:2006wm, deBoer:2006vg, Denef:2007vg, Manschot:2007ha})
\be
\begin{split}
M_{\mu\nu}(T)=&\, e^{\pi\I\(\mu+\frac{p}{2}\)^2+\tfrac{\pi\I}{12}\, c_{2,a}p^a}\,\delta_{\mu\nu},
\\
M_{\mu\nu}(S)=&\, \frac{(-1)^{\chi(\cO_{\cD_p})}}{\sqrt{|\Lambda^*/\Lambda|}}\, e^{(b_2(X)-2)\frac{\pi\I}{4}}\,
e^{-2\pi\I \mu\cdot\nu}\,,
\end{split}
\label{Multsys-hp}
\ee
where $\mu\cdot\nu=\kappa^{ab}\mu_a \nu_b$,
$\delta_{\mu\nu}$ is the Kronecker delta on the discriminant group $\Lambda^*/\Lambda$, and
$\chi(\cO_{\cD_p})=\frac12(b_2^+(\cD_p)+1)$ is the arithmetic genus given by
\be
\label{defL0}
\chi(\cO_{\cD_p}) = \frac16\, \kappa_{abc} p^a p^b p^c+ \frac{1}{12}\, c_{2,a} p^a\, .
\ee
While the modular completion \eqref{defwh} is not holomorphic, due to the kernels
$R_n\left(\{ \check{\gamma}_i\}, \tau_2\right)$, its $\bar\tau$-derivative is simply determined by the 
modular completions of the generating series of the constituents, as 
 \be
 \label{dbarwh}
\partial_{\bar\tau} \whh_{p;\mu}(\tau) = 
\sum_{n\geq 2} \sum_{\check\gamma=\sum_{i=1}^n \check\gamma_i}
\cJ_n\left(\{ \check{\gamma}_i\}, \tau_2\right)\, e^{\I \pi\tau\, Q_n(\{ \check{\gamma}_i\})} \prod_{i=1}^n
\whh_{p_i,\mu_i}(\tau)
\ee
where the kernels $\cJ_n\left(\{ \check{\gamma}_i\}, \tau_2\right)$ are again sums of products 
of generalized error functions, this time leading to a modular theta series, consistent with the fact
that $\tau_2^2 \partial_{\bar\tau}$ raises the modular weight by 2 units.

Up to now, we assumed that the quadratic form $\kappa_{ab}=\kappa_{abc}p^c$ was invertible, and that the vector $p^a$ had positive norm, $\kappa_{ab} p^a p^b>0$. Both conditions are automatically satisfied when the divisor class $p^a H_a$ is ample (see footnote \ref{fooHodge}), however they fail for the D4-D2-D0 indices related to D2-D0 indices by a relative conifold monodromy, as we shall see in \S\ref{sec_DTmodEll}. Generalizations of the modularity constraints were conjectured in \cite{Alexandrov:2020qpb} in the case where  $\kappa_{ab}$ is degenerate, and in \cite{Alexandrov:2025sig} in the case where the vector  $p^a$ is isotropic, $\kappa_{ab} p^a p^b=0$.
We now briefly summarize the resulting prescriptions, while pointing out that it is an open problem to characterize the modular properties of generating series for arbitrary, non-ample effective divisor classes.

Starting with the case where $\kappa_{ab}$ is degenerate, let $\{\lambda_s^a\}$ be a set of null vectors, i.e. $\kappa_{ab}\lambda_s^a=0$, and $\Lambda_p\subset\Lambda$ the sublattice orthogonal to these vectors,
\be
\Lambda_p=\{ q_a\in \IZ^{b_2(X)}  + \frac12 \,\kappa_{ab}p^b \ :\  \lambda_s^a q_a=0\}.
\label{Lam-p}
\ee
We now introduce the pseudo-inverse quadratic form\footnote{This is also known as the Moore-Penrose inverse.} $\kappa^{ab}$ defined by the conditions
i) $\rank(\kappa^{ab})=\rank(\kappa_{ab})$, ii) $\kappa^{ac}\kappa_{cb}=\delta^a_b-\sum_{s,t} e^{st} \lambda_s^a\lambda_t^b$
where $e^{st}$ is the inverse of $e_{st}=\sum_a \lambda_s^a\lambda_t^a$.
Then the claim of \cite{Alexandrov:2020qpb} is that, upon replacing $\Lambda$ by $\Lambda_p$ and $b_2(X)$ by $\rank(\Lambda_p)$, 
all results previously stated for ample divisors continue to apply. In particular, the modular weight of the generating function is now $-1-\frac12 \rank(\Lambda_p)$. In the special case where $p^a$ itself is a null vector (i.e. $\kappa_{ab} p^b=0$), as happens in the case of vertical D4-D2-D0 invariants in K3-fibered CY threefolds \cite{Doran:2024kcb}, the modular anomaly turns out to be absent, so that 
$\whh_{p;\mu}$ coincides with $h_{p;\mu}$ for any $p$, reducible or not.

Next we turn to the case where $p^a$ is an isotropic vector, $\kappa_{ab} p^a p^b=0$ with $\kappa_{ab}$ being non-degenerate. In this case, the modular anomaly is still present but it significantly simplifies. For simplicity, we assume that $p^a$ lies along the boundary of the effective cone, such that the only allowed splittings involve collinear charges, $p_i^a= r_i p_0^a$ where 
$p_0^a$ is a primitive vector with $p_0^3:=\kappa_{abc} p_0^a p_0^b p_0^c=0$. 
In this \textit{collinear case}, only binary splittings occur on the r.h.s. of \eqref{dbarwh} \cite[5.2]{Alexandrov:2019rth}.  The isotropic case $p_0^3=0$ can be obtained as a limit $\epsilon\to 0$
of a positive vector $p_\epsilon= p_0 + \epsilon v_1$, where $v_1$ is any lattice vector
with positive inner product $\xi:= p_0 \cdot v_1>0$.  While
the coefficient $\cJ_2$ is in general proportional to $\sqrt{p_\epsilon^3}$, which vanishes in the limit 
$\epsilon\to 0$, it also involves a theta series which diverges in the same limit. In \cite[\S D]{Alexandrov:2025sig}, it is shown that the result is finite in the limit $\epsilon\to 0$, and given by 
\be
\label{Sergeyhae}
\tau_2^2 \partial_{\bar\tau} \whh_{r p_0,\mu} = \frac{1}{16\pi\I} \sum_{r=r_1+r_2} 
\cJ_{r_1,r_2}\, \whh_{r_1 p_0,\mu} \, \whh_{r_2 p_0,\mu}
\ee 
where 
\be
\label{defcJ}
\cJ_{r_1,r_2} = r_0 \sum_{\mu_1\in \Lambda_1^*/\Lambda_1,
\mu_2\in \Lambda_2^*/\Lambda_2}
\delta_{\mu-\mu_1-\mu_2 \in r_0 \Lambda_0}
\sum_{A=0}^{n_g-1}
\delta^{(\xi r_{12})}_{p_0\cdot(\mu_{12}^{\parallel}+r_{12} g_A^{\parallel})}
\vartheta^{\perp}  _{\mu_{12}^{\perp}+r_{12} g_A^{\perp}}\,.
\ee
Here  $r_0=\gcd(r_1,r_2)$, $r_{12}=\frac{r r_1 r_2}{r_0^2}$;  
$\Lambda_0$ is the lattice $\Lambda$ equipped with the quadratic form $\kappa_{0,ab} =
\kappa_{abc} p_0^c$,  $\Lambda_1$ and $\Lambda_2$ are the same lattices with quadratic form 
rescaled by $r_1$ and $r_2$, respectively; the lattice $\Lambda_0^{\parallel}$ is the two-dimensional 
sublattice of  $\Lambda_0$ spanned by $(p_0,v_1)$, and $ \Lambda_0^\perp$ is its orthogonal complement; the glue vectors $(g_A^{\parallel}, g_A^{\perp}), A=1,\dots n_g$ 
generate the quotient $\Lambda_0 / (\Lambda_0^{\parallel} \oplus \Lambda_0^{\perp})$,
and $\delta^{(N)}_m$ is the Kronecker delta on $\IZ_N$, equal to 1 if $m=0 \mod N$ and 0 otherwise.
Finally, the factor $\mu_{12}$ is defined as
\be
r_0\, \mu_{12,a} = r_2 \mu_{1,a} - r_1  \mu_{2,a} + r_1 r_2 \kappa_{0,ab} \left( \rho_1^b -\rho_2^b + 
\frac12 (r_1-r_2) p_0^b   \right)\,,
\ee
where $\rho_1^a,\rho_2^a$ are any integer solutions of 
\be
\mu_a-\mu_{1,a}-\mu_{2,a} + r_1 r_2  \kappa_{0,ab}  p_0^b  = 
\kappa_{0,ab} \left( r_1 \rho_1^b + r_2 \rho_2^b \right)
\ee
(which exist by virtue of the condition $\mu-\mu_1-\mu_2 \in r_0 \Lambda_0$).
While these various definitions depend on the choice of vectors $v_1, \rho_1, \rho_2$, 
the final result does not.  Of course, it is advisable 
to choose $v_1$ in such a way that the number of glue vectors is minimal. In the case of
interest in \S\ref{sec_DTmodEll}, we shall have to deal at the same time with a 
degenerate quadratic form $\kappa_{ab}$ and an isotropic vector $p^a$, but the lattice 
$\Lambda_0^\perp$ will be null such that the theta series $\vartheta^{\perp}  _{\mu_{12}^{\perp}+r_{12} g_A^{\perp}}$ will actually be a constant.

\subsection{Genus one fibrations and Jacobi forms}
\label{sec_Jacobi}
We now assume that $X$ is a CY threefold with a torus fibration $F\xhookrightarrow{} X\stackrel{\pi}{\to}B$ 
over a generalized del Pezzo surface  $B$. By this we mean that 
$B$ is either $\mathbb{P}^1\times \mathbb{P}^1$, the Hirzebruch surface $\mathbb{F}_2$ or an iterative blow-up of $\mathbb{P}^2$ in $k=0,\ldots,8$ points that are not necessarily generic but do not lie on any curve of self-intersection $-2$.
For simplicity, we shall focus on sufficiently generic fibrations such that  $b_2(X)=b_2(B)+1$.~\footnote{This implies in particular that the fibration does not exhibit any fibral divisors, that resolve singularities over one-dimensional components of the discriminant locus, or additional $N$-sections that correspond to a non-zero Mordell-Weil rank of the associated Jacobian fibration.}
We also assume that the complex structure is sufficiently generic, such that the discriminant divisor in $B$ is reduced and irreducible and has at most isolated nodes and cusps as singularities.
For a  more detailed discussion of such geometries, see~\cite[Section 3]{Duque:2025kaa}.

Let us first introduce some terminology.
A divisor $S$ is called an \textit{$N$-section} if $\pi(S)=B$ and there exists a dense open subset $U\subset B$ such that $S\big|_U$ is an $N$-sheeted branched covering of $U$.
This implies that $F\cap S=N$.
We further assume that the covering is irreducible and refer to divisors that are unions of $N_i$-sections with $\sum_i N_i=N$ as \textit{pseudo $N$-sections}.
A $1$-section is just a section.

If the fibration has a section it is called \textit{elliptic}.
On the other hand, if it only exhibits an $N$-section -- and  there are no $N'$-sections with $N'<N$ -- we refer to it as a \textit{genus one fibration with an $N$-section}.
We use the term \textit{torus fibration} if we want to remain agnostic about the existence of a section.
Every torus fibration has an $N$-section for some $N\in \mathbb{N}$.

We choose a basis $\{\check{D}^\alpha,  \alpha=1,\ldots,b_2(B)\}$ of effective curve classes on $B$, with intersection form $C^{\alpha\beta}=\check{D}^\alpha\cap \check{D}^\beta$, and denote by 
$\check{D}_\alpha=C_{\alpha\beta} \check{D}^{\beta}\in \Pic(B)$ the dual basis such that 
$\check{D}^\alpha \cap \check{D}_\beta=\delta^{\alpha}_\beta$ (here $C_{\alpha\beta}$ is the inverse of $C^{\alpha\beta}$; note that $\det C_{\alpha\beta}=1$ as the lattice $H_2(B,\IZ)$ is self-dual).
If the K\"ahler cone of $B$ is simplicial, we assume that the divisor classes $\check{D}_\alpha$ form a basis.
If it is not simplicial, we assume that the $\check{D}_\alpha$
span a simplicial sub-cone of the K\"ahler cone and that  all  integral divisors inside the K\"ahler cone are linear combinations with integer coefficients of the $\check{D}_\alpha$'s.
 We denote by $c_1(B)= a_{\alpha} 
\check{D}^{\alpha} = a^\alpha \check{D}_{\alpha}$ the first Chern class of $B$. 

A basis of effective divisors $\{D_a\} = \{D_e\} \cup \{D_\alpha, \alpha=1,\ldots,b_2(B)\}$
 in $H_4(X,\IZ)$ can be assembled by combining the class of the $N$-section $D_e$ with pullbacks $D_\alpha=\pi^*(\check{D}_\alpha)$ of effective divisors on the base. By construction, the divisors $D_\alpha$ have vanishing intersections,
$D_\alpha \cap D_\beta \cap D_\gamma=0$, 
while their intersection with $D_e$ is proportional
to their intersection on the base, $D_e \cap D_\alpha \cap D_\beta = N C_{\alpha\beta}$.
We denote the intersection numbers by $\kappa_{abc}:=D_a\cap D_b\cap D_c$ and introduce $\kappa,\ell_\alpha$ such that
\be
\label{kappagen}
\kappa:=\kappa_{eee}\,, \quad \ell_\alpha:=\kappa_{ee\alpha}\,, \quad 
\kappa_{e\alpha\beta}=N C_{\alpha\beta}, \quad \kappa_{\alpha\beta\gamma}=0
\ee
and denote the integrals of the second Chern class $c_2(TX)$ on  $(D_e,D_\alpha)$ by
\be
c_2:=c_{2,e}, \quad c_\alpha:=\frac{1}{12}c_{2,\alpha}\,.
\ee
In fact,  the coefficients $c_\alpha$ appear to be always equal to the coefficients $a_\alpha$ of the anti-canonical class on the base, but since we are not aware of a proof in general, we use a different notation.~\footnote{For elliptic fibrations with a section this is proven for example in~\cite[Appendix D]{Grimm:2013oga} and for certain genus one fibrations with $N$-sections a proof can be found in~\cite[Appendix B]{Cota:2019cjx}. A sketch of a more general proof is discussed in~\cite[Section 3.3]{Duque:2025kaa} and a physical derivation can be found in~\cite[Section 4.5]{Duque:2025kaa} by comparing the Chern-Simons terms of five-dimensional F- and M-theory compactifications.}
Importantly, while the divisors $D_\alpha$ are always nef, and $D_e$ is always effective, $D_e$ is not necessarily nef.
One can however always shift it by a linear combination $D_e\mapsto D_e+ \eta^\alpha D_\alpha$, such that the resulting divisor is nef and such that $\{D_e, D_\alpha\}$ generates the nef cone (defined as the closure of the K\"ahler/ample cone). Under this shift,  the intersection coefficients \eqref{kappagen} transform as 
\be
\label{shiftDe}
\ell_\alpha\mapsto \ell_\alpha + 2 N C_{\alpha\beta} \eta^\beta,\quad
\kappa\mapsto \kappa + 3 \ell_\alpha \eta^\alpha + 3N  C_{\alpha\beta} \eta^\alpha \eta^\beta\,,
\ee
leaving the following combinations invariant, 
\be
\label{defhkappa}
\hatkappa:=\kappa-\frac{3}{4N} \ell_\alpha C^{\alpha\beta} \ell_\beta, \quad
\hatc_2:= c_2 - \frac{6}{N} c^\alpha \ell_\alpha\,.
\ee

We will refer to curves that are supported on a single fiber of the fibration as \textit{fibral curves}.
The fibers of the fibration are irreducible except for $I_2$-fibers over isolated points of the base.
The $I_2$-fibers consist of two rational curves that intersect transversely in two points.
We denote by $N_k$, $k=1,\ldots,N-1$ the number of fibral curves that intersect the $N$-section $k$ times (note the symmetry property $N_{N-k}=N_k$).
By comparing the Chern-Simons terms of the corresponding F- and M-theory vacua, it was conjectured in~\cite[Conj. 3]{Duque:2025kaa} that the invariant combinations \eqref{defhkappa}
are determined in terms of these numbers by 
\bea
\label{FtheoryRel}
\begin{split}
\hatkappa &=& \frac{12-\chi_B}{4} N^3 - \frac{1}{8N} \sum_{k=1}^{N-1} k^2(N-k)^2 N_k\,,\\
\hatc_2&=& 4(15-\chi_B)N- \frac{1}{2N} \sum_{k=1}^{N-1} k(N-k) N_k\,.
\end{split}
\eea
Moreover, the $N_k$'s along with the Euler numbers $\chi_X$ and $\chi_B$ determine all non-trivial Gopakumar-Vafa invariants at base degree 
zero~\cite{Oehlmann:2019ohh}\footnote{This statement generalizes~\cite[Thm 6.9]{toda2012stability} to the case of torus fibrations. We do not know of a rigorous mathematical proof in general, but we shall give a heuristic argument of the mod $N$ periodicity of $\GV^{(0)}_{n,0}$ using
invariance under the relative conifold monodromy in~\S\ref{sec_modgw0}.}, 
\be
\label{GV01eq}
%\GV^{(0)}_{mN+k,0_\alpha}=N_k, \quad \GV^{(0)}_{(m+1)N,0_\alpha}=-\chi_X, \quad
%\GV^{(1)}_{(m+1)N,0_\alpha}=\chi_B
    \text{GV}^{(0)}_{mN,0_\alpha}=-\chi_X\,,\quad \text{GV}^{(0)}_{(m-1)N+a,0_\alpha}=\text{GV}^{(0)}_{mN-a,0_\alpha}=N_a\,,\quad \text{GV}^{(1)}_{mN,0_\alpha}=\chi_B\,,
\ee
for all $m\geq 1$ and $a\in\{ 1,\dots,\lfloor\frac{N}{2}\rfloor\}$. 
Moreover, the following sum rule
\be
\label{GV0sumrule}
\sum_{k=1}^N \GV^{(0)}_{k,0_\alpha} = 60\, (12-\chi_B)\,,
\ee
 is required by the cancellation of gravitational anomalies in F-theory~\cite{Grassi:2011hq}. For $N=1$, this sum rule reduces to $\chi_X=-60(12-\chi_B)$.
 Since the base is a rational surface, we also have $\chi_B=12-a_\alpha a^\alpha$.

Now, we denote by $\{ C^a\} = \{\cE\} \cup \{C^\alpha, \alpha=1,\dots b_2(B)\}$ the basis of curve classes in $H_2(X,\IZ)$ dual to $\{D_a\}$, such that
\begin{align}
    \cE\cap D_e=1\,,\quad \cE\cap D_\alpha=C^\alpha \cap D_e=0\,,\quad  C^\alpha \cap D_{\beta}=\delta^{\alpha}_{\gamma}\,.
\end{align}
The curve class $C^e:=\cE$ is related to the class of the generic fiber $F$ via $F=N \cE$ such that $F\cap D_e=N$, consistent with the fact that  $D_e$ is the class of the $N$-section. The curve classes $C^\alpha$ are linear combinations of $F$ and of the intersections $D_e \cap D_\alpha$,
\be
    C^\alpha=\frac{1}{N}C^{\alpha\beta}\left(D_e \cap D_\beta-\frac{\ell_\beta}{N} F\right)\,.
 \ee
 Note that $C^\alpha$ is an integral class, despite the factors of $1/N$ appearing in this expression.
 For $N=1$, i.e. in the presence of a section $\sigma$, the intersection $D_e D_\alpha$ is identical with the embedding $\sigma(\check{D}_\alpha)$ of the curve on $B$, but differs from $C_{\alpha\beta} C^\beta$ by a multiple of the fiber class. 
 
As in \S\ref{sec:ZtopGW}, the K\"ahler form $\omega$ can be expanded as a linear combination of the Poincar\'e dual of the basis of divisors $\{D_e,D_\alpha\}$,
\be
\omega = T D_e + S^\alpha D_\alpha \,.
\ee
In order to state the modular properties of the generating series of Gromov-Witten invariants, it is useful to expand $\omega$ in a different basis $\{\tilde D_e:= D_e+ \frac{1}{2N} D,D_\alpha\}$ 
where $D$ is the pull-back of the so called height pairing~\cite{Huang:2015sta,Cota:2019cjx},
\begin{align}
\label{heightD}
	D:=-\pi^*\pi_*(D_e D_e)=-\ell_\alpha C^{\alpha\beta}D_\beta\,,
\end{align}
and parametrize instead the K\"ahler form as
\be
\label{defShat}
    \omega=T \tilde D_e +\hatS^\alpha D_\alpha\ ,\quad 
    \hatS^\alpha:=S^\alpha+\frac{T}{2N} C^{\alpha\beta} \ell_\beta\,.
\ee
Although $\tilde D_e$ is in general not an integral class, it is worth noting that it is invariant 
under shifts $D_e\mapsto D_e- \eta^\alpha D_\alpha$ accompanied by the corresponding change 
\eqref{shiftDe} in the intersection numbers. In fact, the invariants $\hatkappa$ and $\hatc_2$
respectively are the self-intersection  $\tilde D_e^3$ and  second Chern class $\tilde D_e \cap c_2(TX)$ for the non-integral class $\tilde D_e$, which satisfies $\tilde D_e \cap \tilde D_e  \cap D_\alpha=0$.
\medskip

In terms of the parametrization \eqref{defShat}, one may define generating series of Gromov-Witten invariants at fixed genus $g$ and base degree $k_\alpha$
\be
\label{deffg}
f^{(g)}_{k_\alpha}(T) = \sum_{n\ge 0} \GW_{k_\alpha C^\alpha+ n F}^{(g)} \, 
e^{2\pi\I T\left(n-\frac{\ell_\alpha C^{\alpha\beta} k_\beta}{2N} \right)}  
\ee
(where the sum over $n$ is bounded from below)
such that 
\be
\mathcal{F}^{(g)}(T,S^\alpha) = \mathcal{F}^{(g)}_{\rm pol}(T,S^\alpha) + \sum_{k_\alpha\geq 0} f^{(g)}_{k_\alpha}(T)  \, e^{2\pi\I k_\alpha \hatS^\alpha}
\ee
where the first term is the contribution of the non-constant polynomial terms that arise at genus zero and one.
In~\cite{Candelas:1994hw,Klemm:2012sx,Alim:2012ss}, it was observed in a few examples that the generating series are quasimodular forms of weight $2g-2$. A more precise statement was  conjectured  for elliptic fibrations in~\cite{Huang:2015sta}, and generalized to torus fibrations
in~\cite{Cota:2019cjx},\footnote{Generalizations for elliptic fibrations with fibral divisors and/or non-zero Mordell-Weil rank have been proposed in~\cite{DelZotto:2017mee,Lee:2018spm,Lee:2018urn} and were also extended to genus one fibrations in~\cite{Cota:2019cjx}.} in terms of the coefficients 
$Z_{k_\alpha}(T,\lambda)$ appearing in the Fourier expansion
of the topological string partition function $Z_{\rm top}$ with respect to $\hatS^\alpha$, normalized by the base degree 0 contribution,
\be
\label{ZtopZ0}
\frac{Z_{\rm top}(T,\hatS^\alpha,\lambda) }{ 
Z_0(T,\lambda)} = 1 + \sum_{k_\alpha>0} Z_{k_\alpha}(T,\lambda)\, e^{2\pi \I \hatS^\alpha k_\alpha} \,.
\ee
It will be convenient to use the notations $Z_{k_\alpha}$ and $Z_{H}$ interchangeably, where
$H=k_\alpha \check{D}^\alpha$ is an effective divisor class on the base $B$. 
Using \eqref{ZtopPT}, $Z_{H}(T,\lambda)$ can equivalently be defined as the generating function of PT invariants with fixed base degree $H$, 
normalized by the generating function of PT invariants with zero base degree,
\begin{equation}
\label{defZbeta}
\begin{gathered}
Z_{k_\alpha}(T,\lambda) = \frac{\PT_{k_\alpha}(T,\lambda)}{\PT_0(T,\lambda)}\,, \\
\PT_{k_\alpha}(T,\lambda) = \sum_{n,m} \PT(k_\alpha C^\alpha + n F,m) \, 
e^{2\pi\I (n -\frac{\ell_\alpha C^{\alpha\beta} k_\beta}{2N} )T+\I m \lambda} \,.
\end{gathered}
\end{equation}
For  base degree zero, the only non-vanishing GV invariants are those listed in \eqref{GV01eq}, so  the GV/PT relation \eqref{eqGVPT} gives
\be
\label{PT0}
\begin{split}
\PT_0(T,\lambda) =& \prod_{n>0} (1- e^{2\pi\I n T})^{-\chi_B} \\
\times & \prod_{\mu=0}^{N-1} \left(  \prod_{k>0} \prod_{n>0} 
\left(1- (-1)^k e^{2\pi\I (\mu + n N)T+\I k\lambda} \right)^{ (\mu+n N) \GV^{(0)}_{\mu,0}} \right) 
\,.
\end{split}
\ee

In terms of the normalized generating functions \eqref{defZbeta}, the statement of \cite{Huang:2015sta,Cota:2019cjx} is that $Z_{H}(T,\lambda)$ should be a meromorphic Jacobi form of  weight 0 and index $h_H-1$ under 
the congruence subgroup $\Gamma_1(N)\subset SL(2,\IZ)$ defined in \eqref{eq:defg1N}, where 
\be
\label{defhH}
h_H = 1+ \frac12 H \cap (H-c_1(B)) = 1+ \frac12 k_\alpha C^{\alpha\beta} (k_\beta-c_\beta)
\ee
is the arithmetic genus of the curve class $H$.
This entails several properties of $Z_{H}(T,\check{\lambda})$ that are most conveniently formulated using $\check{\lambda}:=\lambda/(2\pi)$ and slightly abusing the notation:
\begin{itemize}
\item[i)] Periodicity under $\check{\lambda}\mapsto \check{\lambda}+1$ and $T\mapsto T+1$ (up to a phase), which is manifest from  the definition \eqref{defZbeta}.

\item[ii)] Quasi-periodicity under $\check{\lambda}\mapsto \check{\lambda}+m NT$ for any integer $m\in \IZ$,   
\be
\label{Jacqperiod}
Z_H(T,\check{\lambda}+m NT) = e^{-2\pi\I (h_H-1) (m^2NT +2 m \check{\lambda})} Z_H(T,\check{\lambda})\,.
\ee
This was proven in \cite{Oberdieck:2016nvt} for elliptic fibrations ($N=1$) and reduced\footnote{I.e. for any decomposition $H=\sum H_i$ into effective classes, all
$H_i$'s are primitive.} $H$, but remains conjectural in general. Note that one could rescale $\check{\lambda}\to\check{\lambda}/N$ so as to enforce quasi-periodicity under $\check{\lambda}\to \check{\lambda}+mT$, at the expense of introducing fractional powers of $e^{2\pi\I\check{\lambda}}$.

\item[iii)] Modular invariance under the generator $\scriptsize\begin{pmatrix}1&0 \\ N & 1 \end{pmatrix}$ of $\Gamma_1(N)$,
\be
\label{Jacmod}
Z_H\left(\frac{T}{1+NT},\frac{\check{\lambda}}{1+N T}\right) = e^{-2\pi\I (h_H-1) \frac{\check{\lambda}^2}{1+NT}} Z_H(T,\check{\lambda})\,.
\ee

\item[iv)] For $N\geq 5$, similar relations as \eqref{Jacmod} under the remaining generators of $\Gamma_1(N)$.
Sets of generators for $N=5,6,7$ can be found in Table~\ref{tab:gensG1N} on page \pageref{tab:gensG1N}.

\end{itemize}
We note that ii) follows from i) and iii), for which there is considerable evidence in explicit models but no mathematical proof yet. In fact, the conjectures in \cite{Huang:2015sta,Cota:2019cjx} go beyond these Jacobi modular properties, and provide an Ansatz specifying the analytic structure of 
$Z_H(T,\check{\lambda})$, namely
\be
Z_H(T,\check{\lambda}) = \frac{\Delta_{2N}(T)^{\frac{r_H}{N}}}{\eta(NT)^{12 c_1(B)\cdot H} \prod_{\alpha=1}^{b_2(B)}\prod_{s=1}^{k_\alpha} \varphi_{-2,1}(N T,s\check{\lambda})}\, \varphi_H(T,\check{\lambda})
\label{res-fnz}
\ee
where $\eta(T)=q^{1/24} \prod(1-q^n)$ with $q:=e^{2\pi\I T}$ is the Dedekind eta function, $\Delta_{2N}(T)$ 
is a specific modular form  of weight $2N$ under $\Gamma_1(N)$,  $\varphi_{-2,1}(T,z)=-\vartheta_1(T,z)^2/\eta^6(T)$ is a weak Jacobi form of weight $-2$ and index 1,
and $\varphi_H(T,\check{\lambda})$ is a weak Jacobi form of suitable weight and index,\footnote{For example, for an elliptic fibration over $B=\IP^2$, with  hyperplane class $H$, $r_{nH}=0$ and $\varphi_{nH}(T,\check{\lambda})$ is a weak Jacobi form of weight $16n$ and index $\frac13n(n-1)(n+4)$.} such that 
\eqref{res-fnz} has weight 0 and index $h_H-1$,  as specified in \eqref{defhH}. 
The modular form $\Delta_{2N}(T)$ is given explicitly by\footnote{See \cite[\S D.4]{Duque:2025kaa} for a proof of the modularity properties of $\Delta_{2N}$.}
\be
\label{defDelta2N}
\Delta_{2N}(T) = e^{-2\pi\I T} \varphi_{-2,1}(NT, T)^{-N} = 
\left\{ \begin{array}{cl}
q(1+8q+\ldots) & N=2\\
q^{N-1} \left( 1+ 2N q + \dots \right) & N\ge 3\\
\end{array}
\right.
\ee
while the exponent $r_H$ is determined modulo $N$ by 
\be
\label{defrH}
        r_H=\frac12[N^2c_1(B)-\pi(D)]\cap H\text{ mod }N
        =\frac12(N^2c _\alpha+\ell_\alpha) C^{\alpha\beta} k_\beta \text{ mod }N\,,
\ee
Upon Laurent expanding around $\check{\lambda}=0$, these properties automatically imply that the generating series \eqref{deffg} are quasimodular forms of weight $2g-2$ under $\Gamma_1(N)$, which furthermore satisfy 
a holomorphic anomaly equation of the form
\be  
\begin{split}
\partial_{E_2(NT)} f_H^{(g)}(T)=& 
   -\frac{1}{24} \sum_{g=g_1+g_2 \atop H=H_1+H_2} (H_1\cap H_2)
    f_{H_1}^{(g_1)} f_{H_2}^{(g_2)}
   -\frac{1}{24} H\cap (H-c_1(B)) \, f_{H}^{(g-1)} \\
\end{split}
\label{anom-fnN}
\ee
where the second term is absent for $g=0$. For $N=1$, this recovers~\cite{Alim:2012ss} \cite[\S 3.6]{Oberdieck:2017pqm}.

\section{Base degree zero modularity and Eichler integrals}
\label{sec:basedegreezero}
In this section, we derive general modular expressions for the base degree zero contributions $f^{(g)}_0(T)$ to the topological string free energies, based on the simple structure of the vertical GV invariants summarized in~\eqref{GV01eq}.
This completes and generalizes the results from~\cite[Section 4.5]{Huang:2015sta} for $N=1$ and~\cite[Section 4.3]{Cota:2019cjx} for $N=2,3,4,6$ to the case of genus one fibered Calabi-Yau threefolds with $N$-sections for arbitrary $N$. These results will also serve as consistency checks for the general transformation properties of the topological string partition function that will be the subject of Section~\ref{sec:waveFunctionModularity}.

As reviewed in Section~\ref{sec_Jacobi}, the only non-vanishing base degree zero Gopakumar-Vafa invariants are
\begin{align}
    \text{GV}^{(0)}_{mN,0_\alpha}=-\chi_X\,,\quad \text{GV}^{(0)}_{(m-1)N+a,0_\alpha}=\text{GV}^{(0)}_{mN-a,0_\alpha}=N_a\,,\quad \text{GV}^{(1)}_{mN,0_\alpha}=\chi_B\,,
    \label{eqn:bd0GV}
\end{align}
for $m\ge 1$ and $a\in\{1,\ldots,\lfloor{\tiny\frac{N}{2}}\rfloor\}$.
The base degree zero contributions to the topological string free energies at arbitrary genus are then determined by the GV-formula~\eqref{GWtoGV}, together with the expansion
\begin{align}
    \begin{split}
    \sum\limits_{m\ge 1}\frac{q^m}{m}\left[2\sin\left(\frac{m\lambda}{2}\right)\right]^{-2}
    =\lambda^{-2}\text{Li}_3(q)+\sum\limits_{g\ge 1}\frac{(-1)^{g+1}\lambda^{2g-2}B_{2g}}{2g[(2g-2)!]}\text{Li}_{3-2g}(q)\,,
    \end{split}
\end{align}
and the constant map contributions
\begin{align}
    F_0\vert_{\text{const.}}=-\frac12\zeta(3)\chi_X\,,\quad F_{g\ge 2}\vert_{\text{const.}}=\frac{(-1)^{g-1}}{2}\frac{B_{2g}B_{2g-2}}{2g(2g-2)[(2g-2)!]}\chi_X\,.
\end{align}

At genus zero, this leads to the result
\begin{align}
    f^{(0)}_0(T)=-\chi_X\left(\frac{\zeta(3)}{2}+\sum\limits_{m\ge1}\text{Li}_{3}(q^{mN})\right)+\sum\limits_{m\ge 0}\sum\limits_{k=1}^{N-1}N_k\text{Li}_3\left(q^{mN+k}\right)\,,
    \label{eqn:bd0f0}
\end{align}
while the higher genus free energies  take the form
\begin{align}
    f_0^{(1)}(T)=\frac{1}{12}\left[(12\chi_B-\chi_X)\sum\limits_{m\ge1}\text{Li}_1(q^{mN})+\sum\limits_{m\ge0}\sum\limits_{k=1}^{N-1}N_k\,\text{Li}_1\left(q^{mN+k}\right)\right]\,,
    \label{eqn:bd0f1}
\end{align}
\begin{align}
    \begin{split}
    f_0^{(g\ge 2)}(T)=&(-1)^{g+1}\frac{B_{2g}}{2g[(2g-2)!]}\left(\chi_X\left[\frac{B_{2g-2}}{2(2g-2)}-\sum\limits_{m\ge1}\text{Li}_{3-2g}\left(q^{mN}\right)\right]\right.\\
    &\left.+\sum\limits_{m\ge0}\sum\limits_{k=1}^{N-1}N_k\,\text{Li}_{3-2g}\left(q^{mN+k}\right)\right)\,.
    \label{eqn:bd0f2}
    \end{split}
\end{align}

\subsection{Elliptic fibrations}
\label{sec_elliptic}
Let us first understand the modular properties of these expressions~\eqref{eqn:bd0f0},~\eqref{eqn:bd0f1} and ~\eqref{eqn:bd0f2} in the easiest case $N=1$.
We start by recalling some well-known properties of Eisenstein series and their less familiar Eichler integrals.

For even $k\in \mathbb{N}$, with $k\ge 4$, the Eisenstein series 
\begin{align}
    E_k(\tau):=1-\frac{2k}{B_k}\sum\limits_{n=1}^\infty\sigma_{k-1}(n)q^n=1-\frac{2k}{B_{k}}\sum\limits_{m\ge 1}\text{Li}_{1-k}(q^m)\,,
\end{align}
is a modular form of weight $k$ for $SL(2,\IZ)$.
The weight two Eisenstein series $E_2(\tau)$ is a quasimodular form for $SL(2,\IZ)$ and transforms as~\cite{MR1363056}
\begin{align}
    E_2\left(\frac{a\tau+b}{c\tau+d}\right)=(c\tau+d)^2E_2(\tau)+\frac{12}{2\pi\I}c(c\tau+d)\,,\quad \left(\begin{array}{cc}a&b\\c&d\end{array}\right)\in SL(2,\IZ)\,.
    \label{eqn:e2trafo}
\end{align}
Using the relation
\begin{align}
    \sum\limits_{m\ge 1}\sigma_{k-1}(m)q^m=\sum\limits_{m\ge 1}\sum\limits_{n\vert m}n^{k-1}q^m=\sum\limits_{m,n\ge 1}n^{k-1}q^{mn}=\sum\limits_{m\ge 1}\text{Li}_{1-k}(q^m)\,,
\end{align}
one can rewrite
\begin{align}
    \begin{split}
    E_{k}(\tau)=&1-\frac{2k}{B_{k}}\sum\limits_{m\ge 1}\text{Li}_{1-k}(q^m)\,.
    \end{split}
    \label{eqn:eisensteinPolyLogS}
\end{align}

In general, given a holomorphic modular form $f(\tau)=\sum_n a_n q^n$ of weight $w$, 
its \textit{holomorphic Eichler integral} is defined as $\widetilde{f}(\tau)=P_{w-1}(\tau)+\sum_{n\neq 0} \frac{a_n}{n^{w-1}} q^n$, such that the iterated derivative $(2\pi\I)^{1-w}\partial_\tau^{w-1}\tilde f(\tau)=f(\tau)$. The \textit{constant term} $P_{w-1}(\tau)$ is polynomial in $\tau$ of degree $w-1$, unspecified except for its top degree term, determined by $a_0$. Under $\gamma\in SL(2,\IZ)$, $\widetilde{f}(\tau)$ transforms like a modular form of (negative) weight $2-w$, up to a $\gamma$-dependent polynomial of degree $w-2$ (and a logarithmic term when $w=2$). We are particularly interested in the holomorphic Eichler integral of the weight four Eisenstein series $E_4(\tau)$, which we define following~\cite[(27)]{0990.11041} (see also~\cite[App. B]{Angelantonj:2015rxa} and~\cite{zbMATH06149482,bringmann2021})
\be
\widetilde E_{-2}(\tau) := \frac{(2\pi\I)^3}{1440}  \tau^3 + \frac12 \zeta(3) + \sum_{m\geq 1} \text{Li}_3(q^m)\,,
\label{eqn:defEm2}
\ee
such that
\be
(2\pi\I)^{-3} \partial_{\tau}^3\widetilde E_{-2}(\tau)=\frac{1}{240}\, E_4(\tau)\,.
\ee
Under $SL(2,\IZ)$, $E_{-2}(\tau)$ transforms as a weight $-2$ modular form, up to a quadratic polynomial in $\tau$,
\be
\begin{split}
\tau^2 \widetilde E_{-2}(-1/\tau)  =&\widetilde E_{-2}(\tau)  + \frac{\I\pi^3}{36} \tau\,,\\
\widetilde E_{-2}(\tau+1)=&\widetilde E_{-2}(\tau) -\frac{\I\pi^3}{180} (1+3\tau+3\tau^2)\,.
\end{split}
\label{eqn:Em2trafo}
\ee
Similarly, the holomorphic Eichler integral of the weight two Eisenstein series $E_2(\tau)$ takes the form
\begin{align}
    \widetilde{E}_0(\tau):=\log\eta(\tau)=\frac{2\pi{\rm i}}{24}\tau-\sum\limits_{m\ge 1}\text{Li}_1(q^m)\,,
    \label{eqn:defE0}
\end{align}
and satisfies $(2\pi\I)^{-1}\partial_\tau \widetilde{E}_{0}(\tau)=E_2(\tau)$.
The transformations of the Dedekind eta function
\begin{align}
    \eta(-1/\tau)=e^{-\frac{\pi\I}{4}}\sqrt{\tau}\eta(\tau)\,,\quad \eta(\tau+1)=e^{\frac{\pi\I}{12}}\eta(\tau)\,,
\end{align}
imply that the Eichler integral transforms as
\begin{align}
    \widetilde{E}_0(-1/\tau)=\widetilde{E}_0(\tau)+\frac12\log(\tau)-\frac{\pi\I}{4}\,,\quad \widetilde{E}_0(\tau+1)=\widetilde{E}_0(\tau)+\frac{\pi\I}{12}\,.
    \label{eqn:E0trafo}
\end{align}

Combining the expressions for the base degree zero free energies~\eqref{eqn:bd0f0},~\eqref{eqn:bd0f1} and~\eqref{eqn:bd0f2}, with~\eqref{eqn:defEm2},~\eqref{eqn:defE0} and~\eqref{eqn:eisensteinPolyLogS}, and using $\chi_X=-60(12-\chi_B)$, one obtains the identities
\begin{align}
    \begin{split}
    f^{(0)}_0(T)=&(12-\chi_B)\left[-\frac{(2\pi\I)^3}{24}T^3+60\widetilde{E}_{-2}(T)\right]\,,\\
    f^{(1)}_0(T)=&4(15-\chi_B)\left[\frac{2\pi\I}{24}T-\widetilde{E}_0(T)\right]\,,\\
    f^{(g\ge 2)}_0(T)=&(-1)^{g}\frac{15B_{2g}B_{2g-2}}{g(2g-2)[(2g-2)!]}(12-\chi_B)E_{2g-2}(T)\,.
    \end{split}
    \label{eqn:f0modularN1}
\end{align}
Further defining 
\be
\label{defhatf0}
 \hatf^{(0)}_0(T):=\frac{\widehat{\kappa}}{6}(2\pi {\rm i}T)^3+f^{(0)}_0(T), \quad
 \hatf_0^{(1)}(T):=-\frac{\widehat{c}_2}{12}\pi{\rm i}T+f^{(1)}_0(T)
\ee
and using \eqref{FtheoryRel} with $N=1$, we see that the cubic and linear terms cancel, leading to
\be
\hatf^{(0)}_0(T)=60 (12-\chi_B) \widetilde{E}_{-2}(T)\ ,\quad 
\hatf^{(1)}_0(T)=-4(15-\chi_B) \widetilde{E}_0(T)
\ee
The transformation properties of these expressions under $SL(2,\IZ)$ can be easily deduced from~\eqref{eqn:e2trafo},~\eqref{eqn:Em2trafo} and~\eqref{eqn:E0trafo}. In particular, the third derivative $Y_0(T):= (2\pi\I)^{-3}\partial_T^3\hatf_0^{(0)}$ (also known as Yukawa coupling) is a modular form of weight 4, $e^{\hatf^{(1)}_0}$ is a modular form of weight $2\chi_B-30$, $f^{(2)}$ is a quasimodular form of weight 2
while $f^{(g\ge 3)}_0$ is a modular form of weight $2g-2$.

\subsection{Genus one fibrations}
We shall now generalize the expressions~\eqref{eqn:f0modularN1} to genus one fibrations with $N$-sections.
To this end, it will be necessary to understand the modular properties of expressions
\begin{align}
    \phi^{(g)}_{N,a}(\tau):=\sum\limits_{m\ge 0}\sum\limits_{k=1}^Ng_{N,a}(k)\text{Li}_{3-2g}\left(q^{mN+k}\right)\,,
    \label{eqn:phigdef}
\end{align}
where we use
\begin{align}
     g_{N,a}(k)=\left\{\begin{array}{cl}
        1&\text{ if }\,k\equiv \pm a\text{ mod }N\\
        0&\text{ else}
    \end{array}\right.\,.
\end{align}
We relegate most of the technical work to  
Appendices~\ref{sec:dirichletEisenstein}--\ref{sec:eichlerIntegrals}.
In \S\ref{sec:polyLogAndDEisenstein}, we prove Lemma~\ref{lem:g2g}, which together with Proposition~\ref{prop:eisenstein} from \S\ref{sec:dirichletEisenstein} implies that for $g\ge 2$,
\begin{align}
    \varphi_{N,a}^{(g\ge 2)}(\tau):=-\frac{B_{2g-2}}{2g-2}\delta_{1,N/\gcd(N,a)}+\phi_{N,a}^{(g)}(\tau)
    \label{eqn:varphigdef}
\end{align}
is a modular form of weight $2g-2$ for $\Gamma_1(N)$.
In Appendix~\ref{sec:polyLogAndDEisenstein} we also prove Lemma~\ref{lem:g01}, which implies that for $g\in\{0,1\}$ the expression
\begin{align}
    \Phi_{N,a}^{(g)}(\tau):=&-\frac{(2\pi{\rm i})^{3-2g}\beta_{4-2g,N,a}}{(4-2g)[(3-2g)!]}\tau^{3-2g}+\phi_{N,a}^{(g)}\left(\tau\right)\,,
    \label{eqn:Phigdef}
\end{align}
is a holomorphic Eichler integral of a $\Gamma_1(N)$ (quasi) modular form
\begin{align}
    \varphi^{(g)}_{N,a}(\tau):=(2\pi\I)^{2g-3}\partial_\tau^{3-2g}\Phi^{(g)}_{N,a}(\tau)\,,
\end{align}
of weight $4-2g$. For $g=\{0,1\}$, the coefficient $\beta_{4-2g,N,a}$ defined 
in \eqref{eqn:generalizedBernoulli1} is given explicitly by 
\begin{align}
    \begin{split}
    \beta_{2,N,a}=&\frac{1}{N}\left(\frac{N^2}{6}-a(N-a)\right)\,,\quad \beta_{4,N,a}=\frac{1}{N}\left(-\frac{N^4}{30}+a^2(N-a)^2\right)\,.
    \end{split}
     \label{eqn:BernoulliBeta}
\end{align}
The transformation properties of $\Phi^{(g)}_{N,a}(\tau)$ under $\tau\rightarrow \tau+1$ and $\tau\rightarrow \tau/(N\tau+1)$ are determined in Appendix~\ref{sec:eichlerIntegrals} and summarized in Theorem~\ref{thm:eichler} on page \pageref{thm:eichler}.

\paragraph{Genus 0.}
Using~\eqref{eqn:defEm2} and~\eqref{eqn:Phigdef}, we can express~\eqref{eqn:bd0f0} as
\begin{align}
    \begin{split}
        f^{(0)}_0(T)=&(2\pi\I)^3CT^3-\chi_X\widetilde{E}_{-2}(NT)+\sum\limits_{a=1}^{\lfloor N/2\rfloor}N_a\Phi^{(0)}_{N,a}(T)\,,
    \end{split}
\end{align}
where the constant $C$ is defined as
\begin{align}
    C:=&\frac{\chi_X}{1440}N^3+\sum\limits_{a=1}^{\lfloor N/2\rfloor}\frac{N_a}{24N}\left(-\frac{N^4}{30}+a^2(N-a)^2\right)\,.
\end{align}
Using~\eqref{GV0sumrule} and \eqref{FtheoryRel}, we can rewrite this as
\begin{align}
    \begin{split}
    C=&\frac{N^3}{1440}\left(\chi_X-\sum\limits_{a=1}^{\lfloor N/2\rfloor}N_a\right)+\sum\limits_{a=1}^{N-1}\frac{N_a}{24N}a^2(N-a)^2
    =-\frac16\widehat{\kappa}\,.
    \end{split}
\end{align}
In terms of the function $\hatf^{(0)}_0$ introduced in \eqref{defhatf0}, 
we therefore have
\begin{align}
    \hatf^{(0)}_0(T)
    =-\chi_X \widetilde{E}_{-2}(NT)+\sum\limits_{a=1}^{\lfloor N/2\rfloor}N_a\Phi^{(0)}_{N,a}(T)\,.
\end{align}
Using Theorem~\ref{thm:eichler} and~\eqref{eqn:Em2trafo}, as well as~\eqref{GV0sumrule}, we see that this transforms as
\begin{align}
\label{transf000}
\begin{split}
    (N T +1)^2&\hatf_0^{(0)}\left(\frac{ T }{N T +1}\right)-\hatf_0^{(0)}( T )\\
    =&\frac{{\rm i}\pi^3}{180}\chi_X(4N^2 T ^2-3N T -3)\\
    &-\frac{{\rm i}\pi^3}{90}\sum\limits_{a=1}^{\lfloor N/2\rfloor}N_a\left[2(15a^2-15aN+2N^2) T ^2-3N T -3\right]\\
        =&(2\pi{\rm i})^3\left[\frac{1}{8}(\chi_B-12)(N T +1)+\frac{1}{24}N(\widehat{c}_2-12N) T ^2\right]\,,
    \end{split}
\end{align}
as well as $\hatf_0^{(0)}(T+1)-\hatf_0^{(0)}( T )=(2\pi\I)^3\,\widehat{\kappa}\,(3T^2+3T+1)/6$. 
The third derivative of $\hat{f}^{(0)}(T)$ is a linear combination of  $\Gamma_1(N)$ Eisenstein series of weight four,
\begin{align}
    Y_0(T):= (2\pi\I)^{-3}\partial_T^3\hatf_0^{(0)}(T)=-\chi_XN^3E_4(N\tau)+\sum\limits_{a=1}^{\lfloor N/2\rfloor}N_a\varphi^{(0)}_{N,a}(T)\,.
\end{align}
More generally, under $\gamma=\scriptsize\begin{pmatrix} a & b \\ c & d\end{pmatrix}\in \Gamma_1(N)$, $\hatf_0^{(0)}$ transforms as 
\be
\hatf_0^{(0)}\left(\frac{aT+b}{cT+d}\right) = 
\frac{\hatf_0^{(0)}(T) + x T^2 + y T+z}{(cT+d)^2}
\ee
where $x,y,z$ are $\gamma$-dependent constants, subject to obvious cocycles relations, such that the third derivative transforms as $Y_0\left(\frac{aT+b}{cT+d}\right)=(cT+d)^4 Y_0(T)$. 

\paragraph{Genus 1.}
Using~\eqref{eqn:defEm2} and~\eqref{eqn:Phigdef}, we can express~\eqref{eqn:bd0f1} as
\begin{align}
    \begin{split}
       f^{(1)}_0(T)=&2\pi\I DT-\left(\chi_B-\frac{\chi_X}{12}\right)\widetilde{E}_0(NT)+\frac{1}{12}\sum\limits_{a=1}^{\lfloor N/2\rfloor}N_a\Phi^{(1)}_{N,a}(T)\,,
    \end{split}
\end{align}
where the constant $D$ is defined as
\begin{align}
    \begin{split}
    D:=&\frac{N}{24}\left(\chi_B-\frac{\chi_X}{12}\right)+\frac{1}{24}\sum\limits_{a=1}^{\lfloor N/2\rfloor}\frac{N_a}{N}\left(\frac{N^2}{6}-a(N-a)\right)\,.
    \end{split}
\end{align}
Using also~\eqref{GV0sumrule} and~\eqref{FtheoryRel}, we can rewrite this as
\begin{align}
    \begin{split}
    D=&\frac{N}{24}\left(\chi_B+5(12-\chi_B)\right)-\frac{1}{24N}\sum\limits_{a=1}^{\lfloor N/2\rfloor}N_a a(N-a)
    =\frac{\widehat{c}_2}{24}\,.
    \end{split}
\end{align}
In terms of the function $\hatf^{(1)}_0$ introduced in \eqref{defhatf0}, 
it follows that
\begin{align}
    \hatf_0^{(1)}(T)
    =\left(\frac{\chi_X}{12}-\chi_B\right)\widetilde{E}_0(NT)+\frac{1}{12}\sum\limits_{a=1}^{\lfloor N/2\rfloor}N_a\Phi^{(1)}_{N,a}(T)\,.
\end{align}
Using Theorem~\ref{thm:eichler} and~\eqref{eqn:E0trafo}, as well as~\eqref{GV0sumrule} and~\eqref{FtheoryRel}, we find that
\begin{align}
    \begin{split}
    \hatf_0^{(1)}\left(\tfrac{ T }{N T +1}\right)&-\hatf_0^{(1)}(T)\\
    =&\left(\frac{\chi_X}{12}-\chi_B\right)\left[\frac12\log(N T +1)-\frac{\pi{\rm i}}{12}\right]-\frac{1}{12}\sum\limits_{a=1}^{\lfloor N/2\rfloor}N_a\hat{c}_{N}^{(1)}\\
    =&\frac{1}{24}\left(\chi_X-12\chi_B\right)\log(N T +1)+5\pi{\rm i}-\frac{\pi{\rm i}}{3}\chi_B\,,
    \end{split}
    \label{eqn:f1vtransform}
\end{align}
as well as $\hatf_0^{(1)}(T+1)-\hatf_0^{(1)}(T)=-\pi\I\widehat{c}_2/12$.
More generally, under $\gamma=\scriptsize\begin{pmatrix} a & b \\ c & d\end{pmatrix}\in \Gamma_1(N)$, $e^{\hatf_0^{(1)}}$ transforms as a modular form of weight $\frac{\chi_X}{24}-\frac{\chi_B}{2}$, up to a phase subject to the group relations.
Moreover, the derivative of $\hatf^{(1)}(T)$ is a quasimodular form of weight two under $\Gamma_1(N)$, 
\begin{align}
    (2\pi\I)^{-1}\partial_T\hatf^{(1)}(T)=N\left(\chi_B-\frac{\chi_X}{12}\right)E_2(NT)+\frac{1}{12}\sum\limits_{a=1}^{\lfloor N/2\rfloor}N_a\varphi^{(1)}_{N,a}(T)\,.
\end{align}

\paragraph{Genus $g\ge 2$.}
Using~\eqref{eqn:eisensteinPolyLog} and~\eqref{eqn:phigdef}, we can express~\eqref{eqn:bd0f2} as
\begin{align}
    \begin{split}
        &f^{(g\ge 2)}_0(T)\\
        =&(-1)^{g+1}\frac{B_{2g}}{2g[(2g-2)!]}\left[\chi_X\frac{B_{2g-2}}{2(2g-2)}E_{2g-2}(NT)+\sum\limits_{a=1}^{\lfloor N/2\rfloor}N_a\phi^{(g)}_{N,a}(T)\right]\,.
    \end{split}
\end{align}
For $g=2$, one can use Lemma~\ref{lem:g2g}, together with the fact that $e_{N,2}(\tau):=NE_2(N\tau)-E_2(\tau)$ is a modular form of weight two for $\Gamma_1(N)$, to see that $f^{(2)}_0(T)$ itself is a quasimodular form of weight two for $\Gamma_1(N)$.
For $g\ge 3$, Lemma~\ref{lem:g2g} implies that $f^{(g)}_0(T)$ is a modular form of weight $2g-2$ for $\Gamma_1(N)$.

\subsection{Simplified formulae for $N=2,3,4,6$}
 For $N=2,3,4,6$, the relevant Eisenstein series of $\Gamma_1(N)$ can all be obtained in terms of the Eisenstein series of $SL(2,\IZ)$ with arguments rescaled by divisors of $N$. In this way,
 we get, for $N=2$,
 \be
 \begin{split}
\hatf_0^{(0)}(T)=& N_1 \widetilde{E}_{-2}(T) -(\chi_X+N_1) \widetilde{E}_{-2}(2T)\,, \\
\hatf^{(1)}_0(T)=&-\frac{N_1}{12} \log \eta(T) + \frac{N_1+\chi_X-4\chi_B}{12} \log \eta(2T)\,,
 \end{split}
\ee
for $N=3$,
\be
 \begin{split}
 \hatf_0^{(0)}(T) =& N_1 \widetilde{E}_{-2}(T)-(\chi_X+N_1) \widetilde{E}_{-2}(3T)\,,\\
 \hatf^{(1)}_0(T)=&- \frac{N_1}{12} \log \eta(T) + \frac{N_1+\chi_X-4\chi_B}{12} \log \eta(3T)\,,
 \end{split}
\ee
for $N=4$,
\be
 \begin{split}
 \hatf_0^{(0)}(T)=& 
N_1\widetilde{E}_{-2}(T) +(N_2-N_1) \widetilde{E}_{-2}(2T) -(\chi_X+N_2) \widetilde{E}_{-2}(4T) \,,
\\
\hatf^{(1)}_0(T)=&-\frac{N_1}{12} \log \eta(T) + \frac{N_1-N_2}{12} \log \eta(2T)  
+ \frac{N_2+\chi_X-4\chi_B}{12} \log \eta(4T)\,,
 \end{split}
 \ee
and for $N=6$
\be
 \begin{split}
 \hatf_0^{(0)}(T) =& N_1\widetilde{E}_{-2}(T) +(N_2-N_1) \widetilde{E}_{-2}(2T)  \\&
+(N_3-N_1) \widetilde{E}_{-2}(3T)  -(\chi_X-N_1+N_2+N_3) \widetilde{E}_{-2}(6T)\,,\\
\hatf^{(1)}_0(T)=&-\frac{N_1}{12} \log \eta(T) + \frac{N_1-N_2}{12} \log \eta(2T)  
+ \frac{N_1-N_3}{12} \log \eta(3T)   \\ &
+ \frac{N_2+N_3-N_1+\chi_X-4\chi_B}{12} \log \eta(6T)\,.
 \end{split}
 \ee
For $N=5,6$, the transformation properties of $\hatf_0^{(0,1)}$ under the extra generator of $\Gamma_1(N)$ (see point iv) in \S\ref{sec_Jacobi}) can be obtained from the transformation properties of the Eichler integrals conjectured in \S\ref{sec_extragen}.

\section{Modularity from wave-function property}
\label{sec:waveFunctionModularity}

In this section, we shall derive the modular properties of generating functions of GW invariants, and the Jacobi properties of PT invariants, from the wave-function behavior of the topological string partition function under a suitable monodromy.

\subsection{Wave-function property of the topological string partition function}
\label{sec_Ztop}

First, let us recall  the wave function interpretation of the topological string partition function, following~\cite{ Witten:1993ed,Verlinde:2004ck,Aganagic:2006wq,Gunaydin:2006bz,Schwarz:2006br}. This interpretation is most transparent in terms of the topological B-model, which depends only on the complex structure of the Calabi-Yau manifold, rather than the topological A-model, which depends only on the complexified K\"ahler moduli, but the two are related by mirror symmetry. Thus, the moduli $t^a$ below stand for flat  coordinates on the complex structure moduli space 
$\cM$ of the CY threefold $\widehat{X}$ related to $X$ by mirror symmetry. With this in mind, let us define\footnote{In \cite{Bershadsky:1993ta}, $\chi$ should be understood as the Euler number of the A-model geometry, $\chi=\chi_X=-\chi_{\widehat{X}}$.}
\be
\label{BCOV}
 \Psi_{\rm BCOV}(t,\bar t; \lambda,x) = \lambda^{\frac{\chi_X}{24}-1} 
 \exp\left( \sum_{g=0}^{\infty}
\sum_{n=0}^{\infty} \frac{1}{n!} \lambda^{2g-2}
\ C^{(g)}_{a_1 \cdots a_n}(t,\bar t)\,  x^{a_1}
\cdots x^{a_n} \right) .
\ee
where $C^{(g)}_{a_1 \cdots a_n}$ are the topological correlators, which
vanish unless  $2g-2+n> 0$. Those arise as iterated covariant derivatives of the 
genus $g$ free-energies $\cF^{(g)}(t,\bar t)$ for $g\geq 1$, or of the 
Yukawa couplings $C_{abc}$ for $g=0$. The topological string partition function \eqref{BCOV}
satisfies the holomorphic anomaly equations~\cite[(3.17-18)]{Bershadsky:1993ta}
\be
\label{BCOVhae}
\begin{split}
\left[ \partial_{\bar t^a} -
\frac{\lambda^2}{2} e^{2\cK} \bar 
C_{\bar a\bar b\bar c} g^{b\bar b} g^{c\bar c} \frac{\partial^2}{\partial x^b\partial x^c}
+ g_{\bar a b} x^b \left( \lambda \frac{\partial}{\partial \lambda} + x^a \partial_{x^a}\right)
\right]
 \Psi_{\rm BCOV}  = 0,
\\
\left[\partial_{t^a} + \Gamma_{ab}^c x^b \frac{\partial}{\partial x^c} 
+ \partial_{t^a}\cK \left( \frac{\chi_X}{24}-1-\lambda\partial_{\lambda}\right)
-\partial_{x^a} + \partial_{t^a} \cF_1 +  \frac1{2\lambda^2} C_{abc} x^b x^c 
\right] \Psi_{\rm BCOV} = 0,
\end{split}
\ee
where $\Gamma_{ab}^c$ are Levi-Civita coefficients of the special K\"ahler metric $g_{a\bar b}\de t^a \de \bar t^{\bar b}$ 
with K\"ahler potential $\cK(t,\bar t)$. As explained in \cite{ Witten:1993ed,Verlinde:2004ck}, these
equations naturally arise by viewing  $\Psi_{\rm BCOV}$ as the wave function of a particular
state $|\Psi_{\rm top}\rangle$ in the Hilbert space obtained by quantizing the symplectic space $H^3(\widehat{X},\IR)$, in a suitable complex polarization determined by the complex structure on $\widehat{X}$. 

More precisely,  the real vector space $H^3(\widehat{X},\IR)$ carries a symplectic form
$\omega(C,C')$ $= \int_{\widehat{X}} C\wedge C'$, and inherits a complex structure (called Griffith's complex structure) from  the Hodge decomposition
\be
\label{CH3}
C =  x^0  \Omega_{3,0} + x^a   D_a \Omega_{3,0}+
\bar x^{\bar a}  D_{\bar a} \bar \Omega_{3,0}
+  \bar x^0 \bar\Omega_{3,0}\, ,
\ee
where $\Omega_{3,0}$ is a nowhere vanishing holomorphic 3-form on $\widehat{X}$, depending 
on the moduli $(t,\bar t)$. In the complex coordinates $x^\Lambda=(x^0,x^a)$, the symplectic form becomes 
\be
\label{omGriff}
\omega = e^{-\cK} \left(g_{a\bar b}  \de x^a \wedge \de \bar x^{\bar b} - \de x^0 \wedge 
\de \bar x^0 \right)
= \de x^\Lambda \wedge \de \tilde x_\Lambda
\ee
where $\tilde x_\Lambda := (\tilde x_0,\tilde x_a):= e^{-\cK} ( -\bar x^0, g_{a\bar b} \bar x^{\bar b})$.
Upon quantization, the Darboux coordinates  $(x^\Lambda,\tilde x_\Lambda)$ become operators\footnote{Note the unusual hermiticity property $\tilde x^0=-e^{-\cK}\bar x_0$,
which originates from the minus sign in \eqref{omGriff} and leads to convergence issues. This can be remedied by 
exchanging $(x^0,\bar x^0)$, which amounts to using Weil's complex structure rather
than Griffith's, see \cite{Gunaydin:2006bz,Alexandrov:2010ca}.
}
satisfying the canonical commutation rules 
\be
[ \widehat{x^\Lambda}, \widehat{\tilde x_\Sigma} ] = \I\hbar\, \delta^{\Lambda}_{\Sigma}\,, \quad 
[ \widehat{x^\Lambda}, \widehat{x^\Sigma} ] =[ \widehat{\tilde x_\Lambda}, \widehat{\tilde x_\Sigma} ]=0 \,.
\ee
Introducing a basis of (dual) coherent states $ _{t,\bar t}\langle x^\Lambda |$ 
which diagonalize the operators $\widehat{x^\Lambda}$, the topological string partition function can be interpreted as the overlap 
\be
\Psi_{\rm BCOV}(t,\bar t;\lambda,x) = _{t,\bar t}\langle  x^\Lambda | \Psi_{\rm top} \rangle, \quad 
\ee
where $|\Psi_{\rm top} \rangle$ is a fixed, background independent state
(as will become clear below, the topological string coupling $\lambda$ can be identified with $\sqrt{\hbar}$ in the K\"ahler gauge $X^0=1$). 
The indices $t,\bar t$ indicate that the Hodge decomposition \eqref{CH3} depends in a non-holomorphic way on the moduli 
$t^a$. The resulting dependence of the coherent states implies the holomorphic anomaly equations \eqref{BCOVhae}.

\medskip

On the other hand, after choosing a  symplectic basis $(g^\Lambda, \tilde g_\Lambda)$ of $H^3(\widehat{X},\IZ)$ (also known as marking, satisfying $\int_X g^\Lambda \wedge \tilde g_\Sigma=\delta^\Lambda_\Sigma$, $\int_X g^\Lambda \wedge g^\Sigma=\int_X \tilde g_\Lambda  \wedge \tilde g_\Sigma=0$), one can alternatively parametrize $V$
by real Darboux coordinates $(\zeta^\Lambda,\tilde\zeta_\Lambda)$,
\be
C = \zeta^\Lambda  \tilde g_{\Lambda} - \tilde\zeta_\Lambda g^{\Lambda}\ ,\quad 
\omega = \de\zeta^\Lambda \wedge \de \tilde\zeta_\Lambda
\ee
such that $(\zeta^\Lambda,\tilde\zeta_\Lambda)$ become hermitean operators satisfying the canonical
commutation rules
\be
[ \widehat{\zeta^\Lambda}, \widehat{\tilde\zeta_\Sigma} ] = \I \hbar\, \delta^{\Lambda}_{\Sigma}\,, \quad 
[ \widehat{\zeta^\Lambda}, \widehat{\zeta^\Sigma} ]= 
[ \widehat{\tilde\zeta_\Lambda}, \widehat{\tilde\zeta_\Sigma} ]=0
\ee
and work in a polarization
 where the hermitean operator $\widehat{\zeta^\Lambda}$ is diagonalized, while $\tilde\zeta_\Lambda=\I\hbar\partial_{\zeta^\Lambda}$. The resulting `real polarized' 
 wave function $\Psi_\IR(\zeta^\Lambda)$ no longer depends on a choice of background $(t,\bar t)$, but it depends on a choice of marking. Under monodromies in $\cM$, it transforms under the metaplectic representation of $Sp(b_3(\widehat{X}),\IZ)$. 
 If the monodromy acts via
\be
\begin{pmatrix} \tilde\zeta_\Lambda \\ \zeta^\Lambda   \end{pmatrix} \mapsto 
\begin{pmatrix} A & B \\ C & D \end{pmatrix} \begin{pmatrix}  \tilde\zeta_\Lambda \\ \zeta^\Lambda  \end{pmatrix}\,,
\ee 
where $A,B,C,D$ are square matrices of size $\frac12 b_3(\widehat X)$ with integer entries satisfying \eqref{SpABCD}, 
and if the matrix $C$ is invertible,  the wave function should transform as 
 \be
 \label{metarule}
 \Psi_\IR(\zeta^\Lambda) \mapsto \Psi'_\IR(\zeta'^\Lambda)
 = \frac{1}{\sqrt{\det(\hbar C)}}\int  e^{-S(\zeta,\zeta')/\hbar}  \Psi_\IR(\zeta^\Lambda) \, \prod_{\Lambda} \de \zeta^\Lambda
 \ee
 where, similar to \eqref{SABCDX}, 
\be
\label{SABCD}
S(\zeta^\Lambda,\zeta'^\Lambda)=-\frac12 \zeta^t C^{-1} D \zeta + \zeta^t C^{-1} \zeta'-\frac12 \zeta'^t AC^{-1}  \zeta'\,.
\ee
If the matrix $C$ is not invertible, the Gaussian kernel should be replaced by the product of a delta function along the null directions of $C$, times a Gaussian kernel along its orthogonal complement.
For example, if $C=0$ and therefore $D=A^{-T}$, $\Psi'_\IR(\zeta'^\Lambda)=e^{\frac12 \zeta' B \zeta'}
\Psi_\IR(D^{-1}\zeta')$. 

\medskip

The key question is now to relate the real polarized wave function $\Psi_{\IR}(\zeta)$ to the BCOV wave function $\Psi_{\rm BCOV}(t,\bar t; \lambda,x)$. For this, following \cite{Schwarz:2006br}, we define the quasi-homogeneous function
\be
\begin{split}
\widetilde{\Psi}(t,\bar t, \lambda; x^\Lambda) := & 
\lambda^{1-\frac{\chi_X}{24}}
\Psi_{\rm BCOV}(t,\bar t; \lambda/x^0,x^a/x^0) \\
= & (x^0)^{1-\frac{\chi_X}{24}} 
 \exp\left( \sum_{g=0}^{\infty}
\sum_{n=0}^{\infty} \frac{1}{n!} \left(\frac{\lambda}{x^0}\right)^{2g-2}
C^{(g)}_{a_1 \cdots a_n}(t,\bar t)\,  \frac{x^{a_1}}{x^0}
\cdots \frac{x^{a_n}}{x^0} \right) .
\end{split} 
\ee
and take the `holomorphic limit' $\bar{t}^{a}\to -\I\infty$ keeping $x^a$ and other variables fixed. From this limit one arrives at the wave function in holomorphic polarization~\cite[(36)]{Schwarz:2006br}\footnote{We adjust the normalization by $x^\Lambda$-independent factor such that $\Psi_{\rm hol}$ satisfies the holomorphic limit of the holomorphic anomaly equations.
}
\be
    \Psi_{\rm hol}(t,\lambda;x^\Lambda)=
    \lambda^{\frac{\chi_X}{24}-1}
    e^{\cF^{(1)}(t,-\I\infty)}\widetilde{\Psi}(t,-\I\infty, \lambda; x^\Lambda)\,.
\ee
Using the fact that $C_{a_1\dots a_n}^{(g)}$ are covariant derivatives of $\cF^{(g)}(t,\bar t)$, and that $\partial_{t^a}\cK$ vanishes in this limit, one arrives at 
\be
\label{PsiholF}
\Psi_{\rm hol}(t,\lambda;x^\Lambda) = 
(x^0)^{1-\frac{\chi_X}{24}} 
 \exp\left( F( t^a+\tfrac{x^a}{x^0},\tfrac{\lambda}{x^0}) - \tfrac{(x^0)^2}{\lambda^2}
 \left( \cF^{(0)} + \tfrac{x^a}{x^0} \cF_a^{(0)} +  \tfrac{x^a x^b}{2(x^0)^2}
 \cF_{ab}^{(0)} \right)  \right) 
\ee
where $F(t^a,\lambda)=\sum_{g\geq 0} \lambda^{2g-2} \cF^{(g)}(t^a,-\I\infty)$,
and the genus 0 subtractions originate from the fact that $C^{(g)}_{a_1\dots a_n}$ only includes correlators with 
$2g-2+n>0$ (here we denoted $\cF^{(0)}_a=\partial_{t^a} F^{(0)}, \cF^{(0)}_{ab}=\partial^2_{t^a t^b} \cF^{(0)}$).

Now, let us consider the effect of the `holomorphic limit' ${\bar t}^{a}\to -\I\infty$ on the Hodge decomposition~\eqref{CH3}.
Using
\be
\left(\int_X g^\Lambda \wedge \Omega_{3,0};  \int_X \tilde g_\Lambda \wedge \Omega_{3,0}\right)=
(X^\Lambda; F^{(0)}_\Lambda)
= X^0 ( 1, t^a;  2 \cF^{(0)}- t^a \cF_a^{(0)},  \cF_a^{(0)})\,,
\ee
one finds that in this limit, the real Darboux coordinates
are related to the coordinates in the Hodge decomposition by (see
\cite[Appendix A]{Schwarz:2006br} for a detailed derivation)
\be
\begin{split}
\zeta^0/X^0=&x^0\,, \\
\zeta^a/X^0=&x^a + t^a x^0\,,\\ 
\tilde\zeta_a/X^0=& \tilde x_a + \cF_{ab}^{(0)} x^b + \cF_a^{(0)} x^0\,,\\
\tilde\zeta_0/X^0 = & \tilde x_0  - t^a \tilde x_a+ (\cF_a^{(0)} - t^b \cF_{ab}^{(0)}) x^a + 
(2\cF^{(0)} - t^a \cF_a^{(0)}) x^0 \,.
\end{split}
\ee
Following the prescription given below \eqref{SABCD}, this implies that 
the wave function in the real polarization is related to the wave function in the holomorphic polarization by
\be
\begin{split}
&\Psi_{\rm hol}(t,\lambda;x^\Lambda)\\
=& e^{-\frac{(X^0)^2}{\hbar} \left(  \cF^{(0)} (x^0)^2 +  
\cF_a^{(0)} x^0 x^a +   \frac12 \cF_{ab} ^{(0)} x^a x^b\right) }\, 
\Psi_{\IR}\left(\zeta^0=X^0 x^0 , \zeta^a=X^0(x^a + t^a x^0)\right)\,.
\end{split}
\ee
Comparing with \eqref{PsiholF} and setting $\lambda=\sqrt{\hbar}/X^0$, we see that the genus 0 subtractions cancel and one finds
\be
\begin{split}
\Psi_{\IR}(\zeta^\Lambda) = & 
(\lambda x^0)^{\frac{\chi_X}{24}-1} \exp\left( F\left(t^a+\frac{x^a}{x^0},\frac{\lambda}{x^0}\right) \right)
\\
=&
(\lambda\zeta^0/X^0)^{\frac{\chi_X}{24}-1} 
 \exp\left( F\left(\tfrac{\zeta^a}{\zeta^0},\tfrac{\lambda X^0}{\zeta^0} \right) \right)
\\=&
Z_{\rm top}\left(\tfrac{\zeta^a}{\zeta^0},\tfrac{\sqrt{\hbar}}{\zeta^0}\right)\,,
 \end{split} 
\ee
where $Z_{\rm top}(t^a,\lambda)$ was defined in \eqref{defZtop}. 
Thus, we conclude that the topological string partition function $Z_{\rm top}$ is equal to the real polarized wave function, and should transform under monodromies according to the metaplectic representation \eqref{metarule}, with $\hbar$ replaced by $(\lambda X^0)^2$. Further renaming $\zeta^\Lambda\mapsto X^\Lambda$ and enforcing the K\"ahler gauge choice $\lambda=1/X^0$, we arrive at the transformation property under the 
monodromy~\eqref{SpABCDV},
\be
\label{metaruleZtop}
Z'_{\rm top}(X'^\Lambda) 
 = \frac{1}{\sqrt{\det C}}\int  e^{-S(X,X')}  
 Z_{\rm top}(X^\Lambda)\, \prod_{\Lambda} \de X^\Lambda
\ee
for $Z_{\rm top}(X^\Lambda) := Z_{\rm top}(X^a/X^0,1/X^0) = (X^0)^{1-\frac{\chi_X}{24}}
\exp(\sum_{g\geq 0} F^{(g)}(X^\Lambda))$. 

We emphasize that $Z_{\rm top}(X^\Lambda)$ is only defined as an asymptotic series in powers of $1/X^0$, and correspondingly the integral in \eqref{metaruleZtop} should be evaluated using the saddle point method, by the usual Feynman expansion around extrema of $F(X^\Lambda)-S(X,X')$ with respect to $X^\Lambda$. 
As an example, restricting to a one-parameter model with $T=X^1/X^0, \lambda=\lambda'$ for simplicity,\footnote{See~\cite[(2.11)]{Aganagic:2006wq} for an example of a Feynman expansion up to second order in multi-parameter models, along with a pictorial representation.}
\be
\int \de T  \exp \left( -\frac{\cS(T,T')}{\lambda^2} + 
\sum_{g\geq 0} \cF^{(g)}(T) \lambda^{2g-2} \right)
= \exp \left( \sum_{g\geq 0} \cF'^{(g)}(T') \lambda^{2g-2} \right)
\ee
with $S(X^\Lambda,X'^\Lambda)=(X^0)^2 \cS(T,T')$ and 
\bea
\label{Feynexp}
\cF'^{(0)} &=& \langle \cG(T,T')\rangle_T\,, \quad 
\cG(T,T') := \cF^{(0)}(T) -\cS(T,T')\,, \nn\\
\cF'^{(1)} &=&  \cF^{(1)} - \frac12 \log ( \partial^2 \cG )\,, \nn\\
\cF'^{(2)} &=&  \cF^{(2)} 
- \tfrac{\partial^2 \cF^{(1)}  + (\partial \cF^{(1)} )^2}{2\,  \partial^2 \cG }
+  \tfrac{\partial^4 \cG  + 4\, \partial \cF^{(1)} \, \partial^3 \cG  }{8 ( \partial^2 \cG)^2 }
- \tfrac{5 (\partial^3 \cG )^2}{24  ( \partial^2 \cG)^3} \,,
\nn\\
\cF'^{(3)} &=&  \cF^{(3)} 
-\tfrac{2 {\partial \cF^{(1)}} \,  {\partial \cF^{(2)}}+{\partial^2 \cF^{(2)}}}{2 \, {\partial^2 \cG}}
+\tfrac{4
   {\partial^3 \cG}\, {\partial \cF^{(2)}}+{\partial^4 \cF^{(1)}}+2 ({\partial^2 \cF^{(1)}})^2+4 {\partial^3 \cF^{(1)}}\, {\partial \cF^{(1)}}
   +4 ({\partial \cF^{(1)}})^2 \, {\partial^2 \cF^{(1)}}}{8
   ({\partial^2 \cG})^2}
    \nn\\&&
    - \tfrac{{\partial^6 \cG}+12 {\partial^4 \cG}\,
   {\partial^2 \cF^{(1)}}+12\, {\partial^4 \cG} \, ({\partial \cF^{(1)}})^2+20\,
   {\partial^3 \cG} \, {\partial^3 \cF^{(1)}}+8\, {\partial^3 \cG}
  ( {\partial \cF^{(1)}})^3}{48 ({\partial^2 \cG})^3} 
   - \tfrac{{\partial \cF^{(1)}} \left({\partial^5 \cG}
  +8 \, {\partial^3 \cG} \,{\partial^2 \cF^{(1)}}\right)}{8
   ({\partial^2 \cG})^3} 
   \nn\\&&
   +\tfrac{4 {(\partial^4 \cG})^2+30
  ( {\partial^3 \cG})^2 \, {\partial^2 \cF^{(1)}}+24 (\partial^3 {\cG})^2\,
   ({\partial \cF^{(1)}})^2+7\, {\partial^5 \cG}\, {\partial^3 \cG} +32\,
   {\partial^4 \cG}\, {\partial^3 \cG}\, {\partial \cF^{(1)}}}{48
   ({\partial^2 \cG})^4} \nn\\&&
   -\tfrac{30 ({\partial^3 \cG})^3\, {\partial \cF^{(1)}}+25\,
   {\partial^4 \cG}\, ({\partial^3 \cG})^2}{48
   ({\partial^2 \cG})^5}
   +\tfrac{5( {\partial^3 \cG})^4}{16
   {(\partial^2 \cG})^6} \,,
\eea
where $\partial\equiv\partial_{T}$ evaluated at the extremum of 
$\cG(T,T')$ with respect to $T$, assuming  that this extremum is unique. 

\subsection{Relative conifold monodromy}
\label{sec_relcon}

In order to derive the modular properties of the generating series of GW and PT invariants introduced in~\S\ref{sec_Jacobi}, we shall exploit the wave-function transformation property of $Z_{\rm top}$ under a certain monodromy $U$ in K\"ahler moduli space $\cM$, which acts on the fiber modulus  as 
$T\mapsto \frac{T}{1+NT}$, while preserving the large base limit $S^\alpha\to\I\infty$. Together
with the large volume monodromy $T\mapsto T+1$, this generates a subgroup of $\Gamma_1(N)$ which coincides with $\Gamma_1(N)$ for $N\leq 4$.

Generalizing~\cite{Schimannek:2019ijf} to the case of genus one fibrations with an $N$-section, it turns out that the relevant monodromy $U$ corresponds to an auto-equivalence $g_U$ of the derived category of coherent sheaves $\cC=D^b\Coh X$ given by a Fourier-Mukai transformation with respect to the ideal sheaf of the relative diagonal in the fiber product $X\times_B X$~\cite[Section 3.3]{Cota:2019cjx}.
Its action on the Chern character of a brane $\mathcal{E}^{\bullet}\in\cC$ is given by
\begin{align}
    g_U:\,\text{ch}(\mathcal{E}^{\bullet})\mapsto \text{ch}(\mathcal{E}^{\bullet})-\pi_{2,*}\left[\pi_1^*\left(\text{ch}(\mathcal{E}^{\bullet})\text{Td}(T_{X/B})\right)\right]\,,
\end{align}
where $\text{Td}(T_{X/B})$ is the Todd class of the virtual relative tangent bundle of the fibration.
In Appendix \ref{app_relcon}, we use this relation to compute  the action of $U$ on the integral basis
of the charge lattice $\Gamma$
\begin{align}
\cB=    \left( \text{ch}(\mathcal{O}_X),\,\text{ch}(\mathcal{O}_{D_e}),\,\text{ch}(\mathcal{O}_{D_\alpha}),\,\text{ch}(\mathcal{O}_{\mathcal{E}}),\,\text{ch}(\mathcal{C}^\alpha),\,\text{ch}(\mathcal{O}_{\text{pt.}})\,\right)^T\,.
\end{align}
The end result then takes the form $\cB \mapsto U\cB $ where $U$ is the matrix with integer coefficients~\footnote{The integrality of $\frac{N}{2}\left(a_\gamma+\frac{1}{N}\ell_\gamma\right)C^{\gamma\beta}$ is not obvious but implied by the fact that the Fourier-Mukai transform induces an automorphism of $H^*(X,\mathbb{Z})$.}
\begin{align}
\label{eqU}
	U=\left(\begin{array}{cccccc}
		1&0&-c^\beta&\frac{N}{2}c^\gamma(a_\gamma-C_{\gamma\gamma})&0_\beta&0\\
		-N&1&\frac{N}{2}\left(a_\gamma+\frac{1}{N}\ell_\gamma\right)C^{\gamma\beta}&\rho&0_\beta&0\\
		0_\alpha&0_\alpha&\delta^{\beta}_\alpha&-N c_\alpha&0_{\alpha\beta}&0_\alpha\\
		0&0&0^\beta&1&0_\beta&0\\
		0^\alpha&0^\alpha&-C^{\alpha\beta}&\frac{N}{2}C^{\alpha\gamma}\left(a_\gamma-C_{\gamma\gamma}\right)&\delta^{\alpha}_{\beta}&0^\alpha\\
		0&0&0^\beta&-N&0_\beta &1
	\end{array}\right)\,,
\end{align}
with $C_{\gamma\gamma}$ being the vector of diagonal entries of the matrix $C_{\alpha\beta}$, and 
\begin{align}
    \begin{split}
    \rho =N^2+\frac{N}{4}(\ell_\alpha+N a_\alpha)C^{\alpha\beta}(C_{\beta\beta}-a_\beta)-\frac{N}{12}(2\kappa+c_{2})\,.
    \end{split}
\end{align}
If the base $B$ is a del Pezzo surface, then $C_{\alpha\alpha}-a_\alpha=-2$ and we find the simpler expression
\begin{align}
	\rho=-N\left[\frac{1}{2}\sum_{\alpha=1}^{b_2(B)}(N a^\alpha+\ C^{\alpha\beta} \ell_\beta)-N+\frac16\kappa+ \frac{1}{12}c_{2}\right]\,.
\end{align}
The combination $\frac16\kappa+ \frac{1}{12}c_{2}$ is equal to the holomorphic Euler characteristic of the divisor $D_e$, so is integer, and one may check that $\sum_{\alpha}(N a^\alpha+\ C^{\alpha\beta} \ell_\beta)$ is always even. Using $\ccF^{(0)}:=(2\pi\I)^{-3}\cF^{(0)}$, we introduce the period vector\footnote{This differs from the vector $V=(F_\Lambda,X^\Lambda)$ in \eqref{SpABCDV} by a basis change.} 
\be
\label{eqn:defV}
V=X^0(1,T,S^\alpha,-(2\ccF^{(0)}-T\partial_T \ccF^{(0)}-S^\alpha \partial_{S^\alpha} \ccF^{(0)}), -\partial_T \ccF^{(0)}, -\partial_{S^\alpha} \ccF^{(0)})^T
\ee
associated to the prepotential (in `primed' basis, see footnote \ref{foo:primed})
\be
\label{eqF0}
\cF^{(0)}(T,S^\alpha) = (2\pi\I)^3\left(\frac{\kappa T^3}{6} + \frac{T^2}{2} \ell_\alpha S^\alpha + \frac{NT}{2} S^\alpha C_{\alpha\beta} S^\beta\right) + f^{(0)}(T,S^\alpha)
\ee
where 
\be
f^{(0)}(T,S^\alpha) = \sum_{k_\alpha\geq 0 ,n\geq 0} \GW_{k_\alpha C^\alpha+n F}^{(0)}
e^{2\pi\I (n T + k_\alpha S^\alpha)}, 
\ee
where the constant map contribution is as usual $\GW_{0,0_\alpha}^{(0)}=-\frac12\chi_X \zeta(3)$.
In terms of this basis, the monodromy acts as $V\mapsto U' V$ with\footnote{Here we used the conjectural equality $a_\alpha=c_\alpha$. Note that $U'$ is not integral in general, unlike~$U$.}
\be
\label{eqUp}
U' = \begin{pmatrix}
1 & N & 0_\beta & 0 & 0 & 0_\beta \\
0 & 1 & 0_\beta & 0 & 0 & 0_\beta \\
\frac12 c^\alpha & \frac{N}{2} c^\alpha & \delta^\alpha_\beta & 0^\alpha & 0^\alpha & - C^{\alpha\beta} \\
\frac{12-\chi_B}{4} & \frac{N(12-\chi_B)}{4}  & 0^\beta & 1 & 0 & -\frac12 c^\beta \\
-\frac{N(12-\chi_B)}{4}   &  \frac{(\chi_B-8)N^2}{4}  -\frac{c_2 N}{12}  & -\frac{N}{2} c_\alpha& -N & 1 & \frac{N}{2} c^\beta \\
0_\alpha & -\frac{N}{2} c_\alpha & 0_{\alpha\beta} & 0_\alpha & 0_\alpha & \delta_{\alpha}^{\beta}
\end{pmatrix}\,.
\ee
The two matrices are related as  $U=P U' P^{-1}$ where
\be
P=\begin{pmatrix}
0 & \frac{c_2}{24} & \frac12 c_\beta & 1 & 0 & 0_\beta \\
-\frac{\kappa}{6}-\frac{c_2}{24} & -\frac{\kappa}{2} & -\frac12 \ell_\beta & 0 & 1 & 0_\beta \\
-\frac12 c_\alpha & N_\alpha -\frac{N}{2} c_\alpha &0_{\alpha\beta}&  0_\alpha & 0_\alpha & \delta_{\alpha\beta} \\ 0 & 1 & 0_\alpha & 0 & 0 &0_\beta \\
0_\alpha & 0_\alpha & \delta_{\alpha\beta} & 0_\alpha & 0_\alpha  & 0_{\alpha\beta} \\
-1 & 0 & 0_\beta & 0 & 0 & 0_\beta
\end{pmatrix}
\ee
with $N_\alpha$ being a vector with all entries equal to $N$.

Using \eqref{eqUp}, we find that the relative conifold monodromy acts on $(T,S^\alpha)$ by 
\be
\label{TStrans}
\begin{split}
T\mapsto & \frac{T}{1+NT}, \quad
S^\alpha\mapsto S^\alpha + \frac12 c^\alpha + \frac{C^{\alpha\beta}}{1+NT}
\left( \frac{T^2}{2} \ell_\beta + \frac{\partial_{S^\beta} f^{(0)}}{(2\pi\I)^3} \right) \\
\end{split}
\ee
or in terms of the shifted variable $\hatS_\alpha$ defined in \eqref{defShat},
\be
\hatS^\alpha\mapsto \hatS^\alpha + \frac12 c^\alpha + \frac{1}{(2\pi\I)^3}\frac{C^{\alpha\beta}\partial_{S^\beta} f^{(0)} }{1+NT} 
\ee
while the (instanton part of the) prepotential $f^{(0)}(T,S^\alpha)$ should transform as~\footnote{The algebric manipulations leading to \eqref{transf0} can be found in the companion Mathematica worksheet. }
\be
\label{transf0}
\begin{split}
 f^{(0)}(T,S^\alpha) \mapsto &  \frac{f^{(0)}(T,S^\alpha) }{(1+NT)^2}+\frac{(2\pi\I)^3}{(1+NT)^2}
 \left[ 
 \frac{\hatkappa NT^4}{6(1+NT)}+\frac{NT^2}{24} \left( \hatc_2-12N\right)
 \right. \\ & \left. - \frac{12-\chi_B}{8} (1+NT) \right] 
 + \frac{1}{2(2\pi\I)^3(1+NT)^3} C_{\alpha\beta} \partial_{S^\alpha} f^{(0)} \partial_{S^\beta} f^{(0)}\,.
 \end{split}
\ee
In particular, it follows that the derivative of $f^{(0)}$ with respect to $S_\alpha$ (or, equivalently, with respect to  $\hatS_\alpha$) should transform like a modular form of weight $-2$ under $\Gamma_1(N)$,
\be
\label{transdf0}
\partial_{S^\alpha} f^{(0)} = \partial_{\hatS^\alpha} f^{(0)}  \mapsto  \frac{\partial_{S^\alpha} f^{(0)}}{(1+NT)^2}\,.
\ee
It will be convenient to introduce
\be
\hatf^{(0)}(T,\hatS^\alpha) := \frac{(2\pi\I)^3}{6} \hatkappa T^3 + f^{(0)}(T,\hatS^\alpha)
\ee
whose transformation property is somewhat simpler than \eqref{transf0},
\be
\label{transf0h}
\begin{split}
 \hatf^{(0)}(T,\hatS^\alpha) \mapsto &  \frac{\hatf^{(0)}(T,\hatS^\alpha)}{(1+NT)^2}+
\frac{(2\pi\I)^3}{(1+NT)^2}
 \left[  
 \frac{NT^2}{24} \left( \hatc_2-12N\right)
 - \frac{12-\chi_B}{8} (1+NT) \right]  \\ & 
 + \frac{1}{2(2\pi\I)^3(1+NT)^3} C_{\alpha\beta} 
 \partial_{\hatS^\alpha} \hatf^{(0)} \partial_{\hatS^\beta} \hatf^{(0)}\,.
 \end{split}
\ee

\subsubsection*{Lift of a general $\Gamma_1(N)$ element}
More generally, any element $\gamma=\scriptsize\begin{pmatrix} a & b \\ c& d \end{pmatrix}\in\Gamma_1(N)$ can be obtained by suitable product of large volume monodromies $T\mapsto T+1$, relative conifold monodromies $T\mapsto T/(NT+1)$
and monodromies around other components of the discriminant locus when $N>4$. 
The corresponding 
element  $g_\gamma\in Sp(2b_2(X)+2,\IZ)$ in general acts on $T,S^\alpha$ via transformations of the form
\be
\label{transgen}
\begin{split}
T&\mapsto  \frac{aT+b}{cT+d}, \quad\\
S^\alpha&\mapsto S^\alpha + \frac{\ell^\alpha[1+(cT-a)(cT+d)]}{2N c(cT+d)}+
 \frac{c}{N(cT+d)} \frac{C_{\alpha\beta} \partial_{\hatS^\alpha} \hatf^{(0)}}{(2\pi\I)^3}+
\delta^\alpha, 
\\
\hatf^{(0)}&(T,\hatS^\alpha)  \mapsto   \frac{\hatf^{(0)}(T,\hatS^\alpha)}{(cT+d)^2} + \frac{(2\pi\I)^3\left(x T^2 + y T +z\right)}{(cT+d)^2}
 + \frac{C_{\alpha\beta} \partial_{\hatS^\alpha} \hatf^{(0)} \partial_{\hatS^\beta} \hatf^{(0)}}{2(2\pi\I)^3(cT+d)^3} 
 \end{split}
\ee
for some $\gamma$-dependent constants $\delta^\alpha,x,y,z$ subject to the group relations.
While the transformation of $S^\alpha$ looks complicated, it is such that $\hatS^\alpha$, defined in \eqref{defShat}, transforms as 
\be
\hatS^\alpha  \mapsto \hatS^\alpha+ \frac{c}{N(cT+d)} \frac{C_{\alpha\beta} \partial_{\hatS^\alpha} \hatf^{(0)}}{(2\pi\I)^3}+\delta^\alpha \,,
\ee
and the variable $\tildeS^\alpha$ defined in \eqref{defSt} below transforms in an even simpler way, $\tildeS^\alpha\mapsto \tildeS^\alpha+\delta^\alpha$.
This product of monodromies is realized by an auto-equivalence $g_\gamma\in{\rm Aut}(\cC)$
acting on the period vector \eqref{eqn:defV} by $V\mapsto U'_\gamma V$, where 
  $U'_\gamma$ is a matrix of the form
\be
\label{eqUpgen}
U'_\gamma= \begin{pmatrix}
d & c & 0_\beta & 0 & 0 & 0_\beta \\
b & a & 0_\beta & 0 & 0 & 0_\beta \\
d \delta^\alpha -\frac{b \ell^\alpha }{2 N}  & \frac{\ell^\alpha  (d-a)}{2 N}+c \delta^\alpha 
& d \delta_\alpha^\beta & 0^\alpha & 0^\alpha & -\frac{c}{N} C^{\alpha\beta} \\
\rho_1& \rho_2& b \delta^\beta  N   & a & -b & \frac{b \ell^\beta }{2 N}-a \delta^\beta \\
\rho_3&\rho_4&   -\frac{b \ell_\alpha }{2}-d \delta_\alpha  N & -c & d &  \frac{\ell^\beta  (a-d)}{2 N}+c \delta^\beta \\ 
-b \delta_\alpha  N & -a \delta_\alpha  N-\frac{b \ell_\alpha }{2} &-b N \delta_{\alpha\beta} & 0_\alpha & 0_\alpha & a\delta_{\alpha}^{\beta}
\end{pmatrix}
\ee
where 
\be
\begin{split}
\rho_1=&b \left(y+\frac{\delta ^2 N}{2}\right)-2 a z\,,\\
\rho_2=&\frac{1}{2} \left(a \delta ^2 N-2 a y+4 b x+b \delta  \ell \right)\,, \\
\rho_3=& \frac{1}{2} \left(-b \delta\cdot  \ell  +4 c z-d \left(\delta^2 N+2 y\right)\right) \,, \\
\rho_4=&
-\frac{2 N \left(a \delta\cdot \ell  +c \delta^2 N-2 c y+4 d x+d  \delta\cdot \ell\right)+b \ell^2}{4 N}
\end{split}
\ee
and $\delta^2=\delta_\alpha \delta^\beta, \delta\cdot \ell=\delta_\alpha  \ell^\alpha$, etc.  For $\gamma=\scriptsize\begin{pmatrix} 1 & 1 \\ 0& 1 \end{pmatrix}$, one recovers the standard action of the large volume monodromy $T\mapsto T+1$,
\be
U'_T=\begin{pmatrix}
 1 & 0 & 0_\beta & 0 & 0 & 0_\beta \\
 1 & 1 & 0_\beta & 0 & 0 & 0_\beta \\
 0^\alpha & 0^\alpha & \delta^\alpha_\beta &  0^\alpha & 0^\alpha  & 0^{\alpha\beta} \\
 \frac{\kappa }{6} & \frac{\kappa }{2} & \frac{\ell^\beta }{2} & 1 & -1 & 0^\beta  \\
 -\frac{\kappa }{2} & -\kappa  & -\ell_\alpha  & 0 & 1 & 0^\beta \\
-\frac{\ell_\alpha}{2}  & -\ell_\alpha  & -N C_{\alpha\beta} & 0_\alpha & 0_\alpha & \delta^\alpha_\beta  \\
\end{pmatrix}
\ee
corresponding to $\delta^\alpha=\frac{\ell^\alpha}{2N}, (x,y,z)=(\hatkappa/2,\hatkappa/2,\tilde
\kappa/6)$.
While we do not know how to determine the constants $\delta^\alpha,x,y,z$ for a general group element, an important constraint is that the matrix 
\be
U_\gamma=P U'_\gamma P^{-1}
=\left(\begin{array}{cccccc}
		a&-b&*&*&*&*\\
		-c&d&*&*&*&*\\
		0_\alpha&0_\alpha &a \delta_\alpha^\beta & * &-b N C_{\alpha\beta}&*\\
		0&0&0^\beta&a &0_\beta& -b\\
		0^\alpha&0^\alpha&-\frac{c}{N} C^{\alpha\beta}&* &d \delta^{\alpha}_{\beta}&*\\
		0&0&0^\beta&-c&0_\beta &d
	\end{array}\right)\,,
\ee
where the stars stand for somewhat complicated but easily computed combinations of $x,y,z,\delta,\kappa,c_\alpha,c_2$,
should have integer entries.

\subsection{Modularity of genus 0 GW invariants\label{sec_modgw0}}
At this point, we can deduce the modular properties of generating series of GW invariants \eqref{deffg} at genus zero, by inserting the Fourier expansion
\be
\label{Fourierf0}
\hatf^{(0)}(T,\hatS^\alpha)  =\frac{(2\pi\I)^3}{6} \hatkappa T^3 + \sum_{k_\alpha\geq 0} f^{(0)}_{k_\alpha}(T)  \, e^{2\pi\I k_\alpha \hatS^\alpha}
\ee
into \eqref{transf0h}. As before, we use the notations $f^{(0)}_{k_\alpha}=f^{(0)}_{H}$ interchangeably,
for any nef divisor class $H=k_\alpha \check{D}^\alpha$ on $B$, and in this section omit the genus superscript $(0)$ to lighten the notation. 

First we note that in the large base limit $S^\alpha\to\I\infty$, the quadratic term on the second line of \eqref{transf0h} is exponentially suppressed, so we get the transformation property for the base degree $0_\alpha$, genus zero generating series
\be
\begin{split}
\label{transf00}
(1+NT)^2 \hatf_{0_\alpha}\left( \tfrac{T}{1+NT}\right) -\hatf_{0_\alpha}(T)
= (2\pi\I)^3\left(\tfrac{NT^2}{24} \left( \hatc_2-12N\right) - \tfrac{12-\chi_B}{8} (1+NT) \right)
 \end{split}
\ee
in perfect agreement with \eqref{transf000}.
In particular, the Yukawa coupling
$Y_0(T)=(2\pi\I)^3\partial^3_T \hatf_{0_\alpha}$ is a modular form of weight $4$ under $\Gamma_1(N)$,
and $\hatf_{0_\alpha}(T)$ is its holomorphic Eichler integral.

We now turn to the case of non-vanishing base degree, $k_\alpha> 0$. In this case, we can
determine the modular properties of the generating series $f_{k_\alpha}(T)$ by inserting
\eqref{Fourierf0} into \eqref{transdf0},
\be
\sum_{k_\alpha>0} k_\alpha \, f_{k_\alpha}(\tfrac{T}{1+NT})  \, e^{2\pi\I k_\alpha (\hatS^\alpha
+ \frac12 c^\alpha + \frac{C^{\alpha\beta}\partial_{\hatS^\beta} f }{1+NT})} = 
\frac{1}{(1+NT)^2} \sum_{k_\alpha>0} k_\alpha \, f_{k_\alpha}(T)  \, e^{2\pi\I k_\alpha \hatS^\alpha}
\ee 
and expanding in powers of $e^{2\pi\I \hatS^\alpha}$ on both sides. The only difficulty is that 
the function $\partial_{\hatS^\beta} f$ appearing in the exponent on the left-hand side is itself
a function of $\widehat S^\alpha$ which needs to be expanded.\footnote{Such doubly exponential contributions were ignored in~\cite{Cota:2019cjx}, but in fact, $e^{-e^{-S}}= 1+e^{-S}+\cO(e^{-2S})$ and such effects are key for the modular properties of the topological string amplitude.
} For primitive curve classes, this
complication does not arise and one immediately concludes that the generating series of genus 0 GV invariants is a modular form of weight $-2$ with multiplier system,
\be
f_{k_\alpha}(\tfrac{T}{1+NT})  = \frac{e^{-\I\pi k_\alpha c^\alpha} f_{k_\alpha}(T) }{ (1+NT)^2}, \quad f_{k_\alpha}(T+1)=e^{-\frac{\I\pi}{N} k_\alpha \ell^\alpha} f_{k_\alpha}(T)  \qquad \mbox{(primitive case)}
\ee
where $\ell^\alpha:=C^{\alpha\beta}\ell_\alpha$. 
In general however, there are contributions from lower degrees, leading to a modular anomaly of the form 
\be
\label{eqn:Ak}
 e^{\I\pi k_\alpha c^\alpha} (1+NT)^2 f_{H}(\tfrac{T}{1+NT})  - f_{H}(T)  = 
 \frac{1}{1+NT}
 \sum_{H=\sum H_i} a_{\{H_i\}} \prod_i \frac{f_{H_i}(T)}{1+NT}
\ee
where $a_{\{k_i\}}$ are rational coefficients (see \eqref{Afupto5} below for some examples). This indicates that $ f_{H}$ is in general a quasimodular form of weight $-2$.

To show this more explicitly, it is useful to introduce yet another shifted version of the K\"ahler moduli $S^\alpha$, namely
\be
\label{defSt}
\tildeS^\alpha := \hatS^\alpha - \frac{E_2(NT)}{12} \frac{C^{\alpha\beta} \partial_{\hatS^\beta} f^{(0)}}{(2\pi\I)^2}
\ee
where $\hatS^\alpha$ was defined in \eqref{defShat} and $E_2(T)$ is the quasimodular Eisenstein series of weight 2 already encountered in \S\ref{sec_elliptic}.
Using 
\be
\label{E2Ntrans}
E_2\left( \frac{NT}{1+NT}\right) = (1+NT)^2 E_2(NT) + \frac{12}{2\pi\I} (1+NT)
\ee
which follows from \eqref{eqn:e2trafo}, 
we see that $\tildeS^\alpha$ transforms simply with a constant shift under the relative conifold monodromy,
\be
\tildeS^\alpha \mapsto \tildeS^\alpha + \frac12 c^\alpha\,.
\ee
Thus, upon Fourier expanding $\partial_{\hatS^\alpha} f$ as a function of $\tildeS^\alpha$ rather than
$\hatS^\alpha$,
\be
\label{deff0ta}
\sum_{k_\alpha>0} k_\alpha \, f_{k_\alpha}(T)  \, e^{2\pi\I k_\beta \hatS^\beta} = 
\sum_{k_\alpha>0} k_\alpha \, \tildef_{k_\alpha}(T)  \, e^{2\pi\I k_\beta \tildeS^\beta}
\ee
the new generating series $\tildef_{k_\alpha}(T)$ (which are polynomials in $E_2(NT)$ and in the previous
series $f_{k_\alpha}(T)$) are now actual modular forms of weight $-2$ with a multiplier system but no modular anomaly,
\be
\tildef_{k_\alpha}(\tfrac{T}{1+NT})  = \frac{e^{-\I\pi k_\alpha c^\alpha} \tildef_{k_\alpha}(T) }{ (1+NT)^2}\,, \quad 
\tildef_{k_\alpha}(T+1)=e^{-\frac{\pi\I}{N} k_\alpha \ell^\alpha} 
\tildef_{k_\alpha}(T)\,.
\ee
In fact, one may integrate the vectorial equation \eqref{deff0ta} to a scalar equation
\be
\label{deff0t}
\tildef(T,\tildeS^\alpha):= 
f(T,S^\alpha)-\frac{E_2(NT)}{24(2\pi\I)^2} C^{\alpha\beta} \partial_{\hatS^\alpha} f \partial_{\hatS^\beta} f
= \sum_{k_\alpha\geq 0} \tildef_{k_\alpha}(T)\,  e^{2\pi\I k_\alpha \tildeS^\alpha}\,.
\ee
To see that the last equality defines the same Fourier modes $\tilde f_{k_\alpha>0}$
as in~\eqref{deff0ta}, along with the zero mode $\tilde f_{0_\alpha}(T)$, we use that
\begin{align}
\begin{split}
    \frac{\partial\tilde{f}}{\partial\tildeS^\alpha}=&\frac{\partial\hatS^\beta}{\partial\tildeS^\alpha}\frac{\partial}{\partial\hatS^\beta}\left(f(T,S^\alpha)-\frac{E_2(NT)}{24(2\pi\I)^2} C^{\gamma\delta} \partial_{\hatS^\gamma} f \partial_{\hatS^\delta} f\right)\\
    =&\frac{\partial\hatS^\beta}{\partial\tildeS^\alpha}\left(\delta_{\beta}^\gamma - C^{\gamma\delta}\frac{E_2(NT)}{12}  \frac{
\partial^2_{\hatS^\beta \hatS^\delta} f }{(2\pi\I)^2} \right)\partial_{\hatS^\gamma}f\,,
\end{split}
\end{align}
together with
\begin{align}
\partial_{\hatS^{\alpha}} \tildeS^\beta = 
\delta_{\alpha}^\beta - C^{\gamma\beta}\frac{E_2(NT)}{12}  \frac{
\partial^2_{\hatS^\alpha \hatS^\gamma} f }{(2\pi\I)^2}   \,.
\end{align}
From this follows the important fact that
$\partial_{\hatS^{\alpha}} f= \partial_{\tildeS^{\alpha}} \tildef$.
Taking derivatives with respect to $E_2(NT)$ on both sides of the first equality in~\eqref{deff0t}, keeping $\hatS^\alpha$ fixed, we get 
\be
\label{deff0tE2}
-\frac{1}{12(2\pi\I)^2}
C^{\alpha\beta} \partial_{\hatS^\beta} f\, \partial_{\tildeS^\alpha} \tildef = 
 \partial_{E_2(NT)}|_{\tilde{S}} f  - \frac{1}{24(2\pi\I)^2} C^{\alpha\beta} \partial_{\hatS^\beta} f\, \partial_{\hatS^\alpha}  f
\ee 
and therefore
\be
\label{anomf0}
 \partial_{E_2(NT)}|_{\tilde{S}} f  =
 - \frac{1}{24(2\pi\I)^2} C^{\alpha\beta} \partial_{\hatS^\alpha}  f \partial_{\hatS^\beta} f\, 
\ee
Further expanding in Fourier modes with respect to $\hatS^{\alpha}$, we get the genus zero version of \eqref{anom-fnN}, 
\be  
   \partial_{E_2(NT)} f_{H}^{(0)}(T)=
   -\frac{1}{24} \sum_{H=H_1+H_2} (H_1 \cap H_2)
    f_{H_1}^{(0)} f_{H_2}^{(0)}
    \label{hanom0}
\ee
confirming that the generating series $f_{H}^{(0)}(T)$ (now restoring the genus superscript) are quasimodular forms of weight $-2$, and depth given by the maximal number of elements in possible decompositions of $H=\sum H_i$ into a sum of effective curve classes in $H_2(B,\IZ)$.
 In contrast, the functions $\tildef_{H}^{(0)}(T)$ defined by \eqref{deff0t} are ordinary holomorphic modular forms of weight $-2$  under 
$\Gamma_1(N)$, given by the depth zero part of  $f_{H}^{(0)}(T)$.

\subsubsection*{Explicit formulae for genus one fibrations on $\IP^2$}

To illustrate the constructions above, let us consider genus one  fibrations with $N$-section on $\IP^2$. For brevity we write $f_k(T):=f_{k H}(T)$, where $H$ is the hyperplane
class of $\IP^2$. Denoting by $A_k$ the left-hand side of \eqref{eqn:Ak}, and setting 
 $\hat f_{k_i}= f_{k_i}/(1+NT)$, we find the following anomalous transformations properties of the generating series of genus zero GW invariants,
\be
\label{Afupto5}
\begin{split}
A_1 =& 0\,, \\
A_2 =& -\frac12 \hat f_1^2\,, \\
A_3 =& - \frac13 \hat f_1^3 + 2 \hat f_1 \hat f_2\,,\\
A_4 =& - \frac23 \hat f_1^4 + 4 \hat f_1^2 \hat f_2 - 2 \hat f_2^2 - 3 \hat f_1 \hat f_3\,, \\
A_5 =&- \frac{25}{24} \hat f_{1}^5+\frac{25}{3} \hat f_{2} \hat f_{1}^3-\frac{15}{2} \hat f_{3}
   \hat f_{1}^2-10 \hat f_{2}^2 \hat f_{1}+4 \hat f_{4} \hat f_{1}+6 \hat f_{2} \hat f_{3}\,.
\end{split}   
\ee
In contrast, the combinations $\tildef_k$ defined by \eqref{deff0ta} transform as weakly holomorphic modular forms of weight $-2$,
\be
\label{tfkupto5}
\begin{split}
\tildef_1 =&f_1 \,\\
\tildef_2 =& f_2 + \frac{E_2(NT)}{24} f_1^2\,, \\
\tildef_3 =& f_3 + \frac{E_2(NT)}{6} f_1 f_2 + \frac{E_2(NT)^2}{288} f_1^3 \,,  \\
\tildef_4 =& f_4 + \frac{E_2(NT)}{4} f_1 f_3 
+\frac{E_2(NT)}{6} f_2^2 +  \frac{E_2(NT)^2}{36} f_1^2 f_2 
+ \frac{E_2(NT)^3}{2592} f_1^4 \,,  \\
 \tildef_5 =& f_5
 + \frac{25 E_2^4 }{497664}f_{1}^5 +\frac{25 E_2^3}{5184} f_{2}
   f_{1}^3+\frac{5}{96} E_2^2 f_{3} f_{1}^2+\frac{5}{72} E_2^2
   f_{2}^2 f_{1}+\frac{1}{3} E_2 f_{4} f_{1}+\frac{1}{2} E_2  f_{2} f_{3} 
\end{split}   
\ee
where in the last line, we left the argument $NT$ of $E_2$ implicit. Inverting this triangular system of equations, it is straightforward to check that the $f_k$'s satisfy the holomorphic anomaly equations
\be
\frac{\partial}{\partial E_2(NT)} f_k(T)=   -\frac{1}{24} \sum_{k=k_1+k_2} k_1\,k_2\,  f_{k_1}(T) f_{k_2}(T)\,.
\ee
In fact, the $\tildef_k$'s are simply the depth zero part of the quasimodular forms $f_k$.
Using the ansatz
\be
\label{ansatztfk}
\tildef_k(T) = \frac{[\Delta_{2N}(T)]^{\frac{k(3N^2+\ell)}{2N}}}{\eta^{36k}(NT)} 
P_{18k-2-k(3N^2+\ell)}(T)\,,
\ee
where $\Delta_{2N}(T)$ is the weight $2N$ modular form defined in \eqref{defDelta2N},
we can easily determine the holomorphic modular form $P_w$ of weight $w$ under $\Gamma_1(N)$
from the knowledge of the GV invariants $\GV_{k_1,k_2\leq k}$ for sufficiently large $k_1$.  
The power of $\Delta_{2N}$ in \eqref{ansatztfk} is the lowest one allowed by \eqref{defrH},
although it can be frequently increased by an integer by lowering the weight 
of $P_w$ correspondingly. In the simplest case of the elliptic fibration over $\IP^2$, with $N=1$, the
corresponding modular forms were already identified up to degree 3 in  \cite[\S D.2]{Alim:2012ss} 
\cite[(5.15)]{Klemm:2012sx}, which we record below:
\be
\begin{split}
\tildef_1 =& \frac{31 E_4^4+113 E_4 E_6^2}{48\eta^{36}}  \\
\tildef_2 =&-\frac{196319 E_4 E_6^5 +755906 E_4^4 E_6^3 + 208991 
E_4^7 E_6}{221184 \eta^{72} } \\
\tildef_3 =&\tfrac{
    E_4 (49789907821 E_4^{12} + 
     1904214859592 E_4^9 E_6^2 + 
     6966210848730 E_4^6 E_6^4+  4311836724416 E_4^3 E_6^6+360744024241 E_6^8)}{557256278016 \eta^{108}} 
     \end{split}
\ee
Similar expressions for genus-one  fibrations with $N$-section over $\IP^2$ are collected in Appendix \S\ref{app_P2}.
In order to extract from these expressions the genus zero GV invariants, it is necessary to invert
the system \eqref{tfkupto5} so as to express $f_k$ in terms of $\tildef_k$, and then use the multicover formula 
\be
f^{\GV}_{k}(T) = \sum_{d|k} \frac{\mu(d)}{d^3} f_{k/d}(dT)
\ee
where $\mu(d)$ is the M\"obius function, equal to $(-1)^k$ if  $d$ is the product of $k$ distinct primes, or zero otherwise. For $N=1$, this produces 
\bea
f_1^{\GV}(T) &=& \tildef_1(T) \nn\\
&=&  q^{-3/2} \left( 3 - 1080\q + 143770 \q^2+204071184 \q^3+21772947555 \q^4+\dots\right) \,,
\nn\\
f_2^{\GV}(T)  &=& \tildef_2(T) - \frac{E_2(T)}{24} \tildef_1^2(T) - \frac18 f_1(2T)
\\
&=& q^{-3} \left(-6+2700 \q-574560
   \q^2+74810520 \q^3-49933059660 \q^4+\dots \right) \,,
   \nn\\
f_3^{\GV}(T)  &=& \tildef_3(T) -\frac{1}{3^3} \tildef_1(3T)
-\frac{E_2(T)}{6} \tildef_1 \tildef_2 + \frac{E_2^2(T) }{2\cdot (24)^2} \tildef_1^3
\nn \\
&=& \q^{-9/2} \left( 27-17280 \q+5051970 \q^2-913383000
   \q^3+224108858700 \q^4+\dots \right) \,. \nn
\eea
For $N\geq 1$, the same expressions hold (but not the $q$-expansions) after rescaling the argument of $E_2$ by a factor of $N$.

\subsection{Modularity at higher genus from the wave-function property}

We now turn to the modular properties of generating series of higher genus GW invariants, and for this purpose study the transformation of the topological string partition function $Z_{\rm top}$ under the relative conifold monodromy.

\subsubsection{Modular properties of $Z_{\rm top}$}
As reviewed in \S\ref{sec_Ztop}, under a general monodromy, the topological string partition function transforms according to the metaplectic representation \eqref{metaruleZtop}.   
Now, since the upper-right block in the matrix $U'$ \eqref{eqUp} is not invertible, the definition \eqref{SABCD} of the kernel $\cS$ cannot apply literally. This is simply due to the fact that the transformation $T\mapsto T/(1+NT)$ does not mix $T$ with the conjugate variables
\begin{align}
    (\cF_0, \cF_T, \cF_{S^\alpha}):=(2\cF^{(0)}-T\partial_T \cF^{(0)}-S^\alpha \partial_{S^\alpha} \cF^{(0)}, \partial_T \cF^{(0)}, \partial_{S^\alpha} \cF^{(0)})\,.
\end{align}
The remedy, already hinted at below~\eqref{SABCD}, is that the Gaussian kernel should be restricted 
to the variables $S^\alpha$, while the transformation of $T$ is enforced by a Dirac delta function
$\delta(T'-T/(1+NT))$. The appropriate kernel $\cS(S^\alpha,S'^\alpha)$ can be found by requiring that the prepotential transforms as in \eqref{Ftrans},
\be
\label{F0Legendre}
(1+NT)^2\, \cF'^{(0)}(T',S'^{\alpha}) = 
\langle \cF^{(0)}(T,S^\alpha) - \cS(S^\alpha,S'^\alpha)\rangle_{S^\alpha}
\ee
where the bracket denotes extremization with respect to $S^\alpha$, and 
$T',S'^\alpha$ on the l.h.s. are the transformed variables \eqref{TStrans}. The factor 
of $1+NT$ on the l.h.s. follows from the relation $X'^0=c X^1+d X^0=X^0(cT+d)$, which 
implies that the topological string coupling $\lambda=1/X^0$ should transform as 
\be
\lambda\mapsto \lambda'= \frac{\lambda}{1+NT}\, .
\ee
Moreover, the value
of $S^\alpha$ extremizing \eqref{F0Legendre} should be consistent with \eqref{TStrans}. 
This determines uniquely the quadratic kernel
\be
\begin{split}
\cS=&(2\pi\I)^3\left(-\frac12(S^\alpha-(1+NT) S'^\alpha)C_{\alpha\beta}(S^\beta-(1+NT) S'^\beta+c^\beta)\right.\\
&\left.-\tfrac{NT}{2} S^\alpha C_{\alpha\beta} c^\beta 
+\tfrac{12N^2-c_2 N}{24} T^2\right)\,.
\end{split}
\ee
According to the prescription\eqref{metaruleZtop} in \S\ref{sec_Ztop},  the topological string partition function 
should thus transform as 
\be
\label{Ztoptrans}
Z_{\rm top}(T',S'^{\alpha}, \lambda') =  \int\, 
e^{\frac{-\cS(S^\alpha,S'^\alpha)}{\lambda^2}}\, 
Z_{\rm top}(T,S^{\alpha}, \lambda) \, \prod_{\alpha=1}^{b_2(S)}  
\frac{\de S^\alpha}{\lambda}  \,.
\ee
Inserting the genus expansion in \eqref{defZtop} and trading $S^\alpha$ for the shifted 
variable $\hatS^\alpha$ defined in \eqref{defShat}, we get
\begin{align}
\exp\left( \sum_{g\geq 0} \lambda'^{2g-2}
F'^{(g)}(T',\hatS'^{\alpha}) \right) = & \nn\\
\frac{(1+NT)^{\frac{\chi_X}{24}-1}}{\lambda^{b_2(B)}} 
\int \prod \de \hatS^\alpha & \exp\left( -\frac{\cS}{\lambda^2}+\sum_{g\geq 0} \lambda^{2g-2}
F^{(g)}(T,\hatS^{\alpha}) \right)
\end{align}
where, on the r.h.s, the $\cO(1/\lambda^2)$ terms are given by 
\be
\begin{split}
\cF^{(0)}-\cS =&(2\pi\I)^3\left(\frac{1+NT}{2} \left[ (\hatS^\alpha-\hatS'^\alpha)
C_{\alpha\beta} (\hatS^\beta-\hatS'^\beta+ c^\beta) + N T \hatS'^\alpha
C_{\alpha\beta} \hatS'^\beta
\right]\right.  \\ & \left.+\frac{\hatkappa T^3}{6} 
-T^2 \frac{12N^2-\hatc_2 N}{24}\right)  + f^{(0)}(T,\hatS^{\alpha}) \,.
\end{split}
\ee
Computing this integral in perturbation theory around the saddle point $\hatS^\alpha=\hatS'^\alpha-\frac12 c^\alpha-\frac{C^{\alpha\beta}}{1+NT}\frac{\partial_{S^\alpha}f^{(0)}}{(2\pi\I)^3}$ determines the transformed free energies $F'^{(g)}(T',S'^{\alpha})$ in terms of the original ones.

In the classical, leading order approximation, we recover by construction the transformation property of $f^{(0)}(T,\hatS^{\alpha})$ discussed in the previous subsection. At one-loop order around the saddle point, we get a factor
\be
\frac{(\lambda^2/(1+NT))^{b_2(B)/2}}{\sqrt{\det\left(C_{\alpha\beta}
+\tfrac{1}{(2\pi\I)^3(1+NT)}\partial^2_{S^\alpha S^\beta} f^{(0)}\right)}}\,,
\ee
up to an overall constant factor,
from the fluctuation determinant. Using the fact that $b_2(B)=\chi_B-2$ for a rational surface, and including the contribution from the prefactor $\lambda^{\frac{\chi_X}{24}-1}$ in $Z_{\rm top.}$, we obtain that the genus one amplitude 
\be
\hatf^{(1)}(T,\hatS^\alpha):=-\frac{\hatc_2}{24}(2\pi \I) T + f^{(1)}(T,\hatS^\alpha)
\ee
must transform as 
\be
\label{eqn:transfhatf1full}
\begin{split}
\hatf^{(1)}(T',\hatS'^\alpha)=&\hatf^{(1)}(T,\hatS^\alpha)+\tfrac{\chi_X-12\chi_B}{24}
\log(1+NT)+5\pi{\rm i}-\frac{\pi{\rm i}}{3}\chi_B\\
&-\frac12\log\det
\left(C_{\alpha\beta}+\tfrac{1}{(2\pi\I)^3(1+NT)}\partial^2_{S^\alpha S^\beta} f^{(0)}\right)\,.
\end{split}
\ee

Here, we have fixed the constant term using the known transformation property~\eqref{eqn:f1vtransform} of the generating function of base degree zero, genus 1
Gromov-Witten invariants, which dominates in the large base limit (where the determinant in \eqref{eqn:transfhatf1full} is effectively equal to one). 
Unfortunately, we are not able at this stage to derive this constant term directly from the wave function behavior (although we suspect it is related to the choice of steepest descent integration contour). Fourier expanding with respect to $\hatS$, one finds that 
the generating series of genus 1 Gromov-Witten invariants at non-zero base degree $H$ transforms as a quasimodular form of weight $0$ and depth equal to the maximal number of terms in decompositions 
$H=\sum_i H_i$ into effective divisor classes, plus one.

\subsubsection{Introducing the modular covariant wave function $Z_{\rm mod}$}

In order to characterize the quasimodular properties more precisely, we find it useful to introduce a different polarization, making use of the improved K\"ahler modulus $\tildeS^\alpha$ introduced in 
\eqref{defSt},  which we recall for the reader's convenience,
\be
\label{defSt1}
\tildeS^{\alpha} := \hatS^\alpha - \frac{E_2(NT)}{12} \frac{C^{\alpha\beta} \partial_{\hatS^\beta} f^{(0)}}{(2\pi \I)^2} \mapsto \tildeS^\alpha+ \delta^\alpha\,.
\ee
Since this definition mixes $\hatS^\alpha$ with the derivative $\partial_{\hatS^\beta} f^{(0)}$, we can view it as a canonical transformation (at fixed $T$), and ask for the wave function in a polarization where $\tildeS^{\alpha}$ is diagonalized: since the latter transforms simply by a constant shift under monodromy, we expect  that the new wave function will have much simpler transformation properties.  With this motivation in mind, we define  the topological wave function in `modular polarization' as the transform
\be
\label{Zholdefg}
Z_{\rm mod}(T,\tildeS^\alpha,\lambda) :=
\frac{1}{\left(\lambda \sqrt{E_2(NT)}\right)^{b_2(B)}} \int 
Z_{\rm top}(T,\hatS^\alpha,\lambda) \, e^{\frac{\cH(\tildeS^\alpha, \hatS^\alpha) }{\lambda^2}} \prod_{\alpha} \de \hatS^{\alpha}
\ee
where $\cH$ is the quadratic kernel
\be
\label{defcH}
\cH(\tildeS^\alpha, \hatS^\alpha)=-(2\pi\I)^3\frac{NT}{2} \hatS^\alpha C_{\alpha\beta} \hatS^\beta
-(2\pi\I)^2\frac{6}{E_2(NT)}(\tildeS^\alpha-\hatS^\alpha)C_{\alpha\beta}(\tildeS^\beta-\hatS^\beta)\,.
\ee
The first term in \eqref{defcH} is chosen to cancel the classical term in the prepotential \eqref{eqF0},
while the second term combines with the instanton  contribution to the prepotential $f^{(0)}$ such that the integral has a saddle point precisely when \eqref{defSt1} is obeyed. The overall power of $E_2(NT)$ in \eqref{Zholdefg} is chosen such the inverse formula takes the same form,
\be
\label{Zholdefgi}
Z_{\rm top}(T,\hatS^\alpha,\lambda) = 
\frac{1}{\left(\lambda \sqrt{E_2(NT)}\right)^{b_2(B)}} \int 
Z_{\rm mod}(T,\tildeS^\alpha,\lambda) 
 \, e^{-\frac{\cH(\tildeS^\alpha, \hatS^\alpha) }{\lambda^2}} 
 \prod_{\alpha} \de \tildeS^{\alpha}\,.
\ee
Substituting \eqref{defZtop}, we have, more explicitly,
\be
\label{Zholdefge}
\begin{split}
&Z_{\rm mod}(T,\tildeS^\alpha,\lambda)\\
=& \tfrac{\lambda^{\frac{\chi_X}{24}-b_2(B)-1}}{E_2(NT)^{\frac12b_2(B)}}
 \int   \exp'\left(
 \frac{ \hatkappa T^3 }{6\lambda^2 }- \frac{3 C_{\alpha\beta}(\tildeS^\alpha-\hatS^\alpha)(\tildeS^\beta-\hatS^\beta)}{\pi\I\lambda^2 E_2(NT)}-\frac{\frac{\widehat{c}_2}{24}T+\frac12c_\alpha \hatS^\alpha}{(2\pi\I)^2} \right) \\
  & \times \exp\left( \sum_{g\geq 0} f^{(g)}(T,\hatS^\alpha) \lambda^{2g-2}\right) 
\prod_\alpha \de \hatS^\alpha \,,
\end{split}
\ee
where we use $\exp'(x)=\exp\left((2\pi\I)^3x\right)$.
To obtain the transformation property of $Z_{\rm mod}$, we need to invert \eqref{Zholdefg}, apply the transformation \eqref{Ztoptrans} and then apply \eqref{Zholdefg} again. 
The result is that  $Z_{\rm mod}(\tildeS,T,\lambda)$  
transforms into
\bea
\label{Zholtransg}
Z_{\rm mod}(T',\tildeS',\lambda')
&=&\int \frac{ (1+NT)^{b_2(B)}Z_{\rm mod}(T,\tildeS,\lambda)}{\left(\lambda^3 \sqrt{E_2(NT)E_2(NT')}\right)^{b_2(B)}}\, 
 \exp'\left(  \frac{NT \hatS^2}{2\lambda^2} + \frac{3(\tildeS-\hatS)^2}{\pi\I\lambda^2 E_2(NT)} \right)  
\nn\\
&\times & 
\exp'\left( \tfrac{(\hatS^\alpha-(1+N T)\hatS'^\alpha)C_{\alpha\beta}(\hatS^\beta-(1+N T)\hatS'^\beta+c^\beta)}{2\lambda^2}
+\tfrac{NT}{2\lambda^2} c_\alpha \hatS^\alpha 
- \tfrac{\left(\frac{N^2}{2}-\frac{c_2 N}{24}\right)T^2}{\lambda^2} \right)
\nn\\
&\times &
\exp'\left( -\frac{NT' \hatS'^2}{2\lambda'^2} - \frac{3(\tildeS'-\hatS')^2}{\pi\I\lambda'^2 E_2(NT')} \right) 
\,\prod_\alpha  \de \hatS^\alpha\,  \de \hatS'^\alpha \, \de \tildeS^\alpha\,.
\eea
Using that $T'=T/(NT+1)$, $\lambda'=\lambda/(1+NT)$ and the transformation rule \eqref{E2Ntrans} of $E_2(NT)$, 
one finds that the quadratic form in $(\hatS^\alpha,\hatS'^\alpha)$ appearing in the product of exponentials 
has half-maximal rank, with null space spanned by $( E_2(NT), E_2(NT)+\frac{1}{2\pi\I}\tfrac{12}{1+NT})x^\alpha$ for any $x^\alpha\in\IR^{b_2(B)}$. Decomposing
\be
(\hatS^\alpha, \hatS'^\alpha)=\left( (1+NT)  E_2(NT) x^\alpha+ y^\alpha, \left(\tfrac{12}{2\pi\I}+ (1+NT) E_2(NT)\right)x^\alpha+y^\alpha\right)
\ee
the integral over $x^\alpha$ leads to a delta function $\delta[ (1+NT) (\tildeS'^\alpha-\tildeS^\alpha-\frac12c^\alpha)/\lambda^2 ]$. Setting  $\tildeS'^\alpha=\tildeS^\alpha+\frac12 c^\alpha$, the remaining Gaussian integral over $y^\alpha$
is peaked at
\begin{align}
    y^\alpha=\tildeS^\alpha -\frac{(1+NT)\left(2\pi\I E_2(NT)+12\right)}{24}c^\alpha\,,    
\end{align}
with Hessian $(1+NT) C_{\alpha\beta}/ \left[ \lambda^2 E_2(NT) E_2(\tfrac{NT}{1+NT})\right]$.
Collecting the fluctuation determinant and Jacobian in the delta function, we get
\be
\label{tZholtrans}
\begin{split}
\frac{Z_{\rm mod}\left(\tfrac{T}{1+NT},\tildeS^\alpha+\tfrac12 c^\alpha,  \tfrac{\lambda}{1+NT}\right)}
{Z_{\rm mod}(T,\tildeS^\alpha,  \lambda)}
=& 
\frac{\exp'\left(  - \frac{(12 N^2-c_2 N)T^2}{24 \lambda^2}
+\frac{\chi_B-12}{8\lambda^2}(1+NT)\right)}{(1+NT)^{b_2(B)/2}} \,,
\end{split}
\ee
where we have also used that $\chi_B=12-a^\alpha a_\alpha$ and assumed that $a^\alpha=c^\alpha$. More generally, under a $\Gamma_1(N)$ monodromy of the form \eqref{eqUpgen}, a similar computation shows that $Z_{\rm mod}$ transforms as 
\be
\label{tZholtransg}
\begin{split}
Z_{\rm mod}\left(\tfrac{aT+b}{cT+d},\tildeS^\alpha+\delta^\alpha,  \tfrac{\lambda}{cT+d}\right)
=& 
(cT+d)^{\frac{\chi_X-12\chi_B}{24}} \exp'\left(  \tfrac{x T^2+y T+z}{ \lambda^2}\right) Z_{\rm mod}(T, \tildeS^\alpha,  \lambda) \, .
\end{split}
\ee

As a result, if we compute \eqref{Zholdefge} in perturbation theory around the saddle point and define new topological amplitudes $\tildef^{(g)}(T,\tildeS^\alpha):=\sum_{k_\alpha}
\tildef^{(g)}(T) e^{2\pi\I k_\alpha \tildeS^\alpha}$ via
\be
Z_{\rm mod}(T,\tildeS^\alpha,\lambda)= \lambda^{\frac{\chi_X}{24}-1}
 \exp\left( 
 \frac{(2\pi\I)^3 \hatkappa T^3 }{6\lambda^2 }+ \pi\I c_\alpha \tildeS^\alpha 
+ \sum_{g\geq 0}  \tildef^{(g)}(T,\tildeS^\alpha) \lambda^{2g-2}\right) 
\ee
it follows that each Fourier coefficient $\tildef^{(g)}_{k_\alpha}(T)$ will be a modular form of weight $2-2g$ under $\Gamma_1(N)$, with no anomaly except for $g=0,1$. For $g=0$, this definition reproduces  \eqref{deff0t}, while for $g=1$, it yields
\be
\label{deff1t}
\begin{split}
\tildef^{(1)}(T,\tildeS^\alpha) =& \hatf^{(1)}(T,\hatS^\alpha) 
 -\frac12\log\det
\left(C_{\alpha\beta}-\tfrac{E_2(NT)}{12(2\pi\I)^2}\partial^2_{S^\alpha S^\beta} f^{(0)}\right)\\
&
- \tfrac{E_2(NT)}{12(2\pi\I)^2} C^{\alpha\beta} c_\alpha \partial_{\hatS^\beta} f^{(0)}\,.
\end{split}
\ee
Taking derivatives with respect to $E_2(NT)$ on both sides of \eqref{deff1t} keeping $\hatS^\alpha$ fixed, similar to \eqref{deff0tE2}, and using the genus 0 result \eqref{anomf0}, we obtain
\be
\label{anomf1}
\begin{split}
 \partial_{E_2(NT)} f^{(1)}  = & 
 - \frac{1}{12(2\pi\I)^2} C^{\alpha\beta} \partial_{\hatS^\alpha}  f^{(0)} \partial_{\hatS^\beta} f^{(1)} \\ &-\frac{1}{24(2\pi\I)^2} C^{\alpha\beta} \partial^2_{\hatS^\alpha\hatS^\beta}  f^{(0)}
 +\frac{1}{48\pi\I} c_\alpha C^{\alpha\beta} \partial_{\hatS^\beta}  f^{(0)}
 \end{split}
\ee
or in terms of the Fourier coefficients,
\be  
   \partial_{E_2(NT)} f_H^{(1)}(T)=
   -\frac{1}{24} \sum_{H=H_1+H_2} (H_1\cap H_2)
    f_{H_1}^{(0)} f_{H_2}^{(1)} - \frac{1}{24} H \cap(H-c_1(B))  f_{H}^{(0)} 
    \label{hanom1}
\ee
confirming the quasimodularity of the generating series $f^{(1)}_{H}$.

More generally,
the series $f^{(g)}_H$ are quasimodular forms of weight $2g-2$ and depth equal to $g$ plus the maximal number of elements in possible decompositions of $H=\sum_i H_i$ into  effective divisor classes, while
the series $\tildef^{(g)}_{H}$ are actual modular forms of weight $2g-2$, equal to the depth zero part of $f^{(g)}_{H}$. See Appendix \S\ref{app_P2} for many explicit examples.
We leave it as an exercise to the reader to derive the modular anomaly equations at genus $g>1$ using the same techniques.

\subsubsection{Jacobi properties of $Z_{H}(T,\lambda)$}

We are now ready to derive the Jacobi property of the normalized generating series of PT invariants defined in 
\eqref{defZbeta}. Substituting \eqref{ZtopZ0} in \eqref{Zholdefg} and 
integrating term by term, we get
\be
\begin{split}
Z_{\rm mod}(T,\tildeS^\alpha,  \lambda) =&  Z_{{\rm mod},0}(T, \lambda) \sum_{k_\alpha\geq 0} Z_{k_\alpha}(T,\lambda)e^{2\pi\I\left(k_\alpha-\frac12c_\alpha\right)\tildeS^\alpha}\\
& \times 
  \, \exp\left( \lambda^2
\frac{E_2(NT)}{24}k_\alpha C^{\alpha\beta} \left(k_\beta-\frac12c_\beta\right)
 \right) \,,
  \end{split}
\ee
where
\be
\begin{split}
 Z_{{\rm mod},0}(T, \lambda) =& \lambda^{\frac{\chi_X}{24}-1}  \exp\left( 
 \frac{(2\pi\I)^3\hatkappa T^3 }{6\lambda^2 }-\frac{\pi\I}{12}\widehat{c}_2T+\frac{12-\chi_B}{96} \lambda^2 E_2(NT)\right)\\
 &\times\,\exp\left(
  \sum_{g\geq 0}  f_0^{(g)}(T) \lambda^{2g-2}\right)\,.
  \end{split}
\ee
Note that we have once more used that $b_2(B)=\chi_B-2$, $\chi_B=12-a^\alpha a_\alpha$ and the assumption that $a^\alpha=c^\alpha$.
Eq. \eqref{tZholtrans} then implies that 
$Z_{H}(T,\lambda)$ transforms as a Jacobi form of weight 0, index $\frac12 H\cap (H- c_1(B))$ under $T\mapsto T/(1+NT)$. 
As explained below \eqref{Jacmod}, the quasiperiodicity 
property \eqref{Jacqperiod} then follows from the manifest periodicity under $\check\lambda\mapsto\check\lambda+1$. 

\section{Modularity of DT invariants from monodromy}
\label{sec_DTmodEll}

In this section, we discuss the implications of invariance under the relative conifold monodromy on various types of Donaldson-Thomas invariants. As mentioned in \S\ref{sec_DTreview}, DT invariants are invariant under auto-equivalences of the derived category $g\in \Aut\cC$, provided $g$ acts both on the charge $\gamma$ and on the stability condition $\sigma$, see \eqref{monOm}. 
In general however, $g\cdot \sigma$ and $\sigma$ may not lie in the same chamber, so the invariants 
$\Omega_\sigma(\gamma)$ and $\Omega_\sigma(\gamma\cdot g)$ need not be equal. Here we shall assume that no wall-crossing takes place between $\sigma$ and $g\cdot \sigma$, and see what kind of property this implies for the generating series of DT invariants. 

The auto-equivalence of interest is the Fourier-Mukai transformation $g_U$ with respect to the ideal sheaf of the relative diagonal discussed in \S\ref{sec_relcon}.
It acts on the charge row vector $\gamma=(p^0,p^e,p^\alpha;q_e,q_\alpha,q_0)$ as 
$\gamma\mapsto  \gamma\cdot {\rm U}$ with
\be
{\rm U}= \Sigma  U'^{-1} \Sigma^{-1} = 
\begin{pmatrix}
1 & 0 & \frac12c^\beta &- \frac{N(12-\chi_B)}{4} & 0_\beta & \frac{12-\chi_B}{4}  \\
N & 1 & \frac{N}{2} c^\beta &  -\frac{c_2 N}{12} + \frac{(\chi_B-8)N^2}{4} & -\frac{N}{2} c_\beta & \frac{N(12-\chi_B)}{4} \\
0_\alpha & 0_\alpha & \delta^\alpha_\beta & -\frac{N}{2} c_\alpha & 0^\alpha_\beta & 0_\alpha  \\
0&0&0_\beta & 1 & 0_\beta & 0  \\
0_\alpha & 0_\alpha & -C_{\alpha\beta} & \frac{N}{2}c^\alpha  & \delta^{\alpha}_{\beta} & 
- \frac12 c^\alpha \\
0 & 0 & 0_\beta & -N  &  0_\beta & 1
\end{pmatrix}
\ee
where $\Sigma$ is the matrix appearing in the central charge $Z=\gamma \Sigma V$, 
\be
\Sigma = \begin{pmatrix}
0 & 0 & 0_\beta & -1 & 0 & 0_\beta  \\
0 & 0 & 0_\beta & 0 & -1 & 0_\beta  \\
0^\alpha & 0^\alpha & 0_{\alpha\beta} & 0 & 0 & -\delta^{\alpha}_\beta\\
0 & 1 & 0_\beta & 0 & 0 & 0_\beta \\
0^\alpha & 0^\alpha & \delta^{\alpha}_{\beta} & 0^\alpha & 0^\alpha & 0^{\alpha}_\beta\\
1 & 0 & 0_\beta & 0 & 0 & 0_\beta
\end{pmatrix}\,.
\ee
In view of the action \eqref{TStrans} on K\"ahler moduli, the large volume (LV) limit is mapped to a phase $(T\to \frac{1}{N},S^\alpha\to \I\infty)$ where the base $B$ is still very large but some fibral curves have vanishing volume and fractional $B$-field $1/N$.~\footnote{Note that the limit $(T\to 0,S^\alpha\to \I\infty)$ is instead expected to correspond to the large volume limit of the relative Jacobian fibration that, for $N>1$, carries a flat but topologically non-trivial B-field.}

\subsection{Elliptic property of PT invariants}

First, we briefly discuss how the invariance of PT invariants under the relative conifold monodromy
implies the quasi-periodicity property \eqref{Jacqperiod}, i.e. the `elliptic' part of the Jacobi properties.
Recall from \S\ref{sec:BPS} that the PT invariant $\PT(k_e,k_\alpha,m)$ counts stable 
BPS states with charge 
$\gamma= (-1,0,0^\alpha;-k_e+\frac{c_2}{24},-k_\alpha+\frac12 c^\alpha,-m)$ for large volume and large B-field of suitable sign. Its image under $g_U$ is 
\be
\gamma\cdot {\rm U}  = \left(-1, 0, C^{\alpha\beta}k_\beta-c^\alpha; -k_e+ N m - \tfrac{N(12-\chi_B-k_\alpha c^\alpha)}{2}, -m - \tfrac{12-\chi_B-k_\alpha c^\alpha}{2}
\right)\,.
\ee
Further acting by a large volume monodromy \eqref{lvm} with $(\epsilon^e,\epsilon^\alpha)=(0,
C^{\alpha\beta}k_\beta-c^\alpha)$, so as to set the D4-brane charge to zero, we get a PT charge vector $\gamma'= (-1,0,0^\alpha;-k'_e+\frac{c_2}{24},-k'_\alpha+\frac12 c^\alpha,-m')$ with the
same horizontal D2-brane charge $k'_\alpha=k_\alpha$ but with new vertical 
D2 and D0 brane charges 
\be
 \quad k'_e = k_e 
+ \frac{N}{2}    k_\alpha C^{\alpha\beta} (k_\beta - c_\beta) - m N, \quad 
m'= m - k_\alpha C^{\alpha\beta} ( k_\beta - c_\beta ) \,.
\ee
For $N=1$, this reduces to the result in  \cite[Lemma 9]{Oberdieck:2016nvt} for the action of the auto-equivalence $\Phi_H$ on the Chern character. As explained in \textit{loc.cit.}, this autoequivalence maps stable pairs to so-called $\pi$-stable pairs, but wall-crossing
is trivial between these two notions of stability. Assuming that this statement continues to hold for genus one fibrations, we conclude that  
\be
\label{PTperiod}
\PT( k_e, H, m) = 
\PT\left( k_e+  N (h_H -1)  - m N, H, m - 2( h_H -1 ) \right)\,,
\ee
where $h_H= 1 + \frac12 k_\alpha C^{\alpha\beta} ( k_\beta - c_\beta )$ is the arithmetic genus of the curve class $H=k_\alpha \check{D}^\alpha$ on the base. As explained in  \cite[Corollary 1]{Oberdieck:2016nvt}, this implies the quasi-periodicity property \eqref{Jacqperiod} of the generating series of PT invariants $Z_\beta$. It is worth stressing that \eqref{PTperiod} is only supposed to be valid for any non-trivial nef divisor class $H$ on $B$. In particular it does not apply for $H=0$, due to modular anomalies in the base degree zero sector. 

\subsection{Modularity of D4-D2-D0 invariants}

We now turn to the action of the relative conifold monodromy on D2-D0 brane BPS bound states, with charge 
$\gamma=(0,0,0^\alpha;k_e,k_\alpha,m)$, where $k_e\geq 0$ and $H=k_\alpha \check D^\alpha$ is an ample divisor on the base $B$. As recalled in \S\ref{sec:BPS}, these bound states are counted at large volume 
by the genus zero GV invariant $\GV^{(0)}_{k_e,k_\alpha}$,  independent of the D0-brane charge $m$. Under the action of $g_U$, the charge vector becomes
\be
\gamma \cdot {\rm U} = \left( 0, 0, -C^{\alpha\beta}k_\alpha; -Nm+k_e+\frac{N}{2} c^\alpha k_\alpha, 
k_\alpha, m-\frac12 k_\alpha c^\alpha \right)\,.
\ee
After applying an homological shift $\cE\mapsto \cE[1]$ 
such that $\gamma\mapsto -\gamma$, we 
arrive at a D4-D2-D0 bound state with divisor class $p^\alpha D_\alpha$, 
D2-brane flux $(q_e,q_m)$ and D0-brane charge $q_0$ given by
\be
\label{gamd4d2d0}
 p^\alpha=C^{\alpha\beta}k_\alpha, \ q_e=Nm-k_e-\frac{N}{2} c^\alpha k_\alpha, \ 
q_\alpha=-k_\alpha, \ q_0=\frac12 k_\alpha c^\alpha-m\,.
\ee
In particular, the D4-brane charge has no component along the section $D_e$ ($p^e=0$).
In order to determine the invariants $(\mu, \hat q_0)$ introduced in \S\ref{sec_DTreview}, 
we evaluate the quadratic form $\kappa_{ab}:=\kappa_{abc} p^c$:
\be
\kappa_{ab} = \begin{pmatrix} \ell_\alpha C^{\alpha\beta} k_\beta & N k_\beta \\
N k_\alpha & 0_{\alpha\beta} \end{pmatrix} \,.
\ee
Note that this matrix has rank 2, signature $(1,1)$, and that the vector $p^\alpha$ is
isotropic, $p^a \kappa_{ab} p^b=\kappa_{abc} p^a p^p p^c=0$. This is a consequence of the divisor class $D=p^\alpha D_\alpha$ being not ample.
However, $D$ is nef and we therefore expect the modular properties to be as reviewed in~\S\ref{sec_DTreview}, assuming that there is no wall-crossing between the chamber $g_U \cdot LV$ and the large volume attractor chamber.

If $b_2(B)>1$, one needs to follow the prescription summarized below \eqref{Lam-p} and restrict $\kappa_{ab}$ to the orthogonal complement of the null space spanned by $(0,\lambda_s^\alpha)$ where $\lambda_s^\alpha, s=1,\dots, b_2(B)-1$ is a basis of vectors such that $\lambda_s^\alpha k_\alpha=0$, and take its inverse.
In any case, the result is 
\be
\label{evq0h}
\hat q_0 = -\frac{k_e}{N} + \frac{\ell_\alpha C^{\alpha\beta} k_\beta}{2N^2}\,.
\ee
In particular, it is independent of the original D0-brane charge $m$. Indeed, the action of a large volume monodromy \eqref{lvm} with $\epsilon^e=0$ shifts $m\mapsto m - k_\alpha \eps^\alpha$ without affecting the other charges. Moreover, \eqref{evq0h}  satisfies the 
Bogomolov-Gieseker bound $\hat q_0\leq \frac{\chi(\cD_p)}{24}= \frac{c_\alpha C^{\alpha\beta} k_\beta}{2}$ for any $k_e\geq 0$, so long as the condition 
\be
(\ell_\alpha-  N^2 c_\alpha) C^{\alpha\beta} k_\alpha\leq 0
\ee 
is obeyed.
This inequality implies that $\eta(NT)^{12c_\alpha k^\alpha}f_{k_\alpha}^{(0)}(T)$ is regular at $T\rightarrow \I\infty$.
It is satisfied in all of our examples but we are not aware of a proof.

The discrete fluxes $(\mu_e, \mu_\alpha)$ are in turn given by
\be
\label{mueq}
\mu_e = Nm-k_e  - \frac{N}{2} k_\alpha( k^\alpha+c^\alpha) - \epsilon^e k^\alpha \ell_\alpha - N \epsilon^\alpha k_\alpha\ ,\quad 
\mu_\alpha = -k_\alpha - N \epsilon^e k_\alpha
\ee
where $\epsilon^e,\epsilon^\alpha$ are arbitrary integers. Modulo these spectral flow parameters, 
and restricting to the space $\lambda_s^\alpha \mu_\alpha=0$, the fluxes $(\mu_e,\mu_\alpha)$ 
are a priori valued in the discriminant group $\IZ_{kN} \times \IZ_{kN}$, of cardinality $(kN)^2$, where 
$k$ is the greatest common divisor of the $k_\alpha$'s.
Note that $ k_\alpha( k^\alpha+c^\alpha)=2k_\alpha k^\alpha-k_\alpha(k^\alpha-c^\alpha)$.
Since $k_\alpha(k^\alpha-c^\alpha)$ is the Euler characteristic of a curve in the class $k_\alpha$, which is even, we conclude that $\frac{N}{2}k_\alpha( k^\alpha+c^\alpha)$ is an integer multiple of $N$, as long as the class can be represented by a smooth curve in $B$.

The assumption that there is no wall-crossing between $g_U \cdot LV$ and the large volume attractor chamber implies the equality of 
integer DT invariants
\be
\label{GWeqMSW}
\Omega_{\infty}(0,0,0; k_e,k_\alpha, m) = \GV^{(0)}_{k_e,k_\alpha} = \Omega_{0,k^\alpha;\mu_e,\mu_\alpha}(\hat q_0)\,.
\ee
In particular, the fact that the D2-D0 brane index is independent of the D0-brane charge (or holomorphic Euler characteristic) $m$ implies that the corresponding D4-D2-D0 indices on the r.h.s. of \eqref{GWeqMSW} are invariant under $\mu_e\mapsto \mu_e+N$, rather than just  $\mu_e\mapsto \mu_e+kN$. While the origin of this invariance is obscure on the D4-D2-D0 side, its effect is to cut down the number of independent fluxes to 
$k N^2$, further halved by the symmetry $\mu\mapsto -\mu$. For a fixed value of $\mu_\alpha$, there are effectively $N$ components, and we shall later argue that only the coset $\mu_\alpha=-k_\alpha$ is populated.

\medskip 

Before discussing the implications of the relation \eqref{GWeqMSW} at the level of generating series of GW and DT invariants, we remark that for D2-D0 branes wrapping the genus-one fiber (i.e. for $\gamma=(0,0,0^\alpha;k_e,0,m)$), the image under the relative conifold monodromy $g_U$ has vanishing D4-brane charge,
\be
\label{hUonverticalD2}
\gamma \cdot {\rm U} = (0,0,0; k_e-Nm,0,m)
\ee
and the corresponding BPS states are counted at large volume by the genus 0 GV invariant 
$\GV_{k_e-Nm,0}^{(0)}$. Assuming again the absence of wall-crossing, and using the $m$-independence of the DT invariant counting the original BPS states of charge $\gamma$, we conclude that  $\GV_{k_e,0}^{(0)}$ should be invariant under $k_e\mapsto k_e+N$. This gives a heuristic derivation of the periodicity property stated above \eqref{GV0sumrule}, although it does not say anything about higher genus base degree zero GV invariants. 

\subsubsection{Relating generating series of GW and rank 0 DT invariants}
Now, we explore the relation between the generating series $f_{k^\alpha}(T)$ 
of genus 0 GW invariants and  the generating series $h_{p;\mu}(\tau)$ of D4-D2-D0 invariants, 
which we recall for convenience, 
\be
\label{genGW0}
f_{k_\alpha}(T) = \sum_{k_e\ge 0} \GW_{k_e,k_\alpha}^{(0)} q_T^{k_e -\frac{\ell_\alpha k^\alpha}{2N}}
\, ,\quad 
h_{0,k^\alpha;\mu_e,\mu_\alpha}(\tau) 
= \sum_{\hat q_0 \leq \frac{\chi(\cD_p)}{24}} \bar\Omega_{0,k^\alpha;\mu_e,\mu_\alpha}(\hat q_0)\,
q_\tau^{-\hat q_0}\,.
\ee
Note that we again drop the genus superscript of $f_{k_\alpha}^{(0)}(T)$.
We further introduce the generating series of  \textit{integer} D4-D2-D0 invariants, 
\be
\label{defhint0}
h^{\rm int}_{0,k^\alpha;\mu_e,\mu_\alpha}(\tau) 
:= \sum_{\hat q_0 \leq \frac{\chi(\cD_p)}{24}} \Omega_{0,k^\alpha;\mu_e,\mu_\alpha}(\hat q_0)\,
q_\tau^{-\hat q_0}\,.
\ee

We will first show that
\begin{align}
    f_{k_\alpha}^\GV(T) =& \sum\limits_{\mu_e=0}^{N-1}h_{0,k^\alpha;\mu_e,k_\alpha}^{\rm int}\left(NT\right)\,.
    \label{fkhconjint}
\end{align}
To this end, note that for $p^0=p^e=0$, $p^\alpha=k^\alpha$ and $\mu_\alpha=-k_\alpha$, we have
\begin{align}
    \hat{q}_0=q_0+\frac12k_\alpha k^\alpha-\left(\frac{\mu_e}{N}-\frac{k_\alpha \ell^\alpha}{2N^2}\right)\,,
\end{align}
and since $q_0\in\mathbb{Z}-\frac{1}{2}c_\alpha k^\alpha$, and $k^\alpha(k_\alpha-c_\alpha)\in2\mathbb{Z}$, we can write
\begin{align}
    h^{\rm int}_{0,k^\alpha;-\mu_e,-k_\alpha}(\tau) 
= \sum_{\stackrel{n\in\mathbb{Z}}{n\geq n_{\rm m}}} \Omega_{0,k^\alpha;\mu_e,k_\alpha}\left(-n-\frac{\mu_e}{N}+\frac{k_\alpha \ell^\alpha}{2N^2}\right)\,
q_\tau^{n+\frac{\mu_e}{N}-\frac{k_\alpha \ell^\alpha}{2N^2}}\,,
\end{align}
where $n_{\rm m}:=\mu_e/N-k_\alpha(c^\alpha-\ell^\alpha/N^2)/2$.
Using~\eqref{GWeqMSW} we then observe that
\begin{align}
\begin{split}
    &\sum_{k_e\ge 0} \GV_{k_e,k_\alpha}^{(0)} q_T^{k_e -\frac{\ell_\alpha k^\alpha}{2N}}
		=\sum_{k_e\ge 0}\Omega_{0,k^\alpha;-k_e,-k_\alpha}\left(-\tfrac{k_e}{N}+\tfrac{\ell_\alpha}{2N^2} k^\alpha\right)q_T^{k_e-\frac{\ell_\alpha k^\alpha}{2N}}\\
        =&\sum\limits_{\mu_e=0}^{N-1}\sum\limits_{n\in \mathbb{Z}}\Omega_{0,k^\alpha;-\mu_e,-k_\alpha}\left(-n-\tfrac{\mu_e}{N}+\tfrac{\ell_\alpha}{2N^2} k^\alpha\right)q_T^{Nn+\mu_e-\frac{\ell_\alpha k^\alpha}{2N}}\,.
    \end{split}
\end{align}
The result~\eqref{fkhconjint} follows after taking into account the symmetry under $\mu\rightarrow -\mu$.

In order to obtain the relation between genus 0 GW invariants and rational DT invariants, 
 we use the multicover formulae, which  crucially involve different powers of $d$ in the sum over divisors:
\be
\GW^{(0)}_{k_e,k_\alpha}= \sum_{d|(k_e,k_\alpha)} \frac{1}{d^3} 
\GV^{(0)}_{k_e/d,k_\alpha/d} \ ,\quad
\bar\Omega_t(\gamma) = 
\sum_{d|\gamma} \frac{1}{d^2} \Omega_t(\gamma/d)
\ee
or equivalently
\be
\GV^{(0)}_{k_e,k_\alpha}= \sum_{d|(k_e,k_\alpha)} \frac{\mu(d)}{d^3} 
\GW^{(0)}_{k_e/d,k_\alpha/d} \ ,\quad
\Omega_t(\gamma) = 
\sum_{d|\gamma} \frac{\mu(d)}{d^2} \bar\Omega_t(\gamma/d)
\ee
where $\mu(d)$ is the M\"obius function. 
Morever, it is important to stress that the statement of $m$-independence holds only for the integer DT invariants in \eqref{GWeqMSW}, rather than for their rational counterparts. Nonetheless, we claim that
the relation between generating series of GW and rational DT invariants is almost identical to \eqref{fkhconjint}, namely
\be
\label{fkhconj}
k\, f_{k_\alpha}(T)= \sum_{\mu_e=0}^{kN-1} h_{0,k^\alpha,\mu_e,k_\alpha}(NT)\,,
\ee
where we recall that $k$ is the greatest common divisor of the $k_\alpha$'s.

To show this, we first observe that for a function
\begin{align}
    f(T)=\sum\limits_{n\ge 0}a_ne^{2\pi\I T(n+\phi)}\,,\quad a_n\in\mathbb{C}\,,\,\, \phi\in\mathbb{Q}\,,
\end{align}
and any $N\in\mathbb{N}_{>0}$, $k\in\{0,\ldots,N-1\}$, the restricted sum over $n=k\mod N$ is given by 
\begin{align}
    \frac{1}{N}\sum\limits_{b=0}^{N-1}e^{-\frac{2\pi\I\mu}{N}(k+\phi)}f\left(T+\frac{b}{N}\right)=\sum\limits_{n\ge 0}a_{Nn+k}e^{2\pi\I T(Nn+k+\phi)}\,.
\end{align}
We can therefore invert the relation 
\eqref{fkhconjint} into
\be
h^{\rm int}_{0,k^\alpha;\mu,k_\alpha}(NT) = \frac{1}{N} \sum_{b=0}^{N-1} e^{- \frac{2\pi\I b}{N}
( \mu-\frac{\ell_\alpha k^\alpha}{2N})} f_{k_\alpha}^{\GV}\left(T+\tfrac{b}{N}\right)
\ee
and insert it into the definition of $h_{p;\mu}$, 
\bea
h_{0,k^\alpha;\mu,k_\alpha}(NT) &=& \sum_{d|(\mu,k)} \frac{1}{d^2} 
h^{\rm int}_{0,k^\alpha/d;\mu/d,k_\alpha/d}(dNT) \\
&=& \sum_{d|(\mu,k^\alpha)} \frac{1}{Nd^2} 
\sum_{b=0}^{N-1} e^{-\frac{2\pi\I b}{dN}
( \mu-\frac{\ell_\alpha k^\alpha}{2N})} 
\sum_{e|(k/d)} \frac{\mu(e)}{e^3} f_{k_\alpha/(d e)}\left( d e T+\tfrac{b}{N}\right) \nn\\
&=& \sum_{m|k} \frac{1}{m^3} \left[
\sum_{d|(\mu,m)} \frac{d}{N}  \mu(m/d)  \sum_{b=0}^{N-1} e^{-\frac{2\pi\I b}{dN}( \mu-\frac{\ell_\alpha k^\alpha}{2N})} \right]
 f_{k_\alpha/m}\left( m T+\tfrac{b}{N}\right) \,. \nn
\eea
Summing over $\mu_e$ and using the identities
\be
\forall d, \sum_{\mu=0\dots kN-1 \atop d|\mu }  e^{-\frac{2\pi\I b \mu }{dN}}
= \frac{kN}{d} \delta^{(N)}_{b} \ ;  \quad \forall m, \sum_{d|m} \mu(m/d)=\delta_{m,1}
\ee
we arrive exactly at \eqref{fkhconj}. 

\subsubsection{Comparing holomorphic anomaly equations}
Let us now compare the holomorphic anomaly equations satisfied by generating series of GW and rank 0 DT invariants. For simplicity we continue to restrict to genus one fibrations over $\IP^2$, such that the index $\alpha$ takes only one value, which we omit. 

We assume that the divisor class $p^a=(0,k)$ lies on the boundary of the effective cone, such that the only possible decompositions $p=\sum_{i=1}^n p_i$ are of the form $p_i=(0,k_i)$. As a result, only two-term decompositions contribute, $k=k_1+k_2$. We denote $k_0=\gcd(k_1,k_2)$. 
According to the prescription  \eqref{Sergeyhae}, the holomorphic anomaly for rank 0 DT invariants is given by\footnote{For fibrations on $\IP^2$, the lattice $\Lambda_\perp$ is trivial
and one does not need any non-trivial glue vectors, so $n_g=1, g_A=0$.}
\be
\label{haeDTP2}
\begin{split}
\partial_{\bar\tau} \whh_{0,k;\mu_e,\mu} =&  \frac{1}{16\pi\I \tau_2^2}
\sum_{k_1+k_2=k} k_0 
\sum_ { \mu_{1}\in\IZ_{k_1 N} \atop  \mu_{2}\in\IZ_{k_2 N} }
\delta^{(k_0 N)}_{\mu-\mu_{1}-\mu_{2} } \, 
\delta^{(N k k_1 k_2/k_0^2)}_{\mu_{12}}  \\ & 
\times 
\sum_{\mu_{1,e}\in\IZ_{k_1 N} \atop \mu_{2,e}\in\IZ_{k_2 N}}
\delta^{(k_0 N)}_{ \mu_ e-\mu_{1,e}-\mu_{2,e} -\frac{\ell}{N} (\mu-\mu_{1}-\mu_{2} )}\, 
\whh_{0,k_1;\mu_{1,e},\mu_{1}}
\whh_{0,k_2;\mu_{2,e},\mu_{2}}
\end{split} 
\ee
where 
\be
k_0 \mu_{12} = k_2 \mu_{1}- k_1 \mu_{2} + N k_1 k_2 (\rho_1^e- \rho_2^e)
\label{eqn:mu12def}
\ee
with $(\rho^e_1,\rho^e_2)$ any integral solutions to 
\be
 \mu-\mu_{1}-\mu_{2} = N (k_1 \rho_1^e + k_2 \rho_2^e)\,.
 \label{eqn:rhocond}
\ee

Let us focus on the case $\mu=k$.
Note that for $k=k_1+k_2$ one has $\gcd(k,k_1)=\gcd(k_1,k_2)=k_0$ and introduce $\hat{k}:=k/k_0$, $\hat{k}_1:=k_1/k_0$, $\hat{k}_2:=k_2/k_0$.
Then we can solve~\eqref{eqn:rhocond} for $\rho_2^e$ and find that the second $\delta$-function in~\eqref{haeDTP2} imposes
\begin{align}
    \hat{k}\mu_1-\hat{k}k_1+Nk_1\hat{k}\rho_1^e=0\text{ mod }Nk_1\hat{k}_2\hat{k}\,.
\end{align}
This implies that $\hat{k}\mu_1-\hat{k}k_1=0\text{ mod }Nk_1\hat{k}$.
Solving~\eqref{eqn:rhocond} instead for $\rho_1^e$, one also obtains that $\hat{k}\mu_2-\hat{k}k_2=0\text{ mod }Nk_2\hat{k}$.
We therefore find that the sum only receives contributions from terms with
\begin{align}
    \mu_1=k_1\text{ mod }Nk_1\,,\quad \mu_2=k_2\text{ mod }Nk_2\,.
\end{align}
All of these terms are compatible with $\rho_1^e=\rho_2^e=0$ and the holomorphic anomaly equations take the simplified form
\begin{align}
    \partial_{\bar\tau} \whh_{0,k;\mu_e,k} =&  \frac{1}{16\pi\I \tau_2^2}
\sum_{k_1+k_2=k} k_0 
\sum_{\mu_{1,e}\in\IZ_{k_1 N} \atop \mu_{2,e}\in\IZ_{k_2 N}}\delta^{(k_0 N)}_{ \mu_ e-\mu_{1,e}-\mu_{2,e}}\, 
\whh_{0,k_1;\mu_{1,e},k_1}
\whh_{0,k_2;\mu_{2,e},k_2}\,.
\end{align}

In terms of the sum over all residue classes $\overline{h_{0,k}}:=\sum_{\mu_e\in \IZ_{kN}} \whh_{0,k;\mu_e,k}$, this further simplifies to
\be
\partial_{\bar\tau} \overline{h_{0,k}} =  \frac{k}{16\pi\I \tau_2^2}
\sum_{k_1+k_2=k}  \overline{h_{0,k_1}} \, \overline{h_{0,k_2}} \,.
\ee 
Identifying  $\overline{h_{0,k}}(N\tau)=k \hatf_k(T)$, we arrive at 
\be
\partial_{\bar T} \hatf_k(T) =  \frac{1}{16\pi\I N T_2^2} 
\sum_{k_1+k_2=k} k_1 k_2  \ \hatf_{k_1} \hatf_{k_2} 
\ee
or equivalently, using $\whE_2(NT)=E_2(NT)-\frac{3}{\pi NT_2}$, $\partial_{\bar T} \whE_2(NT)=\frac{3\I}{2\pi N T_2^2}$,
\be
\partial_{\whE_2(NT)} \hatf_k =-\frac{1}{24}
\sum_{k_1+k_2=k} k_1 k_2  \ \hatf_{k_1} \hatf_{k_2} \,.
\ee
This reproduces the holomorphic anomaly equation \eqref{hanom0} for genus 0 GW invariants.
Generalizing this argument to the case of genus-one fibrations on higher del Pezzo surfaces would require to confront both the isotropicity of $p^a$ and non-invertibility of $\kappa_{ab}$ simultaneously, which we leave for future work.

\subsubsection{Vector valued modularity}
So far, we have been able to express the genus zero free energies in modular polarization as sums over modular completions of D4-D2-D0 indices
\begin{align}
     k\hatf_k(T)=\sum_{\mu_e=0}^{kN-1} \whh_{0,k,\mu_e,k}\left(NT\right)\,.
     \label{eqn:hatfD4D2D0}
\end{align}
However, this relation only involves contributions where the coset of the base curve class $\mu\in\Lambda^*/\Lambda$ is equal to the image $k$ of the D4-brane charge in $\Lambda^*$.
This is closely related to the fact that $\hatf_k(T)$ is a modular form for $\Gamma_1(N)$ but transforms itself as a vector valued modular form under the full modular group~\cite{Knapp:2021vkm,Schimannek:2021pau,Duque:2025kaa}.

In~\cite{Schimannek:2021pau} it was proposed that -- for $N>1$ -- the image of $f_k(T)$ under the Fricke involution $T\rightarrow -1/(NT)$ can be interpreted as a genus zero topological string free energy $f^\vee_k(T)$ of the relative Jacobian fibration $J$ of $X$, in the presence of a flat but topologically non-trivial B-field.
Using~\eqref{Multsys-hp}, we find that under S-transformations
\begin{align}
\begin{split}
    	&k\, \hatf_{k}\left(-\frac{1}{NT}\right)= \sum_{\mu_e=0}^{kN-1} \whh_{0,k,\mu_e,k}\left(-\frac{1}{T}\right)\\
	=&\frac{e^{\frac{\pi\I}{4} \left[c_\alpha\left(4k^\alpha-c^\alpha\right)+1\right]}}{kN}\sum_{\mu_e'=0}^{kN-1}\sum_{k'=0}^{kN-1}\sum_{\mu_e=0}^{kN-1} \exp\left(-\tfrac{2\pi\I}{N}\left[\tfrac{k'}{k}\left(\mu_e-\tfrac{k\ell}{N}\right)+\mu_e'\right]\right)\whh_{0,k,\mu_e',k'}\left(T\right)\\
	=&(-1)^{c_\alpha k^\alpha}e^{\frac{\pi \I}{4}(\chi_B-3)}\sum_{\mu_e'=0}^{kN-1}\exp\left(-\frac{2\pi\I\mu_e'}{N}\right)\whh_{0,k,\mu_e',0}\left(T\right)\,.
    \end{split}
    \label{eqn:hatfdualD4D2D0}
\end{align}
This suggests that $f^\vee_k(T)$ also admits a modular completion $\hatf_k^\vee(T)$, that satisfies
\begin{align}
    k\hatf^\vee_k(T)=\sum_{\mu_e=0}^{kN-1}\exp\left(-\frac{2\pi\I\mu_e}{N}\right)\whh_{0,k,\mu_e,0}\left(T\right)\,.
\end{align}
Perhaps the most striking aspect of this relation is that the two sides a priori count objects on different geometries.

More generally, the $SL(2,\mathbb{Z})$ orbit of $\hatf_k(T)$ is composed of topological string free energies associated to different genus one fibrations that share the same relative Jacobian fibration with $X$ and carry different B-field topologies~\cite{Schimannek:2021pau,Duque:2025kaa}.
By applying suitable combinations of the $T$- and $S$-transformations~\eqref{Multsys-hp}, it is possible to obtain relations analogous to~\eqref{eqn:hatfD4D2D0} and~\eqref{eqn:hatfdualD4D2D0} for all of those string backgrounds.
However, while these relations will be interesting to explore further, they would take us too far afield and we will leave this for future work.

\subsection{Comments on the black string SCFT}

We now comment on the growth of D4-D2-D0 indices, and its relation to the central charge of the black string SCFT obtained by wrapping an M5-brane on the corresponding divisor~\cite{Maldacena:1997de}. First, let us consider the generating series of genus 0 GW invariants (which have the same growth as genus 0 GV invariants)
\be
f_H(T) = -\frac{\Delta_{2N}(T)^{r_H/N} \varphi_H(T,0) }{\eta(NT)^{12 c_1(B)\cdot H}} \,.
\ee
Since $f_H(T):=\sum_n C_H(n) q^n$ is a $\Gamma_1(N)$ weak quasimodular form of negative weight, the growth of its Fourier coefficients is controlled by the most polar term $q'^{-\Delta_{\rm max}}$ in the Fourier expansion around the cusp at $T=0$ (where $q':=e^{-2\pi\I/T}$),
\be
C_H( n) \stackrel{n\to+\infty}{ \sim} \exp\left( 4\pi \sqrt{ \Delta_{\rm max} n} \right)\,.
\ee 
Using the fact that $\Delta_{2N}(T)$ is regular near $T=0$~\footnote{Following~\cite[(4.21)]{Schimannek:2021pau}, one finds that $T^{2N}\Delta_{2N}\left(-\frac{1}{NT}\right)=T^{2N}\exp\left(\frac{2\pi\I}{NT}\right)\phi_{-2,1}\left(-\frac{1}{T},-\frac{1}{NT}\right)^{-N}=\phi_{-2,1}\left(T,-\frac{1}{N}\right)^{-N}=\left[2\I\sin(\pi/N)\right]^{-2N}+\mathcal{O}(q')$.}, while $\eta(NT)$ behaves as $q'^{\frac{1}{24 N}}$, we find 
 \be
 \Delta_{\rm max}=\frac{c_\alpha C^{\alpha\beta} k_\beta}{2N} \,.
 \ee
This should be contrasted with the pole at $T=\I\infty$, which has order $\ell_\alpha C^{\alpha\beta} k_\beta/(2N)$.  Going back to \eqref{gamd4d2d0} and identifying the Fourier mode $n$ with $k_e-\frac{\ell_\alpha C^{\alpha\beta} k_\beta}{2N}=- N \hat q_0$, we find that the D4-D2-D0 indices with $p=(0,C^{\alpha\beta}k_\beta)$ grow as
\be
\Omega_{p,\mu}(\hat q_0 ) \sim \exp\left( 2\pi \sqrt{2 c_\alpha C^{\alpha\beta} k_\beta 
(- \hat q_0) } \right)\,.
\ee
This agrees with the macroscopic Bekenstein-Hawking entropy
\be
S_{BH}=2\pi \sqrt{ \frac{p^3+ c_2\cdot p}{6} (-\hat q_0) }
\ee
upon noting that $p^3:=\kappa_{abc}p^ap^bp^c=0$ and $c_\alpha=\frac{1}{12} c_{2,\alpha}$. As explained in \cite{Maldacena:1997de}, this growth is consistent with the central charge of the $(0,4)$ superconformal field theory obtained by wrapping an M5-brane on a very ample divisor\footnote{The prime is used to distinguish this result, based on the ampleness assumption, from the correct result in \eqref{cLRV} below.}
\be
\label{cLRMSW}
\begin{split}
c'_L=& \chi(D) = \kappa_{abc} p^a p^b p^c + c_{2,a} p^a= 12 c_1(B)\cdot C\,, \\
c'_R=&6 \chi(\cO(D)) = \kappa_{abc} p^a p^b p^c + \frac12 c_{2a} p^a=6 c_1(B)\cdot C\,.
\end{split}
\ee

However, as emphasized in \cite{Vafa:1997gr,Haghighat:2015ega}, the divisor $D=\pi^{-1}(C)$ is not ample and the  formulae \eqref{cLRMSW} for the central charges do not apply. 
 In particular, the odd cohomology of $D$
does not in general vanish, rather it is given by the arithmetic genus of the curve $C$,
\be
\label{defhD}
h_{1,0}(D)  = h_D := \frac12 (C^2 - c_1(B)\cdot C )+1\,.
\ee
Using the Euler number $\chi(D)$ and the signature $\sigma(D)$ given by 
\bea
\chi(D) &=& 2-4 h_{1,0}(D)+2 h_{2,0}(D) + h_{1,1}(D) = 12 c_1(B)\cdot C\,, \nn\\
\sigma(D) &=& 2+2 h_{2,0}(D) - h_{1,1}(D) = -8 c_1(B)\cdot C\,,\\
\chi(\cO_D) &=& 1- h_{1,0}(D)+h_{2,0}(D) = c_1(B)\cdot C\,, \nn
\eea
we deduce the even cohomology of $D$,
\be
\begin{split}
h_{2,0}(D) =& \frac12 \left( C^2 + c_1(B)\cdot C \right) \,,\\
h_{1,1}(D) =& C^2 +9  c_1(B)\cdot C +2\,.
\end{split}
\ee
Recall that the Hodge numbers $h_{2,0}, h_{1,1}$ are related to the number of self-dual and anti-self-dual two-forms by
\be
b_2^+(D)=2h_{2,0}(D)+1\,, \quad b_2^-(D)=h_{1,1}(D)-1\,,
\ee
while the holomorphic Euler number $\chi(\cO_D)=\frac14( \chi(D) + \sigma(D))$ determines the cohomology of  the line bundle $\cL=\mathcal{O}_X(D)$,
\be
\chi(\cO_D) = \sum_{i=0}^3 (-1)^i \dim H^i(X,\cL)\,.
\ee
For a very ample line bundle, $h_{1,0}(D)$ vanishes and $\chi(\cO_D)=\dim H^0(X,\cL)$ gives the complex dimension (plus one) of the moduli space $\cM_D$ of the divisor $D$ inside $X$. More generally however, the real dimension of $\cM_D$ is given by $d_D=2h_{2,0}(D)$.

The $(0,4)$ SCFT obtained by wrapping an M5-brane on $D$ (or a D3-brane on $C$ in F-theory) is, at least at large volume,  a non-linear sigma model whose target space is a Narain bundle of signature $(b_2^- -1 ,b_2^+ -1)$ over $\cM_D$, along with a  center of motion factor $\IR^3\times S^1$. The  central charges (including the center of motion factor) are given by~\cite{Vafa:1997gr,Haghighat:2015ega}
\bea
\label{cLRV}
\begin{split}
c_L=&d_D + b_2^- + \frac12\times 4 h_{1,0}(D) + 3 = 3 C^2 + 9 c_1(B)\cdot C+6 \,, \\
c_R=& d_D + b_2^+  + \frac12 \times 4 h_{2,0}(D) + 5 =  3 C^2 + 3 c_1(B)\cdot C+6 \,.
\end{split}
\eea
In terms of the arithmetic genus \eqref{defhD}, this is rewritten as 
\be
c_L = 6h_D + 12 c_1(B)\cdot C\,, \quad c_R = 6h_D + 6 c_1(B)\cdot C\,,
\ee
which differs from \eqref{cLRMSW} by an additional term $6h_D$ on both sides (recall that the
difference  $c_L-c_R$ is fixed by the gravitational anomaly \cite{Kraus:2005vz}).

This presents a puzzle, since the growth of the D4-D2-D0 indices is controlled by $c'_L$, in agreement with the macroscopic entropy, rather than by $c_L$. The most obvious resolution is that even though the
central charge of the SCFT is given by \eqref{cLRV}, the Cardy formula involves the 
 \textit{effective}  central charge  $c_{\rm eff} = c_L - 24 h_{\rm min}$, where $h_{\rm min}$ is the conformal dimension of the lowest conformal primary, which reduces to $c'_L$ if one
identifies $h_{\rm min}=\frac14 h_D$, which is recognized as the energy of the Ramond ground state for $4h_D$ real fermions. We note that the difference $(c_L-c'_L,c_R-c'_R)$ can also be understood as the contribution of one unit of D4-brane charge $p^e$ along the base; on the F-theory side, this corresponds to one unit of Taub-NUT charge, and is consistent with the 4D/5D lift \cite{Couzens:2017way,Grimm:2018weo}. 
 
On the other hand, the authors of \cite{Haghighat:2015ega} presented some support for the formulae 
\eqref{cLRV} from the growth of the Fourier coefficients of the meromorphic Jacobi form $Z_{H}(T,\lambda)= \sum_{n,r} C(n,r) q^n y^r$, tentatively identified as the elliptic genus of the SCFT. In particular, in \cite[\S A]{Haghighat:2015ega} is is argued that for a meromorphic Jacobi form of index $k$ under $SL(2,\IZ)$, Fourier coefficients grow as $C(n,r)\sim \exp(\pi \sqrt{4kn-r^2})$ (similarly as for
a weak Jacobi form of the same index). The same arguments show that for a meromorphic Jacobi form of index $k$ under $\Gamma_1(N)$ (for example one obtained from
 a $SL(2,\IZ)$ modular form by rescaling $(T,\lambda)\to(NT,\lambda)$), Fourier coefficients grow
 as  $C(n,r)\sim \exp(\pi \sqrt{\frac{4kn}{N}-r^2})$. Setting $k=h_D-1$ and taking the large $n$ limit at fixed $r$, we get 
 \be
 \log|C(n,r)| \stackrel{n\to +\infty}{\sim} 2\pi \sqrt{\frac{(C^2-c_1(B)\cdot C+2)n}{2N}}\,.
 \ee 
 
 Let us compare this result with the macroscopic entropy of a static BMPV black hole,  given by \cite{Ferrara:1996um}
\be
S_{5D}(q_a)=2\pi \left( \frac16\kappa_{abc}t^a t^b t^c \right)\, ,\quad 
q_a  = \frac12 \kappa_{abc} t^b t^c\,,
\ee
where $t^a$ corresponds to the Kahler modulus at the attractor point. For genus one fibrations on $\IP^2$, with charge $q_a=(n,C)$, setting $\hat{n}=n-\frac{\ell\cdot C}{2N}$ 
we find
\be
\label{S5Dfull}
S_{5D}(n,C) = 
\frac{2 \pi  \sqrt{\hat{n} N-\sqrt{N \left(\hat{n}^2 N-\tilde{\kappa} C^2\right)}}}
{3 \sqrt{\hatkappa N^3  }}
 \left(\sqrt{N \left(\hat{n}^2 N-\tilde{\kappa}   C^2\right)}+2 \hat{n} N\right)
\ee
provided $|\hat{n}| \geq \sqrt{\frac{\hatkappa}{N}} |C|$. 
In the Cardy regime $\tilde n\gg |C|$,  we get
\be
S_{5D}(n,C) =2\pi \sqrt{  \frac{\hat{n} C^2}{2N}} - \frac{\pi \hatkappa (C^2)^{3/2}}{3 (2N \hat{n})^{3/2} }+ \dots
\simeq 2\pi \sqrt{\frac{ \hat{n} C^2}{2 N}}  \left( 1- \frac{\hatkappa C^2}{24 
N \hat{n}} +\dots \right)  \,.
\ee
The $N$-dependence at leading order is consistent with the expectation for $\Gamma_1(N)$ Jacobi forms, but the microscopic origin of the $\cO(\hatkappa C^2/\tilde m)$ corrections is obscure. Subleading corrections proportional to $c_1(B)$ were matched in \cite[(5.33)]{Castro:2007hc}, 
by noting that higher-derivative couplings of the form $c_\alpha \cA^\alpha \wedge \cR\wedge \cR$ term in 5 
dimensions (where $\cR$ is the Riemann tensor and $\cA^\alpha$ the graviphotons) induce a shift $q_\alpha\mapsto q_\alpha+\frac18 c_{2\alpha}$, which amounts at leading order to shifting $C^2 \mapsto C^2+ 3 c_1(B)\cdot C$ as in the central charge $c_L$ in \eqref{cLRV}.

The upshot of this discussion is that the central charge \eqref{cLRV} appears to control the growth of indices counting 5D black hole in the Cardy limit, or PT indices counting 4D black holes with one unit of D6-brane charge, but that the growth of indices counting D4-D2-D0 black holes wrapping the genus-one fiber is instead controlled by the MSW central charge \eqref{cLRMSW}, even though the pulled back divisor is not ample but only nef. It would be very interesting to reproduce the full black hole entropy \eqref{S5Dfull} from microscopic counting.

\section{Discussion and open problems}
\label{sec_discussion}

In this work, we have revisited the modular properties of generating series of various enumerative invariants of torus-fibered Calabi-Yau threefolds, restricting for simplicity to fibrations with trivial Mordell-Weil group and no fibral divisors  (hence at most discrete gauge symmetries in the corresponding F-theory vacua, see~\cite{Duque:2025kaa}). Our main result is a derivation of the Jacobi property of the generating series of PT invariants at fixed base degree, assuming the wave function property of the topological string partition function $Z_{\rm top}(t^a,\lambda)$ under a suitable monodromy. This wave function property is physically well motivated, but far from being established mathematically, and our work can be viewed as further evidence for its validity for general projective CY threefolds (see e.g. \cite{Lho:2018ayw,Coates:2018hms,Bousseau:2020ckw} for rigorous results in this direction for toric CY singularities and~\cite{Guo:2018mol} for a proof of the BCOV equations for the quintic). 

While we have only considered $Z_{\rm top}(t^a,\lambda)$ (or rather its logarithm) as an asymptotic series in the topological string coupling $\lambda$,  the wave-function property is expected to hold for the putative non-perturbative completion of the topological string amplitude, potentially fixing it uniquely, and it would be extremely interesting to understand its consequences for the resurgent structure of the asymptotic series, along the lines 
of~\cite{Gu:2023mgf,Douaud:2024khu,Bridgeland:2024ecw}.

 Notably, a key step in our derivation of the Jacobi properties was to introduce a variant $Z_{\rm mod}(t^a,\lambda)$ of the topological string amplitude, which is both holomorphic and modular invariant, up to a multiplier system entirely determined by the topological data of the fibration. This was made possible by the introduction of the improved K\"ahler moduli $\tildeS^\alpha$ \eqref{defSt}, which transform simply by an additive constant under $\Gamma_1(N)$.   
 As an asymptotic series in $\lambda$, $Z_{\rm mod}$ is essentially determined by the depth-zero part of the quasimodular series of GW invariants at arbitrary degree and genus. A better understanding of its analytic properties might open the way to determine all-genus GV invariants at higher degrees than currently available. 
 
 On the other hand, our derivation heavily hinges on the structure~\eqref{eqn:bd0GV} of the base degree zero GV invariants, as well as on the relations  \eqref{FtheoryRel} 
to the intersection numbers of the fibrations. In the elliptic case, i.e. $N=1$, the relations~\eqref{eqn:bd0GV} follow from~\cite[Theorem 6.9]{toda2012stability}, the proof of which
relies on the fact that for elliptic fibrations, the smooth generic Weierstra{\ss} fibration is isomorphic to the corresponding relative Jacobian fibration.  We expect that the relations for $N>1$ can be proven along the same lines but using the twisted derived equivalence of the smooth genus one fibered Calabi-Yau with its relative Jacobian fibration~\cite[Theorem 5.1]{Caldararu2002}. Clearly, it would be interesting to generalize our results to general fibrations with fibral divisors and/or non-trivial Mordell-Weil group (see \cite{fierrorcota-to-appear} for new results in this direction).

It was proposed in~\cite{Schimannek:2021pau} and further elaborated on in~\cite{Dierigl:2022zll,Duque:2025kaa}, that the $\Gamma_1(N)$-modular properties of $Z_{\rm top}$ extend to vector valued modular properties under the full modular group, when taking into account the topological string partition functions on all genus one fibered Calabi-Yau threefolds that share the same Jacobian fibration.
This requires including suitable flat but topologically non-trivial B-field backgrounds on fibrations that exhibit $\mathbb{Q}$-factorial terminal singularities, see also~\cite{Katz:2022lyl,Katz:2023zan,Schimannek:2025cok}.
Lemma~\ref{lem:frickeE}, together with Lemma~\ref{lem:g2g} and~\ref{lem:g01}, can be used to prove this vector valued modularity for the vertical invariants under the Fricke involution $T\rightarrow -1/(NT)$, assuming the Ansatz for the topological string partition function in the presence of a flat but topologically non-trivial B-field in terms of torsion refined Gopakumar-Vafa invariants from~\cite{Schimannek:2021pau}.
More generally, it should be possible to generalize these results to understand the full vector valued modularity for the vertical invariants and to use the wave function property of the topological string partition function to extend this beyond base degree zero. 

Our second main  result was to relate the generating series of genus zero Gromov-Witten invariants at fixed base degree $H$ (which, up to multi-cover effects, count BPS bound states of D2-D0 branes wrapped on $H+n F$) to the generating series of rank 0 rational DT invariants supported on the pulled-back divisor $D=\pi^{-1}(H)$ (which, up to different multi-cover effects, count BPS bound states of D4-D2-D0 branes wrapped on $D$). Restricting for simplicity to genus one fibrations over $\IP^2$, the proposed relation \eqref{fkhconj} (extending the earlier proposal of \cite{Klemm:2012sx} to the case of genus one fibrations with $N$-sections)
arises by following DT invariants under the same relative conifold monodromy, assuming the absence of wall-crossing.  It is clearly an important problem to prove this assumption, and extend the relation \eqref{fkhconj} to more general torus fibrations.  
If this relation is correct, the fact that the generating series $f_H$ are quasimodular forms under $\Gamma_1(N)$ implies that the generating series $h_{D,\mu}$ transform as 
vector-valued mock modular forms under $SL(2,\IZ)$, exactly as predicted in \cite{Alexandrov:2018lgp,Alexandrov:2019rth}, given the fact that for isotropic charge vectors, mock modularity reduces to quasimodularity. This is similar to the case of K3-fibrations, where the modularity of the generating series of vertical rank 0 DT invariants can be reduced to the modularity of generating series of Noether-Lefschetz invariants~\cite{Bouchard:2016lfg,Doran:2024kcb}.

Clearly, both in the case of torus and K3-fibrations,  it would be very interesting to compute rank 0 DT invariants for general divisor classes with non-trivial support on the base, so as to provide further support to the S-duality conjectures, and possibly obtain new constraints on higher genus GV invariants, along the lines of \cite{Alexandrov:2023zjb,Alexandrov:2023ltz}. While it is known that in the case of a smooth elliptic fibration over a del Pezzo surface $B$, rank 0 DT invariants supported on $[D]=r[B]$ coincide with rank $r$ Vafa-Witten invariants on $B$~\cite{Tanaka:2017jom,Tanaka:2017bcw,Gholampour:2017bxh}, and therefore have modular generating series, it is an interesting problem to generalize this to the case of torus fibrations with $N$-sections. Similarly, it would be interesting to clarify the relation (if any) between  GV invariants supported on the basis of a torus fibration with $N$-section, and GV invariants of the non-compact CY threefold given by the total space of the canonical bundle $K_{B}$. 
We hope to return to some of these questions in future work. 

\appendix

\section{Modular forms for $\Gamma_1(N)$}
\label{sec:modularity}

In this appendix, we review basic facts about modular forms for the congruence group $\Gamma_1(N)$, introduce Eisenstein series with Dirichlet character
and their holomorphic Eichler integrals,
and obtain  new results for their transformation under Fricke involutions.   
The results are used in Section~\ref{sec:basedegreezero} to understand the modular properties of the base degree zero contributions to the topological string free energies of a generic genus one fibered CY threefold.

The congruence subgroup $\Gamma_1(N)\subset SL(2,\IZ)$ is defined as
\begin{align}
\label{eq:defg1N}
    \Gamma_1(N):=\left\{\,\left(\begin{array}{cc}a&b\\c&d\end{array}\right)\in SL(2,\IZ)\,\,\bigg\vert\,\,a,d\equiv 1\text{ mod }N\,,\quad c\equiv 0\text{ mod }N\,\right\}\,.
\end{align}
For $1\le N\le 4$, the group $\Gamma_1(N)$ is generated by the two elements $\tiny\left(\begin{array}{cc}1&1\\0&1\end{array}\right)$ and $\tiny\left(\begin{array}{cc}1&0\\N&1\end{array}\right)$.
More generally, the minimal number of generators is given by $2g(N)+n_c(N)-3$ in terms of the genus $g(N)$ and the number of cusps $n_c(N)$ of the modular curve $X_1(N)=\mathbb{H}/\Gamma_1(N)$.
We have listed sets of generators for $N=5,6,7$ in Table~\ref{tab:gensG1N}.

\begin{table}[ht!]
    \begin{align*}\begin{array}{|c|c|}\hline
        N & \text{Generators of }\Gamma_1(N) \\\hline
5 & \left(
\begin{array}{cc}
 1 & 1 \\
 0 & 1 \\
\end{array}
\right), \left(
\begin{array}{cc}
 1 & 0 \\
 5 & 1 \\
\end{array}
\right), \left(
\begin{array}{cc}
 11 & -4 \\
 25 & -9 \\
\end{array}
\right)\\
6 & \left(
\begin{array}{cc}
 1 & 1 \\
 0 & 1 \\
\end{array}
\right), \left(
\begin{array}{cc}
 1 & 0 \\
 6 & 1 \\
\end{array}
\right), \left(
\begin{array}{cc}
 7 & -3 \\
 12 & -5 \\
\end{array}
\right)\\
7 & \left(
\begin{array}{cc}
 1 & 1 \\
 0 & 1 \\
\end{array}
\right), \left(
\begin{array}{cc}
 1 & 0 \\
 7 & 1 \\
\end{array}
\right), \left(
\begin{array}{cc}
 15 & -4 \\
 49 & -13 \\
\end{array}
\right), \left(
\begin{array}{cc}
 -13 & 5 \\
 -21 & 8 \\
\end{array}
\right), \left(
\begin{array}{cc}
 22 & -9 \\
 49 & -20 \\
\end{array}
\right)\\\hline
        \end{array}
    \end{align*}
    \caption{Generators of the group $\Gamma_1(N)$ for $5\le N\le 7$, obtained using the algorithm described in~\cite[Section 1.4]{Stein2007} and implemented in {\tt Sage}~\cite{sagemath} with the command {\tt Gamma1(N).generators()}.}
    \label{tab:gensG1N}
\end{table}

We denote the corresponding ring of modular forms by $M(N)$ and let $M_w(N)\subset M(N)$ be the subspace of modular forms of weight $w$.
Since $\Gamma_1(1)=\text{SL}(2,\mathbb{Z})$, we have $M(1)=\langle E_4,E_6\rangle$ in terms of the Eisenstein series $E_4(\tau)$,$E_6(\tau)$.
Using~\eqref{eqn:e2trafo}, it is easy to see that
\begin{align}
    e_{N,2}(\tau):=NE_2(N\tau)-E_2(\tau)\,,
    \label{eqn:defE2n}
\end{align}
is a modular form of weight $2$ for $\Gamma_1(N)$.
A basis for the ring of modular forms for $\Gamma_1(N)$ with $N\le 5$ have been constructed for example in~\cite[Appendix B]{Schimannek:2021pau} and we briefly summarize the relevant forms in Table~\ref{tab:mformsG1N}.
\begin{table}[ht!]
    \centering
    $\begin{array}{|c|c|c|}\hline
        N & \mbox{Generators of } M(N) & \Delta_{2N} \\ \hline
        1 & e_{1,4}=E_4, e_{1,6}=E_6 &  \\\hline
        \multirow{2}{*}{2}    & e_{2,2}=2E_2(2\tau)-E_2(\tau)=1+24q+24q^2+96q^3+24q^4+\ldots &
         \frac{e_{2,4}-e_{2,4}^2 }{192} \\
                                & e_{2,4}=E_4 & 
                               \\\hline
        \multirow{2}{*}{3}    & e_{3,1}=\left(\frac{\eta(\tau)^9}{\eta(3\tau)^3}+27\frac{\eta(3\tau)^9}{\eta(\tau)^3}\right)^{\frac13}=1+6q+6q^3+6q^4+\ldots 
        & e_{3,3}^2 \\
                                & e_{3,3}=\frac{\eta(3\tau)^9}{\eta(\tau)^3}=q+3q^2+9q^3+13q^4+\ldots &                              
                                 \\\hline
       \multirow{2}{*}{4}     & e_{4,1}=\frac{\eta(2\tau)^{10}}{\eta(\tau)^4\eta(4\tau)^4}=1+4q+4q^2+4q^4+\ldots & -\frac{e_{4,1}^2(e_{4,1}^2-e_{2,2})^3}{4096} \\
                                & e_{2,2} & \\\hline
       \multirow{2}{*}{5}     & e_{5,1}=\frac{(q;q)_{\infty}^2}{\left((q;q^5)_\infty(q^4;q^5)_\infty\right)^5}=1+3q+4q^2+2q^3+q^4+\ldots  & e_{5,1}^6 \tilde{e}_{5,1}^4 \\
                                & \tilde{e}_{5,1}=\frac{q\,(q;q)_{\infty}^2}{\left((q^2;q^5)_\infty(q^3;q^5)_\infty\right)^5}=q-2q^2+4q^3-3q^4+\ldots & \\\hline
    \end{array}$
    \caption{Generators of the ring $M(N)$ of modular forms for $\Gamma_1(N)$ with $N\le 5$, and canonical modular form $\Delta_{2N}$ for $2\leq N\leq 5$. We denote by $e_{N,w}$ or $e_{N,w,\bullet}$ a modular form for $\Gamma_1(N)$ of weight $w$ and $(a;q)_n$ is the $q$-Pochhammer symbol $(a;q)_n=\prod_{k=0}^{n-1}(1-aq^k)$.}
    \label{tab:mformsG1N}
\end{table}

\subsection{Dirichlet characters, Bernoulli numbers and L-series}
Let us first recall some definitions and properties of Dirichlet characters, generalized Bernoulli numbers and Dirichlet L-series from~\cite[Chapter 3]{Miyake1989}.

We denote the group of Dirichlet characters with modulus $N$ by
\begin{align}
    D[N]:=\widehat{\left(\mathbb{Z}/N\mathbb{Z}\right)^\times}\,.
\end{align}
Since $\chi(-1)^2=\chi(1)=1$, we have that $\chi(-1)\in\{-1,1\}$.
We denote by $D[N]^{\pm}\subset D[N]$ the subgroup of Dirichlet characters with modulus $N$ and $\chi(-1)=\pm 1$.
Given a character $\chi$ with modulus $N$, the conductor $m_\chi\in\mathbb{N}$ of $\chi$ is the smallest positive integer such that $\chi(m)=\chi(n)$ for all $m,n\in\mathbb{Z}$ with $\gcd(m,N)=\gcd(n,N)=1$ and $m\equiv n\text{ mod }m_\chi$.
The conductor always divides the modulus, $m_\chi\vert N$.

A character $\chi\in D[N]$ is called primitive if $N=m_\chi$.
If $\chi$ is not primitive, there is a unique primitive character $\chi^\star\in D[m_\chi]$, such that
\begin{align}
    \chi(n)=\left\{\begin{array}{cl}
        \chi^\star(n)&\text{ if }\gcd(n,N)=1\\
        0&\text{ if }\gcd(n,N)\ne1\\
    \end{array}\right.\,.
\end{align}
The so-called principal character $\chi_{N,0}$ is given by
\begin{align}
    \chi_{N,0}(n):=\left\{\begin{array}{cl}
        0&\text{ if }\gcd(n,N)\ne 1\\
        1&\text{ if }\gcd(n,N)= 1\\
    \end{array}\right.\,,
    \label{eqn:principalCharacter}
\end{align}
and always has conductor $1$. The associated primitive character is the trivial character $\chi\in D[1]$ with $\chi(n)=1$.
The number of Dirichlet characters of modulus $N$ is given by Euler's totient function
\begin{align}
    \varphi(N)=N\prod\limits_{p\vert N}\left(1-p^{-1}\right)\,,
\end{align}
where the product is over prime factors of $N$.

The generalized Bernoulli numbers $B_{k,\chi}$ are defined via the relation~\footnote{Compared to~\cite{Miyake1989}, we extend the definition of $B_{k,\chi}$
to the case that $\chi$ is not necessarily primitive.}
\begin{align}
    \sum\limits_{a=1}^N\chi(a)\frac{x e^{ax}}{e^{Nx}-1}=\sum\limits_{k=0}^\infty B_{k,\chi}\frac{x^k}{k!}\,.
    \label{eqn:generalizedBernoulli}
\end{align}
Expanding the left-hand side of~\eqref{eqn:generalizedBernoulli} in $x$, we can introduce $\beta_{k,\chi,a}$ via the relation
\begin{align}
\label{eqn:generalizedBernoulli1}
    B_{k,\chi}=\sum\limits_{a=1}^N\chi(a)\beta_{k,N,a}\, ,\quad 
    \sum\limits_{k=0}^\infty \beta_{k,N,a}\frac{x^k}{k!} = \frac{x e^{ax}}{e^{Nx}-1}\,,
\end{align}
and note that $\beta_{k,rN,ra}=r^{k-1}\beta_{k,N,a}$.
We have the relation
\begin{align}
    B_{k,\chi}=&B_{k,\chi^\star}\sum\limits_{d\vert N}\mu(d)d^{k-1}\chi^\star(d)=B_{k,\chi^\star}\prod\limits_{p\vert N}\left(1-p^{k-1}\chi^\star(p)\right)\,,
    \label{eqn:bernoulliMoebius}
\end{align}
in terms of the M\"obius function
\begin{align}
    \mu(n)=&\left\{\begin{array}{cl}
        1&\text{ if }\,n=1\\
        (-1)^k&\text{ if $n$ has $k$ distinct prime factors}\\
        0&\text{ if $n$ is divisible by a square greater than one}
    \end{array}\right.\,.
\end{align}

The Dirichlet L-series is defined as
\begin{align}
    L(s,\chi)=\sum\limits_{n=1}^\infty\frac{\chi(n)}{n^s}\,.
\end{align}
One has the relation
\begin{align}
    L(s,\chi)=L(s,\chi^\star)\prod\limits_{p\vert N}\left(1-\frac{\chi^\star(p)}{p^s}\right)\,,
    \label{eqn:LseriesPrimitive}
\end{align}
where the product is over the prime factors of $N$.
If $\chi\in D[N]^+$ ($\chi\in D[N]^-$), the only zeros $L(s,\chi)$ with $s\in\mathbb{Z}$ are simple zeros at the even (odd) negative integers.
We will also need that
\begin{align}
    L(0,\chi^\star)=\left\{\begin{array}{cl}
        0&\text{ if }\,m_\chi>1\\
        -\frac12&\text{ if }\,m_\chi=1
    \end{array}\right.\,.
    \label{eqn:Lzero}
\end{align}

If $\chi$ is a primitive character with modulus $N>1$, and $\delta\in\{0,1\}$ satisfies $\chi(-1)=(-1)^\delta$, then $L(s,\chi)$ satisfies the functional equation
\begin{align}
    L(s,\chi)=W(\chi)2^s\pi^{s-1}N^{1/2-s}\sin\left(\frac{\pi}{2}(s+\delta)\right)\Gamma(1-s)L(1-s,\overline{\chi})\,.
    \label{eqn:LseriesFunctional}
\end{align}
where $\overline{\chi}$ is the complex conjugate character of $\chi$, and
\begin{align}
    W(\chi)=\frac{\tau(\chi)}{{\rm i}^\delta \sqrt{N}}\,,
\end{align}
in terms of the Gauss sum
\begin{align}
    \tau(\chi)=\sum\limits_{a=1}^N\chi(a)e^{\frac{2\pi{\rm i}a}{N}}\,.
\end{align}
With this, one can show that $L(s,\chi)$, for primitive $\chi$ with modulus $N>1$, is an entire function of $s$.
One then also has~\cite[Theorem 3.3.4]{Miyake1989}
\begin{align}
    L(k,\chi)=(-1)^{1+(k-\delta)/2}\frac{W(\chi)}{2\I^\delta}(2\pi/N)^k\frac{B_{k,\overline{\chi}}}{k!}\,,\quad L(1-k,\chi)=-B_{k,\chi}/k\,.
    \label{eqn:LseriesBernoulli}
\end{align}

\subsection{Eisenstein series with character}
\label{sec:dirichletEisenstein}
We summarize and extend some results about Eisenstein series from~\cite[Chapter 7]{Miyake1989}, see also~\cite[Chapter 5]{Stein2007}.

\begin{definition}
    Let $\chi\in D[N]$ be a Dirichlet character with modulus $N\ge 1$ and let $k>0$ be such that $\chi(-1)=(-1)^k$.
    We define the Eisenstein series
    \begin{align}
        E_{k,\chi}(\tau):=&-\frac{B_k}{2k}\delta_{1,N}+\sum\limits_{m\ge 1}\sum\limits_{n\vert m}\chi(m/n)n^{k-1}q^{m}\,,\label{eqn:defE}\\
        \widetilde{E}_{k,\chi}(\tau):=&-\frac{B_{k,\chi}}{2k}+\sum\limits_{m\ge 1}\sum\limits_{n\vert m}\chi(n)n^{k-1}q^{m}\,,\label{eqn:defEtilde}\\
    \widehat{E}_{k,\chi}(\tau):=&N^{\frac{k}{2}-1}\left(-\frac{B_k}{2k}\delta_{1,N}+\sum\limits_{m\ge 1}\sum\limits_{n\vert m}\sum\limits_{a=1}^N\chi(-a)n^{k-1}e^{\frac{2\pi{\rm i}am}{nN}}q^{m}\right)\,.\label{eqn:defEhat}
    \end{align}
    \label{def:eisenstein}
\end{definition}
We can then show the following:
\begin{proposition}
    Let $\chi$ and $k$ be as in Definition~\ref{def:eisenstein}.
    \begin{enumerate}
        \item If $N>1$ or $k\ne 2$, $\widetilde{E}_{k,\chi}(\tau)$ is a modular form of weight $k$ for $\Gamma_1(N)$.
            If $N=1$ and $k=2$, it is the weight $2$ quasimodular form $\widetilde{E}_{2,\chi_{1,0}}(\tau)=-\frac{1}{24}E_2(\tau)$.
        \item If $k\ne 2$ or $\chi\ne \chi_{N,0}$, $E_{k,\chi}(\tau)$, $\widehat{E}_{k,\chi}(\tau)$ are modular forms of weight $k$ and $E_{2,\chi_{N,0}}(\tau)$, $\widehat{E}_{2,\chi_{N,0}}(\tau)$ are quasimodular forms of weight 2 for $\Gamma_1(N)$.
    \end{enumerate}
    \label{prop:eisenstein}
\end{proposition}
\begin{proof}
    We first focus on $\widetilde{E}_{k,\chi}(\tau)$ and assume that $k\ne 2$ or $m_\chi>1$.
    If $\chi$ is primitive, the modular properties of $\widetilde{E}_{k,\chi}(\tau)$ follow from~\cite[Theorem 7.1.3]{Miyake1989},~\cite[Theorem 7.2.12]{Miyake1989} and~\cite[Theorem 7.2.13]{Miyake1989}, together with~\cite[Lemma 7.1.1]{Miyake1989} and~\cite[Lemma 7.2.19]{Miyake1989}.
 If $\chi$ is not primitive, we observe that for any function $f:\mathbb{Z}\times\mathbb{Z}\rightarrow\mathbb{C}$ one has
\begin{align}
    \begin{split}
        \sum\limits_{m\ge 1}\sum\limits_{n\vert m}\chi(n)f(n,m)=&\sum\limits_{d\vert N}\mu(d)\chi^\star(d)\sum\limits_{m\ge 1}\sum\limits_{n\vert m}\chi^\star(n)f(dn,dm)\,.
    \end{split}
    \label{eqn:moebiusf}
\end{align}
Together with~\eqref{eqn:bernoulliMoebius}, this implies that
\begin{align}
    \begin{split}
    \widetilde{E}_{k,\chi}(\tau)=&-\frac{B_{k,\chi}}{2k}+\sum\limits_{d\vert N}\mu(d)d^{k-1}\chi^\star(d) \left(\widetilde{E}_{k,\chi^\star}(d\tau)+\frac{B_{k,\chi^\star}}{2k}\right)\\
    =&\sum\limits_{d\vert N}\mu(d)d^{k-1}\chi^\star(d) \widetilde{E}_{k,\chi^\star}(d\tau)\,.
    \end{split}
    \label{eqn:etildeInPrim}
\end{align}
We then observe that $\widetilde{E}_{k,\chi^\star}(\tau)$ is a modular form of weight $k$ for $\Gamma_1(m_\chi)$ and in the sum $\chi^\star(d)$ is zero unless $d\vert N/m_\chi$.
    If $k=2$ and $\chi=\chi_{N,0}$, we can write
    \begin{align}
        \begin{split}
        \widetilde{E}_{2,\chi}(\tau)=&-\frac{1}{24}\sum\limits_{d\vert N}\mu(d)dE_2(d\tau)=-\frac{1}{24}\sum\limits_{d\vert N}\mu(d)\left(E_2(\tau)+e_{d,2}(\tau)\right)\\
        =&-\frac{1}{24}\left[\delta_{N,1}E_2(\tau)+\sum\limits_{d\vert N}\mu(d)e_{d,2}(\tau)\right]\,,
        \end{split}
    \end{align}
    where in the last step we have used that $\sum_{d\vert N}\mu(d)=\delta_{1,N}$.

The modular properties of $E_{k,\chi}(\tau)$ for $k\ne 2$ or $m_\chi\ne 1$ again follow from~\cite[Theorem 7.1.3]{Miyake1989} and~\cite[Theorem 7.2.13]{Miyake1989}, together with~\cite[Lemma 7.1.1]{Miyake1989} and~\cite[Lemma 7.2.19]{Miyake1989}.
For $k=2$, $\chi=\chi_{N,0}$ we use again~\eqref{eqn:moebiusf} and observe that
\begin{align}
    E_{2,\chi_{N,0}}(\tau)=&-\frac{1}{24}\sum\limits_{d\vert N}\mu(d)E_2(d\tau)\,.
\end{align}
We can rewrite this again as
\begin{align}
    \begin{split}
    E_{2,\chi_{N,0}}(\tau)=&-\frac{1}{24}\sum\limits_{d\vert N}\frac{\mu(d)}{d}\left(E_2(\tau)+e_{d,2}(\tau)\right)\\
       =&-\frac{1}{24}\frac{\varphi(N)}{N}E_2(\tau)-\frac{1}{24}\sum\limits_{d\vert N}\frac{\mu(d)}{d}e_{d,2}(\tau)\,.
    \end{split}
    \label{eqn:dE2}
\end{align}
where in the last step we have applied M\"obius inversion to the divisor sum formula $\sum\limits_{d\vert N}\phi(d)=N$ to obtain that $\sum_{d\vert N}\mu(d)/d=\phi(N)/N$.

The modular properties of $\widehat{E}_{k,\chi}(\tau)$ follow from Lemma~\ref{lem:ehatrel} below.
\end{proof}

\begin{lemma}
With $\chi$, $k$ as in Definition~\ref{def:eisenstein} and $L:=m_\chi$ we have
    \begin{align}
    \widehat{E}_{k,\chi}(\tau)=\frac{N^{\frac{k}{2}}}{L}\chi(-1)W(\chi^\star)\sum\limits_{d\vert N}\mu(d)d^{-1}\chi^\star(d)E_{k,\overline{\chi}^\star}\left(\frac{N\tau}{Ld}\right)\,.
    \end{align}
    \label{lem:ehatrel}
\end{lemma}
\begin{proof}
We first calculate the non-constant part and define
\begin{align}
    f(\tau):=W(\chi^\star)\sum\limits_{d\vert N}\mu(d)d^{-1}\chi^\star(d)\left(E_{k,\overline{\chi}^\star}\left(\frac{N\tau}{Ld}\right)+\frac{B_{k}}{2k}\delta_{1,N}\right)\,.
\end{align}
Inserting the definition of $E_{k,\chi}(\tau)$, we can rewrite this as
\begin{align}
    \begin{split}
    f(\tau)=&W(\chi^\star)\sum\limits_{d\vert N}\mu(d)d^{-1}\chi^\star(d)\sum\limits_{m\ge 1}\sum\limits_{n\vert m}\overline{\chi}^\star(m/n)n^{k-1}q^{\frac{Nm}{Ld}}\\
    =&W(\chi^\star)\sum\limits_{d\vert N/L}\mu(d)d^{-1}\chi^\star(d)\sum\limits_{m\ge 1}\sum\limits_{n\vert m}\overline{\chi}^\star(m/n)n^{k-1}q^{\frac{Nm}{Ld}}\\
    =&\frac{L}{N}W(\chi^\star)\sum\limits_{d\vert N/L}d\mu\left(\frac{N}{Ld}\right)\chi^\star\left(\frac{N}{Ld}\right)\sum\limits_{m\ge 1}\sum\limits_{n\vert m}\overline{\chi}^\star(m/n)n^{k-1}q^{dm}\,,
    \end{split}
\end{align}
where in the second line we use that $\chi^{\star}(d)$ can only be non-zero if $d\vert N/L$ and in the third line we use $\sum_{d\vert N}f(d)=\sum_{d\vert N}f(N/d)$.
Exchanging the order of summation, and summing over $dm$ instead of $m$, leads to the expression
\begin{align}
    \begin{split}
    f(\tau)=&\frac{L}{N}W(\chi^\star)\sum\limits_{m\ge 1}\sum\limits_{n\vert m}\sum\limits_{d\vert (N/L,m/n)}d\mu\left(\frac{N}{Ld}\right)\chi^\star\left(\frac{N}{Ld}\right)\overline{\chi}^\star\left(\frac{m}{dn}\right)n^{k-1}q^{m}\\
    =&\frac{L}{N}\sum\limits_{m\ge 1}\sum\limits_{n\vert m}\sum\limits_{a=1}^N\chi(a)n^{k-1}e^{2\pi{\rm i}am/(nN)}q^{m}\,,\\
    \end{split}
\end{align}
where in the second line we have used~\cite[Lemma 3.1.3]{Miyake1989}.
To take care of the constant part, we just note that
\begin{align}
\begin{split}
    -\chi(-1)W(\chi^\star)\sum\limits_{d\vert N}\frac{\mu(d)}{d}\chi^\star(d)\frac{B_{k}}{2k}\delta_{1,N}=-\frac{B_{k}}{2k}\delta_{1,N}\,.
    \end{split}
\end{align}
The result follows after comparing with~\eqref{eqn:defEhat}.
\end{proof}

\begin{lemma}
    Given $\chi\in D[N]$ and $k\ge 1$, the (quasi) modular forms $\widetilde{E}_{k,\chi}(\tau)$ and $\widehat{E}_{k,\chi}(\tau)$ are related via the Fricke involution
    \begin{align}
        \widetilde{E}_{k,\chi}\left(-\frac{1}{N\tau}\right)=(\sqrt{N}\tau)^k\widehat{E}_{k,\chi}(\tau)+\delta_{N,1}\delta_{k,2}\frac{{\rm i}N}{4\pi}\tau\,.
    \end{align}
    \label{lem:frickeE}
\end{lemma}
\begin{proof}
Let $L=m_\chi$ be again the conductor of $\chi$.
We first assume that $k\ge 3$.
From~\cite[Lemma 7.1.2]{Miyake1989} one then obtains
\begin{align}
    \widetilde{E}_{k,\chi^\star}\left(-\frac{1}{L\tau}\right)=\tau^k \frac{L^k}{W(\overline{\chi}^\star)}E_{k,\overline{\chi}^\star}(\tau)\,,
\end{align}
where $L$ is the conductor of the primitive character $\chi^\star$.
We also have~\cite[Lemma 3.1.1]{Miyake1989}
\begin{align}
    W(\chi^\star)W(\overline{\chi}^\star)=\chi^\star(-1)L\,.
\end{align}
Together with~\eqref{eqn:etildeInPrim}, this gives
\begin{align}
    \begin{split}
    \widetilde{E}_{k,\chi}\left(-\frac{1}{N\tau}\right)=\tau^k\frac{N^k}{L}W(\chi^\star)\sum\limits_{d\vert N}\mu(d)d^{-1}\chi^\star(-d)E_{k,\overline{\chi}^\star}\left(\frac{N\tau}{Ld}\right)\,.
    \end{split}
\end{align}
Together with Lemma~\ref{lem:ehatrel} the claim follows.
The result can easily extended to $k=1,2$, using~\cite[Theorem 7.2.13]{Miyake1989} and~\cite[Theorem 7.2.12]{Miyake1989}.
\end{proof}

\subsection{Polylogarithms and Eisenstein series with character}
\label{sec:polyLogAndDEisenstein}
The (quasi) modular Eisenstein series $E_k(\tau)$ for even $k\ge 2$ can be expressed as
\begin{align}
    E_{2k-2}(\tau)=1-\frac{2(2k-2)}{B_{2k-2}}\sum\limits_{m\ge 1}\text{Li}_{3-2k}(q^m)\,.
    \label{eqn:eisensteinPolyLog}
\end{align}
On the other hand, one also has
\begin{align}
    \partial_\tau^3\sum\limits_{m\ge 1}\text{Li}_{3}(q^m)=\frac{1}{240}\left(E_4(\tau)-1\right)\,,\quad \partial_\tau\sum\limits_{m\ge 1}\text{Li}_{1}(q^m)=\frac{1}{24}\left(1-E_2(\tau)\right)\,.
    \label{eqn:plEder}
\end{align}
In this section we want to generalize these relations, using the Eisenstein series from Definition~\ref{def:eisenstein}.

First, recall that the indicator function $f_{N,a}(n)$ is given by
\begin{align}
    f_{N,a}(n):=\left\{\begin{array}{cl}
        1&\,\,n\equiv a\text{ mod }N\\
        0&\,\,n\not\equiv a\text{ mod }N
    \end{array}\right.\,,
\end{align}
and we also define $g_{N,a}(n):=f_{N,a}(n)+f_{N,-a}(n)$.

Our goal is to understand the modular properties of the following objects:
\begin{definition}
Given $a,N\in\mathbb{N}$, we define
\begin{align}
    \phi^{(g)}_{N,a}(\tau):=&\sum\limits_{m\ge 0}\sum\limits_{k=1}^Ng_{N,a}(k){\normalfont\text{Li}}_{3-2g}\left(q^{mN+k}\right)\,,\\
    \widehat{\phi}^{(g)}_{N,a}(\tau):=&N^{1-g}\sum\limits_{m\ge 1}\sum\limits_{k=1}^Ng_{N,a}(k){\normalfont\text{Li}}_{3-2g}\left(e^{\frac{2\pi{\rm i}k}{N}}q^m\right)\,.
\end{align}
\end{definition}
We can relate these to the Eisenstein series defined in Appendix~\ref{sec:dirichletEisenstein} using the following lemma.
\begin{lemma}
    Given $N,a\in\mathbb{N}$ with $\gcd(N,a)=1$, one has
    \begin{align}
       g_{N,a}(n)=&\frac{2}{\varphi(N)}\sum\limits_{\chi\in D[N]^+}\frac{\chi(n)}{\chi(a)}\,.
       \label{eqn:ginchar}
\end{align}
\label{lem:gchar}
\end{lemma}
\begin{proof}
The indicator function can be written as
\begin{align}
    f_{N,a}(n)=\frac{1}{\varphi(N)}\sum\limits_{\chi\in D[N]}\chi(n)/\chi(a)\,,
\end{align}
which allows us to express $g_{N,a}(n)$ as
\begin{align}
    \begin{split}
    g_{N,a}(n)=&f_{N,a}(n)+f_{N,-a}(n)=\frac{1}{\varphi(N)}\sum\limits_{\chi\in D[N]}\left(1+\chi(-1))\right)\chi(n)/\chi(a)\,.
    \end{split}
    \label{eqn:gcharacters}
\end{align}
The result then directly follows from the fact that $\chi(-1)\in\{-1,1\}$.
\end{proof}
Using Lemma~\ref{lem:gchar}, it is easy to show the following:
\begin{lemma}
Given $N,a\in\mathbb{N}$ with $r:=\gcd(N,a)$ and $g\ge 2$, one has
    \begin{align}
        \phi^{(g)}_{N,a}(\tau)=&\frac{B_{2g-2}}{2g-2}\delta_{1,N/r}+\frac{2}{\varphi(N/r)}\sum\limits_{\chi\in D[N/r]^+}\frac{\chi(r)}{\chi(a)}E_{2g-2,\chi}(r\tau)\,,\\
        \widehat{\phi}^{(g)}_{N,a}(\tau)=&\frac{(N/r)^{1-g}B_{2g-2}}{2g-2}\delta_{1,N/r}+\frac{2N}{r\varphi(N/r)}\sum\limits_{\chi\in D[N/r]^+}\frac{\chi(r)}{\chi(a)}\widehat{E}_{2g-2,\chi}(r\tau)\,.
        \label{eqn:phiinE}
    \end{align}
    \label{lem:g2g}
\end{lemma}
\begin{proof}
Let us first assume that $\gcd(a,N)=1$.
We can rewrite the Eisenstein series $E_{2g-2,\chi}(\tau)$ as
\begin{align}
    \begin{split}
    E_{2g-2,\chi}(\tau)=&-\frac{B_{2g-2}}{2(2g-2)}\delta_{1,N}+\sum\limits_{m\ge 1}\sum\limits_{n\vert m}\chi(m/n)n^{2g-3}q^{m}\\
    =&-\frac{B_{2g-2}}{2(2g-2)}\delta_{1,N}+\sum\limits_{m\ge 0}\sum\limits_{k=1}^{N}\sum\limits_{n\ge 1}\chi(k)n^{2g-3}q^{n(mN+k)}\\
    =&-\frac{B_{2g-2}}{2(2g-2)}\delta_{1,N}+\sum\limits_{m\ge 0}\sum\limits_{k=1}^{N}\chi(k)\text{Li}_{3-2g}\left(q^{mN+k}\right)\,.
    \end{split}
    \label{eqn:Edrewrite}
\end{align}

Similarly, we can rewrite
\begin{align}
    \begin{split}
    \widehat{E}_{2g-2,\chi}(\tau)=&N^{g-2}\left(-\frac{B_{2g-2}}{2(2g-2)}\delta_{1,N}+\sum\limits_{m\ge 1}\sum\limits_{n\vert m}\sum\limits_{a=1}^N\chi(-a)n^{2g-3}e^{\frac{2\pi{\rm i}am}{nN}}q^{m}\right)\\
    =&N^{g-2}\left(-\frac{B_{2g-2}}{2(2g-2)}\delta_{1,N}+\sum\limits_{m,n\ge 1}\sum\limits_{a=1}^N\chi(-a)n^{2g-3}e^{\frac{2\pi{\rm i}am}{N}}q^{mn}\right)\\
    =&N^{g-2}\left(-\frac{B_{2g-2}}{2(2g-2)}\delta_{1,N}+\sum\limits_{m\ge 1}\sum\limits_{a=1}^N\chi(-a)\text{Li}_{3-2g}\left(e^{\frac{2\pi{\rm i}a}{N}}q^{m}\right)\right)\,.
    \end{split}
\end{align}
In both cases the result then follows from Lemma~\ref{lem:gchar}.
If $r=\gcd(a,N)\ne 1$ we have $\phi^{(g)}_{N,a}=\phi^{(g)}_{N/r,a/r}(r\tau)$ and $\widehat{\phi}^{(g)}_{N,a}=\widehat{\phi}^{(g)}_{N/r,a/r}(r\tau)$.
\end{proof}
\begin{lemma}
Given $N,a\in\mathbb{N}$ with $r:=\gcd(a,N)$ and $g\in\{0,1\}$, one has
\begin{align}
\begin{split}
    &(2\pi{\rm i})^{2g-3}\partial_\tau^{3-2g}\phi^{(g)}_{N,a}(\tau)\\
    =&\frac{\beta_{4-2g,N/r,a/r}}{4-2g}+\frac{2}{\varphi(N/r)}\sum\limits_{\chi\in D[N/r]^+}\frac{\chi(r)}{\chi(a)}\widetilde{E}_{4-2g,\chi}(r\tau)\,,\\
    &(2\pi{\rm i})^{2g-3}\partial_\tau^{3-2g}\widehat{\phi}^{(g)}_{N,a}(\tau)\\
    =&\left(\frac{N}{r}\right)^{1-g}\frac{B_{4-2g}}{4-2g}+\frac{2}{\varphi(N/r)}\sum\limits_{\chi\in D[N/r]^+}\frac{\chi(r)}{\chi(a)}\widehat{E}_{4-2g,\chi}(r\tau)\,.
    \end{split}
\end{align}
\label{lem:g01}
\end{lemma}
\begin{proof}
We first rewrite
\begin{align}
    \begin{split}
    (2\pi{\rm i})^{2g-3}\partial_\tau^{3-2g}\phi^{(g)}_{N,a}(\tau)=&\sum\limits_{m\ge 0}\sum\limits_{k=1}^Ng_{N,a}(k)(mN+k)^{3-2g}\text{Li}_0(q^{mN+k})\\
    =&\sum\limits_{m\ge 1}\sum\limits_{n\vert m}g_{N,a}(n)n^{3-2g}q^{m}\,.
    \end{split}
\end{align}
Using Definition~\ref{def:eisenstein}, this can be expressed as
\begin{align}
    \begin{split}
    &(2\pi{\rm i})^{2g-3}\partial_\tau^{3-2g}\phi^{(g)}_{N,a}(\tau)\\
    =&\frac{2}{\varphi(N)}\sum\limits_{\chi\in D[N]^+}\chi(a^{-1})\left(\widetilde{E}_{4-2g,\chi}(\tau)+\frac{B_{4-2g,\chi}}{2(4-2g)}\right)\\
    =&\frac{2}{\varphi(N)}\sum\limits_{\chi\in D[N]^+}\chi(a^{-1})\left[\widetilde{E}_{4-2g,\chi}(\tau)+\frac{1}{2(4-2g)}\sum\limits_{b=1}^N\chi(b)\beta_{4-2g,N,b}\right]\\
    =&\frac{\beta_{4-2g,N,a}}{4-2g}+\frac{2}{\varphi(N)}\sum\limits_{\chi\in D[N]^+}\chi(a^{-1})\widetilde{E}_{4-2g,\chi}(\tau)\,.
    \end{split}
\end{align}
Similarly, we first rewrite
\begin{align}
\begin{split}
    (2\pi{\rm i})^{2g-3}\partial_\tau^{3-2g}\widehat{\phi}^{(g)}_{N,a}(\tau)=&N^{1-g}\sum\limits_{m\ge 1}\sum\limits_{b=1}^Ng_{N,a}(b)m^{3-2g}\text{Li}_0\left(e^{\frac{2\pi{\rm i}b}{N}}q^m\right)\\
    =&N^{1-g}\sum\limits_{m\ge 1}\sum\limits_{n\ge 1}\sum\limits_{b=1}^Ng_{N,a}(b)m^{3-2g}e^{\frac{2\pi{\rm i}nb}{N}}q^{nm}\,.
    \end{split}
\end{align}
Using again Definition~\ref{def:eisenstein}, this can be expressed as
\begin{align}
\begin{split}
    (2\pi{\rm i})^{2g-3}&\partial_\tau^{3-2g}\widehat{\phi}^{(g)}_{N,a}(\tau)\\
    =&\frac{2}{\varphi(N)}\sum\limits_{\chi\in D[N]^+}\chi(a^{-1})\left(\widehat{E}_{4-2g,\chi}(\tau)+N^{1-g}\frac{B_{4-2g}}{2(4-2g)}\varphi(N)\delta_{m_\chi,1}\right)\\
    =&N^{1-g}\frac{B_{4-2g}}{4-2g}+\frac{2}{\varphi(N)}\sum\limits_{\chi\in D[N]^+}\chi(a^{-1})\widehat{E}_{4-2g,\chi}(\tau)\,.
    \end{split}
\end{align}
\end{proof}

\subsection{Eichler integrals}
\label{sec:eichlerIntegrals}

In this appendix we generalize the relations~\eqref{eqn:defEm2} and~\eqref{eqn:defE0} and derive the modular properties of the functions $\phi^{(g)}_{N,a}(\tau)$ for $g=0,1$.

Let us first make the following definitions.
\begin{definition}
Given $g\in\{0,1\}$ and $N,a\in\mathbb{N}$ we define
\begin{align}
    \begin{split}
    \Phi_{N,a}^{(g)}(\tau):=&c^{(g)}_{N,a}\tau^{3-2g}+\phi_{N,a}^{(g)}\left(\tau\right)\,,\\
    \widehat{\Phi}_{N,a}^{(g)}(\tau):=&\hat{c}^{(g)}_{N}\tau^{3-2g}+\widehat{\phi}_{N,a}^{(g)}\left(\tau\right)\,,
    \end{split}
\end{align}
with
\begin{align}
    c^{(g)}_{N,a}:=-\frac{(2\pi{\rm i})^{3-2g}\beta_{4-2g,N,a}}{(4-2g)[(3-2g)!]}\,,\quad \hat{c}^{(g)}_{N}:=-\frac{(2\pi {\rm i})^{3-2g}N^{1-g}B_{4-2g}}{(4-2g)[(3-2g)!]}\,,
\end{align}
as well as
\begin{align}
    \varphi_{N,a}^{(g)}(\tau):=(2\pi{\rm i})^{2g-3}\partial_\tau^{3-2g}\Phi^{(g)}_{N,a}(\tau)\,,\quad \widehat{\varphi}_{N,a}^{(g)}(\tau):=(2\pi{\rm i})^{2g-3}\partial_\tau^{3-2g}\widehat{\Phi
    }^{(g)}_{N,a}(\tau)\,.
\end{align}
\label{def:Phiphi}
\end{definition}
Note that $\Phi_{1,1}^{(0)}(\tau)=\widehat{\Phi}_{1,1}^{(0)}(\tau)=2\widetilde{E}_{-2}(\tau)-\zeta(3)$ and $\Phi_{1,1}^{(1)}(\tau)=\widehat{\Phi}_{1,1}^{(1)}(\tau)=-2\widetilde{E}_{0}(\tau)$.

From Lemma~\ref{lem:g01}, we know that both $\varphi_{N,a}^{(g)}(\tau)$ and $\widehat{\varphi}_{N,a}^{(g)}(\tau)$ are (quasi) modular forms of weight $4-2g$ for $\Gamma_1(N)$.
Moreover, Lemma~\ref{lem:frickeE} implies that they are related via the Fricke involution
\begin{align}
    \varphi_{N,a}^{(g)}(\tau)=(\sqrt{N}\tau)^{2g-4}\widehat{\varphi}_{N,a}^{(g)}\left(-\frac{1}{N\tau}\right)+\delta_{g,1}\delta_{N,1}\frac{1}{2\pi{\rm i}\tau}\,.
\end{align}
Our goal in this section is to study the transformation behavior under this Fricke involution of $\Phi_{N,a}^{(g)}(\tau)$ and $\widehat{\Phi}_{N,a}^{(g)}(\tau)$ and to use this to deduce the modular properties of $\Phi_{N,a}^{(g)}(\tau)$.
More precisely, we want to show the following:

\begin{theorem}
    Given $N,a\in\mathbb{N}$, and $r:=\gcd(N,a)$, one has
    \begin{align*}
    \begin{split}
    \Phi_{N,a}^{(0)}\left(\tau+1\right)=&\Phi_{N,a}^{(0)}\left(\tau\right)+c^{(0)}_{N,a}(3\tau^2+3\tau+1)\,,\\
        (N\tau+1)^2\Phi_{N,a}^{(0)}\left(\frac{\tau}{N\tau+1}\right)
=&\Phi_{N,a}^{(0)}\left(\tau\right)-\frac{{\rm i}\pi^3}{90}\left[2(15a^2-15aN+2N^2)\tau^2-3N\tau-3\right]\\
&-\zeta(3)\tau(\tau+2)\delta_{1,N/r}\,,
    \end{split}
    \end{align*}
    as well as
    \begin{align*}
    \begin{split}
    \Phi_{N,a}^{(1)}\left(\tau+1\right)=&\Phi_{N,a}^{(1)}\left(\tau\right)+c^{(1)}_{N,a}\,,\\
    \Phi_{N,a}^{(1)}\left(\frac{\tau}{N\tau+1}\right)=&\Phi_{N,a}^{(1)}\left(\tau\right)-\delta_{1,N/r}\log\left(N\tau+1\right)-\hat{c}^{(1)}_{N,a}\,.
    \end{split}
    \end{align*}
    \label{thm:eichler}
\end{theorem}
\begin{proof}
The behavior under $\tau\mapsto\tau+1$ is trivial.
Deriving the transformation under $\tau\rightarrow \tau/(N\tau+1)$ will occupy us for the rest of this section.

If $\gcd(N,a)=r$, with $r>1$, then $\Phi_{N,a}^{(0)}\left(\tau\right)=\Phi_{N',a'}^{(0)}\left(r\tau\right)$, with $N'=N/r$ and $a'=a/r$.
We can therefore assume that $\gcd(a,N)=1$ and the general result follows as well.

Since $\Phi_{N,a}^{(g)}(\tau)$ and $\widehat{\Phi}_{N,a}^{(g)}(\tau)$ are Eichler integrals of the modular forms $\varphi_{N,a}^{(g)}(\tau)$ and $\widehat{\varphi}_{N,a}^{(g)}(\tau)$, we know that
\begin{align}
    \Phi_{N,a}^{(g)}(\tau)-(\sqrt{N}\tau)^{2-2g}\widehat{\Phi}_{N,a}^{(g)}\left(-\frac{1}{N\tau}\right)=P_{N,a}^{(g)}(\tau)\,,
    \label{eqn:PhiFricke}
\end{align}
in terms of the functions
\begin{align}
    P_{N,a}^{(0)}(\tau)=\sum\limits_{k=0}^{2}\alpha^{(0)}_{N,a,k}\tau^k\,,\quad P_{N,a}^{(1)}(\tau)=\tilde{\alpha}_{N,a}^{(1)}\log(\tau)+\alpha^{(1)}_{N,a,0}\,,
\end{align}
with complex coefficients $\tilde{\alpha}_{N,a}^{(1)}$, $\alpha^{(1)}_{N,a,0}$ and $\alpha^{(0)}_{N,a,k}$ for $k=0,1,2$.
This also implies
\begin{align}
    \begin{split}
    \widehat{\Phi}_{N,a}^{(g)}\left(\tau\right)-(\sqrt{N}\tau)^{2-2g}\Phi_{N,a}^{(g)}\left(-\frac{1}{N\tau}\right)=&-(\sqrt{N}\tau)^{2-2g}P_{N,a}^{(g)}\left(-\frac{1}{N\tau}\right)\,.
    \end{split}
\end{align}

Following the strategy laid out in~\cite{0990.11041}, we define
\begin{align}
    \begin{split}
    g_{N,a}^{(g)}(y)=&(-1)^g{\rm i}c^{(g)}_{N,a}y^{3-2g}+\Phi_{N,a}^{(g)}({\rm i}y)=\frac{2}{\varphi(N)}\sum\limits_{\chi\in D[N]^+}\chi(a^{-1})g_{N,\chi}^{(g)}(y)\,,\\
    \hat{g}_{N,a}^{(g)}(y)=&(-1)^g{\rm i}\hat{c}^{(g)}_{N}y^{3-2g}+\widehat{\Phi}_{N,a}^{(g)}({\rm i}y)=\frac{2}{\varphi(N)}\sum\limits_{\chi\in D[N]^+}\chi(a^{-1})\hat{g}_{N,\chi}^{(g)}(y)\,,
    \end{split}
\end{align}
in terms of 
\begin{align}
    \begin{split}
    g_{N,\chi}^{(g)}(y)=&\sum\limits_{m\ge 0}\sum\limits_{k=1}^N\chi(k)\text{Li}_{3-2g}\left(e^{-2\pi (mN+k)y}\right)\,,\\
    \hat{g}_{N,\chi}^{(g)}(y)=&N^{1-g}\sum\limits_{m\ge 0}\sum\limits_{k=1}^N\chi(k)\text{Li}_{3-2g}\left(e^{2\pi{\rm i}k/N}e^{-2\pi m y}\right)\,.
    \end{split}
\end{align}

From~\eqref{eqn:PhiFricke} we then obtain
\begin{align}
    g_{N,a}^{(g)}(y)=(-1)^{1-g}(\sqrt{N}y)^{2-2g}\hat{g}_{N,a}^{(g)}\left(\frac{1}{Ny}\right)+R^{(g)}_{N,a}(y)\,,
\end{align}
in terms of
\begin{align}
    R^{(0)}_{N,a}(y)=\sum\limits_{k=-1}^{3}A^{(0)}_{N,a,k}y^k\,,\quad R^{(1)}_{N,a}(y)=\widetilde{A}^{(1)}_{N,a}\log(y)+\sum\limits_{k=-1}^{1}A^{(1)}_{N,a,k}y^k\,,
\end{align}
with
\begin{align}
    \begin{split}
    A^{(0)}_{N,a,-1}=&{\rm i}\hat{c}^{(0)}_{N}/N^2\,,\quad A^{(0)}_{N,a,0}=\alpha^{(0)}_{N,a,0}\,,\quad A^{(0)}_{N,a,1}={\rm i}\alpha^{(0)}_{N,a,1}\,,\\
    A^{(0)}_{N,a,2}=&-\alpha^{(0)}_{N,a,2}\,,\quad A^{(0)}_{N,a,3}={\rm i}c^{(0)}_{N,a}\,,
    \end{split}
\end{align}
as well as
\begin{align}
    \begin{split}
    \widetilde{A}^{(1)}_{N,a}=&\tilde{\alpha}_{N,a}^{(1)}\,,\quad A^{(1)}_{N,a,-1}={\rm i}\hat{c}^{(1)}_{N}/N\,,\\ A^{(1)}_{N,a,0}=&\alpha^{(1)}_{N,a,0}+\frac12\pi{\rm i}\tilde{\alpha}_{N,a}^{(1)}\,,\quad A^{(1)}_{N,a,1}=-{\rm i}c^{(1)}_{N,a}\,.
    \end{split}
\end{align}

In order to determine the coefficients $A^{(g)}_{N,a,k}$, we calculate the poles of the Mellin transform
\begin{align}
    \tilde{g}_{N,\chi}^{(g)}(s)=&\int\limits_0^\infty y^{s-1}g_{N,\chi}^{(g)}(y)dy\,.
\end{align}
More precisely, if we denote the residue of $\tilde{g}^{(g)}_{N,\chi}(s)$ at $s=-k$ by $A^{(g)}_{N,\chi,k}$, we have
\begin{align}
    A_{N,a,k}^{(g)}=\frac{2}{\varphi(N)}\sum\limits_{\chi\in D[N]^+}\chi(a^{-1})A^{(g)}_{N,\chi,k}\,,
    \label{eqn:polesCoefficients}
\end{align}
After rewriting
\begin{align}
    \begin{split}
    g_{N,\chi}^{(g)}(y)=&\sum\limits_{m\ge 0}\sum\limits_{k=1}^N\sum\limits_{n\ge 1}\chi(k)\frac{e^{-2\pi n(mN+k)y}}{n^{3-2g}}=\sum\limits_{m\ge 1}\sum\limits_{n\vert m}\frac{\chi(m/n)}{n^{3-2g}}e^{-2\pi m y}\,,
    \end{split}
\end{align}
we obtain
\begin{align}
    \begin{split}
    (2\pi)^s\tilde{g}^{(g)}_{N,\chi}(s)=&(2\pi)^s\sum\limits_{m\ge 1}\sum\limits_{n\vert m}\frac{\chi(m/n)}{n^{3-2g}}\int\limits_0^\infty y^{s-1}e^{-2\pi m y}dy\\
    =&\Gamma(s)\sum\limits_{m\ge 1}\sum\limits_{n\vert m}\frac{\chi(m/n)}{n^{3-2g+s}(m/n)^s}=\Gamma(s)\sum\limits_{m\ge 1}\sum\limits_{n\ge 1}\frac{\chi(n)}{n^{3-2g+s}(n)^s}\\
    =&\Gamma(s)\zeta(3-2g+s)L(s,\chi)\,.
    \end{split}
\end{align}
Using the relation~\eqref{eqn:LseriesPrimitive}, we can further rewrite this as
\begin{align}
    \tilde{g}^{(g)}_{N,\chi}(s)=(2\pi)^{-s}\Gamma(s)\zeta(3-2g+s)L(s,\chi^\star)\prod\limits_{p\vert N}\left(1-\frac{\chi^\star(p)}{p^s}\right)\,.
    \label{eqn:gtildeEx}
\end{align}

If the conductor of $\chi$ is greater than one, the Dirichlet L-series $L(s,\chi^\star)$ is an entire function of $s$.
Otherwise, we have $\chi=\chi_{N,0}$ and
\begin{align}
    L(s,\chi_{N,0})=\zeta(s)\prod\limits_{p\vert N}\left(1-p^{-s}\right)\,.
    \label{eqn:LseriesPrincipal}
\end{align}
Recall that the only pole of $\zeta(s)$ is a simple pole at $s=1$ with residue $1$ and the zeros on the real axis lie at the even negative integers.
The gamma function $\Gamma(s)$ has simple poles at values $s=-n$, $n\in\mathbb{N}$, with residue
\begin{align}
    \text{Res}(\Gamma,-n)=\frac{(-1)^n}{n!}\,.
\end{align}

\paragraph{Case $g=0$.}
We first consider $g=0$. Then $\tilde{g}^{(g)}_{N,\chi}(s)$ has potential (simple) poles at $s\in\{-3,\ldots, 1\}$.
One immediately gets
\begin{align}
        A^{(0)}_{N,\chi,3}=\frac46\pi^3L(-3,\chi)\,,\quad A^{(0)}_{N,\chi,1}=-\frac{1}{3}\pi^3L(-1,\chi)\,,\quad A^{(0)}_{N,\chi,0}=\zeta(3)L(0,\chi)\,,
\end{align}
and, using~\eqref{eqn:LseriesPrincipal},
\begin{align}
    A^{(0)}_{N,\chi,-1}=\left\{\begin{array}{cl}
        0&\text{ if }\chi\ne \chi_{N,0}\\
        \frac{\pi^3}{180}\prod\limits_{p\vert N}\left(1-p^{-1}\right)&\text{ if }\chi= \chi_{N,0}
    \end{array}\right.\,.
\end{align}
To obtain $A^{(0)}_{N,\chi,2}$, we first use~\eqref{eqn:LseriesFunctional} and note that for primitive $\chi\in D[N]^+$ one has
\begin{align}
    L(\epsilon-2,\chi)=-\frac{N^{\frac52}W(\chi)}{4\pi^2}L(3,\overline{\chi})\epsilon+\mathcal{O}(\epsilon^2)\,,
\end{align}
such that
\begin{align}
    A^{(0)}_{N,\chi,2}=-\frac12m_\chi^{\frac52}W(\chi)L(3,\overline{\chi})\,.
\end{align}

Using the relation~\eqref{eqn:polesCoefficients}, as well as~\eqref{eqn:Lzero}, we then obtain
\begin{align}
    \begin{split}
    A^{(0)}_{N,a,3}=&-\frac{\pi^3}{3}\beta_{4,N,a}\,,\quad
    A^{(0)}_{N,a,1}=\frac{\pi^3}{3}\beta_{2,N,a}\,,\\
    A^{(0)}_{N,a,0}=&-\zeta(3)\delta_{1,N}\,,\quad
    A^{(0)}_{N,a,-1}=\frac{\pi^3}{90N}\,,
    \end{split}
\end{align}
and this fixes
\begin{align}
    \alpha^{(0)}_{N,a,0}=-\frac12\zeta(3)\delta_{1,N}\,,\quad \alpha^{(0)}_{N,a,1}=-{\rm i}\frac{\pi^3}{3}\beta_{2,N,a}\,.
\end{align}

We can now use these results to calculate
\begin{align}
    \begin{split}
    (N\tau+1)^2&\Phi_{N,a}^{(0)}\left(\frac{\tau}{N\tau+1}\right)=(N\tau+1)^2\Phi_{N,a}^{(0)}\left(-\frac{1}{N\left(-\frac{1}{N\tau}-1\right)}\right)\\
    =&(\sqrt{N}\tau)^2\widehat{\Phi}_{N,a}^{(0)}\left(-\frac{1}{N\tau}-1\right)+\sum\limits_{k=0}^2\alpha^{(0)}_{N,a,k}\tau^k(N\tau+1)^{2-k}\\
=&\Phi_{N,a}^{(0)}\left(\tau\right)-\frac{\hat{c}^{(0)}_{N}}{N}\left[3+3N\tau+N^2\tau^2\right]+\sum\limits_{k=0}^2\alpha^{(0)}_{N,a,k}\tau^k\left[(N\tau+1)^{2-k}-1\right]\\
=&\Phi_{N,a}^{(0)}\left(\tau\right)-\frac{{\rm i}\pi^3}{90}\left[2(15a^2-15aN+2N^2)\tau^2-3N\tau-3\right]-\zeta(3)\tau(\tau+2)\delta_{1,N}\,.
    \end{split}
\end{align}

\paragraph{Case $g=1$.}
The case $g=1$ works analogously and $\tilde{g}^{(1)}_{N,\chi}(s)$ has simple poles $s\in\{-1,1\}$ and a potential double pole at $s=0$.
We immediately obtain the residues
\begin{align}
    A^{(1)}_{N,\chi,1}=\pi L(-1,\chi)\,,\quad  A^{(1)}_{N,\chi,-1}=\left\{\begin{array}{cl}
        0&\text{ if }m_{\chi}\ne 1\\
        \frac{\pi}{12}\prod\limits_{p\vert N}\left(1-p^{-1}\right)&\text{ if }m_{\chi}= 1
    \end{array}\right.\,,
\end{align}
such that
\begin{align}
    A^{(1)}_{N,a,1}=-\pi\beta_{2,N,a}\,,\quad A^{(1)}_{N,a,-1}=\frac{\pi}{6N}\,.
\end{align}
Around the double pole, we expand
\begin{align}
    \tilde{g}^{(1)}_{N,\chi}(s)=\frac{1}{s^2}\widetilde{A}^{(1)}_{N,\chi}+\frac{1}{s}A^{(1)}_{N,\chi,0}+\mathcal{O}(1)\,,
\end{align}
and, after comparing with~\eqref{eqn:gtildeEx}, we find that
\begin{align}
    \tilde{\alpha}^{(1)}_{N,\chi}=\widetilde{A}^{(1)}_{N,\chi}=\left\{\begin{array}{cl}
        0&\text{ if }\,N>1\\
        -\frac12&\text{ if }\,N=1
    \end{array}\right.\,,
\end{align}
where we have used~\eqref{eqn:Lzero}.
Since for $s>0$ one has
\begin{align}
    \int^1_0\log(y)y^{s-1}dy=-\frac{1}{s^2}\,,
\end{align}
we can identify
\begin{align}
    \widetilde{A}_{N,a}^{(1)}=-\frac{2}{\varphi(N)}\sum\limits_{\chi\in D[N]^+}\chi(a^{-1})\widetilde{A}^{(1)}_{N,\chi}\,.
\end{align}
It follows that
\begin{align}
    \widetilde{A}_{N,a}^{(1)}=\left\{\begin{array}{cl}
        0&\text{ if }\,N>1\\
        1&\text{ if }\,N=a=1
    \end{array}\right.\,.
\end{align}
We skip the calculation of $A^{(1)}_{N,\chi,0}$ as it won't be necessary for our purpose.

We can now use these results to calculate
\begin{align}
    \begin{split}
    \Phi_{N,a}^{(1)}&\left(\frac{\tau}{N\tau+1}\right)\\
    =&\Phi_{N,a}^{(1)}\left(-\frac{1}{N\left(-\frac{1}{N\tau}-1\right)}\right)\\
    =&\widehat{\Phi}_{N,a}^{(1)}\left(-\frac{1}{N\tau}-1\right)+\tilde{\alpha}_{N,a}^{(1)}\log\left(\frac{\tau}{N\tau+1}\right)-\alpha^{(1)}_{N,a,0}\\
=&\widehat{\Phi}_{N,a}^{(1)}\left(-\frac{1}{N\tau}\right)-\hat{c}^{(1)}_{N}+\tilde{\alpha}_{N,a}^{(1)}\log\left(\frac{\tau}{N\tau+1}\right)-\alpha^{(1)}_{N,a,0}\\
=&\Phi_{N,a}^{(1)}\left(\tau\right)-\tilde{\alpha}_{N,a}^{(1)}\log\left(\tau\right)+\tilde{\alpha}_{N,a}^{(1)}\log\left(\frac{\tau}{N\tau+1}\right)-\hat{c}^{(1)}_{N}\\
=&\Phi_{N,a}^{(1)}\left(\tau\right)-\delta_{N,1}\log\left(N\tau+1\right)-\hat{c}^{(1)}_{N}\,.
    \end{split}
    \label{eqn:g1Utr}
\end{align}
\end{proof}

\subsubsection{Extra generators}
\label{sec_extragen}
Here, we record conjectural identities for the transformation of the genus 0 and genus 1 Eichler integrals under the additional generators $\scriptsize\begin{pmatrix}11&-4 \\ 25 & -9 \end{pmatrix}$ for $N=5$ and $\scriptsize\begin{pmatrix} 7 &-3 \\ 12 & -5 \end{pmatrix}$ for $N=6$, respectively. While these identities have been checked numerically, we have not attempted to prove them:
\be
\begin{split}
(25\tau-9)^2 \Phi_{5,0}^{(0)}\left(\frac{11\tau-4}{25\tau-9}\right)- \Phi_{5,0}^{(0)}(\tau)  =& 
\frac{(2\pi\I)^3}{144} \left(-123 + 635 \tau - 800 \tau^2\right)
\\
(25\tau-9)^2 \Phi_{5,1}^{(0)}\left(\frac{11\tau-4}{25\tau-9}\right)- \Phi_{5,1}^{(0)}(\tau)  =& 
\frac{(2\pi\I)^3}{720} \left(-399 + 2263 \tau - 3232 \tau^2\right)
\\
(25\tau-9)^2 \Phi_{5,2}^{(0)}\left(\frac{11\tau-4}{25\tau-9}\right)- \Phi_{5,2}^{(0)}(\tau)  =& 
\frac{(2\pi\I)^3}{720} \left( 681 - 3737 \tau + 5108 \tau^2\right) 
\\
(12\tau-5)^2 \Phi_{6,0}^{(0)}\left(\frac{7\tau-3}{12\tau-5}\right)- \Phi_{6,0}^{(0)}(\tau)  =& 
-\frac{3(2\pi\I)^3}{400} \left(984 \tau ^2-888 \tau +197\right)
\\
(12\tau-5)^2 \Phi_{6,1}^{(0)}\left(\frac{7\tau-3}{12\tau-5}\right)- \Phi_{6,1}^{(0)}(\tau)  =& 
\frac{(2\pi\I)^3}{1600} \left( -5781 \tau ^2+5207 \tau -1153 \right)
\\
(12\tau-5)^2 \Phi_{6,2}^{(0)}\left(\frac{7\tau-3}{12\tau-5}\right)- \Phi_{6,2}^{(0)}(\tau)  =& 
\frac{(2\pi\I)^3}{450} \left(333 - 1502 \tau + 1666 \tau^2 \right)
\\
(12\tau-5)^2 \Phi_{6,3}^{(0)}\left(\frac{7\tau-3}{12\tau-5}\right)- \Phi_{6,3}^{(0)}(\tau)  =& 
\frac{27(2\pi\I)^3}{6400} \left( 341 - 1539 \tau + 1707 \tau^2\right)
\end{split}
\ee
\be
\begin{split}
 \Phi_{5,0}^{(1)}\left(\frac{11\tau-4}{25\tau-9}\right)- \Phi_{5,0}^{(1)}(\tau)  =& 
\frac{5}{12} 2\pi\I - \log(25\tau-9)
\\
\Phi_{5,1}^{(1)}\left(\frac{11\tau-4}{25\tau-9}\right)- \Phi_{5,1}^{(1)}(\tau)  =& 
-\frac{11}{60} 2\pi\I 
\\
\Phi_{5,2}^{(1)}\left(\frac{11\tau-4}{25\tau-9}\right)- \Phi_{5,2}^{(1)}(\tau)  =& 
\frac{1}{60} 2\pi\I
\\
 \Phi_{6,0}^{(1)}\left(\frac{7\tau-3}{12\tau-5}\right)- \Phi_{6,0}^{(1)}(\tau)  =& 
\frac{1}{6} 2\pi\I - \log(12\tau-5)
\\
 \Phi_{6,1}^{(1)}\left(\frac{7\tau-3}{12\tau-5}\right)- \Phi_{6,1}^{(1)}(\tau)  =& 
-\frac{1}{12} 2\pi\I 
\\
 \Phi_{6,2}^{(1)}\left(\frac{7\tau-3}{12\tau-5}\right)- \Phi_{6,2}^{(1)}(\tau)   =& 
\frac{1}{6} 2\pi\I
\\
 \Phi_{6,3}^{(1)}\left(\frac{7\tau-3}{12\tau-5}\right)- \Phi_{6,3}^{(1)}(\tau)   =& 
-\frac{1}{12} 2\pi\I
\end{split}
\ee

\section{Relative conifold monodromy}
\label{app_relcon}
In this section, we derive the action~\eqref{eqU} of the relative conifold monodromy on the Chern characters of a basis of branes. 
We use the same notation as in the main text, but distinguish the quantities
\be
 a_\alpha=c_1(T_B)\cap \check{D}_\alpha\,, \quad c_\alpha=\frac{1}{12}c_2(T_X)\cap D_\alpha\,,\quad \alpha=1,\ldots,b_2(B)\,.
\ee
since we cannot prove that they are equal in general. 

Let us denote the two projections from $X\times_B X$ to $X$ by $\pi_i$, $i=1,2$.
As discussed in~\cite[Section 3.3]{Cota:2019cjx}, the relative conifold monodromy $U$ acts on the Chern character of a brane $\mathcal{E}^{\bullet}\in D^b(X)$ as
\begin{align}
    U:\,\text{ch}(\mathcal{E}^{\bullet})\mapsto \text{ch}(\mathcal{E}^{\bullet})-\pi_{2,*}\left[\pi_1^*\left(\text{ch}(\mathcal{E}^{\bullet})\text{Td}(T_{X/B})\right)\right]\,,
\end{align}
where $T_{X/B}$ is the virtual relative tangent bundle of the fibration.
The Chern class of the virtual relative tangent bundle $T_{X/B}$ is given by
\begin{align}
	\begin{split}
		c(T_{X/B})=&\frac{1+c_1(T_X)+c_2(T_X)+c_3(T_X)}{1+c_1(T_B)+c_2(T_B)}\\
		=&1-c_1(T_B)+c_1(T_B)^2-c_2(T_B)+c_2(T_X)-c_1(T_B)c_2(T_X)+c_3(T_X)\,.
	\end{split}
\end{align}
and the corresponding Todd genus takes the form
\begin{align}
    \text{Td}(T_{X/B})=&1-\frac{c_1(T_B)}{2}+\frac{2c_1(T_B)^2-c_2(T_B)+c_2(T_X)}{12}-\frac{c_1(T_B)c_2(T_X)}{24}\,.
\end{align}

A $\mathbb{Q}$-basis of the cohomology of $X$ is given by the classes
\begin{align}
    \{\,1,\,D_e,\,D_\alpha,\,F,\,D_e D_\alpha,\,V\,\}\,,
\end{align}
where $F$ is the 4-form that is Poincar\'e dual to the generic fiber and $V$ is the volume 6-form that is dual to a point on $X$.
They satisfy the relations
\begin{align}
    \begin{split}
    D_e^2=&\frac{1}{N}\left(\kappa-\frac{1}{N}\ell_\alpha\ell_\beta C^{\alpha\beta}\right)F+\frac{1}{N}\ell_\alpha C^{\alpha\beta}D_\beta D_e\,,\quad D_\alpha D_\beta=C_{\alpha\beta} F\,,\\
    D_e^3=&\kappa V\,,\quad D_e^2D_\alpha=\ell_\alpha V\,,\quad D_e D_\alpha D_\beta=NC_{\alpha\beta}V\,,\quad D_\alpha D_\beta D_\gamma=0\,,\\
    FD_e=&NV\,,\quad FD_\alpha=0\,.
    \end{split}
\end{align}
We can interpret $\pi_{2,*}\pi_1^*$ as an endomorphism of $H^\bullet(X,\mathbb{Q})$, where it acts as
\begin{align}
    1\mapsto 0\,,\quad D_\alpha\mapsto 0\,,\quad F\mapsto 0\,,\quad D_e\mapsto N\,,\quad D_eD_\alpha\mapsto ND_\alpha\,,\quad V\mapsto F\,.
\end{align}

The base $B$  is assumed to be a generalized del Pezzo surface and therefore rational.
This implies that the only non-vanishing Hodge numbers are $h^{1,1}(B)=b_2(B)$ and $h^{0,0}(B)=h^{2,2}(B)=1$.
Then $c_2(T_B)=(2+b_2(B))F$ and Noether's formula, together with the fact that the holomorphic Euler characteristic of a rational surface is $\chi(\mathcal{O}_B)=1$, implies that $c_1(T_B)^2=a_\alpha a_\beta C^{\alpha\beta}=(10-b_2(B))F$.

As a result, we can rewrite the Todd genus as
\begin{align}
	\begin{split}
        \text{Td}(T_{X/B})=&1-\frac12 a^\alpha D_\alpha+\left(\frac{1}{4}a_\alpha a_\beta C^{\alpha\beta}-1\right)F+\frac{1}{12}c_2(TX)-\frac{1}{2}c_\alpha a^\alpha V\,.
	\end{split}
\end{align}
By demanding that $c_2(TX)\cap (D_e,D_\alpha)=(c_{2,e},12c_\alpha)$, we also obtain that
\begin{align}
	c_2(T_X)=\frac{12}{N}c_\alpha C^{\alpha\beta}D_\beta D_e+\frac{1}{N}\left(c_{2,e}-\frac{12}{N}c_\alpha C^{\alpha\beta}\ell_\beta\right)F\,.
\end{align}

\paragraph{D6-brane:}
The Chern character of the structure sheaf is $\text{ch}(\mathcal{O}_X)=1$.
We then obtain
\begin{align}
	\pi_{2,*}\left[\pi_1^*\left(\text{ch}(\mathcal{O}_X)\text{Td}(T_{X/B})\right)\right]=c_\alpha C^{\alpha\beta}D_\beta-\frac{1}{2}c_\alpha a^\alpha F\,.
	\label{eqn:toddTrafo1}
\end{align}
The Chern character of a brane $\mathcal{O}_D$ that is supported on a Cartier divisor $D$ is given by
\begin{align}
	\text{ch}(\mathcal{O}_D)=1-\exp(-D)=D-\frac12 D^2+\frac16D^3\,.
\end{align}
We can therefore express the right-hand side of~\eqref{eqn:toddTrafo1} as
\begin{align}
	\begin{split}
		&c_\alpha C^{\alpha\beta}\,\text{ch}(\mathcal{O}_{D_\beta})+\frac{1}{2}c_\alpha C^{\alpha\beta}D_\beta\left(D_\beta-c_1(T_B)\right)\\
		=&c_\alpha C^{\alpha\beta}\,\text{ch}(\mathcal{O}_{D_\beta})+\frac{1}{2}c_\alpha C^{\alpha\beta}\left(c_{\beta\beta}-a_\beta\right)F\,.
	\end{split}
\end{align}

\paragraph{D4-brane on $D_e$:}
The Chern character of a brane $\mathcal{O}_{D_e}$ that is supported on the $N$-section $D_e$ is given by
\begin{align}
	\begin{split}
		\text{ch}(\mathcal{O}_{D_e})=&1-\exp(-{D_e})\\
		=&S-\frac12 \left[\frac{1}{N}\ell_\alpha D^\alpha D_e+\frac{1}{N}\left(\kappa-\frac{1}{N}\ell_\alpha\ell_\beta C^{\alpha\beta}\right)F\right]+\frac16\kappa V\,,
	\end{split}
\end{align}
such that
\begin{align}
	\begin{split}
		&\text{ch}(\mathcal{O}_{D_e})\text{Td}(T_{X/B})\\
        =&D_e-\frac12 \left[\left(a_\alpha+\ell_\alpha\right)C^{\alpha\beta}D_\beta D_e+\frac{1}{N}\left(\kappa-\frac{1}{N}\ell_\alpha\ell_\beta C^{\alpha\beta}\right)F\right]\\
			&+\left[\frac{N}{4}\left(6-b_2(B)\right)+\frac{1}{12}c_{2,e}+\frac{1}{4}\ell_\alpha a^\alpha+\frac16\kappa\right]V\,.
	\end{split}
\end{align}
Then
\begin{align*}
	\begin{split}
		&\pi_{2,*}\left[\pi_1^*\left(\text{ch}(\mathcal{O}_{D_e})\text{Td}(T_{X/B})\right)\right]\\
		=&N-\frac{N}{2} \left(a_\alpha+\frac{1}{N}\ell_\alpha\right)C^{\alpha\beta}D_\beta+\left[\frac{N}{4}\left(6-b_2(B)\right)+\frac{1}{12}c_{2,e}+\frac{1}{4}\ell_\alpha a^\alpha+\frac16\kappa\right]F\\
		=&N\text{ch}(\mathcal{O}_X)-\frac{N}{2} \left(a_\alpha+\frac{1}{N}\ell_\alpha\right)C^{\alpha\beta}\text{ch}(\mathcal{O}_{D_\beta})\\
		&+\left[-\frac{N}{4} \left(a_\alpha+\frac{1}{N}\ell_\alpha\right)C^{\alpha\beta}C_{\beta\beta}+\frac{N}{4}\left(6-b_2(B)\right)+\frac{1}{12}c_{2,e}+\frac{1}{4}\ell_\alpha a^\alpha+\frac16\kappa\right]F\,.
	\end{split}
	\label{eqn:toddTrafo2}
\end{align*}

\paragraph{D4-brane on $D_\alpha$:}
The Chern character of a brane $\mathcal{O}_{D_\alpha}$ that is supported on the vertical divisor $D_\alpha$ is given by
\begin{align}
	\begin{split}
		\text{ch}(\mathcal{O}_{D_\alpha})=&1-\exp(-D_\alpha)=D_\alpha-\frac12 c_{\alpha\alpha}F\,.
	\end{split}
\end{align}
The relevant terms of the product with the Todd genus are
\begin{align}
	\begin{split}
		\text{ch}(\mathcal{O}_{D_\alpha})\text{Td}(T_{X/B})=&\frac{1}{N}c_\gamma C^{\gamma\beta}D_\beta D_\alpha S+\ldots\,,
	\end{split}
\end{align}
such that
\begin{align}
	\pi_{2,*}\left[\pi_1^*\left(\text{ch}(\mathcal{O}_{D_\alpha})\text{Td}(T_{X/B})\right)\right]=c_\gamma C^{\gamma\beta}C_{\beta\alpha}F=c_\alpha F\,.
\end{align}

\paragraph{D2-brane on $C^\alpha$:}
The relevant terms of the product with the Todd genus are
\begin{align}
	\begin{split}
		\text{ch}(\mathcal{C}^\alpha)\text{Td}(T_{X/B})=&C^\alpha-\frac12a^\beta D_\beta C^\alpha+\ldots
		=C^\alpha-\frac12a^\alpha V+\ldots\,,
	\end{split}
\end{align}
such that
\begin{align}
    \begin{split}
	\pi_{2,*}&\left[\pi_1^*\left(\text{ch}(\mathcal{C}^\alpha)\text{Td}(T_{X/B})\right)\right]\\
    =&D^\alpha-\frac12a^\alpha F=C^{\alpha\beta}\text{ch}(\mathcal{O}_{D_\beta})+\frac12C^{\alpha\beta}(C_{\beta\beta}-a_\beta)F\,.
    \end{split}
\end{align}

\paragraph{D2-brane on $F$:}
We have
\begin{align}
	\pi_{2,*}\left[\pi_1^*\left(\text{ch}(\mathcal{O}_{F})\text{Td}(T_{X/B})\right)\right]=0\,.
\end{align}

\paragraph{D0-brane:}
We have
\begin{align}
	\pi_{2,*}\left[\pi_1^*\left(\text{ch}(\mathcal{O}_{\text{pt.}})\text{Td}(T_{X/B})\right)\right]=F\,.
\end{align}

\paragraph{Result}
We choose the following basis of the charge lattice
\begin{align}
    \left\{\,\text{ch}(\mathcal{O}_X),\,\text{ch}(\mathcal{O}_{D_e}),\,\text{ch}(\mathcal{O}_{D_\alpha}),\,\text{ch}(\mathcal{O}_{\mathcal{E}}),\,\text{ch}(\mathcal{C}^\alpha),\,\text{ch}(\mathcal{O}_{\text{pt.}})\,\right\}\,,
\end{align}
where $\text{ch}(\mathcal{O}_{\mathcal{E}})=\frac{1}{N}\text{ch}(\mathcal{O}_{F})$.
The end result then takes the form
\begin{align}
	U=\left(\begin{array}{cccccc}
		1&0&-c^\beta&\frac{N}{2}c^\gamma(a_\gamma-C_{\gamma\gamma})&0_\beta&0\\
		-N&1&\frac{N}{2}\left(a_\gamma+\frac{1}{N}\ell_\gamma\right)C^{\gamma\beta}&\rho&0_\beta&0\\
		0_\alpha&0_\alpha&\delta^{\beta}_\alpha&-N c_\alpha&0_{\alpha\beta}&0_\alpha\\
		0&0&0^\beta&1&0_\beta&0\\
		0^\alpha&0^\alpha&-C^{\alpha\beta}&\frac{N}{2}C^{\alpha\gamma}\left(a_\gamma-C_{\gamma\gamma}\right)&\delta^{\alpha}_{\beta}&0^\alpha\\
		0&0&0^\beta&-N&0_\beta &1
	\end{array}\right)\,,
\end{align}
with
\begin{align}
    \begin{split}
	\rho=&N^2+\frac{N}{4}(\ell_\alpha+Na_\alpha)C^{\alpha\beta}(C_{\beta\beta}-a_\beta)-\frac{N}{12}(2\kappa+c_{2,e})\,.
	% \rho=&-N\left[-\frac{N}{4} \left(a_\alpha+\frac{1}{N}\ell_\alpha\right)C^{\alpha\beta}C_{\beta\beta}+\frac{N}{4}\left(6-b_2(B)\right)+\frac{1}{12}c_{2,e}+\frac{1}{4}\ell_\alpha a^\alpha+\frac16\kappa\right]\\
 %    =&\frac{N^2}{4}\left(a_\alpha c^{\alpha\beta}c_{\beta\beta}+b_2(B)-6\right)+\frac{N}{4}\ell_\alpha C^{\alpha\beta}(C_{\beta\beta}-a_\beta)-\frac{N}{12}(2\kappa+c_{2,e})\\
 %    =&\frac{N^2}{4}a_\alpha C^{\alpha\beta}(C_{\beta\beta}-a_\beta)+N^2+\frac{N}{4}\ell_\alpha C^{\alpha\beta}(C_{\beta\beta}-a_\beta)-\frac{N}{12}(2\kappa+c_{2,e})\\
 %    =&N^2+\frac{N}{4}(\ell_\alpha+Na_\alpha)C^{\alpha\beta}(C_{\beta\beta}-a_\beta)-\frac{N}{12}(2\kappa+c_{2,e})\,.
    \end{split}
\end{align}
If the curves in the classes $C_\alpha$ are rational one can use that $c_{\alpha\alpha}-a_\alpha=-2$ to obtain the expression  
\begin{align}
	\rho=-N\left[\frac{1}{2}\sum_{\beta=1}^{b_2(B)}(Na^\beta+\ell_\alpha C^{\alpha\beta})-N+\frac16\kappa\right]+\frac{1}{12}c_{2,e}\,.
\end{align}

\section{Genus one fibrations}
\label{sec:genusoneconstruction}
In this section we will discuss the construction of genus one fibered CY threefolds $\pi:X\rightarrow B$ that exhibit an $N$-section for $N=1,\ldots 5$.
We assume that the base $B$ is a generalized del Pezzo surface and that $X$ is smooth with $b_2(X)=b_2(B)+1$.

For $N=1$, i.e. when the fibration admits a section, the topology of the fibration is uniquely determined by the choice of basis $B$.
It can be realized as a fibration of Weierstra{\ss} curves, that is a hypersurface
\begin{align}
	\{\,y^2=x^3+fxz^4+gz^6\,\}\subset\mathbb{P}\left(\mathcal{L}^2\oplus\mathcal{L}^3\oplus\mathcal{O}_B\right)\,,
\end{align}
where the CY condition requires $\mathcal{L}$ to be the anti-canonical bundle on $B$ and the Weierstra{\ss} coefficients $f,g$ are respectively sections of $\mathcal{L}^4$ and $\mathcal{L}^6$.
All of the fibers of $X$ are irreducible and the topological invariants are listed in Table~\ref{tab:fibrationsGenericData}.

The situation becomes more interesting for fibrations that only admit an $N$-section with $N\ge 2$.
The fibration will then exhibit isolated $I_2$-fibers, where the torus degenerates into two rational curves that intersect in two points.
Since the $N$-section by definition intersects the generic smooth fiber $N$ times, it intersects the two rational curves of a given $I_2$-fiber respectively $q$ and $N-q$ times for some $1\le q\le\lfloor{\tiny\frac{N}{2}}\rfloor$.~\footnote{Strictly speaking, $q$ could be negative if the $N$-section degenerates and wraps a 1-dimensional component of the $I_2$-fiber. However, we don't encounter this situation for $2\le N\le 5$.}
We will denote the number of such $I_2$-fibers by $n_{\pm q}\in\mathbb{N}$.
It is related to the multiplicity $N_k$ of fibral curves intersecting the $N$-section $k$ times, that has been introduced in Section~\ref{sec_Jacobi}, via $n_{\pm q}=\frac12(N_q+N_{N-q})$.
The topology of the fibration then depends not only on the base $B$, but also on the numbers of $I_2$-fibers $n_{\pm q}$ for $1\le q\le\lfloor{\tiny\frac{N}{2}}\rfloor$ and, as was observed in~\cite{Dierigl:2022zll} for $N=3$, on the height-pairing of the $N$-section.

As we will describe below, for $N\ge 2$ the fibration can always be realized as a double cover (for $N=2$) or a subvariety of a $\mathbb{P}^{N-1}$-bundle on $B$.
For $2\le N\le 4$, this follows from a fiberwise application of the constructions summarized in~\cite{An2001}, for $N=5$ from the construction discussed in~\cite{Fisher2008} and for $N\ge 6$ from the results from~\cite{Fisher2018}.
For $2\le N\le 5$ we will derive closed expressions for the topological invariants of the fibration in terms of the invariants of a set of vector bundles that appear in this construction.
The resulting expressions are also listed in Table~\ref{tab:fibrationsGenericData}.

\begin{table}[H]
	\centering
	\renewcommand{\arraystretch}{1.1}
	\begin{tabular}{|c|c|c|c|l|}\hline
		$N$ & Description & Bundles & \multicolumn{2}{|c|}{Invariants} \\\hline
		\multirow{5}{*}{1} & \multirow{5}{*}{\parbox{4cm}{\centering Hypersurface in\\$\mathbb{P}\left(\mathcal{L}^2\oplus\mathcal{L}^3\oplus \mathcal{O}_B\right)$}} & \multirow{5}{*}{$\mathcal{L}=-K_B$} & $\chi_X$ &$-60c_1(B)^2$ \\\cline{4-5}
		&&&$\kappa$& $c_1(B)^2$\\\cline{4-5}
		&&&$\ell_\alpha$&$\check{D}_\alpha\cap c_1(B)$ \\\cline{4-5}
		&&&$D$&$-c_1(B)$ \\\cline{4-5}
		&&&$c_2$& $12+10c_1(B)^2$\\\hline
		\multirow{6}{*}{2} & \multirow{6}{*}{\parbox{4cm}{\centering Double cover of\\$\mathbb{P}^1$-bundle $\mathbb{P}(V)$}} & \multirow{6}{*}{$\text{rk}(V)=2$} & $\chi_X$ &$-4\left(7c_1(B)^2-2\Delta(V)\right)$ \\\cline{4-5}
		&&&$\kappa$& $2\left(c_1(V)^2-c_2(V)\right)$\\\cline{4-5}
		&&&$\ell_\alpha$&$-2\check{D}_\alpha\cap c_1(V)$ \\\cline{4-5}
		&&&$D$&$2c_1(V)$ \\\cline{4-5}
		&&&$c_2$& $24-6c_1(B)c_1(V)-2\Delta(V)$\\\cline{4-5}
		&&&$n_{\pm 1}$& $4\left(4c_1(B)^2+\Delta(V)\right)$\\\hline
		\multirow{6}{*}{3} & \multirow{6}{*}{\parbox{4cm}{\centering Hypersurface in\\$\mathbb{P}^2$-bundle $\mathbb{P}(V)$}} & \multirow{6}{*}{$\text{rk}(V)=3$} & $\chi_X$ & $3\Delta(V)-18c_1(B)^2$\\\cline{4-5}
		&&&$\kappa$& $2c_1(V)^2-c_1(B)c_1(V)-3c_2(V)$\\\cline{4-5}
		&&&$\ell_\alpha$&$\check{D}_\alpha\cap \left(c_1(B)-2c_1(V)\right)$ \\\cline{4-5}
		&&&$D$&$2c_1(V)-c_1(B)$ \\\cline{4-5}
		&&&$c_2$& $36-4c_1(B)c_1(V)-\Delta(V)$\\\cline{4-5}
		&&&$n_{\pm 1}$& $\frac12\left(42c_1(B)^2+3\Delta(V)\right)$
        \\\hline
		\multirow{9}{*}{4} & \multirow{9}{*}{\parbox{4cm}{\centering Complete intersection\\of two relative quadrics\\in $\mathbb{P}^3$-bundle $\mathbb{P}(V)$\\(vanishing locus of section of $F$)}} & \multirow{9}{*}{\pbox{4cm}{$\text{rk}(V)=4$\\$\text{rk}(E)=2$}} & $\chi_X$ & $-13c_1(B)^2+\Delta(V)+3\Delta(E)$ \\\cline{4-5}
		&&&\multirow{2}{*}{$\kappa$}& $2c_1(V)\left(c_1(V)-c_1(B)\right)$\\
		&&&& $-4c_2(V)+c_2(E)$\\\cline{4-5}
		&&&$\ell_\alpha$&$2\check{D}_\alpha\cap\left(c_1(B)-c_1(V)\right)$ \\\cline{4-5}
		&&&$D$&$2\left(c_1(V)-c_1(B)\right)$ \\\cline{4-5}
		&&&\multirow{2}{*}{$c_2$}& $48-\frac12\left[c_1(B)\left(c_1(B)+6c_1(V)\right)\right.$\\
		&&&& $\left.+\Delta(V)+\Delta(E)\right]$\\\cline{4-5}
		&&&$n_{\pm 1}$& $4\left(4c_1(B)^2+\Delta(E)\right)$\\\cline{4-5}
		&&&$n_{\pm 2}$& $\frac12\left(15c_1(B)^2+\Delta(V)-5\Delta(E)\right)$\\\hline
		\multirow{9}{*}{5} & \multirow{9}{*}{\parbox{4cm}{\centering Pfaffian variety in\\$\mathbb{P}^4$-bundle $\mathbb{P}(V)$\\(rank 2 locus of skew-symmetric map~\eqref{eqn:gen5skewsymmetric})}} & \multirow{9}{*}{\pbox{4cm}{$\text{rk}(V)=5$\\$\text{rk}(E)=5$}} & $\chi_X$ & $-10c_1(B)^2+\Delta(E)$ \\\cline{4-5}
		&&&\multirow{2}{*}{$\kappa$}& $\frac15\left(3 c_1(B)^2 - 9 c_1(B) c_1(V)\right.$\\
		&&&& $\left. + 13 c_1(V)^2 - 25 c_2(V)+\frac12\Delta(E)\right)$\\\cline{4-5}
		&&&$\ell_\alpha$&$\check{D}_\alpha\cap\left(3c_1(B)-2c_1(V)\right)$ \\\cline{4-5}
		&&&$D$&$2c_1(V)-3c_1(B)$ \\\cline{4-5}
            &&&\multirow{2}{*}{$c_2$}& $60-\frac15\left[6c_1(B)\left(c_1(B)+2c_1(V)\right)\right.$\\
		&&&& $\left.+\Delta(V)+\Delta(E)\right]$\\\cline{4-5}
		&&&$n_{\pm 1}$& $13c_1(B)^2+\Delta(E)-\frac12\Delta(V)$\\\cline{4-5}
		&&&$n_{\pm 2}$& $12c_1(B)^2-\frac12\Delta(E)+\frac12\Delta(V)$\\\hline
	\end{tabular}
	\caption{Topological invariants of generic genus one fibered CY threefolds with an $N$-section for $N\le 5$.}
	\label{tab:fibrationsGenericData}
\end{table}

\subsection{Projective bundles}
To construct generic genus one fibrations it will be useful to first recall some generic properties of projective bundles.
Given a vector bundle $\pi_V:V\rightarrow B$ of rank $r$ on a base $B$, the projectivization is a $\mathbb{P}^{r-1}$-bundle $\pi_{\mathbb{P}(V)}:\mathbb{P}(V)\rightarrow B$ on the same base.~\footnote{We are following the convention that $\mathbb{P}(V)$ is, as a complex manifold, the fiberwise projectivization of $V$. This differs from the convention, often used in algebraic geometry, that $\mathbb{P}(V)$ is the fiberwise projectivization of $V^\vee$. When comparing with expressions from the literature the reader is advised to check which conventions have been adopted.}
If the Brauer group of $B$ is trivial, which is the case for generalized del Pezzo surfaces, every projective bundle arises as the projectivization of a vector bundle.

The space $\mathbb{P}(V)$ is equipped with a relative tautological line bundle $\mathcal{O}_{\mathbb{P}(V)}(-1)$ and the inverse is the relative hyperplane bundle $\mathcal{O}_{\mathbb{P}(V)}(1)$.
The so-called tautological exact sequence takes the form
\begin{align}
	0\rightarrow \mathcal{O}_{\mathbb{P}(V)}(-1)\rightarrow \pi_{\mathbb{P}(V)}^*(V)\rightarrow Q\rightarrow 0\,,
	\label{eqn:tautologicalSequence}
\end{align}
in terms of the relative quotient bundle $Q$ on $\mathbb{P}(V)$.
The relative tangent bundle of $\mathbb{P}(V)$ over $B$ is given by
\begin{align}
	T_{\mathbb{P}(V)/B}=\mathcal{O}_{\mathbb{P}(V)}(1)\otimes Q\,,
	\label{eqn:relativeTangent}
\end{align}
and fits into the short exact sequence
\begin{align}
	0\rightarrow T_{\mathbb{P}(V)/B}\rightarrow T_{\mathbb{P}(V)}\rightarrow \pi_{\mathbb{P}(V)}^*T_B\rightarrow 0\,.
	\label{eqn:relativeTangentSequence}
\end{align}

The cohomology of $\mathbb{P}(V)$ is generated by the class $H_V=c_1\left(\mathcal{O}_{\mathbb{P}(V)}(1)\right)$ together with pullbacks of classes from $B$, subject to the relation
\begin{align}
	H_V^r+H_V^{r-1}\pi_V^*\left(c_1(V)\right)+\ldots+\pi_V^*\left(c_r(V)\right)=0\,.
	\label{eqn:relHyperplaneRelation}
\end{align}
By combining~\eqref{eqn:tautologicalSequence},~\eqref{eqn:relativeTangent} and~\eqref{eqn:relativeTangentSequence} one can show that the canonical class of $\mathbb{P}(V)$ takes the form
\begin{align}
	K_{\mathbb{P}(V)}=-rH_V+\pi_{\mathbb{P}(V)}^*\left(K_B-\det(V)\right)\,.
    \label{eqn:pbundleCanonicalClass}
\end{align}

Given any line bundle $L$ on $B$, we denote $V_L:=V\otimes L$.
One then has an isomorphism $f:\mathbb{P}(V)\rightarrow\mathbb{P}(V_L)$ and the corresponding relative hyperplane bundles are related as
\begin{align}
	\mathcal{O}_{\mathbb{P}(V)}(1)=\pi_V^*(L)\otimes f^*(\mathcal{O}_{\mathbb{P}(V_L)}(1))\,.
\end{align}
Therefore, the projective bundle itself is invariant under a twist of the vector bundle but the relative hyperplane bundle depends on the concrete choice of $V$.
The so-called Bogomolov discriminant
\begin{align}
	\Delta(V)=2rc_2(V)-(r-1)c_1(V)^2\,,
\end{align}
does not depend on the choice of twist, i.e. $\Delta(V)=\Delta(V_L)$ for all line bundles $L$ on $B$.
It is therefore an invariant of the projective bundle $\mathbb{P}(V)$.
One can always find a line bundle $L'$ on $B$ such that $H_{V_{L'}}$, together with the vertical divisors that arise via pullback from the base, forms a basis of the K\"ahler cone on $\mathbb{P}(V)$.

We denote a basis of vertical divisors on $\mathbb{P}(V)$ by $J_\alpha=\pi_{\mathbb{P}(V)}^*\check{D}_\alpha$ for $\alpha=1,\ldots,b_2(B)$.

\subsection{2-sections}
\label{sec:generic2sections}
We start by discussing the construction of a generic genus one fibered CY threefold $\pi:X\rightarrow B$ with a $2$-section $D_e$.

The restriction $D_\Sigma= D_e\vert_{\Sigma}$ of the 2-section to the generic fiber $\Sigma$ is a divisor of degree 2 on $\Sigma$.
Recall that the line bundle that is associated to a divisor of degree $d$ on a curve of genus $1$ always has $d$ global sections.
We can therefore denote the global sections of $nD_\Sigma$ for $n=1,2$ by
\begin{align}
	\begin{split}
		\Gamma\left(\Sigma,D_\Sigma\right)=&\{\,X,\,Y\,\}\,,\\
		\Gamma\left(\Sigma,2D_\Sigma\right)=&\{\,Z,\,X^2,\,XY,\,Y^2\,\}\,.
	\end{split}
\end{align}
Then $4D_\Sigma$ has eight sections but there are nine monomials of weighted degree four in $X,Y$ and $Z$.
As a result, $\Sigma$ can be realized as a hypersurface of degree four in $\mathbb{P}^2_{112}$.

Without loss of generality,~\footnote{The coefficient of $Z^2$ is not allowed to vanish because $\Sigma$ is smooth and the terms linear in $Z$ can then be removed by a coordinate redefinition.} we can write this as
\begin{align}
	\Sigma=\{\,Z^2=Q(X,Y)\,\}\subset\mathbb{P}^2_{112}\,,
\end{align}
where $Q(X,Y)$ is a homogeneous polynomial of degree four in $X,Y$.
Since $[0:0:1]\notin\Sigma$, we have a projection
\begin{align}
	\tilde{\pi}_{\text{dc}}:\Sigma\rightarrow\mathbb{P}^1\,,\qquad [X:Y:Z]\rightarrow [X:Y]\,,
\end{align}
which identifies $\Sigma$ as a double cover of $\mathbb{P}^1$.
The double cover is ramified over the four points
\begin{align}
	\{\,Q(X,Y)=0\,\}\subset\mathbb{P}^1\,.
\end{align}

This construction can be applied fiberwise to $\pi:X\rightarrow B$ and we see that $X$ can be realized as a Calabi-Yau double cover of a $\mathbb{P}^1$-bundle over $B$.~\footnote{For a general genus one fibration $\pi:X'\rightarrow B$ with a 2-section that has $b_2(X')>b_2(B)+1$ this construction would give a threefold that is only birationally equivalent to $X'$. This is because $X'$ then exhibits additional $N$-sections, fibral divisors or non-flat fibers and, as a consequence, some components of reducible fibers are contracted in the double cover.}
Assuming that the Brauer group of $B$ is trivial we can then find a rank two vector bundle $V$ such that the $\mathbb{P}^1$-bundle takes the form $\mathbb{P}(V)$.
Using~\eqref{eqn:relHyperplaneRelation} and~\eqref{eqn:pbundleCanonicalClass} we obtain
\begin{align}
	\begin{split}
		H_V\cap H_V\cap H_V=&H_V\cap \pi_V^*\left[c_1(V)^2-c_2(V)\right]=\left(c_1(V)^2-c_2(V)\right)\,,\\
		H_V\cap H_V\cap J_\alpha=&-\check{D}_\alpha\cap c_1(V)\,,\quad H_V\cap J_\alpha\cap J_\beta=\check{D}_\alpha\cap \check{D}_\beta\,.
	\end{split}
	\label{eqn:intP1bundle}
\end{align}

Let us denote the double covering map by $\pi_{\text{dc}}:X\rightarrow\mathbb{P}(V)$.
By twisting $V$ with a line bundle if necessary, we can assume that the relative hyperplane class $H_V$ is part of a K\"ahler cone basis on $\mathbb{P}(V)$ and, in particular, $H_V$ is effective.
Then $\pi_{\text{dc}}^*(H_V)$ is also effective and we can assume that $D_e=\pi_{\text{dc}}^*(H_V)$.

With~\eqref{eqn:intP1bundle}, the non-trivial intersection numbers on $X$ are given by
\begin{align}
	\kappa=2\left(c_1(V)^2-c_2(V)\right)\,,\quad \ell_\alpha=-2\check{D}_\alpha\cap c_1(V)\,,\quad \kappa_{e\alpha\beta}=2\check{D}_\alpha\cap \check{D}_\beta\,.
    \label{eqn:gen2secInts}
\end{align}
This implies that the height pairing of the 2-section is
\begin{align}
	D=-\pi^*\pi_*(D_eD_e)=2c_1(V)\,.
\end{align}
Using for example the method of GV-spectroscopy~\cite{Oehlmann:2019ohh}, one can deduce that the Euler characteristic $\chi_X$ and the number $n_{\pm 1}$ of $I_2$-fibers of $X$ are given by
\begin{align}
	\chi_X=-4\left(7c_1(B)^2-2\Delta(V)\right)\,,\quad n_{\pm 1}=4\left(4c_1(B)^2+\Delta(V)\right)\,,
    \label{eqn:gen2secInvs}
\end{align}
in terms of the Bogomolov discriminant
\begin{align}
	\Delta(V)=4c_2(V)-c_1(V)^2\,.
\end{align}

Using~\eqref{eqn:gen2secInvs} with the relations~\eqref{FtheoryRel} from~\cite{Duque:2025kaa}, we find that
\begin{align}
	\widehat{\kappa}=-\frac12\Delta(V)\,.\quad \widehat{c}_2=24-2\Delta(V)\,.
\end{align}
The expression for $\widehat{\kappa}$ can easily be checked against~\eqref{eqn:gen2secInts}.
On the other hand, from the expression for $\widehat{c}_2$ we can then deduce that
\begin{align}
	c_2=24-6c_1(B)c_1(V)-2\Delta(V)\,.
\end{align}
Note that we could also just calculate $c_2$ directly, which for $N=2$ would only require slightly more work. 
The benefit of this indirect method is that it significantly simplifies the calculation of $c_2$ for $N\ge 4$, as long as the multiplicities of $I_2$-fibers $n_{\pm q}$ can be deduced using GV-spectroscopy.

\subsection{3-sections}
\label{sec:generic3sections}
Let us now consider the case where $\pi:X\rightarrow B$ exhibits a 3-section $D_e$.
We can denote the global sections of the restriction $D_\Sigma= D_e\vert_{\Sigma}$ by
\begin{align}
	\begin{split}
		\Gamma\left(\Sigma,D_\Sigma\right)=&\{\,X,\,Y,\,Z\,\}\,.\\
	\end{split}
\end{align}
One then notes that $3D_\Sigma$ has nine global sections but there are ten monomials of degree three in $X,Y$ and $Z$.
As a result, $\Sigma$ can be realized as a hypersurface of degree three in $\mathbb{P}^2$.

Again this construction can be applied fiberwise to $\pi:X\rightarrow B$.
In this way one obtains a fibration $X'$ that is birationally equivalent or, if $X$ is sufficiently generic, isomorphic to $X$.
To simplify the exposition, we will again assume that $X'$ is isomorphic to $X$.
We then see that $X$ can be realized as an anti-canonical hypersurface in a $\mathbb{P}^2$-bundle over $B$.

We can then find a rank three vector bundle $V$ such that the $\mathbb{P}^2$-bundle takes the form $\mathbb{P}(V)$ and the restriction of the relative hyperplane class is $D_e=H_V\vert_X$.
Using~\eqref{eqn:relHyperplaneRelation} we obtain
\begin{align}
	\begin{split}
		\kappa=&H_V\cap H_V\cap H_V\cap (-K_{\mathbb{P}(V)})=2c_1(V)^2-c_1(B)c_1(V)-3c_2(V)\,,\\
		\ell_\alpha=&H_V\cap H_V\cap J_\alpha\cap(-K_{\mathbb{P}(V)}) =D_\alpha\cap \left(c_1(B)-2c_1(V)\right)\,.
	\end{split}
	\label{eqn:intP2bundle}
\end{align}
This implies that the height pairing of the 3-section is
\begin{align}
	D=-\pi^*\pi_*(D_eD_e)=2c_1(V)-c_1(B)\,.
\end{align}

Using for example the method of GV-spectroscopy~\cite{Oehlmann:2019ohh}, one can deduce that the Euler characteristic $\chi_X$ and the number $n_{\pm 1}$ of $I_2$-fibers of $X$ are given by
\begin{align}
	\chi_X=3\Delta(V)-18c_1(B)^2\,,\quad n_{\pm 1}=\frac12\left(42c_1(B)^2+3\Delta(V)\right)\,,
\end{align}
in terms of the Bogomolov discriminant
\begin{align}
	\Delta(V)=2\left(3c_2(V)-c_1(V)^2\right)\,.
\end{align}

Using~\eqref{eqn:gen2secInvs} with the relations~\eqref{FtheoryRel} from~\cite{Duque:2025kaa}, we find that
\begin{align}
	\widehat{\kappa}=-\frac14\left(c_1(B)^2+2\Delta(V)\right)\,,\quad \widehat{c}_2=36-2c_1(B)^2-\Delta(V)\,.
\end{align}
Again, the expression for $\widehat{\kappa}$ can easily be checked against~\eqref{eqn:intP2bundle} and from the expression for $\widehat{c}_2$ we can deduce that
\begin{align}
	c_2=36-4c_1(B)c_1(V)-\Delta(V)\,.
\end{align}

\subsection{4-sections}
\label{sec:generic4sections}
The restriction $D_\Sigma= D_e\vert_{\Sigma}$ of the 4-section $D_e$ to a generic fiber $\Sigma$ is  a divisor of degree 4 on $\Sigma$.
We can therefore denote the global sections of $D_\Sigma$ by
\begin{align}
	\begin{split}
		\Gamma\left(\Sigma,D_\Sigma\right)=&\{\,W,\,X,\,Y,\,Z\,\}\,.\\
	\end{split}
\end{align}
One then notes that $2D_\Sigma$ has eight global sections but there are ten monomials of degree two in $W,X,Y$ and $Z$.
As a result, $\Sigma$ can be realized as a complete intersection of two quadrics in $\mathbb{P}^3$.

Again this construction can be applied fiberwise to $\pi:X\rightarrow B$ in order to identify $X$ with a complete intersection in a $\mathbb{P}^3$-bundle on $B$.
We can find a rank four vector bundle $V$ such that the $\mathbb{P}^3$-bundle takes the form $\mathbb{P}(V)$.
One can then find a rank two vector bundle $E$ on $B$ such that $X$ is the vanishing locus of a generic section of
\begin{align}
	F=\mathcal{O}_{\mathbb{P}(V)}(2)\otimes \pi_{\mathbb{P}(V)}^*(E)\,.
\end{align}
The Calabi-Yau condition takes the form
\begin{align}
    c_1(E)=c_1(B)+c_1(V)\,.
\end{align}

We assume again that $V$ is chosen such that $H_V$ is effective and the restriction of the relative hyperplane class is $D_e=H_V\vert_X$.
Using~\eqref{eqn:relHyperplaneRelation} we obtain
\begin{align}
	\begin{split}
		\kappa=&H_V\cap H_V\cap H_V\cap c_2(F)\\
		=&2c_1(V)\left(c_1(V)-c_1(B)\right)-4c_2(V)+c_2(E)\,,\\
		\ell_\alpha=&H_V\cap H_V\cap J_\alpha \cap c_2(F) =2\check{D}_\alpha\cap\left(c_1(B)-c_1(V)\right)\,.
	\end{split}
	\label{eqn:intP3bundle}
\end{align}
This implies that the height pairing of the 4-section is
\begin{align}
	D=-\pi^*\pi_*(D_eD_e)=2\left(c_1(V)-c_1(B)\right)\,.
\end{align}
Using the method of GV-spectroscopy~\cite{Oehlmann:2019ohh}, one can deduce that the Euler characteristic $\chi_X$ and the numbers $n_{\pm 1},n_{\pm 2}$ of $I_2$-fibers of $X$ are given by
\begin{align}
	\begin{split}
		\chi_X=&-13c_1(B)^2+\Delta(V)+3\Delta(E)\,,\\
		n_{\pm 1}=&4\left(4c_1(B)^2+\Delta(E)\right)\,,\\
		n_{\pm 2}=&\frac12\left(15c_1(B)^2+\Delta(V)-5\Delta(E)\right)\,.
	\end{split}
    \label{eqn:gen3secInvs}
\end{align}

Using~\eqref{eqn:gen3secInvs} with the relations~\eqref{FtheoryRel} from~\cite{Duque:2025kaa}, we find that
\begin{align}
    \begin{split}
	\widehat{\kappa}=&-\frac14\left(2c_1(B)^2+2\Delta(V)-\Delta(E)\right)\,,\\
    \widehat{c}_2=&\frac12\left(96-7c_1(B)^2-\Delta(V)-\Delta(E)\right)\,.
    \end{split}
\end{align}
The expression for $\widehat{\kappa}$ can again be checked against~\eqref{eqn:intP3bundle} and from the expression for $\widehat{c}_2$ we deduce that
\begin{align}
	c_2=\frac12\left(96+5c_1(B)^2-6c_1(B)c_1(E)-\Delta(V)-\Delta(E)\right)\,.
\end{align}

\subsection{5-sections}
\label{sec:generic5sections}
The case where the fibration $\pi:X\rightarrow B$ only exhibits a 5-section $D_e$ is the first that requires us to go beyond complete intersections.
A detailed study of such CY threefolds was carried out in~\cite{Knapp:2021vkm}.

One can show that a curve of degree 5 can always be realized as the vanishing locus of the Pfaffian of a skew-symmetric $5\times 5$ matrix with entries that are linear homogeneous polynomials in the homogeneous coordinates on $\mathbb{P}^4$~\cite{Fisher2008}.
The genus one curve is then the locus of codimension 3 inside $\mathbb{P}^4$ where the rank of the matrix drops to two.
Since the matrix is skew-symmetric, the locus where the rank is strictly less than two is equal to that where the rank is zero.
This locus has codimension 10 in the space of skew-symmetric $5\times 5$ matrices and is therefore avoided by the curve, if the entries of the matrix are sufficiently generic.

Lifting this construction to fibrations, one can always represent a generic genus one fibration with a 5-section (up to birational equivalence) as the rank $2$ locus $D_2(\phi)$ of a skew-symmetric map
\begin{align}
	\phi:\,\pi_{\mathbb{P}(V)}^*E\rightarrow \pi_{\mathbb{P}(V)}^*\left(E^\vee \otimes L\right)\otimes\mathcal{O}_{\mathbb{P}(V)}(1)\,,
	\label{eqn:gen5skewsymmetric}
\end{align}
where $V$ is a vector bundle of rank $5$ on $B$, $E$ is a vector bundle of rank $5$ on $B$ and $L$ is a line bundle on $B$.
The Calabi-Yau condition takes the form~\cite{Knapp:2021vkm}
\begin{align}
	c_1(L)=\frac15\left(2c_1(E)+c_1(V)+c_1(B)\right)\,.
	\label{eqn:5secCYcondition}
\end{align}

In order to calculate the cohomology class of the degeneracy locus $D_2(\phi)$ we use~\cite[Theorem 8]{Harris1984}, which implies that
\begin{align}
	[D_2(\phi)]=\det\left(\begin{array}{cc}
		c_2&c_3\\
		c_0&c_1
	\end{array}\right)\,,
\end{align}
where, using~\cite[Theorem 10]{Harris1984}, we formally define
\begin{align}
	c_i=c_i\left(\mathcal{O}_{\mathbb{P}(V)}(1/2)\otimes \pi_{\mathbb{P}(V)}^*\left(E^\vee \otimes L^{1/2}\right)\right)\,.
\end{align}
One then obtains
\begin{align}
	\begin{split}
	[D_2(\phi)]=&H_V \left(5H_V^2- 3H_V\left(2 c_1(E) + 5  c_1(L)\right)+2 c_1(E)^2\right.\\
		&\left.+ c_2(E)  - 12 c_1(E) c_1(L)  + 15 c_1(L)^2\right)\,.
	\end{split}
	\label{eqn:5secFundamental}
\end{align}

We again assume that $V$ is chosen such that $H_V$ is effective and therefore the restriction $D_e=H_V\vert_X$ corresponds to a 5-section on $X$.
We denote the restriction of the relative hyperplane class by $D_e=H_V\vert_X$.
One can always choose $V$, by twisting with a line bundle if necessary, such that $D_e$ is part of a K\"ahler cone basis of $X$.
It is then in particular effective and corresponds to a 5-section of the fibration.

Using~\eqref{eqn:5secFundamental} and~\eqref{eqn:5secCYcondition} together with~\eqref{eqn:relHyperplaneRelation}, we obtain
\begin{align}
	\begin{split}
		\kappa=&H_V\cap H_V\cap H_V\cap [D_2(\phi)]\\
		=&\frac15\left(3 c_1(B)^2 - 9 c_1(B) c_1(V) + 13 c_1(V)^2 - 25 c_2(V)+\frac12\Delta(E)\right)\,,\\
		\ell_\alpha=&H_V\cap H_V\cap J_\alpha \cap [D_2(\phi)] =\check{D}_\alpha\left(3c_1(B)-2c_1(V)\right)\,,
	\end{split}
	\label{eqn:gen5secInts}
\end{align}
where we have used the Bogomolov discriminant of $E$,
\begin{align}
	\Delta(E)=2\left(5c_2(E)-2c_1(E)^2\right)\,.
	\label{eqn:bogomolov5sec}
\end{align}
This implies that the height pairing of the 5-section is
\begin{align}
	D=-\pi^*\pi_*(D_eD_e)=2c_1(V)-3c_1(B)\,.
	\label{eqn:heightPairing5sec}
\end{align}
In~\cite{Knapp:2021vkm} the method of GV-spectroscopy~\cite{Oehlmann:2019ohh} was used to deduce that the Euler characteristic $\chi_X$ and the numbers $n_{\pm 1},n_{\pm 2}$ of $I_2$-fibers of $X$ are given by
\begin{align}
	\begin{split}
		\chi_X=&-10c_1(B)^2+\Delta(E)\,,\\
		n_{\pm1}=&13c_1(B)^2+\Delta(E)-\frac12\Delta(V)\,,\\
		n_{\pm2}=&12c_1(B)^2-\frac12\Delta(E)+\frac12\Delta(V)\,,
	\end{split}
\end{align}
where the Bogomolov discriminant $\Delta(V)$ of $V$ also takes the form~\eqref{eqn:bogomolov5sec}.

Using~\eqref{eqn:gen2secInvs} with the relations~\eqref{FtheoryRel} from~\cite{Duque:2025kaa}, we find that
\begin{align}
	\begin{split}
		\widehat{\kappa}=&\frac{1}{20}\left(-15c_1(B)^2-10\Delta(V)+2\Delta(E)\right)\,,\\
		\widehat{c}_2=&\frac15\left(300-24c_1(B)^2-\Delta(V)-\Delta(E)\right)\,.
	\end{split}
\end{align}
The expression for $\widehat{\kappa}$ can be checked against~\eqref{eqn:gen5secInts} and from the expression for $\widehat{c}_2$ we can deduce that
\begin{align}
	c_2=\frac15\left(300-6c_1(B)^2-12c_1(B)c_1(V)-\Delta(V)-\Delta(E)\right)\,.
    \label{eqn:gen5secC2}
\end{align}

\paragraph{Relative homologically projective dual fibrations}
As was discussed in~\cite{Knapp:2021vkm}, a genus one curve of degree $5$ also admits a dual realization as a codimension five complete intersection in the Grassmanian $\text{Gr}(2,5)$.

At the level of the genus one fibered threefolds with a 5-section, the fiberwise application of this duality leads to a second fibration $\pi^\vee:X^\vee\rightarrow B$ that also exhibits a 5-section but is, in general, topologically different from $X$.
However, both geometries share the same relative Jacobian fibration and are therefore elements of the same Tate-Shafarevich group.
The geometries $X$ and $X^\vee$ are related by a relative version of homological projective duality~\cite{Kuznetsov2007} and they correspond to different large volume limits in the same stringy K\"ahler moduli space.

The fibration $X^\vee$ can be constructed as a complete intersection in a Grassmanian bundle on $B$ and the bundles that are involved in this construction are closely related to the bundles $V$ and $E$ that determine $X$ itself~\cite{Knapp:2021vkm}.
However, we can use the 5-section on $X^\vee$ to obtain a second realization as a Pfaffian variety in a $\mathbb{P}^4$-bundle $\mathbb{P}(V')$ on $B$, with a rank five bundle $E'$ and a line bundle $L'$ taking the place of $E$ and $L$.

We are currently not able to determine $V'$, $E'$ and $L'$ directly in terms of $V$, $E$ and $L$.
However, it was observed in~\cite{Knapp:2021vkm} that the fibrations $X$ and $X^\vee$ share the same Euler characteristic but the corresponding values of $n_{\pm 1}$ and $n_{\pm 2}$ are exchanged.
Formally, this can be achieved by assuming that
\begin{align}
\begin{split}
    c_1(E')=&-c_1(E)\,,\quad c_2(E')=c_2(E)\,,\quad c_1(V')=c_1(V)-c_1(E)\,,\\
    c_2(V')=&\frac15\left(c_1(B)^2+2\left(c_1(E)-c_1(V)\right)^2-\frac12\Delta(V)+\frac32\Delta(E)\right)\,,\\
   \end{split}
   \label{sec:gen5secBundleInvolution}
\end{align}
with $c_1(L')$ again being determined by the Calabi-Yau condition~\eqref{eqn:5secCYcondition}.
This choice also ensures that the intersection numbers~\eqref{eqn:gen5secInts} associated to the bundles are 
integral.
In fact, we have checked that -- up to a change of basis -- the involution~\eqref{sec:gen5secBundleInvolution} together with~\eqref{eqn:gen5secInts} and~\eqref{eqn:gen5secC2} correctly reproduces the topological invariants for $X^\vee$ in all of the examples that have been provided in~\cite[Table 19]{Knapp:2021vkm}.

We therefore make the following conjecture:
\begin{conj}
    Let $\pi:X\rightarrow B$ be a smooth genus one fibered CY threefold with a 5-section over a generalized del Pezzo surface $B$ with $b_2(X)=b_2(B)+1$.
    Assume that this is the degenaracy locus $D_2(\phi)$ of a generic skew-symmetric map~\eqref{eqn:gen5skewsymmetric} that is associated to the rank five vector bundles $V$, $E$ and the line bundle $L$ on $B$.
    Then there exist rank five vector bundles $V'$, $E'$ on $B$ with Chern classes~\eqref{sec:gen5secBundleInvolution}, and a line bundle $L'$ satisfying the Calabi-Yau condition
    \begin{align}
	c_1(L')=\frac15\left(2c_1(E')+c_1(V')+c_1(B)\right)\,,
	\label{eqn:5secCYconditionDual}
\end{align}
    such that the corresponding genus one fibered CY threefold $\pi^\vee:X^\vee\rightarrow B$ is the relative homological projective dual of $X$.
\end{conj}

\subsection{Projective bundles from monad bundles}
\label{sec:monadBundles}
In order to study the topological string free energies using mirror symmetry, we would like to work with geometries that can be constructed using the tools from toric geometry.
However, given a toric base $B$, the only projective bundles on $B$ that are themselves toric varieties are projectivizations of sums of line bundles.
To obtain a somewhat larger class of geometries we will also consider projective bundles that are complete intersections in toric ambient spaces.

We will therefore use monad bundles $V$, that are defined via a short exact sequence
\begin{align}
	0\rightarrow V\rightarrow W\xrightarrow{f} U\rightarrow 0\,,
	\label{eqn:monadExactSequence}
\end{align}
where $W$ and $U$ are vector bundles of respective rank $r_W$ and $r_U$, with $r_U<r_W$, such that the rank $r=r_V$ of $V$ is $r_V=r_W-r_U$.
Let us assume that $B$ is a smooth toric variety and that the divisors $\check{D}_\alpha$, $\alpha=1,\ldots,b_2(B)$ form a basis of the K\"ahler cone.~\footnote{Or, if the K\"ahler cone is non-simplicial, a suitable simplicial subcone of the K\"ahler cone.}
We choose $W$ and $U$ to take the form
\begin{align}
	\begin{split}
	W=&\mathcal{O}_B\left(-w^\alpha_1 \check{D}_\alpha\right)\oplus\ldots\oplus\mathcal{O}_B\left(-w^\alpha_{r_W} \check{D}_\alpha\right) \,,\\
	U=&\mathcal{O}_B\left(u^\alpha_1 \check{D}_\alpha\right)\oplus\ldots\oplus\mathcal{O}_B\left(u^\alpha_{r_U} \check{D}_\alpha\right) \,,
	\end{split}
\end{align}
in terms of $w^\alpha_i,u^\alpha_j\in\mathbb{N}^{b_2(B)}$ for $i=1,\ldots,r_W$ and $j=1,\ldots,r_U$.
In order for $V$ to be a vector bundle, we have to require that $u^\alpha_{j}\ge-w^\alpha_i$ for all $i,j,\alpha$ and that the map $f$ is chosen to be sufficiently generic, see for example~\cite{Anderson:2007nc,Anderson:2008uw}.

\section{Generic genus one fibrations on $\mathbb{P}^2$ \label{app_P2}}
\label{sec:examplesP2}
In this section we will construct a large number of generic genus one fibered CY threefolds over $\mathbb{P}^2$ that have an $N$-section with $2\le N\le 5$.

To construct the vector bundles that appear in the constructions for different $N$ we will use monad bundles as described in Section~\ref{sec:monadBundles}.
Over $\mathbb{P}^2$, the vector bundles $W,U$ take the form
\begin{align}
    \begin{split}
    W=&\mathcal{O}_{\mathbb{P}^2}(-w_1)\oplus\ldots\oplus\mathcal{O}_{\mathbb{P}^2}(-w_{r_W})\,,\\
    U=&\mathcal{O}_{\mathbb{P}^2}(u_1)\oplus\ldots\oplus\mathcal{O}_{\mathbb{P}^2}(u_{r_U})\,,
    \end{split}
\end{align}
for $\vec{w},\vec{u}\in\mathbb{N}$.
Using the exact sequence~\eqref{eqn:monadExactSequence}, we find the Chern classes of $V$,
\begin{align}
	c_1(V)=-\left(e_1(\vec{w})+e_1(\vec{u})\right)\,,\quad c_2(V)=e_2(\vec{w})+e_1(\vec{u})e_1(\vec{w})+e_1(\vec{u})^2-e_2(\vec{u})\,,
\end{align}
in terms of the $i$-th elementary symmetric polynomials $e_i$.
In fact, we will sometimes relax the condition that the entries of $\vec{w}$ are non-negative $\vec{u}$.
In those cases we have checked that the corresponding Calabi-Yau is still smooth and has $h^{1,1}=2$.

\subsection{2-sections}
Without loss of generality we can restrict to $r_W=5$, $r_U=3$, as well as $w_4=w_5=0$, and write
\begin{align}
	\begin{split}
	W=&\mathcal{O}_{\mathbb{P}^2}(-w_1)\oplus\mathcal{O}_{\mathbb{P}^2}(-w_2)\oplus\mathcal{O}_{\mathbb{P}^2}(-w_3)\oplus\mathcal{O}_{\mathbb{P}^2}^{\oplus 2}\,,\\
	U=&\mathcal{O}_{\mathbb{P}^2}(u_1)\oplus\mathcal{O}_{\mathbb{P}^2}(u_2)\oplus\mathcal{O}_{\mathbb{P}^2}(u_3)\,.
	\end{split}
\end{align}
The toric data associated to the Calabi-Yau double cover $X$ of $\mathbb{P}(V)$ is given in Table~\ref{tab:tdata2p2}, where we have introduced
\begin{align}
	p=3-e_1(\vec{w})-e_1(\vec{u})\,.
\end{align}
In order for $X$ to be smooth we have to impose $e_1(\vec{w})+e_1(\vec{u})\le 3$.
For all inequivalent values of $\vec{u},\vec{w}$ that satisfy this inequality we use CohomCalg~\cite{Blumenhagen:2010pv,cohomCalg:Implementation} to check if $h^{1,1}(X)=2$.
The resulting geometries, together with their topological invariants, are listed in Table~\ref{tab:gp22sec}.

\begin{table}[H]
\begin{align*}
	\left[\begin{array}{ccccccc|cc}
	\multicolumn{7}{c|}{\vec{p}\in\mathbb{Z}^7}& l^{(1)} & l^{(2)} \\\hline
	 1 & 0 & 0 & 0 & 0 & 0 & 0 & 2 & p   \\
	 0 & 1 & 0 & 0 & 0 & 0 & 0 & 1 &-w_1 \\
	 0 & 0 & 1 & 0 & 0 & 0 & 0 & 1 &-w_2 \\
	 0 & 0 & 0 & 1 & 0 & 0 & 0 & 1 &-w_3 \\
	 0 & 0 & 0 & 0 & 1 & 0 & 0 & 1 & 0   \\
	-2 &-1 &-1 &-1 &-1 & 0 & 0 & 1 & 0   \\
	 0 & 0 & 0 & 0 & 0 & 1 & 0 & 0 & 1   \\
	 0 & 0 & 0 & 0 & 0 & 0 & 1 & 0 & 1   \\
	-p &w_1&w_2&w_3& 0 &-1 &-1 & 0 & 1   \\\hline\hline
	 0 & 0 & 0 & 0 & 0 & 0 & 0 & 1 &u_1  \\
	 0 & 0 & 0 & 0 & 0 & 0 & 0 & 1 &u_2  \\
	 0 & 0 & 0 & 0 & 0 & 0 & 0 & 1 &u_3  \\
	 0 & 0 & 0 & 0 & 0 & 0 & 0 & 4 &2p   \\
	\end{array}\right]
\end{align*}
	\caption{The toric data associated to the genus one fibered CY threefolds over $\mathbb{P}^2$ with a 2-section listed in Table~\ref{tab:gp22sec}.}
	\label{tab:tdata2p2}
\end{table}

\begin{table}[H]
\begin{align*}
\begin{array}{|c|cc|ccc|cc|ccc|ccc|}\hline
	\#& \chi_X & N_1/2  & \kappa & \ell & c_2 & c_{1,V} & c_{2,V} & w_1 & w_2 & w_3 & u_1 & u_2 & u_3\\\hline\hline
  2.1 & -284 & 128 & 8 & 4 & 68 & -2 & 0 & 2 & 0 & 0 & 0 & 0 & 0 \\
  2.2 & -260 & 140 & 2 & 2 & 44 & -1 & 0 & 1 & 0 & 0 & 0 & 0 & 0 \\
  2.3 & -252 & 144 & 0 & 0 & 24 &  0 & 0 & 0 & 0 & 0 & 0 & 0 & 0 \\
  2.4 & -228 & 156 & 0 & 2 & 36 & -1 & 1 & 0 & 0 & 0 & 0 & 1 & 0 \\
  2.5 & -220 & 160 & 4 & 4 & 52 & -2 & 2 & 1 & 0 & 0 & 0 & 1 & 0 \\
  2.6 & -196 & 172 &10 & 6 & 64 & -3 & 4 & 1 & 1 & 0 & 0 & 1 & 0 \\
  2.7 & -188 & 176 & 2 & 4 & 44 & -2 & 3 & 0 & 0 & 0 & 1 & 1 & 0 \\
  2.8 & -164 & 188 & 8 & 6 & 56 & -3 & 5 & 1 & 0 & 0 & 1 & 1 & 0 \\
  2.9 & -156 & 192 & 0 & 4 & 36 & -2 & 4 & 0 & 0 & 0 & 0 & 2 & 0 \\
 2.10 & -132 & 204 & 6 & 6 & 48 & -3 & 6 & 1 & 0 & 0 & 0 & 2 & 0 \\
 2.11 & -100 & 220 & 4 & 6 & 40 & -3 & 7 & 0 & 0 & 0 & 1 & 2 & 0 \\\hline
\end{array}
\end{align*}
	\caption{Some CY threefolds with $h^{1,1}=2$ that exhibit a genus one fibration over $\mathbb{P}^2$ with a 2-section.}
\label{tab:gp22sec}
\end{table}

\begin{sidewaystable}
\centering
$\scriptsize
\begin{array}{|r|r|r|}
\hline
\# & \mbox{Base degree }1 & \mbox{Base degree } 2\\
\hline
 2.1 & \frac{1}{9} \Delta _4^2 \left(-44 e_{2,2}^4-29 e_4 e_{2,2}^2+37 e_4^2\right) &
   \frac{\Delta _4^4 \left(-342544 e_{2,2}^9+357352 e_4 e_{2,2}^7+109695 e_4^2
   e_{2,2}^5-285122 e_4^3 e_{2,2}^3+101003 e_4^4 e_{2,2}\right)}{15552} \\
 2.2 & \frac{8}{3} \Delta _4^{5/2} e_{2,2} \left(41 e_4-20 e_{2,2}^2\right) & \frac{1}{27}
   \Delta _4^5 \left(44400 e_{2,2}^7-111064 e_4 e_{2,2}^5+51527 e_4^2 e_{2,2}^3+11510 e_4^3
   e_{2,2}\right) \\
 2.3 & 192 \Delta _4^3 \left(12 e_{2,2}^2+e_4\right) & 64 \Delta _4^6 \left(-1072
   e_{2,2}^5+7832 e_4 e_{2,2}^3+797 e_4^2 e_{2,2}\right) \\
 2.4 & 24 \Delta _4^{5/2} e_{2,2} \left(12 e_{2,2}^2+e_4\right) & \frac{1}{3} \Delta _4^5
   \left(-1296 e_{2,2}^7+8488 e_4 e_{2,2}^5+5695 e_4^2 e_{2,2}^3+82 e_4^3 e_{2,2}\right) \\
 2.5 & -\frac{1}{3} \Delta _4^2 \left(-28 e_{2,2}^4-33 e_4 e_{2,2}^2+e_4^2\right) &
   \frac{\Delta _4^4 \left(70768 e_{2,2}^9-105688 e_4 e_{2,2}^7+55847 e_4^2 e_{2,2}^5+8422
   e_4^3 e_{2,2}^3+315 e_4^4 e_{2,2}\right)}{1728} \\
 2.6 & -\frac{1}{72} \Delta _4^{3/2} e_{2,2} \left(-20 e_{2,2}^4-131 e_4 e_{2,2}^2+7
   e_4^2\right) & \frac{\Delta _4^3 \left(907120 e_{2,2}^{11}-667704 e_4 e_{2,2}^9-17593
   e_4^2 e_{2,2}^7+193448 e_4^3 e_{2,2}^5-885 e_4^4 e_{2,2}^3+334 e_4^5
   e_{2,2}\right)}{995328} \\
 2.7 & \frac{1}{9} \Delta _4^2 \left(532 e_{2,2}^4-29 e_4 e_{2,2}^2+e_4^2\right) &
   \frac{\Delta _4^4 \left(2048240 e_{2,2}^9-748376 e_4 e_{2,2}^7+844239 e_4^2
   e_{2,2}^5-34622 e_4^3 e_{2,2}^3+407 e_4^4 e_{2,2}\right)}{15552} \\
 2.8 & -\frac{1}{72} \Delta _4^{3/2} e_{2,2} \left(-404 e_{2,2}^4-35 e_4 e_{2,2}^2+7
   e_4^2\right) & \frac{\Delta _4^3 \left(4635760 e_{2,2}^{11}-4009464 e_4 e_{2,2}^9+1702583
   e_4^2 e_{2,2}^7-90892 e_4^3 e_{2,2}^5+1443 e_4^4 e_{2,2}^3+58 e_4^5
   e_{2,2}\right)}{995328} \\
 2.9 & \frac{1}{3} \Delta _4^2 \left(412 e_{2,2}^4-95 e_4 e_{2,2}^2+7 e_4^2\right) &
   \frac{\Delta _4^4 \left(1734256 e_{2,2}^9-841816 e_4 e_{2,2}^7+295431 e_4^2
   e_{2,2}^5-35842 e_4^3 e_{2,2}^3+1411 e_4^4 e_{2,2}\right)}{1728} \\
 2.10 & \frac{1}{8} \Delta _4^{3/2} e_{2,2} \left(116 e_{2,2}^4-21 e_4 e_{2,2}^2+e_4^2\right)
   & \frac{\Delta _4^3 \left(1503728 e_{2,2}^{11}-1037304 e_4 e_{2,2}^9+413359 e_4^2
   e_{2,2}^7-52544 e_4^3 e_{2,2}^5+2211 e_4^4 e_{2,2}^3-10 e_4^5 e_{2,2}\right)}{110592} \\
 2.11 & \frac{1}{72} \Delta _4^{3/2} e_{2,2} \left(1940 e_{2,2}^4-541 e_4 e_{2,2}^2+41
   e_4^2\right) & \frac{\Delta _4^3 \left(38872688 e_{2,2}^{11}-24445304 e_4
   e_{2,2}^9+7526295 e_4^2 e_{2,2}^7-907508 e_4^3 e_{2,2}^5+19667 e_4^4 e_{2,2}^3+1938 e_4^5
   e_{2,2}\right)}{995328} \\
   \hline
\end{array}$
\caption{Modular generating series $\tildef_1, \tildef_2$ for genus 0 GW invariants for models with 2-sections}
\end{sidewaystable}
\FloatBarrier

\subsection{3-sections}
Without loss of generality we can restrict to $r_W=6$, $r_U=3$, as well as $w_4=w_5=w_6=0$, and write
\begin{align}
	\begin{split}
	W=&\mathcal{O}_{\mathbb{P}^2}(-w_1)\oplus\mathcal{O}_{\mathbb{P}^2}(-w_2)\oplus\mathcal{O}_{\mathbb{P}^2}(-w_3)\oplus\mathcal{O}_{\mathbb{P}^2}^{\oplus 3}\,,\\
	U=&\mathcal{O}_{\mathbb{P}^2}(u_1)\oplus\mathcal{O}_{\mathbb{P}^2}(u_2)\oplus\mathcal{O}_{\mathbb{P}^2}(u_3)\,.
	\end{split}
\end{align}
The toric data associated to the Calabi-Yau double cover $X$ of $\mathbb{P}(V)$ is given in Table~\ref{tab:tdata3p2}, where we again use
\begin{align}
	p=3-e_1(\vec{w})-e_1(\vec{u})\,.
\end{align}

In order for $X$ to be smooth we again have to impose $e_1(\vec{w})+e_1(\vec{u})\le 3$.
For all inequivalent values of $\vec{u},\vec{w}$ that satisfy this inequality we use CohomCalg~\cite{Blumenhagen:2010pv,cohomCalg:Implementation} to check if $h^{1,1}(X)=2$.
The resulting geometries, together with their topological invariants, are listed in Table~\ref{tab:gp23sec}.

Note that for geometry $3.2$ the vector $\vec{w}$ contains a negative entry.
We make this choice in order for the divisor that is induced by the relative hyperplane class of the projective bundle $\mathbb{P}(V)$ to be part of a K\"ahler cone basis of the Calabi-Yau.
However, the CY threefold itself is equivalent to the one that is associated to $\vec{w}=(2,1,0)$, $\vec{u}=(0,0,0)$.
On the other hand, the geometry 3.14 is associated to a $\mathbb{P}^2$-bundle that does not admit a monad bundle construction of this type.
Instead we have obtained this geometry via a flop transition from the geometry $5.7_a$ in Table~\ref{tab:gp25sec}.

\begin{table}[H]
\begin{align*}
	\left[\begin{array}{ccccccc|cc}
	\multicolumn{7}{c|}{\vec{p}\in\mathbb{Z}^7}& l^{(1)} & l^{(2)} \\\hline
	 1 & 0 & 0 & 0 & 0 & 0 & 0 & 1 &-w_1   \\
	 0 & 1 & 0 & 0 & 0 & 0 & 0 & 1 &-w_2 \\
	 0 & 0 & 1 & 0 & 0 & 0 & 0 & 1 &-w_3 \\
	 0 & 0 & 0 & 1 & 0 & 0 & 0 & 1 & 0   \\
	 0 & 0 & 0 & 0 & 1 & 0 & 0 & 1 & 0   \\
	-1 &-1 &-1 &-1 &-1 & 0 & 0 & 1 & 0   \\
	 0 & 0 & 0 & 0 & 0 & 1 & 0 & 0 & 1   \\
	 0 & 0 & 0 & 0 & 0 & 0 & 1 & 0 & 1   \\
	w_1&w_2&w_3& 0 & 0 &-1 &-1 & 0 & 1   \\\hline\hline
	 0 & 0 & 0 & 0 & 0 & 0 & 0 & 1 &u_1  \\
	 0 & 0 & 0 & 0 & 0 & 0 & 0 & 1 &u_2  \\
	 0 & 0 & 0 & 0 & 0 & 0 & 0 & 1 &u_3  \\
	 0 & 0 & 0 & 0 & 0 & 0 & 0 & 3 & p   \\
	\end{array}\right]
\end{align*}
	\caption{The toric data associated to the genus one fibered CY threefolds over $\mathbb{P}^2$ with a 3-section listed in Table~\ref{tab:gp23sec}.}
	\label{tab:tdata3p2}
\end{table}

\begin{table}[H]
\begin{align*}
\begin{array}{|c|c|c|ccc|cc|ccc|ccc|}\hline
\# & \chi_X & N_{1,2} & \kappa & \ell & c_2 & c_{1,V} & c_{2,V} & w_1 & w_2 & w_3 & u_1 & u_2 & u_3 \\\hline\hline
 3.1 & -186 & 177 & 14 & 7 & 68 & -2 & 0 & 2 & 0 & 0 & 0 & 0 & 0 \\
 3.2 & -180 & 180 & 3 & 3 & 42 & 0 & -1 & 1 & 0 & -1 & 0 & 0 & 0 \\
 3.3 & -168 & 186 & 5 & 5 & 50 & -1 & 0 & 1 & 0 & 0 & 0 & 0 & 0 \\
 3.4 & -168 & 186 & 11 & 7 & 62 & -2 & 1 & 1 & 1 & 0 & 0 & 0 & 0 \\
 3.5 & -162 & 189 & 0 & 3 & 36 & 0 & 0 & 0 & 0 & 0 & 0 & 0 & 0 \\
 3.6 & -150 & 195 & 2 & 5 & 44 & -1 & 1 & 0 & 0 & 0 & 0 & 0 & 1 \\
 3.7 & -150 & 195 & 8 & 7 & 56 & -2 & 2 & 1 & 0 & 0 & 0 & 0 & 1 \\
 3.8 & -144 & 198 & 15 & 9 & 66 & -3 & 4 & 1 & 1 & 0 & 0 & 0 & 1 \\
 3.9 & -132 & 204 & 5 & 7 & 50 & -2 & 3 & 0 & 0 & 0 & 0 & 1 & 1 \\
 3.10 & -132 & 204 & 23 & 11 & 74 & -4 & 7 & 1 & 1 & 1 & 0 & 0 & 1 \\
 3.11 & -126 & 207 & 12 & 9 & 60 & -3 & 5 & 1 & 0 & 0 & 0 & 1 & 1 \\
 3.12 & -114 & 213 & 2 & 7 & 44 & -2 & 4 & 0 & 0 & 0 & 0 & 0 & 2 \\
 3.13 & -108 & 216 & 9 & 9 & 54 & -3 & 6 & 0 & 0 & 0 & 1 & 1 & 1 \\
 3.14 & -96 & 222 & 29 & 13 & 74 & -5 & 12 & \multicolumn{6}{|c|}{\text{(flop of 5.7$_a$)}} \\
 3.15 & -90 & 225 & 6 & 9 & 48 & -3 & 7 & 0 & 0 & 0 & 0 & 1 & 2 \\\hline
\end{array}
\end{align*}
	\caption{Some CY threefolds with $h^{1,1}=2$ that exhibit a genus one fibration over $\mathbb{P}^2$ with a 3-section.}
\label{tab:gp23sec}
\end{table}

\begin{sidewaystable}
\centering
$\scriptsize
\begin{array}{|r|r|r|}
\hline
\# &  \mbox{Base degree }1 & \mbox{Base degree } 2\\
\hline
3.1 & -\Delta _6^{2/3} e_{3,3}^2 \left(2 e_{3,1}^6-270 e_{3,3} e_{3,1}^3-13851 e_{3,3}^2\right) & -\frac{1}{12}
   \Delta _6^{4/3} e_{3,1}^2 e_{3,3}^4 \left(e_{3,1}^{12}+216 e_{3,3} e_{3,1}^9+2187 e_{3,3}^2 e_{3,1}^6-8109396
   e_{3,3}^3 e_{3,1}^3-691404741 e_{3,3}^4\right) \\
 3.2 & 27 e_{3,1} e_{3,3}^4 \left(e_{3,1}^3+216 e_{3,3}\right) & -\frac{81}{4} e_{3,1} e_{3,3}^8 \left(e_{3,1}^3+216
   e_{3,3}\right) \left(e_{3,1}^6-1296 e_{3,3} e_{3,1}^3-11664 e_{3,3}^2\right) \\
 3.3 & 9 \sqrt[3]{\Delta _6} e_{3,1}^2 e_{3,3}^3 \left(2 e_{3,1}^3+189 e_{3,3}\right) & \frac{1}{4} \Delta _6^{2/3}
   e_{3,3}^6 \left(55 e_{3,1}^{12}+11232 e_{3,3} e_{3,1}^9+2143260 e_{3,3}^2 e_{3,1}^6+30744846 e_{3,3}^3
   e_{3,1}^3+69087330 e_{3,3}^4\right) \\
 3.4 & \Delta _6^{2/3} e_{3,3}^2 \left(e_{3,1}^6+594 e_{3,3} e_{3,1}^3+2916 e_{3,3}^2\right) & \frac{1}{12} \Delta
   _6^{4/3} e_{3,1}^2 e_{3,3}^4 \left(2 e_{3,1}^{12}+432 e_{3,3} e_{3,1}^9+299619 e_{3,3}^2 e_{3,1}^6+18856314
   e_{3,3}^3 e_{3,1}^3+166872474 e_{3,3}^4\right) \\
 3.5 & 27 e_{3,1} e_{3,3}^4 \left(7 e_{3,1}^3+54 e_{3,3}\right) & \frac{243}{4} e_{3,1} e_{3,3}^8 \left(28
   e_{3,1}^9+2142 e_{3,3} e_{3,1}^6+57105 e_{3,3}^2 e_{3,1}^3+65610 e_{3,3}^3\right) \\
 3.6 & 9 \sqrt[3]{\Delta _6} e_{3,1}^2 e_{3,3}^3 \left(8 e_{3,1}^3+27 e_{3,3}\right) & \frac{1}{4} \Delta _6^{2/3}
   e_{3,3}^6 \left(904 e_{3,1}^{12}+37152 e_{3,3} e_{3,1}^9+1979235 e_{3,3}^2 e_{3,1}^6+2165130 e_{3,3}^3
   e_{3,1}^3-1062882 e_{3,3}^4\right) \\
 3.7 & \Delta _6^{2/3} e_{3,3}^2 \left(13 e_{3,1}^6+432 e_{3,3} e_{3,1}^3-1458 e_{3,3}^2\right) & \frac{1}{12} \Delta
   _6^{4/3} e_{3,1}^2 e_{3,3}^4 \left(104 e_{3,1}^{12}+3510 e_{3,3} e_{3,1}^9+778572 e_{3,3}^2 e_{3,1}^6+6928416
   e_{3,3}^3 e_{3,1}^3+24977727 e_{3,3}^4\right) \\
 3.8 & 3 e_{3,1} e_{3,3}^3 \left(e_{3,1}^6+63 e_{3,3} e_{3,1}^3-243 e_{3,3}^2\right) & \frac{3}{4} e_{3,1} e_{3,3}^6
   \left(e_{3,1}^{15}+7155 e_{3,3}^2 e_{3,1}^9+180063 e_{3,3}^3 e_{3,1}^6+492075 e_{3,3}^4 e_{3,1}^3+1594323
   e_{3,3}^5\right) \\
 3.9 & \Delta _6^{2/3} e_{3,3}^2 \left(34 e_{3,1}^6-216 e_{3,3} e_{3,1}^3+729 e_{3,3}^2\right) & \frac{1}{12} \Delta
   _6^{4/3} e_{3,1}^2 e_{3,3}^4 \left(629 e_{3,1}^{12}-4104 e_{3,3} e_{3,1}^9+1080378 e_{3,3}^2 e_{3,1}^6-5353776
   e_{3,3}^3 e_{3,1}^3+7440174 e_{3,3}^4\right) \\
 3.10 & \sqrt[3]{\Delta _6} e_{3,1}^2 e_{3,3}^2 \left(e_{3,1}^6+72 e_{3,3} e_{3,1}^3-486 e_{3,3}^2\right) &
   \frac{1}{12} \Delta _6^{2/3} e_{3,3}^4 \left(2 e_{3,1}^{18}-120 e_{3,3} e_{3,1}^{15}+12414 e_{3,3}^2
  e_{3,1}^{12}-14904 e_{3,3}^3 e_{3,1}^9+5622777 e_{3,3}^4 e_{3,1}^6-59049000 e_{3,3}^5 e_{3,1}^3+274223556
   e_{3,3}^6\right)
    \\
 3.11 & 9 e_{3,1} e_{3,3}^3 \left(e_{3,1}^6+3 e_{3,3} e_{3,1}^3-81 e_{3,3}^2\right) & \frac{9}{4} e_{3,1} e_{3,3}^6
   \left(2 e_{3,1}^{15}-54 e_{3,3} e_{3,1}^{12}+5850 e_{3,3}^2 e_{3,1}^9-23085 e_{3,3}^3 e_{3,1}^6+85293 e_{3,3}^4
   e_{3,1}^3-177147 e_{3,3}^5\right) \\
 3.12 & \Delta _6^{2/3} e_{3,3}^2 \left(64 e_{3,1}^6-1350 e_{3,3} e_{3,1}^3+9477 e_{3,3}^2\right) & \frac{1}{12}
   \Delta _6^{4/3} e_{3,1}^2 e_{3,3}^4 \left(2468 e_{3,1}^{12}-96768 e_{3,3} e_{3,1}^9+3330801 e_{3,3}^2
   e_{3,1}^6-40310784 e_{3,3}^3 e_{3,1}^3+157306536 e_{3,3}^4\right) \\
 3.13 & 9 e_{3,1} e_{3,3}^3 \left(2 e_{3,1}^6-33 e_{3,3} e_{3,1}^3+162 e_{3,3}^2\right) & \frac{9}{4} e_{3,1}
   e_{3,3}^6 \left(7 e_{3,1}^{15}-288 e_{3,3} e_{3,1}^{12}+12609 e_{3,3}^2 e_{3,1}^9-152118 e_{3,3}^3 e_{3,1}^6+597051
   e_{3,3}^4 e_{3,1}^3-354294 e_{3,3}^5\right) \\
 3.14 & \Delta _6^{2/3} e_{3,3}^2 \left(103 e_{3,1}^6-2970 e_{3,3} e_{3,1}^3+24786 e_{3,3}^2\right) & \frac{1}{12} \Delta _6^{4/3} e_{3,1}^2 e_{3,3}^3 \left(12 e_{3,1}^{15}+5459 e_{3,3} e_{3,1}^{12}-322596 e_{3,3}^2  e_{3,1}^9+9773703 e_{3,3}^3 e_{3,1}^6-133017714 e_{3,3}^4 e_{3,1}^3+646763697 e_{3,3}^5\right) 
 \\
 3.15 & 3 e_{3,1} e_{3,3}^3 \left(10 e_{3,1}^6-261 e_{3,3} e_{3,1}^3+1944 e_{3,3}^2\right) & \frac{3}{4} e_{3,1}^4
   e_{3,3}^6 \left(59 e_{3,1}^{12}-3240 e_{3,3} e_{3,1}^9+102978 e_{3,3}^2 e_{3,1}^6-1414260 e_{3,3}^3
   e_{3,1}^3+6692220 e_{3,3}^4\right) \\
   \hline
\end{array}$
\caption{Modular generating series $\tildef_1, \tildef_2$ for genus 0 GW invariants for models with 3-sections}
\end{sidewaystable}        
 \FloatBarrier

\subsection{4-sections}
To construct genus one fibered CY threefolds with a 4-section over $\mathbb{P}^2$, we restrict to $r_W=7$, $r_U=3$, as well as $w_5=w_6=w_7=0$, and write
\begin{align}
	\begin{split}
        W=&\mathcal{O}_{\mathbb{P}^2}(-w_1)\oplus\ldots\oplus\mathcal{O}_{\mathbb{P}^2}(-w_4)\oplus\mathcal{O}_{\mathbb{P}^2}^{\oplus 3}\,,\\
	U=&\mathcal{O}_{\mathbb{P}^2}(u_1)\oplus\mathcal{O}_{\mathbb{P}^2}(u_2)\oplus\mathcal{O}_{\mathbb{P}^2}(u_3)\,.
	\end{split}
\end{align}
We also choose $E=\mathcal{O}_{\mathbb{P}^2}(q_1)\oplus \mathcal{O}_{\mathbb{P}^2}(q_2)$, such that the Calabi-Yau condition takes the form
\begin{align}
    e_1(\vec{q})=3-e_1(\vec{w})-e_1(\vec{u})\,.
    \label{eqn:4secP2CYcond}
\end{align}
The toric data associated to the CY threefold $X$ inside $\mathbb{P}(V)$ is given in Table~\ref{tab:tdata2p2}.

One way to obtain smooth CY threefolds with $h^{1,1}=2$ is to consider all inequivalent choices $\vec{w}\in\mathbb{N}^4$, $\vec{u}\in\mathbb{N}^3$ and $\vec{q}\in\mathbb{N}^2$ that satisfy~\eqref{eqn:4secP2CYcond} and then use the package CohomCalg~\cite{Blumenhagen:2010pv,cohomCalg:Implementation} to check if $h^{1,1}(X)=2$.
We also obtain some geometries for which some of the entries of $\vec{w}$ and $\vec{q}$ are negative and in those cases we check explicitly that one still obtains a smooth complete intersection.
Two of the geometries, 4.19 and 4.21, are complete intersections in projective bundles that do not fall under our monad construction and those are related via flop transitions to the respective fibrations $5.11_{ab}$ and $5.9_a$ from Table~\ref{tab:gp25sec}.

\begin{table}[H]
{\small
\begin{align*}
	\left[\begin{array}{cccccccc|cc}
	\multicolumn{8}{c|}{\vec{p}\in\mathbb{Z}^8}& l^{(1)} & l^{(2)} \\\hline
	 1 & 0 & 0 & 0 & 0 & 0 & 0 & 0 & 1 &-w_1   \\
	 0 & 1 & 0 & 0 & 0 & 0 & 0 & 0 & 1 &-w_2 \\
	 0 & 0 & 1 & 0 & 0 & 0 & 0 & 0 & 1 &-w_3 \\
	 0 & 0 & 0 & 1 & 0 & 0 & 0 & 0 & 1 &-w_4 \\
	 0 & 0 & 0 & 0 & 1 & 0 & 0 & 0 & 1 & 0   \\
	 0 & 0 & 0 & 0 & 0 & 1 & 0 & 0 & 1 & 0   \\
	-1 &-1 &-1 &-1 &-1 &-1 & 0 & 0 & 1 & 0   \\
	 0 & 0 & 0 & 0 & 0 & 0 & 1 & 0 & 0 & 1   \\
	 0 & 0 & 0 & 0 & 0 & 0 & 0 & 1 & 0 & 1   \\
	w_1&w_2&w_3&w_4& 0 & 0 &-1 &-1 & 0 & 1   \\\hline\hline
	 0 & 0 & 0 & 0 & 0 & 0 & 0 & 0 & 1 &u_1  \\
	 0 & 0 & 0 & 0 & 0 & 0 & 0 & 0 & 1 &u_2  \\
	 0 & 0 & 0 & 0 & 0 & 0 & 0 & 0 & 1 &u_3  \\
	 0 & 0 & 0 & 0 & 0 & 0 & 0 & 0 & 2 & q_1 \\
	 0 & 0 & 0 & 0 & 0 & 0 & 0 & 0 & 2 & q_2 \\
	\end{array}\right]
\end{align*}
}
	\caption{The toric data associated to the genus one fibered CY threefolds over $\mathbb{P}^2$ with a 4-section listed in Table~\ref{tab:gp24sec}.}
	\label{tab:tdata4p2}
\end{table}

\begin{table}[H]
{\small
\begin{align*}
\begin{array}{|c|c|c|c|ccc|cc|cc|cc|cccc|ccc|}\hline
	\# & \chi_X & N_{1,3} & N_2/2 & \kappa & \ell & c_2 & c_{1,V} & c_{2,V} & c_{1,E} & c_{2,E} & q_1 & q_2 & w_1 & w_2 & w_3 & w_4 & u_1 & u_2 & u_3\\\hline\hline
 4.1 & -156 & 108 & 84 & 0 & 2 & 36 & 2 & 0 & 5 & 4 & 1 & 4 & -2 & 0 & 0 & 0 & 0 & 0 & 0 \\
 4.2 & -140 & 128 & 72 & 0 & 0 & 24 & 3 & 2 & 6 & 8 & 4 & 2 & -2 & -1 & 0 & 0 & 0 & 0 & 0
   \\
 4.3 & -132 & 128 & 76 & 23 & 12 & 74 & -3 & 3 & 0 & -1 & -1 & 1 & 1 & 1 & 1 & 0 & 0 & 0 &
   0 \\
 4.4 & -132 & 128 & 76 & 8 & 8 & 56 & -1 & 0 & 2 & 0 & 0 & 2 & 1 & 0 & 0 & 0 & 0 & 0 & 0 \\
 4.5 & -132 & 140 & 64 & 20 & 10 & 68 & -2 & 0 & 1 & 0 & 0 & 1 & 2 & 0 & 0 & 0 & 0 & 0 & 0
   \\
 4.6 & -128 & 140 & 66 & 6 & 6 & 48 & 0 & -1 & 3 & 2 & 1 & 2 & -1 & 1 & 0 & 0 & 0 & 0 & 0
   \\
 4.7 & -128 & 144 & 62 & 4 & 4 & 40 & 1 & -1 & 4 & 4 & 2 & 2 & 1 & -1 & -1 & 0 & 0 & 0 & 0
   \\
 4.8 & -124 & 128 & 80 & 4 & 8 & 52 & -1 & 1 & 2 & 0 & 0 & 2 & 0 & 0 & 0 & 0 & 0 & 0 & 1 \\
 4.9 & -124 & 140 & 68 & 16 & 10 & 64 & -2 & 1 & 1 & 0 & 0 & 1 & 1 & 1 & 0 & 0 & 0 & 0 & 0
   \\
 4.10 & -120 & 140 & 70 & 2 & 6 & 44 & 0 & 0 & 3 & 2 & 1 & 2 & 0 & 0 & 0 & 0 & 0 & 0 & 0 \\
 4.11 & -120 & 144 & 66 & 9 & 8 & 54 & -1 & 0 & 2 & 1 & 1 & 1 & 1 & 0 & 0 & 0 & 0 & 0 & 0
   \\
 4.12 & -120 & 144 & 66 & 0 & 4 & 36 & 1 & 0 & 4 & 4 & 2 & 2 & -1 & 0 & 0 & 0 & 0 & 0 & 0
   \\
 4.13 & -116 & 140 & 72 & 12 & 10 & 60 & -2 & 2 & 1 & 0 & 0 & 1 & 1 & 0 & 0 & 0 & 0 & 0 & 1
   \\
 4.14 & -112 & 144 & 70 & 20 & 12 & 68 & -3 & 4 & 0 & 0 & 0 & 0 & 1 & 1 & 0 & 0 & 0 & 0 & 1
   \\
 4.15 & -112 & 144 & 70 & 5 & 8 & 50 & -1 & 1 & 2 & 1 & 1 & 1 & 0 & 0 & 0 & 0 & 0 & 0 & 1
   \\
 4.16 & -108 & 140 & 76 & 8 & 10 & 56 & -2 & 3 & 1 & 0 & 0 & 1 & 0 & 0 & 0 & 0 & 0 & 1 & 1
   \\
 4.17 & -104 & 144 & 74 & 37 & 16 & 82 & -5 & 11 & -2 & 1 & -1 & -1 & 1 & 1 & 1 & 1 & 0 & 0
   & 1 \\
 4.18 & -104 & 144 & 74 & 16 & 12 & 64 & -3 & 5 & 0 & 0 & 0 & 0 & 1 & 0 & 0 & 0 & 0 & 1 & 1
   \\
 4.19 & -100 & 156 & 64 & 29 & 14 & 74 & -4 & 7 & -1 & 1 & \multicolumn{9}{|c|}{\text{(flop of 5.11$_{ab}$)}}
    \\
 4.20 & -100 & 140 & 80 & 4 & 10 & 52 & -2 & 4 & 1 & 0 & 0 & 1 & 0 & 0 & 0 & 0 & 0 & 0 & 2
   \\
 4.21 & -96 & 144 & 78 & 33 & 16 & 78 & -5 & 12 & -2 & 1 & \multicolumn{9}{|c|}{\text{(flop of 5.9$_a$)}}
    \\
 4.22 & -96 & 144 & 78 & 12 & 12 & 60 & -3 & 6 & 0 & 0 & 0 & 0 & 1 & 0 & 0 & 0 & 0 & 0 & 2
    \\
 4.23 & -88 & 144 & 82 & 8 & 12 & 56 & -3 & 7 & 0 & 0 & 0 & 0 & 0 & 0 & 0 & 0 & 0 & 1 & 2
   \\\hline
\end{array}
\end{align*}
}
\caption{
	Some CY threefolds with $h^{1,1}=2$ that exhibit a genus one fibration over $\mathbb{P}^2$ with a 4-section.
}
\label{tab:gp24sec}
\end{table}

\begin{sidewaystable}
\centering
$\scriptsize
\begin{array}{|r|r|r|}
\hline
\# &  \mbox{Base degree }1 & \mbox{Base degree } 2\\
\hline
 4.1 & \frac{3 \Delta_{4}^{1/5} e_{4,1}^4 \left(e_{2,2}-e_{4,1}^2\right){}^5}{8192} & -\frac{3 \Delta_{8}^{1/2}
   e_{4,1}^2 \left(e_{4,1}^2-e_{2,2}\right){}^{10} \left(432 e_{4,1}^8-724 e_{2,2} e_{4,1}^6-288 e_{2,2}^2
   e_{4,1}^4+55 e_{2,2}^3 e_{4,1}^2-4 e_{2,2}^4\right)}{274877906944} \\
 4.2 & \frac{e_{4,1}^2 \left(e_{2,2}-e_{4,1}^2\right){}^6 \left(10 e_{4,1}^2+e_{2,2}\right)}{262144} &
   -\frac{e_{4,1}^4 \left(e_{4,1}^2-e_{2,2}\right){}^{12} \left(1740 e_{4,1}^6-4240 e_{2,2} e_{4,1}^4-377 e_{2,2}^2
   e_{4,1}^2+23 e_{2,2}^3\right)}{26388279066624} \\
 4.3 & \frac{\Delta_{8}^{1/2} e_{4,1}^4 \left(11 e_{2,2}-10 e_{4,1}^2\right)
   \left(e_{2,2}-e_{4,1}^2\right){}^3}{4096} & \frac{\Delta _8 e_{4,1}^2 \left(e_{4,1}^2-e_{2,2}\right){}^6 \left(396
   e_{4,1}^{12}-1076 e_{2,2} e_{4,1}^{10}+3419 e_{2,2}^2 e_{4,1}^8-5354 e_{2,2}^3 e_{4,1}^6+2964 e_{2,2}^4
   e_{4,1}^4-306 e_{2,2}^5 e_{4,1}^2+21 e_{2,2}^6\right)}{6442450944} \\
 4.4 & -\frac{e_{4,1}^2 \left(e_{2,2}-e_{4,1}^2\right){}^5 \left(2 e_{4,1}^4-9 e_{2,2}
   e_{4,1}^2+e_{2,2}^2\right)}{262144} & \frac{e_{4,1}^4 \left(e_{4,1}^2-e_{2,2}\right){}^{10} \left(4692
   e_{4,1}^{10}-9320 e_{2,2} e_{4,1}^8+5493 e_{2,2}^2 e_{4,1}^6-581 e_{2,2}^3 e_{4,1}^4+339 e_{2,2}^4 e_{4,1}^2-23
   e_{2,2}^5\right)}{26388279066624} \\
 4.5 & \frac{\Delta_{4}^{1/5} e_{2,2} e_{4,1}^2 \left(7 e_{2,2}-8 e_{4,1}^2\right)
   \left(e_{2,2}-e_{4,1}^2\right){}^4}{32768} & \frac{\Delta_{8}^{1/2} e_{4,1}^4 \left(e_{4,1}^2-e_{2,2}\right){}^8
   \left(7056 e_{4,1}^{10}-31692 e_{2,2} e_{4,1}^8+54456 e_{2,2}^2 e_{4,1}^6-42307 e_{2,2}^3 e_{4,1}^4+12874
   e_{2,2}^4 e_{4,1}^2-403 e_{2,2}^5\right)}{824633720832} \\
 4.6 & \frac{1}{256} \Delta _8^{3/4} e_{4,1}^2 \left(5 e_{2,2}-4 e_{4,1}^2\right) \left(e_{2,2}-e_{4,1}^2\right){}^3
   & \frac{\Delta _8^{3/2} e_{4,1}^2 \left(e_{4,1}^2-e_{2,2}\right){}^6 \left(4944 e_{4,1}^8-7252 e_{2,2}
   e_{4,1}^6-1172 e_{2,2}^2 e_{4,1}^4+3419 e_{2,2}^3 e_{4,1}^2+189 e_{2,2}^4\right)}{201326592} \\
 4.7 & \frac{\Delta_{8}^{1/2} e_{4,1}^2 \left(11 e_{2,2}-10 e_{4,1}^2\right)
   \left(e_{2,2}-e_{4,1}^2\right){}^4}{4096} & \frac{\Delta _8 e_{4,1}^4 \left(e_{4,1}^2-e_{2,2}\right){}^8
   \left(4140 e_{4,1}^6-7460 e_{2,2} e_{4,1}^4+1907 e_{2,2}^2 e_{4,1}^2+1402 e_{2,2}^3\right)}{6442450944} \\
 4.8 & \frac{e_{4,1}^2 \left(e_{2,2}-e_{4,1}^2\right){}^5 \left(22 e_{4,1}^4-7 e_{2,2}
   e_{4,1}^2+e_{2,2}^2\right)}{262144} & \frac{e_{4,1}^4 \left(e_{4,1}^2-e_{2,2}\right){}^{10} \left(16500
   e_{4,1}^{10}-25208 e_{2,2} e_{4,1}^8+16285 e_{2,2}^2 e_{4,1}^6-3633 e_{2,2}^3 e_{4,1}^4+415 e_{2,2}^4 e_{4,1}^2-23
   e_{2,2}^5\right)}{26388279066624} \\
 4.9 & \frac{\Delta_{4}^{1/5} e_{4,1}^2 \left(e_{2,2}-e_{4,1}^2\right){}^4 \left(-8 e_{4,1}^4+8 e_{2,2}
   e_{4,1}^2+e_{2,2}^2\right)}{32768} & -\frac{\Delta_{8}^{1/2} e_{4,1}^4 \left(e_{4,1}^2-e_{2,2}\right){}^8
   \left(816 e_{4,1}^{10}-4628 e_{2,2} e_{4,1}^8+4600 e_{2,2}^2 e_{4,1}^6+1563 e_{2,2}^3 e_{4,1}^4-2450 e_{2,2}^4
   e_{4,1}^2+19 e_{2,2}^5\right)}{824633720832} \\
 4.10 & \frac{1}{256} \Delta _8^{3/4} e_{4,1}^2 \left(e_{2,2}-e_{4,1}^2\right){}^3 \left(4 e_{4,1}^2+e_{2,2}\right) &
   \frac{\Delta _8^{3/2} e_{4,1}^2 \left(e_{4,1}^2-e_{2,2}\right){}^6 \left(6672 e_{4,1}^8-10612 e_{2,2}
   e_{4,1}^6+7036 e_{2,2}^2 e_{4,1}^4+323 e_{2,2}^3 e_{4,1}^2-3 e_{2,2}^4\right)}{201326592} \\
 4.11 & \frac{3 e_{4,1}^4 \left(3 e_{2,2}-2 e_{4,1}^2\right) \left(e_{2,2}-e_{4,1}^2\right){}^5}{262144} & \frac{3
   e_{4,1}^6 \left(e_{4,1}^2-e_{2,2}\right){}^{10} \left(180 e_{4,1}^8-208 e_{2,2} e_{4,1}^6-111 e_{2,2}^2
   e_{4,1}^4+121 e_{2,2}^3 e_{4,1}^2+38 e_{2,2}^4\right)}{8796093022208} \\
 4.12 & \frac{3 \Delta_{8}^{1/2} e_{4,1}^2 \left(e_{2,2}-e_{4,1}^2\right){}^4 \left(2 e_{4,1}^2+e_{2,2}\right)}{4096}
   & \frac{3 \Delta _8 e_{4,1}^4 \left(e_{4,1}^2-e_{2,2}\right){}^8 \left(172 e_{4,1}^6-340 e_{2,2} e_{4,1}^4+323
   e_{2,2}^2 e_{4,1}^2+14 e_{2,2}^3\right)}{2147483648} \\
 4.13 & -\frac{\Delta_{4}^{1/5} e_{2,2} e_{4,1}^2 \left(e_{2,2}-8 e_{4,1}^2\right)
   \left(e_{2,2}-e_{4,1}^2\right){}^4}{32768} & \frac{\Delta_{8}^{1/2} e_{4,1}^4 \left(e_{4,1}^2-e_{2,2}\right){}^8
   \left(9744 e_{4,1}^{10}-15372 e_{2,2} e_{4,1}^8+6744 e_{2,2}^2 e_{4,1}^6+269 e_{2,2}^3 e_{4,1}^4+538 e_{2,2}^4
   e_{4,1}^2-19 e_{2,2}^5\right)}{824633720832} \\
 4.14 & -\frac{\Delta_{8}^{1/2} e_{4,1}^2 \left(e_{2,2}-e_{4,1}^2\right){}^3 \left(6 e_{4,1}^4-11 e_{2,2}
   e_{4,1}^2+e_{2,2}^2\right)}{4096} & \frac{\Delta _8 e_{4,1}^4 \left(e_{4,1}^2-e_{2,2}\right){}^6 \left(1260
   e_{4,1}^{10}+228 e_{2,2} e_{4,1}^8-2697 e_{2,2}^2 e_{4,1}^6+1240 e_{2,2}^3 e_{4,1}^4+383 e_{2,2}^4 e_{4,1}^2+34
   e_{2,2}^5\right)}{6442450944} \\
 4.15 & \frac{e_{4,1}^4 \left(e_{2,2}-e_{4,1}^2\right){}^5 \left(10 e_{4,1}^2+e_{2,2}\right)}{262144} &
   \frac{e_{4,1}^6 \left(e_{4,1}^2-e_{2,2}\right){}^{10} \left(5748 e_{4,1}^8-8096 e_{2,2} e_{4,1}^6+4193 e_{2,2}^2
   e_{4,1}^4+229 e_{2,2}^3 e_{4,1}^2-6 e_{2,2}^4\right)}{26388279066624} \\
 4.16 & \frac{\Delta_{4}^{1/5} e_{4,1}^2 \left(e_{2,2}-e_{4,1}^2\right){}^4 \left(24 e_{4,1}^4-8 e_{2,2}
   e_{4,1}^2+e_{2,2}^2\right)}{32768} & \frac{\Delta_{8}^{1/2} e_{4,1}^4 \left(e_{4,1}^2-e_{2,2}\right){}^8
   \left(44880 e_{4,1}^{10}-67116 e_{2,2} e_{4,1}^8+39336 e_{2,2}^2 e_{4,1}^6-7627 e_{2,2}^3 e_{4,1}^4+610 e_{2,2}^4
   e_{4,1}^2-19 e_{2,2}^5\right)}{824633720832} \\
 4.17 & -\frac{e_{4,1}^4 \left(e_{2,2}-5 e_{4,1}^2\right) \left(3 e_{2,2}-2 e_{4,1}^2\right)
   \left(e_{2,2}-e_{4,1}^2\right){}^4}{262144} & \frac{e_{4,1}^6 \left(e_{4,1}^2-e_{2,2}\right){}^8 \left(22932
   e_{4,1}^{12}-66200 e_{2,2} e_{4,1}^{10}+85085 e_{2,2}^2 e_{4,1}^8-61145 e_{2,2}^3 e_{4,1}^6+24885 e_{2,2}^4
   e_{4,1}^4-4947 e_{2,2}^5 e_{4,1}^2+414 e_{2,2}^6\right)}{26388279066624} \\
 4.18 & -\frac{\Delta_{8}^{1/2} e_{4,1}^2 \left(e_{2,2}-5 e_{4,1}^2\right) \left(e_{2,2}-e_{4,1}^2\right){}^3 \left(2
   e_{4,1}^2+e_{2,2}\right)}{4096} & \frac{\Delta _8 e_{4,1}^4 \left(e_{4,1}^2-e_{2,2}\right){}^6 \left(15564
   e_{4,1}^{10}-24236 e_{2,2} e_{4,1}^8+13407 e_{2,2}^2 e_{4,1}^6-2180 e_{2,2}^3 e_{4,1}^4+351 e_{2,2}^4 e_{4,1}^2-26
   e_{2,2}^5\right)}{6442450944} \\
 4.19 & \frac{1}{256} \Delta _8^{3/4} e_{4,1}^4 \left(5 e_{2,2}-4 e_{4,1}^2\right) \left(e_{2,2}-e_{4,1}^2\right){}^2
   & \frac{\Delta _8^{3/2} e_{4,1}^2 \left(e_{4,1}^2-e_{2,2}\right){}^4 \left(12336 e_{4,1}^{12}-36964 e_{2,2}
   e_{4,1}^{10}+49372 e_{2,2}^2 e_{4,1}^8-36373 e_{2,2}^3 e_{4,1}^6+13983 e_{2,2}^4 e_{4,1}^4-2043 e_{2,2}^5
   e_{4,1}^2+201 e_{2,2}^6\right)}{201326592} \\
 4.20 & \frac{\Delta_{4}^{1/5} e_{4,1}^2 \left(e_{2,2}-e_{4,1}^2\right){}^4 \left(64 e_{4,1}^4-40 e_{2,2}
   e_{4,1}^2+7 e_{2,2}^2\right)}{32768} & \frac{\Delta_{8}^{1/2} e_{4,1}^4 \left(e_{4,1}^2-e_{2,2}\right){}^8
   \left(221328 e_{4,1}^{10}-359500 e_{2,2} e_{4,1}^8+234488 e_{2,2}^2 e_{4,1}^6-69795 e_{2,2}^3 e_{4,1}^4+9194
   e_{2,2}^4 e_{4,1}^2-403 e_{2,2}^5\right)}{824633720832} \\
 4.21 & -\frac{3 e_{4,1}^4 \left(e_{2,2}-e_{4,1}^2\right){}^4 \left(-2 e_{4,1}^4-3 e_{2,2}
   e_{4,1}^2+e_{2,2}^2\right)}{262144} & \frac{3 e_{4,1}^6 \left(e_{4,1}^2-e_{2,2}\right){}^8 \left(5908
   e_{4,1}^{12}-13032 e_{2,2} e_{4,1}^{10}+12797 e_{2,2}^2 e_{4,1}^8-7405 e_{2,2}^3 e_{4,1}^6+2769 e_{2,2}^4
   e_{4,1}^4-575 e_{2,2}^5 e_{4,1}^2+50 e_{2,2}^6\right)}{8796093022208} \\
 4.22 & \frac{3 \Delta_{8}^{1/2} e_{4,1}^2 \left(e_{2,2}-e_{4,1}^2\right){}^3 \left(14 e_{4,1}^4-7 e_{2,2}
   e_{4,1}^2+e_{2,2}^2\right)}{4096} & \frac{3 \Delta _8 e_{4,1}^4 \left(e_{4,1}^2-e_{2,2}\right){}^6 \left(6988
   e_{4,1}^{10}-11228 e_{2,2} e_{4,1}^8+7055 e_{2,2}^2 e_{4,1}^6-1856 e_{2,2}^3 e_{4,1}^4+199 e_{2,2}^4 e_{4,1}^2-6
   e_{2,2}^5\right)}{2147483648} \\
 4.23 & \frac{\Delta_{8}^{1/2} e_{4,1}^2 \left(e_{2,2}-e_{4,1}^2\right){}^3 \left(90 e_{4,1}^4-61 e_{2,2}
   e_{4,1}^2+11 e_{2,2}^2\right)}{4096} & \frac{\Delta _8 e_{4,1}^4 \left(e_{4,1}^2-e_{2,2}\right){}^6 \left(216972
   e_{4,1}^{10}-368460 e_{2,2} e_{4,1}^8+247407 e_{2,2}^2 e_{4,1}^6-76892 e_{2,2}^3 e_{4,1}^4+10463 e_{2,2}^4
   e_{4,1}^2-434 e_{2,2}^5\right)}{6442450944} \\
   \hline
\end{array}$
\caption{Modular generating series $\tildef_1, \tildef_2$ for genus 0 GW invariants for models with 4-sections}
\end{sidewaystable}
\FloatBarrier

\subsection{5-sections}
A large number of genus one fibered CY threefolds with a 5-section over $\mathbb{P}^2$ and $h^{1,1}=2$ was constructed in~\cite{Knapp:2021vkm}.
We will  reproduce those geometries here in terms of the data described in Section~\ref{sec:generic5sections}.
To this end, we can restrict to the case $V=W$, such that $r_W=5$ and $r_U=0$, and write
\begin{align}
	\begin{split}
        W=&\mathcal{O}_{\mathbb{P}^2}(-w_1)\oplus\ldots\oplus\mathcal{O}_{\mathbb{P}^2}(-w_5)\,.
        \end{split}
\end{align}
We also choose
\begin{align}
    E=\mathcal{O}_{\mathbb{P}^2}(-e_1)\oplus\mathcal{O}_{\mathbb{P}^2}(-e_2)\oplus\mathcal{O}_{\mathbb{P}^2}(-e_3)\oplus\mathcal{O}_{\mathbb{P}^2}^{\oplus 2}\,,
\end{align}
The line bundle $L$ is determined by the Calabi-Yau condition~\eqref{eqn:5secCYcondition}.

This allows us to construct the geometries $5.n_b$, for $n=1,\ldots,12$ listed in Table~\eqref{tab:gp25sec}, with the exception of $5.3_b$.
To construct $5.3_b$ one has to consider a rank 5 bundle $V$ that is not a sum of line bundles, as was explained in~\cite{Knapp:2021vkm}.

The fibrations $5.n_a$, that are conjecturally fiberwise homologically projective dual to the fibrations $5.n_a$ have also been constructed in~\cite{Knapp:2021vkm}.
The Chern classes of the corresponding bundles $V$ and $E$ are determined by~\eqref{sec:gen5secBundleInvolution} and~\eqref{eqn:5secCYconditionDual}.

In each case we have performed a suitable twist of $V$ such that the restriction of the relative hyperplane divisor of $\mathbb{P}(V)$ to the CY threefold is part of a K\"ahler cone basis.

\newpage

\begin{table}[H]
\begin{align*}
\begin{array}{|c|c|c|c|ccc|cc|cc|ccccc|ccc|}\hline
\# & \chi_X& N_{1,4} & N_{2,3} & \kappa & \ell & c_2 & c_{1,V} & c_{2,V} & c_{1,E} & c_{2,E} & w_1 & w_2 & w_3 & w_4 & w_5 & e_1 & e_2 & e_3 \\\hline
 5.1_a & -90 & 100 & 125 & 10 & 15 & 64 & -3 & 7 & 0 & 0 & \multicolumn{5}{|c|}{}&\multicolumn{3}{|c|}{} \\
 5.1_b & -90 & 125 & 100 & 5 & 5 & 38 & 2 & 0 & 0 & 0 & 1 & 0 &-1 &-1 &-1 & 0 & 0 & 0 \\\hline
 5.2_a & -90 & 105 & 120 & 15 & 15 & 66 & -3 & 6 & 0 & 0 & \multicolumn{5}{|c|}{}&\multicolumn{3}{|c|}{} \\
 5.2_b & -90 & 120 & 105 & 0 & 5 & 36 & 2 & 1 & 0 & 0 & 0 & 0 & 0 &-1 &-1 & 0 & 0 & 0 \\\hline
 5.3_a & -90 & 110 & 115 & 20 & 15 & 68 & -3 & 5 & 0 & 0 & \multicolumn{5}{|c|}{}&\multicolumn{3}{|c|}{} \\
 5.3_b & -90 & 115 & 110 & 25 & 15 & 70 & -3 & 4 & 0 & 0 & \multicolumn{8}{|c|}{\text{(exceptional)}} \\\hline
 5.4_a & -94 & 105 & 118 & 51 & 21 & 90 & -6 & 16 & -1 & 0 & \multicolumn{5}{|c|}{}&\multicolumn{3}{|c|}{} \\
 5.4_b & -94 & 118 & 105 & 10 & 9 & 52 & 0 & -1 & 1 & 0 & 1 & 0 & 0 & 0 &-1 & -1 & 0 & 0 \\\hline
 5.5_a & -94 & 108 & 115 & 42 & 19 & 84 & -5 & 11 & 1 & 0 & \multicolumn{5}{|c|}{}&\multicolumn{3}{|c|}{} \\
 5.5_b & -94 & 115 & 108 & 13 & 11 & 58 & -1 & 0 & -1 & 0 & 1 & 0 & 0 & 0 & 0 & 1 & 0 & 0 \\\hline
 5.6_a & -94 & 110 & 113 & 8 & 11 & 56 & -1 & 1 & -1 & 0 & \multicolumn{5}{|c|}{}&\multicolumn{3}{|c|}{} \\
 5.6_b & -94 & 113 & 110 & 5 & 9 & 50 & 0 & 0 & 1 & 0 & 0 & 0 & 0 & 0 & 0 & -1 & 0 & 0 \\\hline
 5.7_a & -96 & 103 & 119 & 94 & 27 & 112 & -9 & 34 & 3 & 3 & \multicolumn{5}{|c|}{}&\multicolumn{3}{|c|}{}
   \\
 5.7_b & -96 & 119 & 103 & 26 & 13 & 68 & -2 & 0 & -3 & 3 & 2 & 0 & 0 & 0 & 0 & 1 & 1 &
   1 \\\hline
 5.8_a & -96 & 104 & 118 & 65 & 23 & 98 & -7 & 21 & 2 & 1 & \multicolumn{5}{|c|}{}&\multicolumn{3}{|c|}{} \\
 5.8_b & -96 & 118 & 104 & 7 & 7 & 46 & 1 & -1 & -2 & 1 & 1 & 0 & 0 & -1 & -1 & 1 & 1 & 0 \\\hline
 5.9_a & -96 & 108 & 114 & 33 & 17 & 78 & -4 & 7 & -2 & 1 & \multicolumn{5}{|c|}{}&\multicolumn{3}{|c|}{} \\
 5.9_b & -96 & 114 & 108 & 21 & 13 & 66 & -2 & 1 & 2 & 1 & 1 & 1 & 0 & 0 & 0 & -1 & -1 & 0
   \\\hline
 5.10_a & -96 & 109 & 113 & 16 & 13 & 64 & -2 & 2 & 2 & 1 & \multicolumn{5}{|c|}{}&\multicolumn{3}{|c|}{} \\
 5.10_b & -96 & 113 & 109 & 2 & 7 & 44 & 1 & 0 & -2 & 1 & 0 & 0 & 0 & 0 & -1 & 1 & 1 & 0 \\\hline
 5.11_{ab} & -100 & 110 & 110 & 29 & 15 & 74 & -3 & 3 & 0 & -1 & 1 & 1 & 1 & 0 & 0 & -1 & 1 & 0
   \\\hline
 5.12_a & -104 & 108 & 110 & 9 & 9 & 54 & 0 & -1 & 1 & -1 &\multicolumn{5}{|c|}{}&\multicolumn{3}{|c|}{} \\
 5.12_b & -104 & 110 & 108 & 17 & 11 & 62 & -1 & -1 & -1 & -1 & 1 & 1 & 0 & 0 & -1 & 1 & 1 &
   -1 \\\hline
\end{array}
\end{align*}
\caption{
	Some CY threefolds with $h^{1,1}=2$ that exhibit a genus one fibration over $\mathbb{P}^2$ with a 5-section, first constructed in~\cite{Knapp:2021vkm}.
}
\label{tab:gp25sec}
\end{table}

\begin{table}
\centering
$\begin{array}{|l|r|}
\hline
\# &  \mbox{Base degree }1 \\
\hline
5.1_a & 5 \Delta _{10} e_1 \tilde{e}_1^2 \left(-108 e_1^2 \tilde{e}_1+343 e_1 \tilde{e}_1^2-2 \tilde{e}_1^3+10
   e_1^3\right) \\
 5.1_b & 5 \Delta _{10} \tilde{e}_1^3 \left(343 e_1^2 \tilde{e}_1+108 e_1 \tilde{e}_1^2+10 \tilde{e}_1^3+2 e_1^3\right)
   \\
 5.2_a & 15 \Delta _{10} e_1 \tilde{e}_1^2 \left(-14 e_1^2 \tilde{e}_1+35 e_1 \tilde{e}_1^2-8 \tilde{e}_1^3+2
   e_1^3\right) \\
 5.2_b & 15 \Delta _{10} \tilde{e}_1^3 \left(35 e_1^2 \tilde{e}_1+14 e_1 \tilde{e}_1^2+2 \tilde{e}_1^3+8 e_1^3\right)
   \\
 5.3_a & 5 \Delta _{10} \tilde{e}_1^2 \left(2 e_1^3 \tilde{e}_1-14 e_1^2 \tilde{e}_1^2-24 e_1
   \tilde{e}_1^3+\tilde{e}_1^4+3 e_1^4\right) \\
 5.3_b & 5 \Delta _{10} \tilde{e}_1^2 \left(24 e_1^3 \tilde{e}_1-14 e_1^2 \tilde{e}_1^2-2 e_1 \tilde{e}_1^3+3
   \tilde{e}_1^4+e_1^4\right) \\
 5.4_a& \Delta _{10}^{3/5} e_1^3 \tilde{e}_1^3 \left(49 e_1^3 \tilde{e}_1-275 e_1^2 \tilde{e}_1^2+569 e_1
   \tilde{e}_1^3-11 \tilde{e}_1^4+e_1^4\right) \\
 5.4_b & \Delta _{10}^{2/5} e_1^3 \tilde{e}_1^5 \left(569 e_1^3 \tilde{e}_1+275 e_1^2 \tilde{e}_1^2+49 e_1
   \tilde{e}_1^3-\tilde{e}_1^4+11 e_1^4\right) \\
 5.5_a & \Delta _{10}^{2/5} e_1^4 \tilde{e}_1^4 \left(48 e_1^3 \tilde{e}_1-35 e_1^2 \tilde{e}_1^2-302 e_1
   \tilde{e}_1^3+13 \tilde{e}_1^4+2 e_1^4\right) \\
 5.5_b & \Delta _{10}^{3/5} e_1^2 \tilde{e}_1^4 \left(302 e_1^3 \tilde{e}_1-35 e_1^2 \tilde{e}_1^2-48 e_1
   \tilde{e}_1^3+2 \tilde{e}_1^4+13 e_1^4\right) \\
 5.6_a & \Delta _{10}^{3/5} e_1^3 \tilde{e}_1^4 \left(-24 e_1^2 \tilde{e}_1+137 e_1 \tilde{e}_1^2-52 \tilde{e}_1^3+41
   e_1^3\right) \\
 5.6_b & \Delta _{10}^{2/5} e_1^4 \tilde{e}_1^5 \left(137 e_1^2 \tilde{e}_1+24 e_1 \tilde{e}_1^2+41 \tilde{e}_1^3+52
   e_1^3\right) \\
 5.7_a & \Delta _{10}^{1/5} e_1^5 \tilde{e}_1^4 \left(-24 e_1^4 \tilde{e}_1+386 e_1^3 \tilde{e}_1^2-1548 e_1^2
   \tilde{e}_1^3+153 e_1 \tilde{e}_1^4+2 \tilde{e}_1^5+e_1^5\right) \\
 5.7_b & \Delta _{10}^{4/5} \tilde{e}_1^3 \left(153 e_1^4 \tilde{e}_1+1548 e_1^3 \tilde{e}_1^2+386 e_1^2
   \tilde{e}_1^3+24 e_1 \tilde{e}_1^4+\tilde{e}_1^5-2 e_1^5\right) \\
 5.8_a & \Delta _{10}^{4/5} e_1^2 \tilde{e}_1^2 \left(4 e_1^3 \tilde{e}_1+195 e_1^2 \tilde{e}_1^2-986 e_1
   \tilde{e}_1^3+14 \tilde{e}_1^4+e_1^4\right) \\
 5.8_b & \Delta _{10}^{1/5} e_1^4 \tilde{e}_1^6 \left(986 e_1^3 \tilde{e}_1+195 e_1^2 \tilde{e}_1^2-4 e_1
   \tilde{e}_1^3+\tilde{e}_1^4+14 e_1^4\right) \\
 5.9_a & 3 \Delta _{10}^{1/5} e_1^5 \tilde{e}_1^5 \left(29 e_1^3 \tilde{e}_1-65 e_1^2 \tilde{e}_1^2+74 e_1
   \tilde{e}_1^3-\tilde{e}_1^4+e_1^4\right) \\
 5.9_b & 3 \Delta _{10}^{4/5} e_1 \tilde{e}_1^3 \left(74 e_1^3 \tilde{e}_1+65 e_1^2 \tilde{e}_1^2+29 e_1
   \tilde{e}_1^3-\tilde{e}_1^4+e_1^4\right) \\
 5.10_a & \Delta _{10}^{4/5} e_1^2 \left(2 e_1-\tilde{e}_1\right) \tilde{e}_1^3 \left(62 e_1 \tilde{e}_1-83
   \tilde{e}_1^2+8 e_1^2\right) \\
 5.10_b & \Delta _{10}^{1/5} e_1^5 \tilde{e}_1^6 \left(2 \tilde{e}_1+e_1\right) \left(62 e_1 \tilde{e}_1-8
   \tilde{e}_1^2+83 e_1^2\right) \\
 5.11_{ab} & \Delta _{10} \tilde{e}_1^2 \left(138 e_1^3 \tilde{e}_1+234 e_1^2 \tilde{e}_1^2-138 e_1
   \tilde{e}_1^3+\tilde{e}_1^4+e_1^4\right) \\
 5.12_a & \Delta _{10}^{2/5} e_1^4 \tilde{e}_1^5 \left(507 e_1^2 \tilde{e}_1-386 e_1 \tilde{e}_1^2+\tilde{e}_1^3+22
   e_1^3\right) \\
 5.12_b & \Delta _{10}^{3/5} e_1^3 \tilde{e}_1^4 \left(386 e_1^2 \tilde{e}_1+507 e_1 \tilde{e}_1^2-22
   \tilde{e}_1^3+e_1^3\right) \\
      \hline
\end{array}$
\caption{Modular generating series $\tildef_1$ for genus 0 GW invariants for models with 5-sections. The notations $e_1, \widetilde{e}_1$ are shorthand for the weight 1 Eisenstein series 
$e_{5,1}, \widetilde{e}_{5,1}$.}
\end{table}

\begin{sidewaystable}
\centering
$\scriptsize
\begin{array}{|r|l|}
\hline
\# &  \mbox{Base degree }2 \\
\hline
  5.1 & \frac{5}{12} \Delta _{10}^2 e_1^2 \tilde{e}_1^4 \left(-7170 e_1^7 \tilde{e}_1+92169 e_1^6 \tilde{e}_1^2-546748
   e_1^5 \tilde{e}_1^3+1438365 e_1^4 \tilde{e}_1^4-1499036 e_1^3 \tilde{e}_1^5+785000 e_1^2 \tilde{e}_1^6-5542 e_1
   \tilde{e}_1^7-11 \tilde{e}_1^8+277 e_1^8\right) \\
 5.2 & \frac{5}{12} \Delta _{10}^2 \tilde{e}_1^6 \left(5542 e_1^7 \tilde{e}_1+785000 e_1^6 \tilde{e}_1^2+1499036
   e_1^5 \tilde{e}_1^3+1438365 e_1^4 \tilde{e}_1^4+546748 e_1^3 \tilde{e}_1^5+92169 e_1^2 \tilde{e}_1^6+7170 e_1
   \tilde{e}_1^7+277 \tilde{e}_1^8-11 e_1^8\right) \\
 5.3 & \frac{15}{4} \Delta _{10}^2 e_1^2 \tilde{e}_1^4 \left(-214 e_1^7 \tilde{e}_1+2960 e_1^6 \tilde{e}_1^2-13874
   e_1^5 \tilde{e}_1^3+28595 e_1^4 \tilde{e}_1^4-37378 e_1^3 \tilde{e}_1^5+43112 e_1^2 \tilde{e}_1^6-3600 e_1
   \tilde{e}_1^7+192 \tilde{e}_1^8+11 e_1^8\right) \\
 5.4 & \frac{15}{4} \Delta _{10}^2 \tilde{e}_1^6 \left(3600 e_1^7 \tilde{e}_1+43112 e_1^6 \tilde{e}_1^2+37378 e_1^5
   \tilde{e}_1^3+28595 e_1^4 \tilde{e}_1^4+13874 e_1^3 \tilde{e}_1^5+2960 e_1^2 \tilde{e}_1^6+214 e_1
   \tilde{e}_1^7+11 \tilde{e}_1^8+192 e_1^8\right) \\
 5.5 & \frac{5}{12} \Delta _{10}^2 \tilde{e}_1^4 \left(-270 e_1^9 \tilde{e}_1+8170 e_1^8 \tilde{e}_1^2-8582 e_1^7
   \tilde{e}_1^3+50729 e_1^6 \tilde{e}_1^4-99154 e_1^5 \tilde{e}_1^5+97890 e_1^4 \tilde{e}_1^6-30530 e_1^3
   \tilde{e}_1^7+3486 e_1^2 \tilde{e}_1^8-42 e_1 \tilde{e}_1^9+4 \tilde{e}_1^{10}+27 e_1^{10}\right) \\
 5.6 & \frac{5}{12} \Delta _{10}^2 \tilde{e}_1^4 \left(42 e_1^9 \tilde{e}_1+3486 e_1^8 \tilde{e}_1^2+30530 e_1^7
   \tilde{e}_1^3+97890 e_1^6 \tilde{e}_1^4+99154 e_1^5 \tilde{e}_1^5+50729 e_1^4 \tilde{e}_1^6+8582 e_1^3
   \tilde{e}_1^7+8170 e_1^2 \tilde{e}_1^8+270 e_1 \tilde{e}_1^9+27 \tilde{e}_1^{10}+4 e_1^{10}\right) \\
 5.7 & \frac{1}{12} \Delta _{10}^{6/5} e_1^6 \tilde{e}_1^6 \left(-26 e_1^9 \tilde{e}_1+2216 e_1^8 \tilde{e}_1^2-20495
   e_1^7 \tilde{e}_1^3+363989 e_1^6 \tilde{e}_1^4-1712685 e_1^5 \tilde{e}_1^5+3516578 e_1^4 \tilde{e}_1^6-2326093
   e_1^3 \tilde{e}_1^7+408725 e_1^2 \tilde{e}_1^8-5962 e_1 \tilde{e}_1^9+77 \tilde{e}_1^{10}+2 e_1^{10}\right) \\
 5.8 & \frac{1}{12} \Delta _{10}^{4/5} e_1^6 \tilde{e}_1^{10} \left(5962 e_1^9 \tilde{e}_1+408725 e_1^8
   \tilde{e}_1^2+2326093 e_1^7 \tilde{e}_1^3+3516578 e_1^6 \tilde{e}_1^4+1712685 e_1^5 \tilde{e}_1^5+363989 e_1^4
   \tilde{e}_1^6+20495 e_1^3 \tilde{e}_1^7+2216 e_1^2 \tilde{e}_1^8+26 e_1 \tilde{e}_1^9+2 \tilde{e}_1^{10}+77
   e_1^{10}\right) \\
 5.9 & \frac{1}{12} \Delta _{10}^{4/5} e_1^8 \tilde{e}_1^8 \left(-54 e_1^9 \tilde{e}_1+3223 e_1^8 \tilde{e}_1^2+13784
   e_1^7 \tilde{e}_1^3+214740 e_1^6 \tilde{e}_1^4-627588 e_1^5 \tilde{e}_1^5+1852552 e_1^4 \tilde{e}_1^6-737734 e_1^3
   \tilde{e}_1^7+148508 e_1^2 \tilde{e}_1^8-2990 e_1 \tilde{e}_1^9+104 \tilde{e}_1^{10}+5 e_1^{10}\right) \\
 5.10 & \frac{1}{12} \Delta _{10}^{6/5} e_1^4 \tilde{e}_1^8 \left(2990 e_1^9 \tilde{e}_1+148508 e_1^8
   \tilde{e}_1^2+737734 e_1^7 \tilde{e}_1^3+1852552 e_1^6 \tilde{e}_1^4+627588 e_1^5 \tilde{e}_1^5+214740 e_1^4
   \tilde{e}_1^6-13784 e_1^3 \tilde{e}_1^7+3223 e_1^2 \tilde{e}_1^8+54 e_1 \tilde{e}_1^9+5 \tilde{e}_1^{10}+104
   e_1^{10}\right) \\
 5.11 & \frac{1}{12} \Delta _{10}^{6/5} e_1^6 \tilde{e}_1^8 \left(-2532 e_1^7 \tilde{e}_1+199011 e_1^6
   \tilde{e}_1^2-79154 e_1^5 \tilde{e}_1^3+260475 e_1^4 \tilde{e}_1^4-487726 e_1^3 \tilde{e}_1^5+371896 e_1^2
   \tilde{e}_1^6-11588 e_1 \tilde{e}_1^7+1442 \tilde{e}_1^8+902 e_1^8\right) \\
 5.12 & \frac{1}{12} \Delta _{10}^{4/5} e_1^8 \tilde{e}_1^{10} \left(11588 e_1^7 \tilde{e}_1+371896 e_1^6
   \tilde{e}_1^2+487726 e_1^5 \tilde{e}_1^3+260475 e_1^4 \tilde{e}_1^4+79154 e_1^3 \tilde{e}_1^5+199011 e_1^2
   \tilde{e}_1^6+2532 e_1 \tilde{e}_1^7+902 \tilde{e}_1^8+1442 e_1^8\right) \\
 5.13 & \frac{1}{12} \Delta _{10}^{2/5} e_1^{10} \tilde{e}_1^8 \left(-108 e_1^{11} \tilde{e}_1+2584 e_1^{10}
   \tilde{e}_1^2-34488 e_1^9 \tilde{e}_1^3+291556 e_1^8 \tilde{e}_1^4-1238134 e_1^7 \tilde{e}_1^5+1849178 e_1^6
   \tilde{e}_1^6+2301330 e_1^5 \tilde{e}_1^7+7585194 e_1^4 \tilde{e}_1^8-418598 e_1^3 \tilde{e}_1^9+10079 e_1^2
   \tilde{e}_1^{10}+198 e_1 \tilde{e}_1^{11}-\tilde{e}_1^{12}+2 e_1^{12}\right) \\
 5.14 & \frac{1}{12} \Delta _{10}^{8/5} \tilde{e}_1^6 \left(-198 e_1^{11} \tilde{e}_1+10079 e_1^{10}
   \tilde{e}_1^2+418598 e_1^9 \tilde{e}_1^3+7585194 e_1^8 \tilde{e}_1^4-2301330 e_1^7 \tilde{e}_1^5+1849178 e_1^6
   \tilde{e}_1^6+1238134 e_1^5 \tilde{e}_1^7+291556 e_1^4 \tilde{e}_1^8+34488 e_1^3 \tilde{e}_1^9+2584 e_1^2
   \tilde{e}_1^{10}+108 e_1 \tilde{e}_1^{11}+2 \tilde{e}_1^{12}-e_1^{12}\right) \\
 5.15 & \frac{1}{12} \Delta _{10}^{8/5} e_1^4 \tilde{e}_1^4 \left(-74 e_1^9 \tilde{e}_1+1675 e_1^8
   \tilde{e}_1^2-15436 e_1^7 \tilde{e}_1^3+224444 e_1^6 \tilde{e}_1^4-1278246 e_1^5 \tilde{e}_1^5+4337072 e_1^4
   \tilde{e}_1^6-4166990 e_1^3 \tilde{e}_1^7+1299718 e_1^2 \tilde{e}_1^8-15502 e_1 \tilde{e}_1^9+95
   \tilde{e}_1^{10}+2 e_1^{10}\right) \\
 5.16 & \frac{1}{12} \Delta _{10}^{2/5} e_1^8 \tilde{e}_1^{12} \left(15502 e_1^9 \tilde{e}_1+1299718 e_1^8
   \tilde{e}_1^2+4166990 e_1^7 \tilde{e}_1^3+4337072 e_1^6 \tilde{e}_1^4+1278246 e_1^5 \tilde{e}_1^5+224444 e_1^4
   \tilde{e}_1^6+15436 e_1^3 \tilde{e}_1^7+1675 e_1^2 \tilde{e}_1^8+74 e_1 \tilde{e}_1^9+2 \tilde{e}_1^{10}+95
   e_1^{10}\right) \\
 5.17 & \frac{3}{4} \Delta _{10}^{2/5} e_1^{10} \tilde{e}_1^{10} \left(4 e_1^9 \tilde{e}_1+776 e_1^8
   \tilde{e}_1^2+5408 e_1^7 \tilde{e}_1^3+41157 e_1^6 \tilde{e}_1^4-90035 e_1^5 \tilde{e}_1^5+195288 e_1^4
   \tilde{e}_1^6-55833 e_1^3 \tilde{e}_1^7+4356 e_1^2 \tilde{e}_1^8-51 e_1
   \tilde{e}_1^9+\tilde{e}_1^{10}+e_1^{10}\right) \\
 5.18 & \frac{3}{4} \Delta _{10}^{8/5} e_1^2 \tilde{e}_1^6 \left(51 e_1^9 \tilde{e}_1+4356 e_1^8 \tilde{e}_1^2+55833
   e_1^7 \tilde{e}_1^3+195288 e_1^6 \tilde{e}_1^4+90035 e_1^5 \tilde{e}_1^5+41157 e_1^4 \tilde{e}_1^6-5408 e_1^3
   \tilde{e}_1^7+776 e_1^2 \tilde{e}_1^8-4 e_1 \tilde{e}_1^9+\tilde{e}_1^{10}+e_1^{10}\right) \\
 5.19 & \frac{1}{12} \Delta _{10}^{8/5} e_1^4 \tilde{e}_1^6 \left(332 e_1^7 \tilde{e}_1+67870 e_1^6
   \tilde{e}_1^2+126688 e_1^5 \tilde{e}_1^3+531900 e_1^4 \tilde{e}_1^4-589156 e_1^3 \tilde{e}_1^5+868299 e_1^2
   \tilde{e}_1^6-45720 e_1 \tilde{e}_1^7+3761 \tilde{e}_1^8+152 e_1^8\right) \\
 5.20 & \frac{1}{12} \Delta _{10}^{2/5} e_1^{10} \tilde{e}_1^{12} \left(45720 e_1^7 \tilde{e}_1+868299 e_1^6
   \tilde{e}_1^2+589156 e_1^5 \tilde{e}_1^3+531900 e_1^4 \tilde{e}_1^4-126688 e_1^3 \tilde{e}_1^5+67870 e_1^2
   \tilde{e}_1^6-332 e_1 \tilde{e}_1^7+152 \tilde{e}_1^8+3761 e_1^8\right) \\
 5.21 & \frac{1}{12} \Delta _{10}^2 \tilde{e}_1^4 \left(\tilde{e}_1^2+e_1^2\right) \left(72 e_1^7 \tilde{e}_1+11679
   e_1^6 \tilde{e}_1^2+189636 e_1^5 \tilde{e}_1^3+1270900 e_1^4 \tilde{e}_1^4-189636 e_1^3 \tilde{e}_1^5+11679 e_1^2
   \tilde{e}_1^6-72 e_1 \tilde{e}_1^7+2 \tilde{e}_1^8+2 e_1^8\right) \\
 5.22 & \frac{1}{12} \Delta _{10}^{4/5} e_1^8 \tilde{e}_1^{10} \left(10674 e_1^7 \tilde{e}_1+467322 e_1^6
   \tilde{e}_1^2+939646 e_1^5 \tilde{e}_1^3+2927115 e_1^4 \tilde{e}_1^4-1600048 e_1^3 \tilde{e}_1^5+105238 e_1^2
   \tilde{e}_1^6-326 e_1 \tilde{e}_1^7+2 \tilde{e}_1^8+251 e_1^8\right) \\
 5.23 & \frac{1}{12} \Delta _{10}^{6/5} e_1^6 \tilde{e}_1^8 \left(326 e_1^7 \tilde{e}_1+105238 e_1^6
   \tilde{e}_1^2+1600048 e_1^5 \tilde{e}_1^3+2927115 e_1^4 \tilde{e}_1^4-939646 e_1^3 \tilde{e}_1^5+467322 e_1^2
   \tilde{e}_1^6-10674 e_1 \tilde{e}_1^7+251 \tilde{e}_1^8+2 e_1^8\right) \\
      \hline
\end{array}$
\caption{Modular generating series $\tildef_2$ for genus 0 GW invariants for models with 5-sections}
\end{sidewaystable}
\FloatBarrier

\subsection{Extended K\"ahler cone and effective cone for fibrations over $\IP^2$}

For the purposes of this paper, it is sufficient to restrict attention to the K\"ahler cone (or its closure, the nef cone). However, it is useful to map out the extended K\"ahler cone, as it provides useful information on the support of GV invariants as well as on the effective cone, which determines the possible divisors that a D4-brane can wrap.
Note that the effective cone, unlike the K\"ahler cone, can depend on the complex structure of the Calabi-Yau manifold.

With a single exception, all of the models constructed above have a K\"ahler cone given by the upper right quadrant $t:=\Im T>0, s:=\Im S>0$.\footnote{The model 3.14, in the basis where the base degree GV invariants are periodic modulo 3, has a larger K\"ahler cone $t>\min(0,-2s)$.}
The boundary at $t=0$ coincides with a boundary of the effective cone, which is therefore a boundary of the geometric phase.
At the other boundary, i.e. $s=0$, the mass of M2-branes wrapped on $n H$ goes to zero. Several things may happen, in increasing degree of severity \cite{Witten:1993yc,Brodie:2021nit,Gendler:2022ztv}:
\begin{enumerate}
\item a set of $\IP^1$'s in class $n[C]$ may shrink to zero size;
\item a divisor fibered by those $\IP^1$'s may shrink to a genus $g$ curve $\Sigma$ of $A_1$ singularities;
\item a divisor $D$ may shrink to a point, at the same time as the $\IP^1$'s shrink to zero size;
\item the whole CY three-fold may shrink to a point.
\end{enumerate}
The first three cases necessarily occur at finite distance in moduli space, while the last one
occurs at infinite distance, so corresponds to an asymptotic boundary,  similar to the one occurring at $t=0$.  

The first case occurs when there exists some $k_{\text{max.}}\in\mathbb{N}$ such that the GV invariants $\GV_{(0,k)}^{(0)}$ vanish for all $k>k_{\text{max.}}$ \footnote{
In most cases $k_\text{max.}\in\{1, 2\}$, but the example 5.10b exhibits a length 3 flop and $k_{\text{max.}}= 3$.} (i.e, the vector $(0,1)$ is  a nilpotent ray as defined in \S\ref{sec:ZtopGW}), but further requires that the vector $(0,1)$ is not a generator of the infinity cone (i.e. the GV invariants $\GV_{(k_1,k_2)}^{(0)}$ vanish for 
$k_2>\alpha k_1+\beta$, with $\alpha$ a positive number). 
Crossing such a wall corresponds to a flop transition, leading to a birationally equivalent geometry with the same Betti numbers but different intersection numbers, namely 
\be
\kappa_{222} \mapsto \kappa_{222} - \sum_{k\geq 1} k^3 \GV^{(0)}_{0,k}\,, \quad 
c_{2,2} \mapsto c_{2,2} + \sum_{k\geq 1} k \GV^{(0)}_{0,k} \,.
\ee
The K\"ahler cone of this new model lives in a subset of the lower half-plane $s<0$, and can be glued to the K\"ahler cone of the original model at $s=0$. The effective cone is unchanged, since it is invariant under birational transformations.

The GV invariants $\GV_{(k_1,k_2)}^{(g)}$ of the new geometry are, up to a change of basis, the same as the original ones, except for the invariants $\GV_{(0,k)}^{(g)}$ on the nilpotent ray, which are now formally attached to the vectors $(0,-k)$ outside the positive quadrant. One may then perform a linear transformation $(T,S)\mapsto (T+x S,-S)$ such that all GV invariants are now inside the positive quadrant, and study the boundary of its K\"ahler cone by the same method as above. In general, the flop does not preserve the genus one fibration, although it is possible to obtain a new genus one fibration (or even an isomorphic one, or a K3-fibration) after one or several flops.~\footnote{However, in the case $b_2(X)=2$ each birationality class can only contain at most two fibration structures~\cite{Caldararu:1997}. This is because the large base limit always corresponds to a wall of the extended K\"ahler cone and the cone has exactly two walls.}

The second case, referred to a `Zariski wall of type $a$'  in \cite{Brodie:2021nit},
requires that the class $(0,1)$ is both nilpotent and is a generator of the infinity cone, 
in other words  $\GV_{(0,k)}^{(0)}$ is non-zero for a finite number of wrappings $k$, but
 there exists $(k_1,k_2)$ such that $\GV_{(k_1,k_2+k)}^{(0)}\neq 0$ for an infinite set of integers $k$. Physically, M2-branes wrapped on the vanishing $\IP^1$'s lead to an enhanced $SU(2)$ gauge symmetry along with $g$ hypermultiplets in the adjoint of $SU(2)$, as well as $N_F$ fundamental hypermultiplets if the $\IP^1$ fibration degenerates at $2N_F$ points on $\Sigma$.
 Formally, the flop transition leads to an isomorphic geometry, identified with the original one under the Weyl group $\IZ_2$.

The third case, referred to a `Zariski wall of type $b$'  in \cite{Brodie:2021nit}, may occur when an infinite set of genus 0 invariants $\GV_{(0,k)}^{(0)}$ are non-zero. The contracting divisor $D$ is typically associated to the boundary of the effective cone, and its volume vanishes quadratically in $s$, whereas it vanishes only linearly for a Zariski wall of type $a$.
Physically, the M5-brane wrapped on $D$ leads to a tensionless string, accompanied with an infinite tower of massless particles corresponding to the M2-branes wrapped on $C$. The K\"ahler cone cannot be extended beyond such a wall, which therefore signifies a boundary of the  extended K\"ahler cone. As shown in \cite[\S 3.2]{Brodie:2021nit}, in the case of two-parameter models a necessary condition for a type $b$ Zariski wall to arise is that the cubic form $\kappa_{abc} t^a t^b t^c$ has a single zero in $\IR \IP^2$.

Defining the hyperextended K\"ahler cone as the union of the images of the extended K\"ahler cone
under the Weyl reflections due to Zariski walls of type $a$, one of the main results of \cite{Gendler:2022ztv} is that the infinity cone is dual to the hyperextended K\"ahler cone. This gives important information on the support of GV invariants. Moreover,  the effective cone is claimed to be dual to the union of the images of the K\"ahler cones of the CY threefolds in the birational orbit of $X$ under the quadratic map\footnote{Note that the map is only locally quadratic, since the intersection form $\kappa_{abc}$ changes along the birational orbit.}
 $t^a \mapsto \hat t_a = \frac12 \kappa_{abc} t^b t^c$ \cite{Alim:2021vhs,Gendler:2022ztv}. This result, while still conjectural, is based on the expectation that each boundary of the effective cone should correspond to an M5-brane of suitable charge becoming tensionless. Assuming the validity of this claim, this gives an efficient way of computing the effective cone, which can then be checked by line bundle cohomology computations.  The results are reported in Tables \ref{tab:effcones123}-- \ref{tab:effcones5} below. As an illustration, we discuss a few models in great detail below:

\noindent{\bf 1.1.} Let us first consider the simplest model, the smooth elliptic fibration over $\IP^2$, realized for example as a degree 18 hypersurface in weighted projectve space $\IP^4_{9,6,1,1,1}[18]$, 
with $(\kappa,\ell,N,c_2)=(9,3,1,102)$. In the usual basis such that GV invariants are supported in the right positive quadrant, the K\"ahler cone is $s>0,t>0$, bounded by the rays $(0,1)$ and $(1,0)$. 
There are no nilpotent rays, thus no flop.
The boundary at $t=0$ is an asymptotic boundary, while the boundary at $s=0$ is a Zariski wall of type $b$. On this wall, the tension of an M5-brane with charge $(1,-3)$ vanishes quadratically, 
\be
\hat t - 3 \hatS = \left( \frac92 t^2 + 3 st +\frac12 s^2 \right) - 3 \left(\frac32 t^2 +s t \right) = \frac12 s^2
\ee 
The image of right positive quadrant in the $(\hat t,\hatS)$ plane is the wedge $0< \hat t<3 \hatS$, from which one deduces that the effective cone is bounded by $(0,1)$ and $(1,-3)$.  
 
\noindent{\bf 2.6.} This model admits one nilpotent ray, $\GV_{(0,k)}^{(0)}=2$ for $k=1$ and zero otherwise, which we denote by  $(0,1)_{2}$. The initial K\"ahler cone is $s,t>0$ and its image in dual coordinates is the wedge $0<\hat s<\frac35 \hat t$. After flopping the curve of degree $(0,1)$, one finds a new phase with K\"ahler cone $0<-s<t$, which is mapped to the positive quadrant 
$s',t'>0$ under the variables change $(t',s')=(t+s,-s)$. In these coordinates, one recognizes the intersection form of a fibration by degree $4$ K3 surfaces with $\kappa'=10, c_2'=64$, which we denote by $[K3_{10,2,64}]$. The image of the new K\"ahler cone  in dual coordinates (using the 
quadratic map associated to the intersection form after the flop) is the wedge $0<\frac35 \hat t <\hat s< \hat t$. The union of the images of the two K\"ahler cones is the wedge $0< \hat s < \hat t$, which shows that the effective cone is bounded by $(0,1)$ (as for all models) and $(1,-1)$. The boundary of the extended K\"ahler coincides with the boundary of the effective cone $(1,-1)$, which therefore corresponds to an asymptotic boundary. The dual of the extended K\"ahler cone gives the infinity cone $0\leq k_2\leq k_1$, which is consistent with the support of GV invariants. Upon restricting to $S=0$, we observe that the genus 0 GV invariants $\sum_{0\leq k_2 \leq k_1} \GV_{k_1,k_2}^{(0)} = \{920, 50520, 5853960,\dots\}$ coincide with the invariants of the CY operator $\#51$ in the AESZ database, which indicates that this model has a conifold transition to a one-parameter model described by that operator. 

\noindent{\bf 4.4.} This model admits one nilpotent ray with $\GV_{(0,1)}^{(0)}=24$, $\GV_{(0,2)}^{(0)}=-2$, which we denote by  $(0,1)_{24,-2}$. The initial K\"ahler cone is the first quadrant $s,t>0$ and its image in dual coordinates is the wedge $0<\hat s<\hat t$. The wall at $s=0$ is a Zariski wall of type $a$. A flop with respect to the curves of degree $(0,1)$ and multiples thereof corresponds to a Weyl symmetry $(t,s)\mapsto (t+s,-s)$, identifying the first quadrant with its image $-t<s<0$. The union of these two covers the effective cone bounded by $(0,1)$ and $(1,-1)$, consistently with the image of the initial K\"ahler cone in dual coordinates. The support of the GV invariants $\GV_{(k_1,k_2)}^{(0)}$ extends to arbitrary large $k_2$ for fixed $k_1$, but we observe that the free energy 
$f_{k_1}(S)=\sum_{k_2\geq 1} \GV_{(k_1,k_2)}^{(0)} \Li_3(q_S^{k_2})$ at fixed degree $k_1$ is a rational function of $q_S$, invariant under $q_S\mapsto 1/q_S$, with a pole of degree $2k_1$ at $q_S=-1$. E.g. for $k_1\leq 2$,
\be
\begin{split}
f_1(S)=& \frac{128 + 1280 q_S + 2304 q_S^2 + 1280 q_S^3 + 128 q_S^4}{(1 + q_S)^2}, \quad \\
f_2(S)=& \frac{168 + 13632 q_S + 118528 q_S^2 + 356160 q_S^3 + 502192 q_S^4 + 356160 q_S^5 + 
 \dots+ 168 q_S^8}{(1 + q_S)^4}
 \end{split}
\ee
etc. As a result, the limit $S\to 0$ is smooth, and we find that it reduces to the free energy
of a one-parameter  with genus 0 GV invariants  $\{1280, 92288,15655168, \dots\}$, which 
we recognize as the  hypergeometric model $X_{4,2}$.

 \newpage

\begin{figure}[H]
\begin{center}
\includegraphics[height=6.5cm]{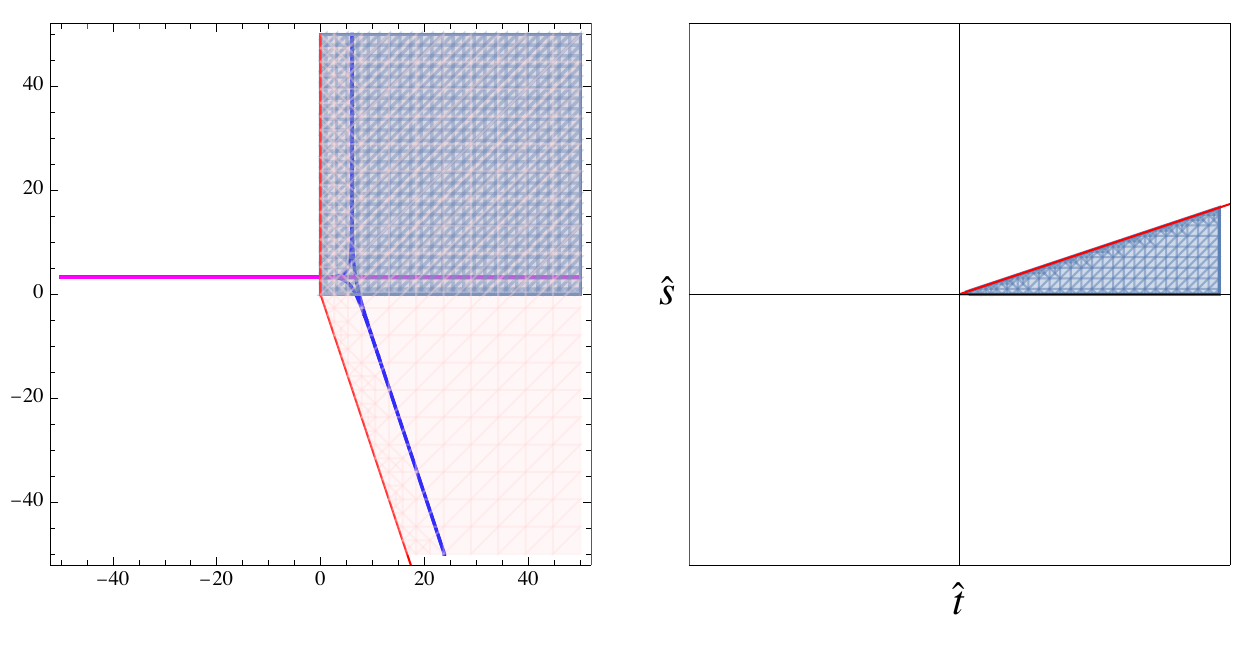}
\\
\includegraphics[height=6.5cm]{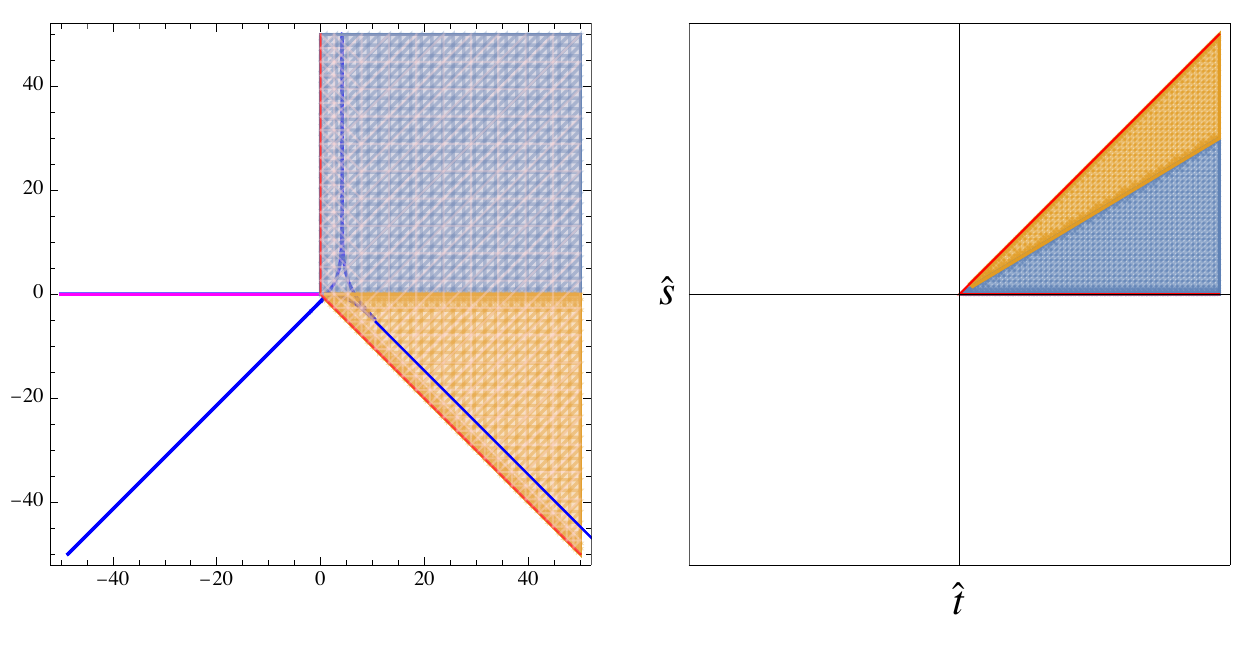}
\\
\includegraphics[height=6.5cm]{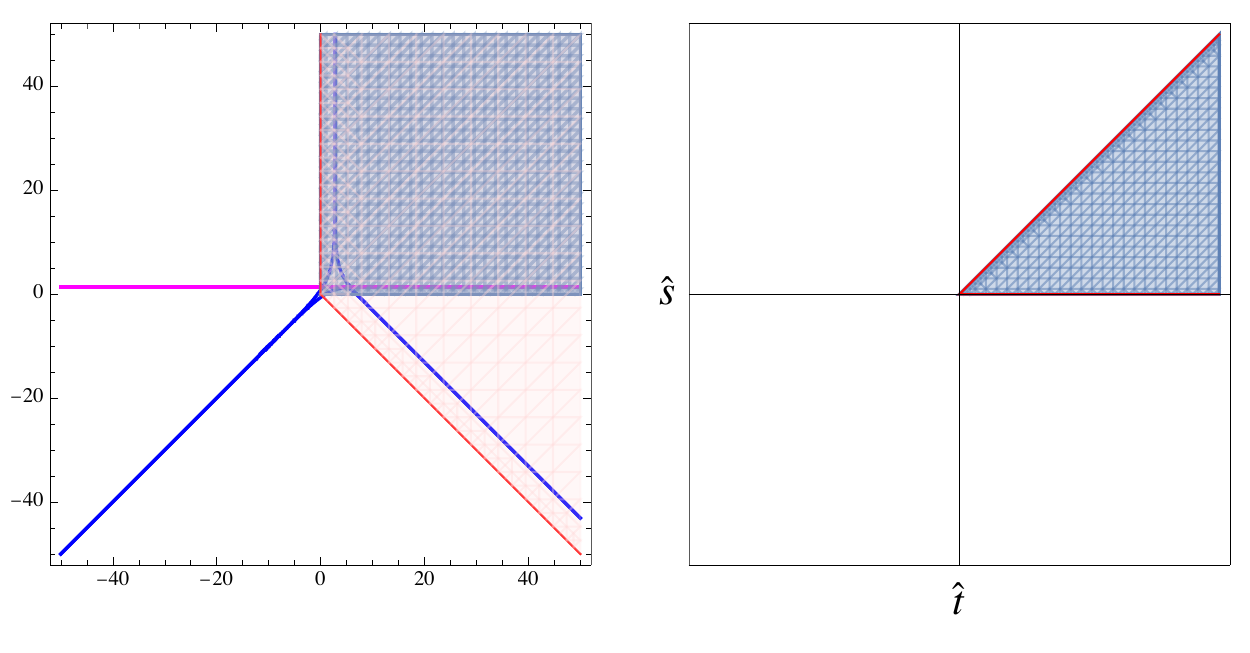}
\end{center}
  \caption{Extended K\"ahler cone (left) and its image in dual coordinates (right) for models 1.1, 2.6, 4.4
  (from top to bottom). The light red region indicates the effective cone for generic complex structure. On the left side, we have superposed the two components of the discriminant locus in blue and magenta.   \label{fig_Cones}}
 \end{figure}

 \begin{table}[H]
\begin{align*}
   \begin{array}{|l|r|c|r|r|l|c|l|}
  \hline
  X & \chi_X & \mbox{Nilp. rays} 
& {\rm Eff.cone}  & {\rm Ext.cone} &  {\rm Flop} & {\rm Conifold} & {\rm Ref.} \\[1mm]
  \hline
  \hline
1.1 & -540 & [\IP_2] &  (1, -3)  & (1,0)  & & &
\begin{array}{c}\IP^4_{9,6,1,1,1}[18]\\ \cite{Alim:2012ss,Klemm:2012sx,Huang:2015sta}\end{array}
 \\  \hline\hline
2.1       & -284   & [\IF_0]    & (1,-2) & (1,0)       & & &
        \cite{Braun:2011ux} \cite[\S6.2]{Schimannek:2021pau}
         \\  \hline
2.2     & -260     & [dP_7]     &  (1,-1) & (1,0)   & &&  \cite{Braun:2011ux} 
        \\  \hline
2.3          & -252     & [K3]    &  (1,0)  & (1,0)  & &&    \cite[\S7.2.1]{Cota:2019cjx}    
                \\    \hline
2.4    & -228     & [Sym]   &  (1,0) & (1,0)   & \mbox{double fib.} & & 
        \\  \hline
2.5          & -220& (0,1)_{20,-2} & (2,-1) & (1,0)    & Z_a &  X_{6,2} &
        \\    \hline  
2.6   & -196   & (0,1)_{2} & (1,-1)  & (1,-1) &  \to K3_{10, 2, 64} &  \AESZ{51} &
        \\ \hline 
2.7   & -188 & (0,1)_{56, -2} & (2, -1) & (1,0) &   Z_a &  X_8 &
       \\   \hline 
2.8   & -164 & (0,1)_{6} &  (1,-1)   & (2,-1) & \to [dP_7] &  X_{4,2} & 
        \\   \hline   
2.9     & -156 & [4.1]  & (1,0)  & (1,0)  &   \mbox{double fib.} & & 
        \\  \hline   
2.10   & -132 &  (0,1)_{12} &  (2,-1) & (2,-1)  &   \mbox{isoflop} &   X_{4,3} & 
        \\ \hline
2.11  & -100  & (0,1)_{20,2} & (3,-1) & (3,-1) &   \mbox{isoflop} &   X_{4,4} &
        \\        \hline     \hline           
3.1  & -186    & (0,1)_{-2} &  (1,-2)  & (1,0) & Z_a &   &  \cite[(18,1)]{Brodie:2021nit}  
         \\  \hline
3.2 & -180 &  [dP_6]  & (1,-1) & (1,0)  & & & 
        \\ \hline     
3.3             & -168 &  (0,1)_{18,-2}  &  (1,-1)  & (1,0)  & Z_a &  X_5 & \cite[(10,1)]{Brodie:2021nit} 
        \\  \hline                         
 3.4      & -168   & (0,1)_{1} &  (1,-1)  & (1,-1) &  \to K3_{11,2, 62} &   & 
        \\ \hline
 3.5  & -162 &  [Sym] &  (1,0)  & (1,0) & & &  \begin{array}{c}\left(\tiny\begin{array}{c|cc}  \IP^2 & 3 \\ \IP^2 & 3  \end{array}\right)_{\#7884} \\ \cite[\S B.2]{Hosono:1993qy} \end{array}
   \\
   \hline
 3.6  & -150   & (0,1)_{72,1}  &  (5 ,-1) & (5,-1) &  \mbox{isoflop} &  X_{8} &  \left(\tiny\begin{array}{c|cc}  \IP^2 & 2 & 1 \\ \IP^3 & 3 & 1  \end{array}\right)_{\#7883}  \\
         \hline
  3.7  & -150 &  (0,1)_{13}& (1,-1) & (4,-1) & \to [\IP_2] &  
         X_{4,2} &   \\
        \hline              
  3.8        & -144 & (0,1)_{3} & (1,-1)  & (1,-1)  &    \to K3_{ 15,3, 66} & 
  \AESZ{24}& 
        \\  \hline
  3.9  & -132 &   \begin{array}{l} (0,1)_{34} \\ (1,4)_{-2} \end{array} 
        &   (2,-1) & (4,-1)
              &   \mbox{flop}+Z_a &  X_5 &  \left(\tiny\begin{array}{c|ccc}  \IP^2 & 1 & 1 & 1 \\ \IP^3 & 3 & 1 & 1  \end{array}\right)_{\#7868} 
 \\   \hline           
  3.10   & -132 & (0,1)_{1} & (1,-1) & (1,-1) & \to [4.3] &    & 
\\ \hline
  3.11 & -126 &        (0,1)_{9} & (1,-1) & (2,-1) & \to  [dP_6] & 
         X_{3,2,2} &  
                \\ \hline
  3.12     & -114 &     (0,1)_{64,27}     &   (7,-1) & (7,-1) 
        &  \mbox{isoflop} &   X_8 &   \left(\tiny\begin{array}{c|cc}  \IP^2 & 2 & 1 \\ \IP^3 & 1 & 3  \end{array}\right)_{\#7833} \\
 \hline
  3.13         & -108 &  (0,1)_{18}    &  (2,-1) & (2,-1)&    \mbox{isoflop} &   X_{3,3}& 
          \left(\tiny\begin{array}{c|cccc}  \IP^2 & 0 & 1 & 1 & 1 \\ \IP^5 & 3 & 1 & 1 & 1  \end{array}\right)_{\#7808}
 \\  \hline   
  3.14 & -96 &(1,2)_1 & (1,-1)  & (1,-1) & \to [5.7_a] &   & 
\\ \hline
  3.15     & -90 &    (0,1)_{30,3}       &   (3,-1) & (3,-1)
        & \mbox{isoflop} &  X_{4,3} & 
         \left(\tiny\begin{array}{c|ccc}  \IP^2 & 0 & 2 & 1 \\ \IP^4 & 3 & 1 & 1  \end{array}\right)_{\#7668}   \\
\hline   
          \end{array}
\end{align*}
\caption{Nilpotent rays, boundaries of the effective cone and extended K\"ahler cone, allowed flops and conifold transitions for 1,2 or 3-section fibrations on $\IP^2$. The notation $[dP_n]$ indicates a genus one fibration whose horizontal GV invariants are those of a del Pezzo surface $dP_n$ (similarly for $[\IP_2]$ and $[\IF_0]$). $[K3]$ indicates a model admitting both a genus one and a K3 fibration, while $[Sym]$ indicates a model invariant under $S\leftrightarrow T$. The notation $\to K3_{\kappa,m,c_2}$ indicates a flop transition to a fibration by  degree $2m$ K3 surfaces, with
$c_{111}=\kappa$ and $c_{2,1}=c_2$. }
\label{tab:effcones123}
\end{table}

  \begin{table}[H]
\begin{align*}
  \begin{array}{|l|r|c|r|r|l|c|l|}
  \hline
  X & \chi_X & \mbox{Nilp. rays} 
& {\rm Eff.cone}  & {\rm Ext.cone} &  {\rm Flop} & {\rm Conifold} & {\rm Ref.} \\[1mm]
  \hline
  \hline
 4.1 & -156 &  [2.9] & (1,0) & (1,0) & \mbox{double fib} & &    \\
      \hline
 4.2         & -140&  [K3_{0,2,36}] &      (1,0)       &  (1,0)  & &
      &  \\    \hline
4.3 & -132 & (0,1)_{1} & (1,-1) & (1,-1) & \to [3.10] &   &
    \\   \hline
4.4    & -132&  (0,1)_{24,-2} & (1,-1) & (1,0)   & 
     Z_a  &   X_{4,2} & \cite[\S 7.2]{Schimannek:2021pau}
       \\   \hline
4.5    & -132&  (0,1)_{-2} &(1,-2)  & (1,0)    & Z_a  &   &
        \\    \hline
4.6 & -128& (0,1)_{16,-2} &(1,-1) & (1,0)   & Z_a  &  X_{4,3} &
       \\  \hline
4.7 & -128& [dP_5] & (1,-1)  & (1,0)  &    &  & 
  \\  \hline          
4.8  & -124     &  (0,1)_{64,2} &  (4,-1) & (4,-1)
               &  \mbox{isoflop} &   X_{6,2} &  \left(\tiny\begin{array}{c|ccc}  \IP^2 & 0 & 2 & 1 \\ \IP^4 & 2 & 2 & 1  \end{array}\right)_{\#7853} 
    \\ \hline
4.9   & -124& (0,1)_{2} &  (1,-1)   & (1,-1) &  \to K3_{16,3, 64} &   X_{2,2,2,2} &
       \\  \hline 
4.10   & -120  & (0,1)_{80,8} & (6,-1)  & (6,-1)
        &  \mbox{isoflop}&   X_8 & \left(\tiny\begin{array}{c|cc}  \IP^2 & 2 & 1 \\ \IP^3 & 2 & 2  \end{array}\right)_{\#7844}
        \\  \hline      
4.11   & -120&  (0,1)_{12} &   (1,-1)  & (4,-1)     & \to [\IP_2]&   X_{3,3} &
        \\  \hline
4.12 & -120&  [Sym]  & (1,0)  & (1,0) &  & &
       \\   \hline   
  4.13 & -116 & (0,1)_{14} &(1,-1) & (3,-1) &  \to [\IF_0] &  X_{3,2,2} &
       \\   \hline   
 4.14 & -112 &  (0,1)_{4} & (1,-1) & (1,-1)& \to  K3_{20,4,68} &   \AESZ{25} &
       \\     \hline   
4.15&   -112 &     \begin{array}{l} (0,1)_{44}  \\ (1,4)_1 \end{array}     
        &  (3,-1) & (3,-1)
        &  \to  K3_{ 28,4,76}  &  X_5&  \left(\tiny\begin{array}{c|ccc}  \IP^2 & 1 & 1 & 1 \\ \IP^4 & 2 & 2 & 1  \end{array}\right)_{\#7821} 
         \\ \hline
 4.16   & -108 &   \begin{array}{l} (0,1)_{34} \\ (1,3)_{12,-2} \end{array}
        &  (2,-1) & (3,-1)   & \mbox{flop} + Z_a &  X_{4,2} & 
         \left(\tiny\begin{array}{c|cccc}  \IP^2 & 0 & 1 & 1 & 1 \\ \IP^5 & 3 & 1 & 1 & 1  \end{array}\right)_{\#7807} 
           \\   \hline
 4.17& -104  &   \begin{array}{l} (0,1)_{1} \\ (1,1)_{76,20} \\ (9,8)_1 \end{array} 
    & (7,-8) & (7,-8)& \mbox{isoflop} &   &
   \\    \hline
 4.18 & -104 & (0,1)_{12} &  (1,-1)  & (2,-1) &\to [dP5] &   X_{2,2,2,2} &
 \\        \hline
 4.19 & -100 &(0,1)_{1} &   (1,-1)& (1,-1) & \to [5.11_{ab}] &    &
 \\      \hline
 4.20    & -100 &   (0,1)_{62,16}        &  (5,-1) & (5,-1)
        &       \mbox{isoflop} &   X_{6,2}&   \left(\tiny\begin{array}{c|ccc}  \IP^2 & 0 & 2 & 1 \\ \IP^4 & 2 & 1 & 2  \end{array}\right)_{\#7758}  
                 \\   \hline
 4.21 & -96 & (0,1)_{3} &  (1,-1) & (1,-1)& \to [5.9_{a}] &   \AESZ{198} & 
  \\ \hline
   4.22  & -96 &   (0,1)_{24}   & (2,-1)  & (2,-1) 
        &  \mbox{isoflop}  &  X_{3.2.2}&  \left(\tiny\begin{array}{c|ccccc}  \IP^2 & 0 & 0 & 1 & 1 & 1 \\ 
   \IP^6 & 2&2 & 1 & 1 & 1  \end{array}\right)_{\#7725}  
   \\  \hline
  4.23  & -88 &   (0,1)_{40,4}  &  (3,-1) & (3,-1) 
        &  \mbox{isoflop} &    X_{4,2}& \left(\tiny\begin{array}{c|cccc}  \IP^2 & 0 & 0 & 2 & 1  \\ \IP^5 & 2&2 & 1 & 1   \end{array}\right)_{\#7643}   
        \\  \hline   
 \end{array}
\end{align*}
\caption{Nilpotent rays, boundaries of the effective cone and extended K\"ahler cone, allowed flops and conifold transitions for 4-section fibrations on $\IP^2$. Notations similar as in Table 
\ref{tab:effcones123}}
\label{tab:effcones4}
\end{table}

  \begin{table}[H]
\begin{align*}
  \begin{array}{|l|r|c|r|r|l|c|l|}
  \hline
  X & \chi_X & \mbox{Nilp. rays} 
& {\rm Eff.cone}  & {\rm Ext.cone} &  {\rm Flop} & {\rm Conifold} & {\rm Ref.} \\[1mm]
\hline
	5.1_a & -90 & (0,1)_{50,5} & (3,-1) & (3,-1) & \mbox{isoflop}&  \AESZ{51} &  \\ \hline
	5.1_b & -90 & [dP_4?]  & (1,0) & (1,0) & & & \\ \hline
	5.2_a &-90  & (0,1)_{30} & (2,-1) & (2,-1) & \mbox{isoflop} & 
	 \AESZ{24}  & \cite{Knapp:2021vkm} \\ \hline
	5.2_b &  -90   & [Sym] &(1,0)  & (1,0) & &  & \cite{Inoue:2019jle,Knapp:2021vkm} \\\hline
	5.3_a & -90  & (0,1)_{15}& (1,-1) & (2,-1) & \to  [dP_4?] &  \AESZ{25} & \\ \hline
	5.3_b & -90 & (0,1)_{5} & (1,-1) & (1,-1) &  
	\to K3_{25,5,66} &   \AESZ{101}  & \\ \hline
	5.4_a & -94   & 
	\begin{array}{l} (0,1)_1 \\ (1,1)_{70,11} \\ (7,6)_1 
	\end{array}
	& (5,-6)  &(5,-6) &  \mbox{isoflop} &    & \\ \hline
	5.4_b & -94   & (0,1)_{11} &(1,-1) & (4,-1) & 
	\to [\IP_2] &  \AESZ{210} &
	\\ \hline
	5.5_a & -94 & 
	\begin{array}{l}  (0,1)_2 \\ (1,1)_{88,13} \\ (9,8)_2 
	\end{array} &(7,-8) &(7,-8) & \mbox{isoflop}&  \AESZ{27} &  \\ \hline
	5.5_b & -94  &  (0,1)_{13} &(1,-1) &(3,-1) & \to [\IF_0] &
	 \AESZ{99} & 
	 \\ \hline
	5.6_a &  -94  & 
	\begin{array}{l} (0,1)_{41} \\ (1,3)_{52,1}
	\end{array}  &  (11,-4)&(11,-4) &  \to[5.6_b] &  X_{4,2} &
	 \\ \hline
	5.6_b & -94  &
	 \begin{array}{l} (0,1)_{52,1} \\ (1,4)_{41}
	 \end{array}
	  &(11,-3) &(11,-3) &\to[5.6]_a &  X_5 &
	   \\ \hline
	5.7_a & -96 &  (0,1)_{1} &(1,-2) & (1,-2) & 
	\to [3.14] &
	  & 
	\\ \hline
	5.7_b &  -96  &  (0,1)_{-2}  &(1,-2) & (1,0) & Z_a &   &
	 \\ \hline
	5.8_a &  -96  &
	  \begin{array}{l} (0,1)_1 \\ (1,1)_{26,-2} \\ (4,3)_1 
	  \end{array}  & (2,-3) & (2,-3) & \mbox{isoflop} &   &   \\ \hline
	5.8_b & -96  & (0,1)_{14,-2}  & (1,-1) & (1,0) & Z_a &   \AESZ{109} &   \\\hline
	5.9_a &-96 &  (0,1)_3 &(1,-1) &  (1,-1)  &  \to [4.21] & \AESZ{198} &
	 \\ \hline
	5.9_b &-96  & (0,1)_3 & (1,-1)& (1,-1)  & \to K3_{21,4,66}&   \AESZ{193} &  \\\hline
	5.10_a & -96  & 
	\begin{array}{l} (0,1)_{16} \\ (1,3)_{-2} \end{array} &(1,-1) &(3,-1) & \mbox{flop} + Z_a & 
	 X_{2,2,2,2} & \\ \hline
	5.10_b & -96  &  (0,1)_{83,16,1}  &(7,-1) &(7,-1) & \mbox{isoflop}&  X_8& \\\hline
        5.11_{ab} & -100  & (0,1)_1 & (1,-1)& (1,-1) & \to [4.19]&  &  \\\hline
	5.12_a & -104&  (0,1)_{22,-2} &(1,-1)  & (1,0) & Z_a&  X_{3,3} & \\ \hline
	5.12_b  & -104 & (0,1)_1 & (1,-1) & (1,-1) &\to K3_{17,3, 62} & &  \\ \hline
   \end{array}
\end{align*}
\caption{Nilpotent rays, boundaries of the effective cone and extended K\"ahler cone, allowed flops and conifold transitions for 5-section fibrations on $\IP^2$. Notations similar as in Table 
\ref{tab:effcones123}. The notation $[dP_4?]$ refers to a list of GV invariants $\{10, -10, 15, -40, 135, -510, 2100, -9280, 43245,\dots\}$ which we tentatively identify 
as those of a $dP_4$ surface, complementing the list of vanishing 4-cycles in \cite[Table 6]{Klemm:1996hh}. }
\label{tab:effcones5}
\end{table}

\begin{table}[H]
\begin{align*}
   \begin{array}{|c|c|c|c|c|c|c|}
  \hline
 {\rm AESZ} & \widetilde{X}   & \chi_{\widetilde{X}} & \kappa & c_2 & X \\[1mm]
  \hline
  \hline
24& X^{2,5}_{\mathcal{O}(1)^{\oplus 2}\oplus\mathcal{O}(3)}~\cite{Haghighat:2008ut,Ueda:2016wfa,Doran:2024kcb} & -150 &15 &66 & 3.8, 5.2_a, X_5^{[2,1]} \\ \hline 
25  & X^{2,5}_{\mathcal{O}(1)\oplus\mathcal{O}(2)^{\oplus 2}}~\cite{Haghighat:2008ut,Katz:2022lyl,Doran:2024kcb} & -120 &20 &68 & 4.14, 5.3a, X_5^{[1,1]} 
\\  \hline  
26 & X^{2,6}_{\mathcal{O}(1)^{\oplus 4}\oplus\mathcal{O}(2)}~\cite{Haghighat:2008ut,Ueda:2016wfa,Doran:2024kcb} &-116 &28 & 76 & X_7^{[1,1]}
 \\ \hline
27 & X^{4,6}_{\mathcal{S}^\vee(1)\oplus\mathcal{O}(1)}~\cite{Ueda:2016wfa} &-98  &42 &84 & 5.5_a 
\\ \hline
29 & X^{2,5}_{\mathcal{S}^\vee(1)\oplus\mathcal{O}(2)}~\cite{Ueda:2016wfa,Doran:2024kcb} & -116 &24 & 72 & X_6^{[1,1]}
 \\ \hline
 42 & X^{3,6}_{\bigwedge^2\mathcal{S}^\vee\oplus\mathcal{O}(1)^{\oplus 2}\oplus\mathcal{O}(2)}~\cite{Ueda:2016wfa,Doran:2024kcb} &-116 &32 & 80 & X_8^{[1,1]} 
\\ \hline
51  & \begin{array}{c}\text{Smooth d.c. of Fano 3-fold}\\B_5=X^{2,5}_{\mathcal{O}(1)^{\oplus 3}}~\cite{Katz:2022lyl,Doran:2024kcb}\end{array}& -200 &10 & 64  & 2.6, 5.1_a, X_5^{[2,2]}  
 \\ \hline 
99 & \text{Pfaffian in }\mathbb{P}^6~\cite{Kanazawa2012,Hori2013,Kapustka_2015,Katz:2022lyl} &-120 &13 &58 & 5.5_b 
\\ \hline
101 & \begin{array}{c}X^{3,5}_{(\bigwedge^2\mathcal{S}^\vee)\otimes\mathcal{O}(1)}\text{, def. equivalent to}\\(\text{Gr}(2,5)\cap\text{Gr}(2,5))\subset\mathbb{P}^9~\cite{inoue2019completeintersectioncalabiyaumanifolds,Ueda:2016wfa}\end{array} &-100 &25 &70 &  5.3_b
 \\ \hline
 109 & & -120 &7 &46 & 5.8_b
 \\ \hline
 185 & X^{5,7}_{\bigwedge^4\mathcal{S}^\vee\oplus\mathcal{O}(1)\oplus\mathcal{O}(2)}~\cite{Ueda:2016wfa,Doran:2024kcb} & -120 &36 & 84 & X_9^{[1,1]}
 \\ \hline
 193 & \text{Derived dual of }X^{2,6}_{\mathcal{S}^\vee(1)\oplus\mathcal{O}(1)^{\oplus}3}~\cite{miura2017}&-102 &21 &66  & 5.9_b 
\\ \hline
198^* & X^{2,6}_{\mathcal{S}^\vee(1)\oplus\mathcal{O}(1)^{\oplus 3}}~\cite{Ueda:2016wfa} &-102 &33 &78 &  4.21, 5.9_a  
\\ \hline
210 & \text{Pfaffian in }\mathbb{P}^6_{1111112}~\cite{Kanazawa2012,Hori2013,Katz:2022lyl} &-116  &10 & 52&  5.4_b 
\\ \hline
4.3.31 & \begin{array}{c}\text{Smooth d.c. of Fano 3-fold}\\A_{22}=X^{3,7}_{(\bigwedge^2 \mathcal{S}^\vee)^{\oplus 3}}~\cite{Doran:2024kcb}\end{array} & -128 & 44 & 92 & X_{11}^{[1,1]} 
\\ \hline 
   \end{array}
\end{align*}
\caption{One-parameter Calabi-Yau threefolds $\widetilde{X}$ with mirror periods associated to non-hypergeometric AESZ~\cite{almkvist2010tablescalabiyauequations} operators, obtainable from two-parameter genus-one or K3-fibered models through a conifold transition. The Euler numbers are related by $\chi_{\widetilde{X}}=\chi_X-2 \sum_{d\geq 1} \GV_{(0,d)}^{X(0)}$. We use the notation $X^{k,n}_E$ for complete intersections in Grassmanians $\text{Gr}(k,n)$ and $X_{m}^{[i,j]}$ for K3-fibered CY threefolds from~\cite{Doran:2024kcb}. The geometries $\widetilde{X}$ were identified by comparing topological and enumerative invariants after the transition. For AESZ operators marked with a star, the relevant periods are obtained after a change of sign of the complex structure coordinate.}
\end{table}

\section{Some examples over bases with higher Picard rank}
\label{app_higher_rank}
In this appendix we will discuss examples of genus one fibered CY threefolds over bases $\mathbb{F}_0=\IP^1\times \IP^1$, $\mathbb{F}_1$ and dP$_2$.
Other examples can also be found in~\cite[Section 5.3]{Cota:2019cjx},~\cite{Banlaki:2019bxr} and~\cite[Appendix D]{Knapp:2021vkm}.

\subsection{Elliptic fibrations over $\IF_k$}
Recall that the Hirzebruch surface  $\IF_k$ is the projectivization of the $\cO(0)\oplus\cO(-k)$ bundle over $\IP_1$. We denote by $\check{D}_F$, $\check{D}_B$ the divisors on $\mathbb{F}_1$ that are respectively associated to the generic $\IP^1$ fiber and the base of the $\mathbb{P}^1$-bundle. In the basis $(\check{D}_F,\check{D}_B)$, the intersection form and canonical class are given by 
\begin{align}
    C_{\alpha\beta}=\left(\begin{array}{cc}0&1\\1&k\end{array}\right)\,, \quad 
    c_{\alpha}=(2,k+2)\,,
\end{align}
such that 
\begin{align}
    C^{\alpha\beta}=\left(\begin{array}{cc}-k &1\\1&0 \end{array}\right)\,, \quad 
    c^{\alpha}=(2-k,2)\,.
\end{align}
Note that $\IF_0$ is the product $\IP^1\times \IP^1$, and $\IF_1$ is the blow-up of $\IP^2$ at one point, with $H=\check{D}_F+\check{D}_B$ the hyperplane class of $\IP^2$ and $\check{D}_B$ the exceptional divisor. $\IF_k$ is Fano for $k=0,1$ and almost Fano for $k=2$. As usual, we denote by $S_1,S_2$ the K\"ahler moduli associated to the pull-back divisors $(D_F,D_B)$, and by $T$ the moduus of the elliptic fiber.

The smooth elliptic fibration over $\IF_0$, known as the STU-model, has been intensively studied, as one of the first examples of heterotic-type II duality \cite{Hosono:1993qy,Kachru:1995wm,Harvey:1995fq,Marino:1998pg} \cite[\S 6.10]{Klemm:2004km}.
Its Euler characteristic and non-trivial Hodge numbers are
\begin{align}
    \chi_X=-480\,,\quad h^{1,1}=3\,,\quad h^{2,1}=243\,.
\end{align}
while the intersection numbers are determined by
\begin{align}
    \kappa=8\,,\quad \ell_\alpha=(2,2)\,,\quad c_2=92\,,\quad c_\alpha=a_\alpha=(2,2)
\end{align}
Denoting $S=S_2$ and $U=T+S_1$, one recognizes the intersection numbers of a fibration by Picard rank 2 K3 surfaces. The moduli $S$ and $(T,U)$ correspond to the heterotic axio-dilaton and torus moduli, respectively.  

Generating series for low base degree $(k_1,k_2)$  were found in \cite[(6.68)]{Klemm:2004km} (although expressions for $k_1 k_2=0$ were omitted). At genus 0, using the symmetry under exchange $k_1\leftrightarrow k_2$, we have
\bea
f_{1,0}^{(0)} &=& -2 \frac{{E_4} {E_6}}{\eta^{24}} \,, \quad 
\tildef_{2,0}^{(0)} =f_{2,0}^{(0)}  =-\frac{ {E_4} {E_6} \left(17 {E_4}^3+7 {E_6}^2\right)
}{96\eta^{48}} 
\,,\nn\\
\tildef_{1,1}^{(0)} &=& -\frac{ {E_4} {E_6}  \left(67 {E_4}^3+65 {E_6}^2\right)}{36\eta^{48}}\,,
\nn\\
\tildef_{1,2}^{(0)} &=& -\frac{{E_4} {E_6}  \left(7751 {E_4}^6+23178
   {E_4}^3 {E_6}^2+5551 {E_6}^4\right)}{6912 \eta^{72}} \,.\nn
\eea
For primitive base degree, the topological string partition function is recognized as the elliptic genus of the heterotic string,
\be
 Z_{\check{D}_B}(T,\check{\lambda})=  Z_{\check{D}_F}(T,\check{\lambda})= \frac{2{E_4} {E_6}}{\eta^{24} \phi_{-2,1}} \,.
\ee 

We now turn to the smooth elliptic fibration over $\IF_1$, or KMV model \cite{Klemm:1996hh,Klemm:2012sx}, studied more recently in \cite[\S 7.3]{Alim:2021vhs}. 
Its Euler characteristic and  Hodge numbers are the same as for the STU model, but
the intersection numbers are now
\begin{align}
    \kappa=8\,,\quad \ell_\alpha=(2,3)\,,\quad c_2=92\,,\quad c_\alpha=a_\alpha=(2,3)\,.
\end{align}
At genus 0, we find
\bea
f_{0,1}^{(0)} &=&  -2 \frac{{E_4} {E_6}}{\eta^{24}}
\,, \quad 
f_{1,0}^{(0)} =\frac{E_4}{\eta^{12}}\,,
\nn\\
\tildef_{0,2}^{(0)} &=&  f_{0,2}^{(0)} = -\frac{ {E_4} {E_6} \left(17 {E_4}^3+7 {E_6}^2\right)}{96\eta^{48}} 
\,, \quad 
\tildef_{2,0}^{(0)} = \frac{2 E_4 E_6 }{24\eta^{24}}
\,,\nn\\
\tildef_{1,1}^{(0)} &=&\frac{ E_4 \left(31 E_4^3+105
  E_6^2\right)}{48\eta^{36}}\,,
\nn\\
\tildef_{2,1}^{(0)} &=&-\frac{5 E_{4} E_{6}  \left(E_{6}^2+E_{4}^3\right)}{288 \eta^{48}} \,,
\nn\\
\tildef_{1,2}^{(0)} &=&
\frac{E_{4} \left(15935 E_{4}^6+161186
   E_{4}^3 E_{6}^2+70175 E_{6}^4\right)}{55296 \eta^{60}}\,.
\eea
For  base degree $(0,1)$, the topological string partition function is again the elliptic genus of the heterotic string, For base degree $(1,0)$, it is the elliptic genus of the E-string,
\be
 Z_{\check{D}_B}(T,\check{\lambda})=   \frac{2{E_4} {E_6}}{\eta^{24} \phi_{-2,1}} \ ,\quad
 Z_{\check{D}_F}(T,\check{\lambda}) = -\frac{E_4}{\eta^{12} \phi_{-2,1}}\,.
\ee

\subsection{Fibration with $2$-section over $\IF_0$}
We now take $X$ to be the geometry $(q_1,q_2)=(4,3)$ from~\cite[Section 5.3]{Cota:2019cjx}, which has also appeared as $X(I_0)$ in~\cite{Kachru:1997bz}, as the dual of a CHL heterotic string model.

In terms of the general construction discussed in Section~\ref{sec:generic2sections}, we choose the base $B=\mathbb{F}_0$ with the bundle $V$ being $V=\mathcal{O}_{\mathbb{F}_0}(-2,-1)\oplus\mathcal{O}_{\mathbb{F}_0}\rightarrow \mathbb{F}_0$.
The pullback of the relative hyperplane class on $\mathbb{P}(V)$ to the double cover gives a $2$-section on $X$ and the fibration does not exhibit a section, such that $N=2$. 
For convenience, we denote by $\check{D}_1=\check{D}_F$, $\check{D}_1=\check{D}_B$ the divisors on $\mathbb{F}_0=\mathbb{P}^1\times\mathbb{P}^1$ that are respectively represented by the first and the second $\mathbb{P}^1$ factor.
The Chern classes of the bundle $V$ are
\begin{align}
    c_1(V)=-2\check{D}_1-\check{D}_2\,,\quad c_2(V)=0\,,
\end{align}
and we have $c_1(V)^2=4$, $c_1(V)c_1(B)=-6$ and $c_1(B)^2=8$.
Using the expressions from Table~\ref{tab:fibrationsGenericData} we obtain the Euler characteristic and the non-trivial Hodge numbers
\begin{align}
    \chi_X=-256\,,\quad h^{1,1}=3\,,\quad h^{2,1}=131\,,
\end{align}
as well as the intersection numbers
\begin{align}
    \kappa=8\,,\quad \ell_\alpha=(2,4)\,,\quad c_2=68\,,\quad c_\alpha=a_\alpha=(2,2)\,,
\end{align}
and the height-pairing of the 2-section
\begin{align}
    D=-\pi^*\pi_*(D_eD_e)=-(4D_1+2D_2)\,.
\end{align}

The curves on $\mathbb{F}_0$ dual to $\check{D}_1,\check{D}_2$ are $\check{D}^1=\check{D}_B$ and $\check{D}^2=\check{D}_F$.
At genus 0, we find the generating series of GV invariants of base degree $(1,0)$, $(0,1)$ and $(1,1)$,
\bea
 f_{\check{D}_B}^{(0)}(T) &=&-\frac{2e_{2,2} \left(2 e_{2,2}^2-e_{2,4}\right)}{\eta(T)^8\eta(2T)^{8}} ,\quad \nn\\
f_{\check{D}_F}^{(0)}(T)&=&-\frac{32 \left(e_{2,2}^2-e_{2,4}\right)}{3\eta(T)^{12}}\ ,\quad \nn\\
\tildef_{\check{D}_B+\check{D}_F}^{(0)}(T)&=&
\frac{4(3 e_{2,4}^3 - 4  e_{2,4}^2  e_{2,2}^2-3  e_{2,4}  e_{2,2}^4+4 e_{2,2}^6)}
{9 \eta(T)^{20} \eta(2T)^{8}}
\eea
Using information from genus 1 GV invariants we find 
the base degree $(1,0)$ and $(0,1)$ topological string partition functions
\begin{align}
    \begin{split}
        Z_{\check{D}_B}(T,\check{\lambda})=&-\frac{1}{96}\frac{e_{2,2} \left(e_{2,2}^2-e_{2,4}\right) \left(2 e_{2,2}^2-e_{2,4}\right)}{\eta(2T)^{24}\phi_{-2,1}(2T,\check{\lambda})}\,,\\
        Z_{\check{D}_F}(T,\check{\lambda})=&-\frac{1}{18}\frac{\Delta_4^{\frac12}\left(e_{2,2}^2-e_{2,4}\right)^2}{\eta(2T)^{24}\phi_{-2,1}(2T,\check{\lambda})}\,.
    \end{split}
\end{align}
We refrain from displaying the base degree $(1,1)$ topological string partition function.

\subsection{Fibration with $3$-section over $\IF_1$}
We now take $X$ to be a smooth anti-canonical hypersurface in $\mathbb{P}^2\times \mathbb{F}_1$ with the induced torus fibration $\pi:X\rightarrow \mathbb{F}_1$.
The hyperplane class of the $\mathbb{P}^2$ induces a 3-section on $X$ and there is no $N'$-section with $N'<3$, such that $N=3$.

In terms of the construction discussed in Section~\ref{sec:generic2sections} this corresponds to the choice $V=\mathcal{O}_{\mathbb{F}_1}^{\oplus 3}$, such that $c_1(V)=c_2(V)=0$.
The modular properties of the corresponding topological string partition function have been studied for example in~\cite{Hosono:1999qc}.

Using the expressions from Table~\ref{tab:fibrationsGenericData}, as well as $c_1(B)^2=8$, we obtain the Euler characteristic and the non-trivial Hodge numbers
\begin{align}
    \chi_X=-144\,,\quad h^{1,1}=3\,,\quad h^{2,1}=75\,.
\end{align}
while the intersection numbers are determined by
\begin{align}
    \kappa=0\,,\quad \ell_\alpha=(2,3)\,,\quad c_2=36\,,\quad c_\alpha=a_\alpha=(2,3)\,,
\end{align}
and the height-pairing of the 3-section is
\begin{align}
    D=-\pi^*\pi_*(D_eD_e)=-(D_1+2D_2)\,.
\end{align}
where $D_1,D_2$ are the pull-back of $\check{D}_1=\check{D}_F$ and $\check{D}_2=\check{D}_F+\check{D}_B$ and $D_e$ is the 3-section divisor on $X$
The curves on $\mathbb{F}_1$ dual to $\check{D}_1,\check{D}_2$ are $\check{D}^1=\check{D}_B$, $\check{D}^2=\check{D}_F$. 

At genus 0, we find the generating series of GV invariants of base degree $(1,0)$, $(0,1)$ and $(1,1)$,
\be
\tildef_{\check{D}_B}= \frac{9\Delta_6^{\frac23}}{\eta(3T)^{12}}\ ,\quad 
\tildef_{\check{D}_F}=\frac{54\Delta_6^{\frac13}e_{3,1}^2 e_{3,3}^2}{\eta(3T)^{24}}\ ,\quad
\tildef_{\check{D}_B+\check{D}_F}=\frac{27 \Delta_6\, e_{3,3}^3 ( 13 e_{3,1}^3+108 e_{3,3})}
{2 \eta(3T)^{36}}
\ee
Using information from genus 1 GV invariants we find 
the base degree $(1,0)$ and $(0,1)$ topological string partition functions
\begin{align}
    Z_{\check{D}_B}(T,\check{\lambda})=-\frac{9\Delta_6^{\frac23}}{\eta(3T)^{12}\phi_{-2,1}(3T,\check{\lambda})}\,,\quad Z_{\check{D}_F}(T,\check{\lambda})=-\frac{54\Delta_6^{\frac13}e_{3,1}^2 e_{3,3}^2}{\eta(3T)^{24}\phi_{-2,1}(3T,\check{\lambda})}\,.
\end{align}
We again refrain from displaying the base degree $(1,1)$ topological string partition function.

\subsection{Fibration with $4$-section over dP$_3$}
We will now discuss an example where the K\"ahler cone of the base is not simplicial.
To this end we choose the base to be $B=\text{dP}_3$ with $b_2(\text{dP}_3)=4$.
The del Pezzo surface is a toric variety and we summarize the toric data in Table~\ref{tab:toricDatadP3}.
\begin{table}[ht!]
    \begin{align*}
        \left[ 
        \begin{array}{c|cc|cccccc}
            \text{Divisor}&\multicolumn{2}{|c|}{\vec{p}\in \Delta^\circ}& \check{C}_1' & \check{C}_2' & \check{C}_3' & \check{C}_4' & \check{C}_5' & \check{C}_6' \\\hline
            \check{D}_1' & 1& 0&-1& 0& 0& 0& 1& 1\\
            \check{D}_2' & 0& 1& 0&-1& 0& 1& 0& 1\\
            \check{D}_3' &-1&-1& 0& 0&-1& 1& 1& 0\\
            \check{D}_4' &-1& 0& 0& 1& 1&-1& 0& 0\\
            \check{D}_5' & 0&-1& 1& 0& 1& 0&-1& 0\\
            \check{D}_6' & 1& 1& 1& 1& 0& 0& 0&-1\\\hline
              -K_{\text{dP}_3} & 0& 0&-1&-1&-1&-1&-1&-1
        \end{array}
        \right]
    \end{align*}
    \caption{The toric data of $\text{dP}_3$.}
    \label{tab:toricDatadP3}
\end{table}

The toric divisors are simultaneously the curves that generate the Mori cone of dP$_3$, such that $\check{D}'_i=\check{C}'_i$ for $i=1,\ldots,6$.
The intersection number $\check{D}'_i\cap\check{C}_j'$ is the entry corresponding to the divisor $\check{D}'_i$ in the linear relation among the points that corresponds to the curve $\check{C}_j'$ in Table~\ref{tab:toricDatadP3}.
It is easy to check that in terms of
\begin{align}
    \check{D}^1=\check{C}_1'-\check{C}_2'\,,\quad \check{D}^2=\check{C}_2'\,,\quad \check{D}^3=\check{C}_3'-\check{C}_2'\,,\quad \check{D}^4=\check{C}_5'\,,
\end{align}
we have
\begin{align}
    \begin{split}
    \check{C}_1'=&\check{D}^1+\check{D}^2\,,\quad \check{C}_2'=\check{D}^2\,,\quad \check{C}_3'=\check{D}^2+\check{D}^3\,,\\ \check{C}_4'=&\check{D}^1+\check{D}^4\,,\quad \check{C}_5'=\check{D}^4\,,\quad \check{C}_6'=\check{D}^3+\check{D}^4\,,
    \end{split}
\end{align}
and the dual divisors $\check{D}_i$, $i=1,\ldots,4$, such that $\check{D}_i\cap \check{D}^j=\delta_i^j$, are
\begin{align}
    \begin{split}
    \check{D}_1=&\check{D}_2'+\check{D}_6'\,,\quad\check{D}_2=\check{D}_1'+\check{D}_5'+\check{D}_6'\,,\\ \check{D}_3=&\check{D}_1'+\check{D}_5'\,,\quad \check{D}_4=\check{D}_1'+\check{D}_2'+\check{D}_6'\,.
    \end{split}
\end{align}
We see that the divisors $\check{D}_i$, $i=1,\ldots,4$ form a basis of the Picard lattice as well as a basis of a simplicial subcone of the K\"ahler cone of dP$_3$ and we have $c_1(B)=\check{D}_2+\check{D}_4$.
The intersection numbers $C_{\alpha\beta}=\check{D}_\alpha\cap\check{D}_\beta$ are
\begin{align}
    C_{\alpha\beta}=\left(\begin{array}{cccc}0&1&1&1\\1&1&1&2\\1&1&0&1\\1&2&1&1\end{array}\right)\,.
\end{align}

We apply the construction from Section~\ref{sec:generic4sections} and choose $X$ to be the 4-section fibration that corresponds to the bundles
\begin{align}
V=\mathcal{O}_{\text{dP}_3}^{\oplus 4}\,,\quad E=\mathcal{O}_{\text{dP}_3}(\check{D}_2)\oplus\mathcal{O}_{\text{dP}_3}(\check{D}_4)\,.
\end{align}
Using the Chern classes
\begin{align}
    c_1(V)=c_2(V)=0\,,\quad c_1(E)=c_1(B)\,,\quad c_2(E)=\check{D}_2\check{D}_4=2\,,
\end{align}
as well as $c_1(B)^2=6$, together with the expressions from Table~\ref{tab:fibrationsGenericData}, we obtain the Euler characteristic and the non-trivial Hodge numbers
\begin{align}
    \chi_X=-72\,,\quad h^{1,1}=5\,,\quad h^{2,1}=41\,.
\end{align}
as well as the intersection numbers
\begin{align}
    \kappa=2\,,\quad \ell_\alpha=2c_\alpha=2a_\alpha=(4,6,4,6)\,,\quad c_2=44\,,
\end{align}
and the height-pairing of the 4-section
\begin{align}
    D=-\pi^*\pi_*(D_eD_e)=-2(\check{D}_2+\check{D}_4)\,.
\end{align}

We find the base degree one topological string partition functions
\begin{align}
    Z_{\check{C}'_i}(T,\check{\lambda})=\frac{1}{2^{17}}\frac{\Delta_8^{\frac14}e_{4,1}^2 \left(e_{4,1}^2-e_{2,2}\right){}^5 \left(2 e_{4,1}^2-e_{2,2}\right)}{\eta(4T)^{36}\phi_{-2,1}(4T,\check{\lambda})}\,,\quad i=1,\ldots,6\,.
\end{align}

\bibliography{DTtorusfib.bib}

@article{Grassi:2011hq,
    author = "Grassi, Antonella and Morrison, David R.",
    title = "{Anomalies and the Euler characteristic of elliptic Calabi-Yau threefolds}",
    eprint = "1109.0042",
    archivePrefix = "arXiv",
    primaryClass = "hep-th",
    reportNumber = "UCSB-MATH-2011-08, IPMU-11-0106",
    doi = "10.4310/CNTP.2012.v6.n1.a2",
    journal = "Commun. Num. Theor. Phys.",
    volume = "6",
    pages = "51--127",
    year = "2012"
}

@article{Katz:2023zan,
    author = "Katz, Sheldon and Schimannek, Thorsten",
    title = "{New non-commutative resolutions of determinantal Calabi-Yau threefolds from hybrid GLSM}",
    eprint = "2307.00047",
    archivePrefix = "arXiv",
    primaryClass = "hep-th",
    month = "6",
    year = "2023",
    note= "To appear in Adv. Theo. Math. Phys. (2025)"
}

@article{Katz:2022lyl,
    author = "Katz, Sheldon and Klemm, Albrecht and Schimannek, Thorsten and Sharpe, Eric",
    title = "{Topological Strings on Non-commutative Resolutions}",
    eprint = "2212.08655",
    archivePrefix = "arXiv",
    primaryClass = "hep-th",
    doi = "10.1007/s00220-023-04896-2",
    journal = "Commun. Math. Phys.",
    volume = "405",
    number = "3",
    pages = "62",
    year = "2024"
}

@article{Haghighat:2014vxa,
    author = "Haghighat, Babak and Klemm, Albrecht and Lockhart, Guglielmo and Vafa, Cumrun",
    title = "{Strings of Minimal 6d SCFTs}",
    eprint = "1412.3152",
    archivePrefix = "arXiv",
    primaryClass = "hep-th",
    doi = "10.1002/prop.201500014",
    journal = "Fortsch. Phys.",
    volume = "63",
    pages = "294--322",
    year = "2015"
}

@article{Haghighat:2013gba,
    author = "Haghighat, Babak and Iqbal, Amer and Koz{\c{c}}az, Can and Lockhart, Guglielmo and Vafa, Cumrun",
    title = "{M-Strings}",
    eprint = "1305.6322",
    archivePrefix = "arXiv",
    primaryClass = "hep-th",
    doi = "10.1007/s00220-014-2139-1",
    journal = "Commun. Math. Phys.",
    volume = "334",
    number = "2",
    pages = "779--842",
    year = "2015"
}

@article{Candelas:1994hw,
    author = "Candelas, Philip and Font, Anamaria and Katz, Sheldon H. and Morrison, David R.",
    title = "{Mirror symmetry for two parameter models. 2.}",
    eprint = "hep-th/9403187",
    archivePrefix = "arXiv",
    reportNumber = "UTTG-25-93, IASSNS-HEP-94-12, OSU-M-94-1",
    doi = "10.1016/0550-3213(94)90155-4",
    journal = "Nucl. Phys. B",
    volume = "429",
    pages = "626--674",
    year = "1994"
}

@article{Bridgeland:2024ecw,
    author = "Bridgeland, Tom and Tulli, Iv{\'a}n",
    title = "{Resurgence and Riemann{\textendash}Hilbert Problems for Elliptic Calabi{\textendash}Yau Threefolds}",
    eprint = "2407.06974",
    archivePrefix = "arXiv",
    primaryClass = "hep-th",
    doi = "10.1007/s00220-025-05310-9",
    journal = "Commun. Math. Phys.",
    volume = "406",
    number = "6",
    pages = "132",
    year = "2025"
}

@Article{	  Aganagic:2006wq,
  author	= "Aganagic, Mina and Bouchard, Vincent and Klemm, Albrecht",
  title		= "{Topological Strings and (Almost) Modular Forms}",
  eprint	= "hep-th/0607100",
  archiveprefix	= "arXiv",
  doi		= "10.1007/s00220-007-0383-3",
  journal	= "Commun. Math. Phys.",
  volume	= "277",
  pages		= "771--819",
  year		= "2008"
}

@Article{	  Deboer:2006vg,
  author	= "de Boer, Jan and Cheng, Miranda C. N. and Dijkgraaf,
		  Robbert and Manschot, Jan and Verlinde, Erik",
  title		= "{A farey tail for attractor black holes}",
  journal	= "JHEP",
  volume	= "11",
  year		= "2006",
  pages		= "024",
  eprint	= "hep-th/0608059",
  archiveprefix	= "arXiv",
  slaccitation	= "%%CITATION = HEP-TH/0608059;%%"
}

@Article{	  Denef:2007vg,
  author	= "Denef, Frederik and Moore, Gregory W.",
  title		= "{Split states, entropy enigmas, holes and halos}",
  journal	= "JHEP",
  volume	= "1111",
  pages		= "129",
  doi		= "10.1007/JHEP11(2011)129",
  year		= "2011",
  eprint	= "hep-th/0702146",
  archiveprefix	= "arXiv",
  primaryclass	= "HEP-TH",
  slaccitation	= "%%CITATION = HEP-TH/0702146;%%"
}

@Article{	  Gaiotto:2006wm,
  author	= "Gaiotto, Davide and Strominger, Andrew and Yin, Xi",
  title		= "{The M5-brane elliptic genus: Modularity and BPS states}",
  journal	= "JHEP",
  volume	= "08",
  year		= "2007",
  pages		= "070",
  eprint	= "hep-th/0607010",
  archiveprefix	= "arXiv",
  doi		= "10.1088/1126-6708/2007/08/070",
  slaccitation	= "%%CITATION = HEP-TH/0607010;%%"
}

@Article{	  Gunaydin:2006bz,
  author	= {G\"unaydin, Murat and Neitzke, Andrew and Pioline, Boris},
  title		= "Topological wave functions and heat equations",
  journal	= "JHEP",
  volume	= "12",
  year		= "2006",
  pages		= "070",
  eprint	= "hep-th/0607200",
  slaccitation	= "%%CITATION = HEP-TH/0607200;%%"
}

@Article{	  Harvey:1995fq,
  author	= {Harvey, Jeffrey A. and Moore, Gregory W.},
  pages		= {315-368},
  journal	= {Nucl. Phys.},
  volume	= {B463},
  title		= {{Algebras, BPS States, and Strings}},
  eprint	= {hep-th/9510182},
  slaccitation	= {%%CITATION = HEP-TH 9510182;%%},
  year		= {1996},
  bibliodb	= {yes},
  bibliodb_last_updated={2007-01-30T07:50:52}
}

@Article{	  Hosono:1999qc,
  author	= {Hosono, S. and Saito, M. H. and Takahashi, A.},
  pages		= {177-208},
  journal	= {Adv. Theor. Math. Phys.},
  volume	= {3},
  title		= {{H}olomorphic anomaly equation and {BPS} state counting of
		  rational elliptic surface},
  eprint	= {hep-th/9901151},
  slaccitation	= {%%CITATION = HEP-TH 9901151;%%},
  year		= {1999},
  bibliodb	= {yes},
  bibliodb_last_updated={2006-07-30T08:26:48}
}

@InCollection{	  mr1363056,
  author	= {Kaneko, Masanobu and Zagier, Don},
  pages		= {165--172},
  volume	= {129},
  title		= {A generalized {J}acobi theta function and quasimodular
		  forms},
  year		= {1995},
  bibliodb	= {yes},
  bibliodb_last_updated={2006-06-04T16:43:56},
  mrreviewer	= {Bruce Hunt},
  mrnumber	= {MR1363056 (96m:11030)},
  mrclass	= {11F11 (11F03)},
  address	= {Boston, MA},
  publisher	= {Birkh\"auser Boston},
  series	= {Progr. Math.},
  booktitle	= {The moduli space of curves (Texel Island, 1994)}
}

@Article{	  Klemm:2004km,
  author	= {Klemm, A. and Kreuzer, Maximilian and Riegler, E. and
		  Scheidegger, E.},
  title		= {{Topological string amplitudes, complete intersection
		  Calabi-Yau spaces and threshold corrections}},
  eprint	= {hep-th/0410018},
  slaccitation	= {%%CITATION = HEP-TH 0410018;%%},
  year		= {2004},
  bibliodb	= {yes},
  bibliodb_last_updated={2007-01-30T07:50:53}
}

@Article{	  Verlinde:2004ck,
  author	= {Verlinde, Erik P.},
  title		= {Attractors and the holomorphic anomaly},
  eprint	= {hep-th/0412139},
  slaccitation	= {%%CITATION = HEP-TH 0412139;%%},
  year		= {2004},
  bibliodb	= {yes},
  bibliodb_last_updated={2005-03-18T14:15:16}
}

@Article{	  Witten:1993ed,
  author	= {Witten, Edward},
  title		= {Quantum background independence in string theory},
  eprint	= {hep-th/9306122},
  slaccitation	= {%%CITATION = HEP-TH 9306122;%%},
  year		= {1993},
  bibliodb	= {yes},
  bibliodb_last_updated={2003-11-03T01:45:00},
  archive	= "http://arXiv.org/abs"
}

@Article{	  Kachru:1995wm,
  author	= "Kachru, Shamit and Vafa, Cumrun",
  title		= "{Exact results for N=2 compactifications of heterotic
		  strings}",
  journal	= "Nucl. Phys.",
  volume	= "B450",
  year		= "1995",
  pages		= "69-89",
  eprint	= "hep-th/9505105",
  archiveprefix	= "arXiv",
  doi		= "10.1016/0550-3213(95)00307-E",
  slaccitation	= "%%CITATION = HEP-TH/9505105;%%"
}

@Article{	  Hosono:1993qy,
  author	= "Hosono, S. and Klemm, A. and Theisen, S. and Yau,
		  Shing-Tung",
  title		= "{Mirror symmetry, mirror map and applications to
		  Calabi-Yau hypersurfaces}",
  journal	= "Commun. Math. Phys.",
  volume	= "167",
  year		= "1995",
  pages		= "301-350",
  eprint	= "hep-th/9308122",
  archiveprefix	= "arXiv",
  doi		= "10.1007/BF02100589",
  slaccitation	= "%%CITATION = HEP-TH/9308122;%%"
}

@Article{	  Manschot:2007ha,
  author	= "Manschot, Jan and Moore, Gregory W.",
  title		= "{A Modern Fareytail}",
  journal	= "Commun. Num. Theor. Phys.",
  volume	= "4",
  year		= "2010",
  pages		= "103-159",
  eprint	= "0712.0573",
  archiveprefix	= "arXiv",
  primaryclass	= "hep-th",
  slaccitation	= "%%CITATION = 0712.0573;%%"
}

@Article{	  Angelantonj:2015rxa,
  author	= "Angelantonj, Carlo and Florakis, Ioannis and Pioline,
		  Boris",
  title		= "{Threshold corrections, generalised prepotentials and
		  Eichler integrals}",
  year		= "2015",
  eprint	= "1502.00007",
  archiveprefix	= "arXiv",
  primaryclass	= "hep-th",
  reportnumber	= "CERN-PH-TH-2015-011",
  slaccitation	= "%%CITATION = ARXIV:1502.00007;%%"
}

@Article{	  Klemm:2012sx,
  author	= "Klemm, Albrecht and Manschot, Jan and Wotschke, Thomas",
  title		= "{Quantum geometry of elliptic Calabi-Yau manifolds}",
  journal	= "Comm. Number Theor. Phys.",
  volume	= "6",
  year		= "2012",
  pages		= "849-917",
  eprint	= "1205.1795",
  archiveprefix	= "arXiv",
  primaryclass	= "hep-th",
  slaccitation	= "%%CITATION = ARXIV:1205.1795;%%"
}

@Article{	  Alexandrov:2019rth,
  author	= "Alexandrov, Sergei and Manschot, Jan and Pioline, Boris",
  title		= "{S-duality and refined BPS indices}",
  eprint	= "1910.03098",
  archiveprefix	= "arXiv",
  primaryclass	= "hep-th",
  reportnumber	= "L2C:19-204",
  doi		= "10.1007/s00220-020-03854-6",
  journal	= "Commun. Math. Phys.",
  volume	= "380",
  number	= "2",
  pages		= "755--810",
  year		= "2020"
}

@Article{	  Huang:2015sta,
  author	= "Huang, Min-xin and Katz, Sheldon and Klemm, Albrecht",
  title		= "{Topological String on elliptic CY 3-folds and the ring of
		  Jacobi forms}",
  eprint	= "1501.04891",
  archiveprefix	= "arXiv",
  primaryclass	= "hep-th",
  reportnumber	= "USTC-ICTS-15-02, BONN-TH-2015-01",
  doi		= "10.1007/JHEP10(2015)125",
  journal	= "JHEP",
  volume	= "10",
  pages		= "125",
  year		= "2015"
}

@Article{	  Alexandrov:2020qpb,
  author	= "Alexandrov, Sergei and Nampuri, Suresh",
  title		= "{Refinement and modularity of immortal dyons}",
  eprint	= "2009.01172",
  archiveprefix	= "arXiv",
  primaryclass	= "hep-th",
  reportnumber	= "L2C:20-102",
  doi		= "10.1007/JHEP01(2021)147",
  journal	= "JHEP",
  volume	= "01",
  pages		= "147",
  year		= "2021"
}

@Article{	  Cota:2019cjx,
  author	= "Cota, Cesar Fierro and Klemm, Albrecht and Schimannek,
		  Thorsten",
  title		= "{Topological strings on genus one fibered Calabi-Yau
		  3-folds and string dualities}",
  eprint	= "1910.01988",
  archiveprefix	= "arXiv",
  primaryclass	= "hep-th",
  reportnumber	= "BONN-TH-2019-05, UWThPh-2019-29",
  doi		= "10.1007/JHEP11(2019)170",
  journal	= "JHEP",
  volume	= "11",
  pages		= "170",
  year		= "2019"
}

@Article{	  Alim:2012ss,
  author	= "Alim, Murad and Scheidegger, Emanuel",
  title		= "{Topological Strings on Elliptic Fibrations}",
  eprint	= "1205.1784",
  archiveprefix	= "arXiv",
  primaryclass	= "hep-th",
  doi		= "10.4310/CNTP.2014.v8.n4.a4",
  journal	= "Commun. Num. Theor. Phys.",
  volume	= "08",
  pages		= "729--800",
  year		= "2014"
}

@Article{	  Braun:2011ux,
  author	= "Braun, Volker",
  title		= "{Toric Elliptic Fibrations and F-Theory
		  Compactifications}",
  eprint	= "1110.4883",
  archiveprefix	= "arXiv",
  primaryclass	= "hep-th",
  doi		= "10.1007/JHEP01(2013)016",
  journal	= "JHEP",
  volume	= "01",
  pages		= "016",
  year		= "2013"
}

@Article{	  Dierigl:2022zll,
  author	= "Dierigl, Markus and Oehlmann, Paul-Konstantin and
		  Schimannek, Thorsten",
  title		= "{The discrete Green-Schwarz mechanism in 6D F-theory and
		  elliptic genera of non-critical strings}",
  eprint	= "2212.04503",
  archiveprefix	= "arXiv",
  primaryclass	= "hep-th",
  reportnumber	= "LMU-ASC 33/22",
  doi		= "10.1007/JHEP03(2023)090",
  journal	= "JHEP",
  volume	= "03",
  pages		= "090",
  year		= "2023"
}

@Article{	  Schimannek:2021pau,
  author	= "Schimannek, Thorsten",
  title		= "{Modular curves, the Tate-Shafarevich group and
		  Gopakumar-Vafa invariants with discrete charges}",
  eprint	= "2108.09311",
  archiveprefix	= "arXiv",
  primaryclass	= "hep-th",
  reportnumber	= "UWThPh-2021-13",
  doi		= "10.1007/JHEP02(2022)007",
  journal	= "JHEP",
  volume	= "02",
  pages		= "007",
  year		= "2022"
}

@Article{	  Schimannek:2019ijf,
  author	= "Schimannek, Thorsten",
  title		= "{Modularity from Monodromy}",
  eprint	= "1902.08215",
  archiveprefix	= "arXiv",
  primaryclass	= "hep-th",
  doi		= "10.1007/JHEP05(2019)024",
  journal	= "JHEP",
  volume	= "05",
  pages		= "024",
  year		= "2019"
}

@Article{	  Marino:1998pg,
  author	= "Marino, Marcos and Moore, Gregory W.",
  title		= "{Counting higher genus curves in a Calabi-Yau manifold}",
  eprint	= "hep-th/9808131",
  archiveprefix	= "arXiv",
  reportnumber	= "YCPT-P23-98, YCPT-23",
  doi		= "10.1016/S0550-3213(98)00847-5",
  journal	= "Nucl. Phys. B",
  volume	= "543",
  pages		= "592--614",
  year		= "1999"
}

@InCollection{	  0990.11041,
  author	= "Zagier, Don and Gangl, Herbert",
  title		= "{Classical and elliptic polylogarithms and special values
		  of $L$-series.}",
  publisher	= "{Dordrecht: Kluwer Academic Publishers}",
  year		= "2000"
}

@Article{	  zbmath06149482,
  author	= {Bringmann, Kathrin and Guerzhoy, Pavel and Kent, Zachary and
		  Ono, Ken},
  title		= {{Eichler-Shimura theory for mock modular forms.}},
  fjournal	= {{Mathematische Annalen}},
  journal	= {{Math. Ann.}},
  issn		= {0025-5831; 1432-1807/e},
  volume	= {355},
  number	= {3},
  pages		= {1085--1121},
  year		= {2013},
  publisher	= {Springer-Verlag, Berlin},
  language	= {English},
  doi		= {10.1007/s00208-012-0816-y},
  msc2010	= {11F67 11F03}
}

@Misc{		  Bringmann2021,
  title		= {Eichler integrals of Eisenstein series as $q$-brackets of
		  weighted $t$-hook functions on partitions},
  author	= {Kathrin Bringmann and Ken Ono and Ian Wagner},
  year		= {2021},
  eprint	= {2009.07236},
  archiveprefix	= {arXiv},
  primaryclass	= {math.NT},
  url		= {https://arxiv.org/abs/2009.07236}
}

@Article{	  Knapp:2021vkm,
  title = {On genus one fibered Calabi-Yau threefolds with 5-sections},
  volume = {29},
  ISSN = {1095-0753},
  url = {http://dx.doi.org/10.4310/ATMP.251023012710},
  DOI = {10.4310/atmp.251023012710},
  number = {5},
  journal = {Advances in Theoretical and Mathematical Physics},
  publisher = {International Press of Boston},
  author = {Knapp,  Johanna and Scheidegger,  Emanuel and Schimannek,  Thorsten},
  year = {2025},
  pages = {1165–1364}
}

@Article{	  Klemm:1996hh,
  author	= "Klemm, Albrecht and Mayr, Peter and Vafa, Cumrun",
  editor	= "Froehlich, J. and Rittenberg, V. and Schwimmer, A.",
  title		= "{BPS states of exceptional noncritical strings}",
  eprint	= "hep-th/9607139",
  archiveprefix	= "arXiv",
  reportnumber	= "CERN-TH-96-184, HUTP-96-A031",
  doi		= "10.1016/S0920-5632(97)00422-2",
  journal	= "Nucl. Phys. B Proc. Suppl.",
  volume	= "58",
  pages		= "177",
  year		= "1997"
}

@Article{	  Banlaki:2019bxr,
  author	= "Banlaki, Andreas and Chattopadhyaya, Aradhita and Kidambi,
		  Abhiram and Schimannek, Thorsten and Schimpf, Maria",
  title		= "{Heterotic strings on $(K3\times T^2)/\mathbb{Z}_3$ and
		  their dual Calabi-Yau threefolds}",
  eprint	= "1911.09697",
  archiveprefix	= "arXiv",
  primaryclass	= "hep-th",
  doi		= "10.1007/JHEP04(2020)203",
  journal	= "JHEP",
  volume	= "04",
  pages		= "203",
  year		= "2020"
}

@Article{	  Kachru:1997bz,
  author	= "Kachru, Shamit and Klemm, Albrecht and Oz, Yaron",
  title		= "{Calabi-Yau duals for CHL strings}",
  eprint	= "hep-th/9712035",
  archiveprefix	= "arXiv",
  reportnumber	= "EFI-97-55, LBNL-41125, LBL-41125, UCB-PTH-97-59",
  doi		= "10.1016/S0550-3213(98)00228-4",
  journal	= "Nucl. Phys. B",
  volume	= "521",
  pages		= "58--70",
  year		= "1998"
}

@article{Oberdieck:2017pqm,
    author = "Oberdieck, Georg and Pixton, Aaron",
    title = "{Gromov-Witten theory of elliptic fibrations: Jacobi forms and holomorphic anomaly equations}",
    eprint = "1709.01481",
    archivePrefix = "arXiv",
    primaryClass = "math.AG",
    doi = "10.2140/gt.2019.23.1415",
    journal = "Geom. Topol.",
    volume = "23",
    pages = "1415--1489",
    year = "2019"
}

@article{Doran:2024kcb,
    author = "Doran, Charles and Pioline, Boris and Schimannek, Thorsten",
    title = "{Enumerative geometry and modularity in two-modulus K3-fibered Calabi-Yau threefolds}",
    eprint = "2408.02994",
    archivePrefix = "arXiv",
    primaryClass = "hep-th",
    month = "8",
    year = "2024",
    note= "To appear in Adv. Theo. Math. Phys. (2025)"
}

@article{Alexandrov:2010ca,
      author         = "Alexandrov, Sergei and Persson, Daniel and Pioline,
                        Boris",
      title          = "{Fivebrane instantons, topological wave functions and
                        hypermultiplet moduli spaces}",
      journal        = "JHEP",
      volume         = "1103",
      pages          = "111",
      doi            = "10.1007/JHEP03(2011)111",
      year           = "2011",
      eprint         = "1010.5792",
      archivePrefix  = "arXiv",
      primaryClass   = "hep-th",
}

@Article{Gopakumar:1998ii,
     author    = "Gopakumar, Rajesh and Vafa, Cumrun",
     title     = "{M-theory and topological strings. I}",
     year      = "1998",
     eprint    = "hep-th/9809187",
     archivePrefix = "arXiv",
     SLACcitation  = "%%CITATION = HEP-TH/9809187;%%"
}

@Article{Gopakumar:1998jq,
     author    = "Gopakumar, Rajesh and Vafa, Cumrun",
     title     = "{M-theory and topological strings. II}",
     year      = "1998",
     eprint    = "hep-th/9812127",
     archivePrefix = "arXiv",
     SLACcitation  = "%%CITATION = HEP-TH/9812127;%%"
}

@article{Alexandrov:2016enp,
    author = "Alexandrov, Sergei and Banerjee, Sibasish and Manschot, Jan and Pioline, Boris",
    title = "{Indefinite theta series and generalized error functions}",
    eprint = "1606.05495",
    archivePrefix = "arXiv",
    primaryClass = "math.NT",
    reportNumber = "L2C:16-078, IPHT-T16/058, TCDMATH 16-09, CERN-TH-2016-142, IPHT-T16-058, TCDMATH-16-09",
    doi = "10.1007/s00029-018-0444-9",
    journal = "Selecta Math.",
    volume = "24",
    pages = "3927",
    year = "2018"
}

@article{Nazaroglu:2016lmr,
    author = "Nazaroglu, Caner",
    title = "{$r$-Tuple Error Functions and Indefinite Theta Series of Higher-Depth}",
    eprint = "1609.01224",
    archivePrefix = "arXiv",
    primaryClass = "math.NT",
    doi = "10.4310/CNTP.2018.v12.n3.a4",
    journal = "Commun. Num. Theor. Phys.",
    volume = "12",
    pages = "581--608",
    year = "2018"
}

@article{Alexandrov:2018lgp,
    author = "Alexandrov, Sergei and Pioline, Boris",
    title = "{Black holes and higher depth mock modular forms}",
    eprint = "1808.08479",
    archivePrefix = "arXiv",
    primaryClass = "hep-th",
    reportNumber = "L2C:18-112",
    doi = "10.1007/s00220-019-03609-y",
    journal = "Commun. Math. Phys.",
    volume = "374",
    number = "2",
    pages = "549--625",
    year = "2019"
}

@article{Alexandrov:2016tnf,
    author = "Alexandrov, Sergei and Banerjee, Sibasish and Manschot, Jan and Pioline, Boris",
    title = "{Multiple D3-instantons and mock modular forms I}",
    eprint = "1605.05945",
    archivePrefix = "arXiv",
    primaryClass = "hep-th",
    reportNumber = "L2C:16-056, IPHT-T16-037, CERN-TH-2016-121, TCDMATH-16-08",
    doi = "10.1007/s00220-016-2799-0",
    journal = "Commun. Math. Phys.",
    volume = "353",
    number = "1",
    pages = "379--411",
    year = "2017"
}

@article{Alexandrov:2017qhn,
    author = "Alexandrov, Sergei and Banerjee, Sibasish and Manschot, Jan and Pioline, Boris",
    title = "{Multiple D3-instantons and mock modular forms II}",
    eprint = "1702.05497",
    archivePrefix = "arXiv",
    primaryClass = "hep-th",
    reportNumber = "L2C:17-011, CERN-TH-2017-040, IPHT-T17-020",
    doi = "10.1007/s00220-018-3114-z",
    journal = "Commun. Math. Phys.",
    volume = "359",
    number = "1",
    pages = "297--346",
    year = "2018"
}

@article{Oberdieck:2016nvt,
    author = "Oberdieck, Georg and Shen, Junliang",
    title = "{Curve counting on elliptic Calabi\textendash{}Yau threefolds via derived categories}",
    eprint = "1608.07073",
    archivePrefix = "arXiv",
    primaryClass = "math.AG",
    doi = "10.4171/jems/938",
    journal = "J. Eur. Math. Soc.",
    volume = "22",
    number = "3",
    pages = "967--1002",
    year = "2019"
}

@article{Gendler:2022ztv,
    author = "Gendler, Naomi and Heidenreich, Ben and McAllister, Liam and Moritz, Jakob and Rudelius, Tom",
    title = "{Moduli space reconstruction and Weak Gravity}",
    eprint = "2212.10573",
    archivePrefix = "arXiv",
    primaryClass = "hep-th",
    reportNumber = "ACFI-T22-10",
    doi = "10.1007/JHEP12(2023)134",
    journal = "JHEP",
    volume = "12",
    pages = "134",
    year = "2023"
}

@article{Alim:2021vhs,
    author = "Alim, Murad and Heidenreich, Ben and Rudelius, Tom",
    title = "{The Weak Gravity Conjecture and BPS Particles}",
    eprint = "2108.08309",
    archivePrefix = "arXiv",
    primaryClass = "hep-th",
    reportNumber = "ACFI-T21-09",
    doi = "10.1002/prop.202100125",
    journal = "Fortsch. Phys.",
    volume = "69",
    number = "11-12",
    pages = "2100125",
    year = "2021"
}

@article{Witten:1988xj,
    author = "Witten, Edward",
    title = "{Topological Sigma Models}",
    reportNumber = "IASSNS-HEP-88/7",
    doi = "10.1007/BF01466725",
    journal = "Commun. Math. Phys.",
    volume = "118",
    pages = "411",
    year = "1988"
}

@article{Bershadsky:1993ta,
author = {Bershadsky, Michael and Cecotti, Sergio and Ooguri, Hirosi and Vafa, Cumrun},
pages = {279-304},
journal = {Nucl. Phys.},
volume = {B405},
title = {Holomorphic anomalies in topological field theories},
eprint = {hep-th/9302103},
year = {1993},
}

@article{Antoniadis:1993ze,
    author = "Antoniadis, Ignatios and Gava, E. and Narain, K. S. and Taylor, T. R.",
    title = "{Topological amplitudes in string theory}",
    eprint = "hep-th/9307158",
    archivePrefix = "arXiv",
    reportNumber = "NUB-3071, IC-93-202, CPTH-A258-0793",
    doi = "10.1016/0550-3213(94)90617-3",
    journal = "Nucl. Phys. B",
    volume = "413",
    pages = "162--184",
    year = "1994"
}

@article{Katz:1999xq,
    author = "Katz, Sheldon H. and Klemm, Albrecht and Vafa, Cumrun",
    title = "{M theory, topological strings and spinning black holes}",
    eprint = "hep-th/9910181",
    archivePrefix = "arXiv",
    reportNumber = "HUTP-99-A056, IASSNS-HEP-98-107, OSU-M-99-9",
    doi = "10.4310/ATMP.1999.v3.n5.a6",
    journal = "Adv. Theor. Math. Phys.",
    volume = "3",
    pages = "1445--1537",
    year = "1999"
}

@article {gw-dt,
    AUTHOR = {Maulik, D. and Nekrasov, N. and Okounkov, A. and
              Pandharipande, R.},
     TITLE = {Gromov-{W}itten theory and {D}onaldson-{T}homas theory. {I}},
   JOURNAL = {Compos. Math.},
  FJOURNAL = {Compositio Mathematica},
    VOLUME = {142},
      YEAR = {2006},
    NUMBER = {5},
     PAGES = {1263--1285},
      ISSN = {0010-437X},
   MRCLASS = {14N35 (14J32)},
  MRNUMBER = {MR2264664 (2007i:14061)},
MRREVIEWER = {Hsian-Hua Tseng},
       DOI = {10.1112/S0010437X06002302},
       URL = {http://dx.doi.org/10.1112/S0010437X06002302},
}

@article{Feyzbakhsh:2021nds,
    author = "Feyzbakhsh, Soheyla and Thomas, Richard P.",
    title = "{Rank $r$ DT theory from rank $1$}",
    eprint = "2108.02828",
    archivePrefix = "arXiv",
    primaryClass = "math.AG",
  journal={Journal of the American Mathematical Society},
  volume={36},
  number={3},
  pages={795--826},
  year={2023}
}

@article{Aganagic:2003db,
    author = "Aganagic, Mina and Klemm, Albrecht and Marino, Marcos and Vafa, Cumrun",
    title = "{The Topological vertex}",
    eprint = "hep-th/0305132",
    archivePrefix = "arXiv",
    reportNumber = "CALT-68-2439, HUTP-03-A032, HU-EP-03-24, CERN-TH-2003-111",
    doi = "10.1007/s00220-004-1162-z",
    journal = "Commun. Math. Phys.",
    volume = "254",
    pages = "425--478",
    year = "2005"
}

@article{Alexandrov:2022pgd,
    author = "Alexandrov, Sergei and Gaddam, Nava and Manschot, Jan and Pioline, Boris",
    title = "{Modular bootstrap for D4-D2-D0 indices on compact Calabi\textendash{}Yau threefolds}",
    eprint = "2204.02207",
    archivePrefix = "arXiv",
    primaryClass = "hep-th",
    doi = "10.4310/ATMP.2023.v27.n3.a2",
    journal = "Adv. Theor. Math. Phys.",
    volume = "27",
    number = "3",
    pages = "683--744",
    year = "2023"
}

@article{Alexandrov:2023zjb,
    author = "Alexandrov, Sergei and Feyzbakhsh, Soheyla and Klemm, Albrecht and Pioline, Boris and Schimannek, Thorsten",
    title = "{Quantum geometry, stability and modularity}",
    eprint = "2301.08066",
    archivePrefix = "arXiv",
    primaryClass = "hep-th",
    doi = "10.4310/CNTP.2024.v18.n1.a2",
    journal = "Commun. Num. Theor. Phys.",
    volume = "18",
    number = "1",
    pages = "49--151",
    year = "2024"
}

@article{Alexandrov:2023ltz,
    author = "Alexandrov, Sergei and Feyzbakhsh, Soheyla and Klemm, Albrecht and Pioline, Boris",
    title = "{Quantum geometry and mock modularity}",
    eprint = "2312.12629",
    archivePrefix = "arXiv",
    primaryClass = "hep-th",
    month = "12",
    year = "2023"
}

@article{Huang:2006hq,
    author = "Huang, Min-xin and Klemm, Albrecht and Quackenbush, Seth",
    title = "{Topological string theory on compact Calabi-Yau: Modularity and boundary conditions}",
    eprint = "hep-th/0612125",
    archivePrefix = "arXiv",
    reportNumber = "MAD-TH-06-12",
    doi = "10.1007/978-3-540-68030-7_3",
    journal = "Lect. Notes Phys.",
    volume = "757",
    pages = "45--102",
    year = "2009"
}

@article{klemm2010noether,
  title="{Noether-Lefschetz theory and the Yau-Zaslow conjecture}",
  author={Klemm, Albrecht and Maulik, Davesh and Pandharipande, Rahul and Scheidegger, Emanuel},
  journal={Journal of the American Mathematical Society},
  volume={23},
  number={4},
  pages={1013--1040},
  year={2010},
        eprint = "0807.2477",
    archivePrefix = "arXiv",
    primaryClass = "math.AG"
}

@article{maulik2007gromov,
  title="{Gromov-Witten theory and Noether-Lefschetz theory}",
  author={Maulik, Davesh and Pandharipande, Rahul},
         eprint = "0705.1653",
    archivePrefix = "arXiv",
    primaryClass = "math.AG",
  year={2007}
}

@article{Bouchard:2016lfg,
      author         = "Bouchard, Vincent and Creutzig, Thomas and Diaconescu,
                        Duiliu-Emanuel and Doran, Charles and Quigley, Callum and
                        Sheshmani, Artan",
      title          = "{Vertical D4-D2-D0 Bound States on K3 Fibrations and
                        Modularity}",
      journal        = "Commun. Math. Phys.",
      volume         = "350",
      year           = "2017",
      number         = "3",
      pages          = "1069-1121",
      doi            = "10.1007/s00220-016-2772-y",
      eprint         = "1601.04030",
      archivePrefix  = "arXiv",
      primaryClass   = "hep-th",
      SLACcitation   = "%%CITATION = ARXIV:1601.04030;%%"
}

@article{Berglund:1997eb,
    author = "Berglund, Per and Henningson, Mans and Wyllard, Niclas",
    title = "{Special geometry and automorphic forms}",
    eprint = "hep-th/9703195",
    archivePrefix = "arXiv",
    reportNumber = "CERN-TH-97-054, CERN-TH-97-54, GOTEBORG-ITP-97-02, NSF-ITP-97-026",
    doi = "10.1016/S0550-3213(97)00396-9",
    journal = "Nucl. Phys. B",
    volume = "503",
    pages = "256--276",
    year = "1997"
}

@article{Henningson:1996jf,
    author = "Henningson, Mans and Moore, Gregory W.",
    title = "{Counting curves with modular forms}",
    eprint = "hep-th/9602154",
    archivePrefix = "arXiv",
    reportNumber = "YCTP-P4-96, RU-96-10",
    doi = "10.1016/0550-3213(96)00245-3",
    journal = "Nucl. Phys. B",
    volume = "472",
    pages = "518--528",
    year = "1996"
}

@article{Sheshmani,
author = "Vladimir Baranovsky and Ludmil Katzarkov and Maxim Kontsevich and Artan Sheshmani",
title= "Degenerations, derived Lagrangians and categorification of DT invariants",
note= "International Congress of Basic Science, 2024"
}

@article{Andreas:2001ve,
    author = "Andreas, Bjorn and Curio, Gottfried and Hernandez Ruiperez, Daniel and Yau, Shing-Tung",
    title = "{Fiber wise T duality for D-branes on elliptic Calabi-Yau}",
    eprint = "hep-th/0101129",
    archivePrefix = "arXiv",
    reportNumber = "HUB-EP-01-03",
    doi = "10.1088/1126-6708/2001/03/020",
    journal = "JHEP",
    volume = "03",
    pages = "020",
    year = "2001"
}

@article{Andreas:2000sj,
    author = "Andreas, Bjorn and Curio, Gottfried and Ruiperez, Daniel Hernandez and Yau, Shing-Tung",
    title = "{Fourier-Mukai transform and mirror symmetry for D-branes on elliptic Calabi-Yau}",
    eprint = "math/0012196",
    archivePrefix = "arXiv",
    reportNumber = "HUB-EP-00-60",
    month = "12",
    year = "2000"
}

@article{Katz:2006gn,
      author         = "Katz, Sheldon H.",
      title          = "{Genus zero Gopakumar-Vafa invariants of contractible
                        curves}",
      journal        = "J. Diff. Geom.",
      volume         = "79",
      year           = "2008",
      number         = "2",
      pages          = "185-195",
      eprint         = "math/0601193",
      archivePrefix  = "arXiv",
      primaryClass   = "math-ag",
      SLACcitation   = "%%CITATION = MATH/0601193;%%"
}

@article {gw-dt2,
    AUTHOR = {Maulik, D. and Nekrasov, N. and Okounkov, A. and
              Pandharipande, R.},
     TITLE = {Gromov-{W}itten theory and {D}onaldson-{T}homas theory. {II}},
   JOURNAL = {Compos. Math.},
  FJOURNAL = {Compositio Mathematica},
    VOLUME = {142},
      YEAR = {2006},
    NUMBER = {5},
     PAGES = {1286--1304},
      ISSN = {0010-437X},
   MRCLASS = {14N35 (14C05)},
  MRNUMBER = {MR2264665 (2007i:14062)},
MRREVIEWER = {Hsian-Hua Tseng},
       DOI = {10.1112/S0010437X06002314},
       URL = {http://dx.doi.org/10.1112/S0010437X06002314},
}

@article{Pandharipande:2007kc,
    author = "Pandharipande, R. and Thomas, R. P.",
    title = "{Curve counting via stable pairs in the derived category}",
    eprint = "0707.2348",
    archivePrefix = "arXiv",
    primaryClass = "math.AG",
    doi = "10.1007/s00222-009-0203-9",
    journal = "Invent. Math.",
    volume = "178",
    pages = "407--447",
    year = "2009"
}

@article{Alexandrov:2025sig,
    author = "Alexandrov, Sergei",
    title = "{Mock Modularity at Work, or Black Holes in a Forest}",
    eprint = "2505.02572",
    archivePrefix = "arXiv",
    primaryClass = "hep-th",
    reportNumber = "L2C:25-038",
    doi = "10.3390/e27070719",
    journal = "Entropy",
    volume = "27",
    number = "7",
    pages = "719",
    year = "2025"
}

@article{Alexandrov:2024jnu,
    author = "Alexandrov, Sergei and Bendriss, Khalil",
    title = "{Modular anomaly of BPS black holes}",
    eprint = "2408.16819",
    archivePrefix = "arXiv",
    primaryClass = "hep-th",
    doi = "10.1007/JHEP12(2024)180",
    journal = "JHEP",
    volume = "12",
    pages = "180",
    year = "2024"
}

@article{Gholampour:2017bxh,
    author = "Gholampour, Amin and Sheshmani, Artan and Yau, Shing-Tung",
    title = "{Localized Donaldson-Thomas theory of surfaces}",
    eprint = "1701.08902",
    archivePrefix = "arXiv",
    primaryClass = "math.AG",
    doi = "10.1353/ajm.2020.0011",
    journal = "Am. J. Math.",
    volume = "142",
    number = "2",
    pages = "405--442",
    year = "2020"
}

@article{Tanaka:2017jom,
    author = "Tanaka, Yuuji and Thomas, Richard P.",
    title = "{Vafa-Witten invariants for projective surfaces I: stable case}",
    eprint = "1702.08487",
    archivePrefix = "arXiv",
    primaryClass = "math.AG",
    doi = "10.1090/jag/738",
    journal = "J. Alg. Geom.",
    volume = "29",
    number = "4",
    pages = "603--668",
    year = "2020"
}

@article{Tanaka:2017bcw,
    author = "Tanaka, Yuuji and Thomas, Richard P.",
    title = "{Vafa{\textendash}Witten invariants for projective surfaces II: semistable case}",
    eprint = "1702.08488",
    archivePrefix = "arXiv",
    primaryClass = "math.AG",
    doi = "10.4310/PAMQ.2017.v13.n3.a6",
    journal = "Pure Appl. Math. Quart.",
    volume = "13",
    number = "3",
    pages = "517--562",
    year = "2017"
}

@article{toda2012stability,
  title={Stability conditions and curve counting invariants on Calabi--Yau 3-folds},
  author={Toda, Yukinobu},
  year={2012},
  journal = "Kyoto Journal of Mathematics",
number = "1",
pages = "1--50",
publisher = "Duke University Press",
volume = "52",
}

@Article{Schwarz:2006br,
     author    = "Schwarz, Albert and Tang, Xiang",
     title     = "{Quantization and holomorphic anomaly}",
     journal   = "JHEP",
     volume    = "03",
     year      = "2007",
     pages     = "062",
     eprint    = "hep-th/0611281",
     archivePrefix = "arXiv",
     SLACcitation  = "%%CITATION = HEP-TH/0611281;%%"
}

@article{Grimm:2013oga,
    author = "Grimm, Thomas W. and Kapfer, Andreas and Keitel, Jan",
    title = "{Effective action of 6D F-Theory with U(1) factors: Rational sections make Chern-Simons terms jump}",
    eprint = "1305.1929",
    archivePrefix = "arXiv",
    primaryClass = "hep-th",
    reportNumber = "MPP-2013-125",
    doi = "10.1007/JHEP07(2013)115",
    journal = "JHEP",
    volume = "07",
    pages = "115",
    year = "2013"
}

@Article{IonelParker:2013,
 Author = {Ionel, Eleny-Nicoleta and Parker, Thomas},
 Title = {The {Gopakumar}-{Vafa} formula for symplectic manifolds},
 Journal = {Ann. Math. (2)},
 Volume = {187},
 Number = {1},
 Pages = {1--64},
 Year = {2018},
 DOI = {10.4007/annals.2018.187.1.1},
     eprint = "1306.1516",
    archivePrefix = "arXiv"
}

@article{Doan:2021,
  doi = {10.48550/ARXIV.2103.08221},
  url = {https://arxiv.org/abs/2103.08221},
  author = {Doan, Aleksander and Ionel, Eleny-Nicoleta and Walpuski, Thomas},
  title = {{The Gopakumar-Vafa finiteness conjecture}},
  year = {2021},
    eprint = "2103.08221",
    archivePrefix = "arXiv",
}

@article{Costello:2012cy,
    author = "Costello, Kevin J. and Li, Si",
    title = "{Quantum BCOV theory on Calabi-Yau manifolds and the higher genus B-model}",
    eprint = "1201.4501",
    archivePrefix = "arXiv",
    primaryClass = "math.QA",
    month = "1",
    year = "2012"
}

@article{Caldararu:2024muy,
    author = "Caldararu, Andrei and Tu, Junwu",
    title = "{Effective Categorical Enumerative Invariants}",
    eprint = "2404.01499",
    archivePrefix = "arXiv",
    primaryClass = "math.AG",
    month = "4",
    year = "2024"
}

@article{Fang:2016svw,
    author = "Fang, Bohan and Liu, Chiu-Chu Melissa and Zong, Zhengyu",
    title = "{On the Remodeling Conjecture for Toric Calabi-Yau 3-Orbifolds}",
    eprint = "1604.07123",
    archivePrefix = "arXiv",
    primaryClass = "math.AG",
    doi = "10.1090/jams/934",
    journal = "J. Am. Math. Soc.",
    volume = "33",
    number = "1",
    pages = "135--222",
    year = "2020"
}

@article{Oehlmann:2019ohh,
    author = "Oehlmann, Paul-Konstantin and Schimannek, Thorsten",
    title = "{GV-Spectroscopy for F-theory on genus-one fibrations}",
    eprint = "1912.09493",
    archivePrefix = "arXiv",
    primaryClass = "hep-th",
    doi = "10.1007/JHEP09(2020)066",
    journal = "JHEP",
    volume = "09",
    pages = "066",
    year = "2020"
}

@article{Anderson:2007nc,
    author = "Anderson, Lara B. and He, Yang-Hui and Lukas, Andre",
    title = "{Heterotic Compactification, An Algorithmic Approach}",
    eprint = "hep-th/0702210",
    archivePrefix = "arXiv",
    doi = "10.1088/1126-6708/2007/07/049",
    journal = "JHEP",
    volume = "07",
    pages = "049",
    year = "2007"
}

@article{Anderson:2008uw,
    author = "Anderson, Lara B. and He, Yang-Hui and Lukas, Andre",
    title = "{Monad Bundles in Heterotic String Compactifications}",
    eprint = "0805.2875",
    archivePrefix = "arXiv",
    primaryClass = "hep-th",
    doi = "10.1088/1126-6708/2008/07/104",
    journal = "JHEP",
    volume = "07",
    pages = "104",
    year = "2008"
}

@article{Fisher2008,
  title = {The invariants of a genus one curve},
  volume = {97},
  ISSN = {0024-6115},
  url = {http://dx.doi.org/10.1112/plms/pdn021},
  DOI = {10.1112/plms/pdn021},
  number = {3},
  journal = {Proceedings of the London Mathematical Society},
  publisher = {Wiley},
  author = {Fisher,  Tom},
  year = {2008},
  month = apr,
  pages = {753-782}
}

@article{An2001,
  title = {Jacobians of Genus One Curves},
  volume = {90},
  ISSN = {0022-314X},
  url = {http://dx.doi.org/10.1006/jnth.2000.2632},
  DOI = {10.1006/jnth.2000.2632},
  number = {2},
  journal = {Journal of Number Theory},
  publisher = {Elsevier BV},
  author = {An,  Sang Yook and Kim,  Seog Young and Marshall,  David C. and Marshall,  Susan H. and McCallum,  William G. and Perlis,  Alexander R.},
  year = {2001},
  month = oct,
  pages = {304-315}
}

@article{Fisher2018,
  title = {A formula for the Jacobian of a genus one curve of arbitrary degree},
  volume = {12},
  ISSN = {1937-0652},
  url = {http://dx.doi.org/10.2140/ant.2018.12.2123},
  DOI = {10.2140/ant.2018.12.2123},
  number = {9},
  journal = {Algebra \& Number Theory},
  publisher = {Mathematical Sciences Publishers},
  author = {Fisher,  Tom},
  year = {2018},
  month = dec,
  pages = {2123-2150}
}

@article{Harris1984,
  title = {On symmetric and skew-symmetric determinantal varieties},
  volume = {23},
  ISSN = {0040-9383},
  url = {http://dx.doi.org/10.1016/0040-9383(84)90026-0},
  DOI = {10.1016/0040-9383(84)90026-0},
  number = {1},
  journal = {Topology},
  publisher = {Elsevier BV},
  author = {Harris,  Joe and Tu,  Loring W.},
  year = {1984},
  pages = {71-84}
}

@article{Kuznetsov2007,
  title = {Homological projective duality},
  volume = {105},
  ISSN = {1618--1913},
  url = {http://dx.doi.org/10.1007/s10240-007-0006-8},
  DOI = {10.1007/s10240-007-0006-8},
  number = {1},
  journal = {Publications math{\'e}matiques de l'IH{\'E}S},
  publisher = {Springer Science and Business Media LLC},
  author = {Kuznetsov,  Alexander},
  year = {2007},
  month = jun,
  pages = {157--220}
}

@Article{Blumenhagen:2010pv,
   author    = "Blumenhagen, Ralph and Jurke, Benjamin 
                and Rahn, Thorsten and Roschy, Helmut",
   title     = "{Cohomology of Line Bundles: A Computational Algorithm}",
   journal   = "J. Math. Phys.",
   volume    = "51",
   pages     = "103525",
   issue     = "10",
   year      = "2010",
   doi       = "10.1063/1.3501132",
   eprint    = "1003.5217",
   archivePrefix = "arXiv",
   primaryClass  = "hep-th"}

@Misc{cohomCalg:Implementation,
   title     = "{cohomCalg package}",
   howpublished  = "Download link",
   url       = "https://github.com/BenjaminJurke/cohomCalg",
   note      = "High-performance line bundle cohomology computation based on \cite{Blumenhagen:2010pv}",
   year      = "2010"}

@article{Haghighat:2015ega,
    author = "Haghighat, Babak and Murthy, Sameer and Vafa, Cumrun and Vandoren, Stefan",
    title = "{F-Theory, Spinning Black Holes and Multi-string Branches}",
    eprint = "1509.00455",
    archivePrefix = "arXiv",
    primaryClass = "hep-th",
    doi = "10.1007/JHEP01(2016)009",
    journal = "JHEP",
    volume = "01",
    pages = "009",
    year = "2016"
}

@article{Vafa:1997gr,
    author = "Vafa, Cumrun",
    title = "{Black holes and Calabi-Yau threefolds}",
    eprint = "hep-th/9711067",
    archivePrefix = "arXiv",
    reportNumber = "HUTP-97-A066",
    doi = "10.4310/ATMP.1998.v2.n1.a8",
    journal = "Adv. Theor. Math. Phys.",
    volume = "2",
    pages = "207--218",
    year = "1998"
}

@article{Maldacena:1997de,
    author = "Maldacena, Juan Martin and Strominger, Andrew and Witten, Edward",
    title = "{Black hole entropy in M theory}",
    eprint = "hep-th/9711053",
    archivePrefix = "arXiv",
    doi = "10.1088/1126-6708/1997/12/002",
    journal = "JHEP",
    volume = "12",
    pages = "002",
    year = "1997"
}

@article{Inoue:2019jle,
    author = "Inoue, Daisuke",
    title = "{Calabi{\textendash}Yau 3-folds from projective joins of del Pezzo manifolds}",
    eprint = "1902.10040",
    archivePrefix = "arXiv",
    primaryClass = "math.AG",
    doi = "10.1016/j.aim.2022.108243",
    journal = "Adv. Math.",
    volume = "400",
    pages = "108243",
    year = "2022"
}

@article{Castro:2007hc,
    author = "Castro, Alejandra and Davis, Joshua L. and Kraus, Per and Larsen, Finn",
    title = "{5D Black Holes and Strings with Higher Derivatives}",
    eprint = "hep-th/0703087",
    archivePrefix = "arXiv",
    doi = "10.1088/1126-6708/2007/06/007",
    journal = "JHEP",
    volume = "06",
    pages = "007",
    year = "2007"
}

@article{Ferrara:1996um,
    author = "Ferrara, Sergio and Kallosh, Renata",
    title = "{Universality of supersymmetric attractors}",
    eprint = "hep-th/9603090",
    archivePrefix = "arXiv",
    reportNumber = "CERN-TH-96-66, SU-ITP-96-10",
    doi = "10.1103/PhysRevD.54.1525",
    journal = "Phys. Rev. D",
    volume = "54",
    pages = "1525--1534",
    year = "1996"
}

@article{Brodie:2021nit,
    author = "Brodie, Callum R. and Constantin, Andrei and Lukas, Andre and Ruehle, Fabian",
    title = {{Geodesics in the extended K{\"a}hler cone of Calabi-Yau threefolds}},
    eprint = "2108.10323",
    archivePrefix = "arXiv",
    primaryClass = "hep-th",
    reportNumber = "CERN-TH-2021-123",
    doi = "10.1007/JHEP03(2022)024",
    journal = "JHEP",
    volume = "03",
    pages = "024",
    year = "2022"
}

@article{Pioline:2025xgf,
    author = "Pioline, Boris and Raj, Rishi",
    title = "{Multi-centered Black Hole Quantum Mechanics and Generalized Error Functions}",
    eprint = "2507.08551",
    archivePrefix = "arXiv",
    primaryClass = "hep-th",
    month = "7",
    year = "2025"
}

@article{Witten:1993yc,
    author = "Witten, Edward",
    editor = "Greene, B. and Yau, Shing-Tung",
    title = "{Phases of N=2 theories in two-dimensions}",
    eprint = "hep-th/9301042",
    archivePrefix = "arXiv",
    reportNumber = "IASSNS-HEP-93-3",
    doi = "10.1016/0550-3213(93)90033-L",
    journal = "Nucl. Phys. B",
    volume = "403",
    pages = "159--222",
    year = "1993"
}

@article{Caldararu:1997,
  title = "{Non-birational Calabi-Yau threefolds that are derived equivalent}",
  volume = {18},
  ISSN = {1793-6519},
  url = {http://dx.doi.org/10.1142/S0129167X07004205},
  DOI = {10.1142/s0129167x07004205},
  number = {05},
  journal = {International Journal of Mathematics},
  publisher = {World Scientific Pub Co Pte Ltd},
  author = {Caldararu, Andrei},
  year = {2007},
  month = may,
  pages = {491-504}
}

@book{Stein2007,
  title = {Modular Forms,  a Computational Approach},
  ISBN = {9781470418038},
  ISSN = {1065-7339},
  url = {http://dx.doi.org/10.1090/gsm/079},
  DOI = {10.1090/gsm/079},
  journal = {Graduate Studies in Mathematics},
  publisher = {American Mathematical Society},
  author = {Stein,  William},
  year = {2007},
  month = feb 
}

@book{Miyake1989,
  title = {Modular Forms},
  ISBN = {9783540295938},
  ISSN = {2196-9922},
  url = {http://dx.doi.org/10.1007/3-540-29593-3},
  DOI = {10.1007/3-540-29593-3},
  journal = {Springer Monographs in Mathematics},
  publisher = {Springer Berlin Heidelberg},
  author = {Miyake,  Toshitsune},
  year = {1989}
}

@article{Duque:2025kaa,
    author = "Duque, David Jaramillo and Kashani-Poor, Amir-Kian and Schimannek, Thorsten",
    title = "{The twisted geometry of 6d F-theory vacua with discrete gauge symmetries}",
    eprint = "2508.16500",
    archivePrefix = "arXiv",
    primaryClass = "hep-th",
    month = "8",
    year = "2025"
}

@article{Kraus:2005vz,
    author = "Kraus, Per and Larsen, Finn",
    title = "{Microscopic black hole entropy in theories with higher derivatives}",
    eprint = "hep-th/0506176",
    archivePrefix = "arXiv",
    doi = "10.1088/1126-6708/2005/09/034",
    journal = "JHEP",
    volume = "09",
    pages = "034",
    year = "2005"
}

@article{Dabholkar:2005dt,
    author = "Dabholkar, Atish and Denef, Frederik and Moore, Gregory W. and Pioline, Boris",
    title = "{Precision counting of small black holes}",
    eprint = "hep-th/0507014",
    archivePrefix = "arXiv",
    reportNumber = "LPTHE-05-14, LPTENS-05-21, TIFR-TH-05-27",
    doi = "10.1088/1126-6708/2005/10/096",
    journal = "JHEP",
    volume = "10",
    pages = "096",
    year = "2005"
}

@article{Pioline:2006ni,
    author = "Pioline, Boris",
    title = "{Lectures on black holes, topological strings and quantum attractors}",
    eprint = "hep-th/0607227",
    archivePrefix = "arXiv",
    reportNumber = "LPTENS-06-27",
    doi = "10.1088/0264-9381/23/21/S05",
    journal = "Class. Quant. Grav.",
    volume = "23",
    pages = "S981",
    year = "2006"
}

@article{Ceresole:1995jg,
    author = "Ceresole, Anna and D'Auria, R. and Ferrara, S. and Van Proeyen, Antoine",
    title = "{Duality transformations in supersymmetric Yang-Mills theories coupled to supergravity}",
    eprint = "hep-th/9502072",
    archivePrefix = "arXiv",
    reportNumber = "CERN-TH-7547-94, POLFIS-TH-01-95, UCLA-94-TEP-45A, KUL-TF-95-4, POLFIS-TH.-01-95, UCLA-94-TEP-45",
    doi = "10.1016/0550-3213(95)00175-R",
    journal = "Nucl. Phys. B",
    volume = "444",
    pages = "92--124",
    year = "1995"
}

@article{Haghighat:2008ut,
    author = "Haghighat, Babak and Klemm, Albrecht",
    title = "{Topological Strings on Grassmannian Calabi-Yau manifolds}",
    eprint = "0802.2908",
    archivePrefix = "arXiv",
    primaryClass = "hep-th",
    reportNumber = "BONN-TH-08-03",
    doi = "10.1088/1126-6708/2009/01/029",
    journal = "JHEP",
    volume = "01",
    pages = "029",
    year = "2009"
}

@article{Ueda:2016wfa,
    author = "Ueda, Kazushi and Yoshida, Yutaka",
    title = "{Equivariant A-twisted GLSM and Gromov--Witten invariants of CY 3-folds in Grassmannians}",
    eprint = "1602.02487",
    archivePrefix = "arXiv",
    primaryClass = "hep-th",
    reportNumber = "KIAS-P16013",
    doi = "10.1007/JHEP09(2017)128",
    journal = "JHEP",
    volume = "09",
    pages = "128",
    year = "2017"
}

@misc{inoue2019completeintersectioncalabiyaumanifolds,
      title="{Complete intersection Calabi--Yau manifolds with respect to homogeneous vector bundles on Grassmannians}", 
      author={Daisuke Inoue and Atsushi Ito and Makoto Miura},
      year={2019},
      eprint={1607.07821},
      archivePrefix={arXiv},
      primaryClass={math.AG},
      url={https://arxiv.org/abs/1607.07821}, 
}

@article{Hori2013,
  title = "{Linear sigma models with strongly coupled phases - one parameter models}",
  volume = {2013},
  ISSN = {1029-8479},
  url = {http://dx.doi.org/10.1007/JHEP11(2013)070},
  DOI = {10.1007/jhep11(2013)070},
  number = {11},
  journal = {Journal of High Energy Physics},
  publisher = {Springer Science and Business Media LLC},
  author = {Hori,  Kentaro and Knapp,  Johanna},
  year = {2013},
  month = nov 
}

@article{Kanazawa2012,
  title = "{Pfaffian Calabi-Yau threefolds and mirror symmetry}",
  volume = {6},
  ISSN = {1931-4531},
  url = {http://dx.doi.org/10.4310/CNTP.2012.v6.n3.a3},
  DOI = {10.4310/cntp.2012.v6.n3.a3},
  number = {3},
  journal = {Communications in Number Theory and Physics},
  publisher = {International Press of Boston},
  author = {Kanazawa,  Atsushi},
  year = {2012},
  pages = {661-696}
}

@article{miura2017,
   title="{Minuscule Schubert Varieties and Mirror Symmetry}",
   ISSN={1815-0659},
   url={http://dx.doi.org/10.3842/SIGMA.2017.067},
   DOI={10.3842/sigma.2017.067},
   journal={Symmetry, Integrability and Geometry: Methods and Applications},
   publisher={SIGMA (Symmetry, Integrability and Geometry: Methods and Application)},
   author={Miura, Makoto},
   year={2017},
   month=aug }

@article{Kapustka_2015,
   title="{Calabi-Yau threefolds in $\mathbb{P}^6$}",
   volume={195},
   ISSN={1618-1891},
   url={http://dx.doi.org/10.1007/s10231-015-0476-0},
   DOI={10.1007/s10231-015-0476-0},
   number={2},
   journal={Annali di Matematica Pura ed Applicata (1923 -)},
   publisher={Springer Science and Business Media LLC},
   author={Kapustka, Grzegorz and Kapustka, Michał},
   year={2015},
   month=jan, pages={529-556} }

@misc{almkvist2010tablescalabiyauequations,
      title="{Tables of Calabi--Yau equations}", 
      author={Gert Almkvist and Christian van Enckevort and Duco van Straten and Wadim Zudilin},
      year={2010},
      eprint={math/0507430},
      archivePrefix={arXiv},
      primaryClass={math.AG},
      url={https://arxiv.org/abs/math/0507430}, 
}

@article{Caldararu2002,
  title = {Derived categories of twisted sheaves on elliptic threefolds},
  volume = {2002},
  ISSN = {1435-5345},
  url = {http://dx.doi.org/10.1515/crll.2002.022},
  DOI = {10.1515/crll.2002.022},
  number = {544},
  journal = {Journal f\"{u}r die reine und angewandte Mathematik (Crelles Journal)},
  publisher = {Walter de Gruyter GmbH},
  author = {Caldararu,  Andrei},
  year = {2002},
  month = jan,
  pages = {161-179}
}

@article{Seidel2001,
  title = {Braid group actions on derived categories of coherent sheaves},
  volume = {108},
  ISSN = {0012-7094},
  url = {http://dx.doi.org/10.1215/S0012-7094-01-10812-0},
  DOI = {10.1215/s0012-7094-01-10812-0},
  number = {1},
  journal = {Duke Mathematical Journal},
  publisher = {Duke University Press},
  author = {Seidel,  Paul and Thomas,  Richard},
  year = {2001},
  month = may 
}

@article{HernndezRuiprez2009,
  title = {Relative integral functors for singular fibrations and singular partners},
  volume = {11},
  ISSN = {1435-9863},
  url = {http://dx.doi.org/10.4171/JEMS/162},
  DOI = {10.4171/jems/162},
  number = {3},
  journal = {Journal of the European Mathematical Society},
  publisher = {European Mathematical Society - EMS - Publishing House GmbH},
  author = {Hernández Ruipérez,  Daniel and López Martín,  Ana Cristina and Sancho de Salas,  Fernando},
  year = {2009},
  month = jun,
  pages = {597–625}
}

@article{Schimannek:2025cok,
    author = "Schimannek, Thorsten",
    title = "{In search of almost generic Calabi-Yau 3-folds}",
    eprint = "2504.06115",
    archivePrefix = "arXiv",
    primaryClass = "hep-th",
    month = "4",
    year = "2025",
    note = "To appear in Comm. Number Theor. Phys."
}

@article{Knapp:2025hnf,
    author = "Knapp, Johanna and McGovern, Joseph",
    title = "{Noncommutative resolutions and CICY quotients from a non-abelian GLSM}",
    eprint = "2504.06147",
    archivePrefix = "arXiv",
    primaryClass = "hep-th",
    month = "4",
    year = "2025"
}

@article{DelZotto:2017mee,
    author = "Del Zotto, Michele and Gu, Jie and Huang, Min-Xin and Kashani-Poor, Amir-Kian and Klemm, Albrecht and Lockhart, Guglielmo",
    title = "{Topological Strings on Singular Elliptic Calabi-Yau 3-folds and Minimal 6d SCFTs}",
    eprint = "1712.07017",
    archivePrefix = "arXiv",
    primaryClass = "hep-th",
    doi = "10.1007/JHEP03(2018)156",
    journal = "JHEP",
    volume = "03",
    pages = "156",
    year = "2018"
}

@article{Lee:2018spm,
    author = "Lee, Seung-Joo and Lerche, Wolfgang and Weigand, Timo",
    title = "{A Stringy Test of the Scalar Weak Gravity Conjecture}",
    eprint = "1810.05169",
    archivePrefix = "arXiv",
    primaryClass = "hep-th",
    reportNumber = "CERN-TH-2018-220",
    doi = "10.1016/j.nuclphysb.2018.11.001",
    journal = "Nucl. Phys. B",
    volume = "938",
    pages = "321--350",
    year = "2019"
}

@article{Lee:2018urn,
    author = "Lee, Seung-Joo and Lerche, Wolfgang and Weigand, Timo",
    title = "{Tensionless Strings and the Weak Gravity Conjecture}",
    eprint = "1808.05958",
    archivePrefix = "arXiv",
    primaryClass = "hep-th",
    reportNumber = "CERN-TH-2018-190",
    doi = "10.1007/JHEP10(2018)164",
    journal = "JHEP",
    volume = "10",
    pages = "164",
    year = "2018"
}

@article{Gu:2023mgf,
    author = "Gu, Jie and Kashani-Poor, Amir-Kian and Klemm, Albrecht and Marino, Marcos",
    title = "{Non-perturbative topological string theory on compact Calabi-Yau 3-folds}",
    eprint = "2305.19916",
    archivePrefix = "arXiv",
    primaryClass = "hep-th",
    doi = "10.21468/SciPostPhys.16.3.079",
    journal = "SciPost Phys.",
    volume = "16",
    number = "3",
    pages = "079",
    year = "2024"
}

@article{Douaud:2024khu,
    author = "Douaud, Simon and Kashani-Poor, Amir-Kian",
    title = "{Borel singularities and Stokes constants of the topological string free energy on one-parameter Calabi-Yau threefolds}",
    eprint = "2412.16140",
    archivePrefix = "arXiv",
    primaryClass = "hep-th",
    doi = "10.1007/JHEP06(2025)253",
    journal = "JHEP",
    volume = "06",
    pages = "253",
    year = "2025"
}

@article{Bousseau:2020ckw,
    author = "Bousseau, Pierrick and Fan, Honglu and Guo, Shuai and Wu, Longting",
    title = "{Holomorphic anomaly equation for $(\mathbb{P}^2,E)$ and the Nekrasov-Shatashvili limit of local $\mathbb{P}^2$}",
    eprint = "2001.05347",
    archivePrefix = "arXiv",
    primaryClass = "math.AG",
    doi = "10.1017/fmp.2021.3",
    journal = "Forum Math. Pi",
    volume = "9",
    pages = "e3",
    year = "2021"
}

@article{Guo:2018mol,
    author = "Guo, Shuai and Janda, Felix and Ruan, Yongbin",
    title = "{Structure of Higher Genus Gromov-Witten Invariants of Quintic 3-folds}",
    eprint = "1812.11908",
    archivePrefix = "arXiv",
    primaryClass = "math.AG",
    month = "12",
    year = "2018"
}

@article{Coates:2018hms,
    author = "Coates, Tom and Iritani, Hiroshi",
    title = "{Gromov-Witten Invariants of Local P{\textasciicircum}2 and Modular Forms}",
    eprint = "1804.03292",
    archivePrefix = "arXiv",
    primaryClass = "math.AG",
    doi = "10.1215/21562261-2021-0010",
    journal = "J. Math. (Kyoto)",
    volume = "61",
    pages = "543--706",
    year = "2021"
}

@article{Lho:2018ayw,
    author = "Lho, Hyenho and Pandharipande, Rahul",
    title = "{Stable quotients and the holomorphic anomaly equation}",
    doi = "10.1016/j.aim.2018.05.020",
    journal = "Adv. Math.",
    volume = "332",
    pages = "349--402",
    year = "2018"
}

@manual{sagemath,
  Key          = {SageMath},
  Author       = {{The Sage Developers}},
  Title        = {{S}ageMath, the {S}age {M}athematics {S}oftware {S}ystem ({V}ersion 9.4)},
  note         = {{\tt https://www.sagemath.org}},
  Year         = {2021},
}

@article{Manschot:2010sxc,
    author = "Manschot, Jan",
    title = "{Stability and duality in N=2 supergravity}",
    eprint = "0906.1767",
    archivePrefix = "arXiv",
    primaryClass = "hep-th",
    reportNumber = "RUNHETC-2009-04",
    doi = "10.1007/s00220-010-1104-x",
    journal = "Commun. Math. Phys.",
    volume = "299",
    pages = "651--676",
    year = "2010"
}

@article{Manschot:2010xp,
    author = "Manschot, Jan",
    title = "{Wall-crossing of D4-branes using flow trees}",
    eprint = "1003.1570",
    archivePrefix = "arXiv",
    primaryClass = "hep-th",
    reportNumber = "IPHT-T10-027",
    doi = "10.4310/ATMP.2011.v15.n1.a1",
    journal = "Adv. Theor. Math. Phys.",
    volume = "15",
    number = "1",
    pages = "1--42",
    year = "2011"
}

@article{Manschot:2010qz,
    author = "Manschot, Jan and Pioline, Boris and Sen, Ashoke",
    title = "{Wall Crossing from Boltzmann Black Hole Halos}",
    eprint = "1011.1258",
    archivePrefix = "arXiv",
    primaryClass = "hep-th",
    doi = "10.1007/JHEP07(2011)059",
    journal = "JHEP",
    volume = "07",
    pages = "059",
    year = "2011"
}

@article{Grimm:2018weo,
    author = "Grimm, Thomas W. and het Lam, Huibert and Mayer, Kilian and Vandoren, Stefan",
    title = "{Four-dimensional black hole entropy from F-theory}",
    eprint = "1808.05228",
    archivePrefix = "arXiv",
    primaryClass = "hep-th",
    doi = "10.1007/JHEP01(2019)037",
    journal = "JHEP",
    volume = "01",
    pages = "037",
    year = "2019"
}

@misc{fierrorcota-to-appear,
author = {Cesar Fierro Cota},
TITLE = "{Supergravity anomaly equations from modularity of Calabi-Yau threefolds}",
NOTE = "In preparation.",
}

@article{Couzens:2017way,
    author = "Couzens, Christopher and Lawrie, Craig and Martelli, Dario and Schafer-Nameki, Sakura and Wong, Jin-Mann",
    title = "{F-theory and AdS$_{3}$/CFT$_{2}$}",
    eprint = "1705.04679",
    archivePrefix = "arXiv",
    primaryClass = "hep-th",
    doi = "10.1007/JHEP08(2017)043",
    journal = "JHEP",
    volume = "08",
    pages = "043",
    year = "2017"
}

@article{Kontsevich:1994dn,
    author = "Kontsevich, Maxim",
    title = "{Homological Algebra of Mirror Symmetry}",
    eprint = "alg-geom/9411018",
    archivePrefix = "arXiv",
    month = "11",
    year = "1994"
}

@article{Bena:2006qm,
    author = "Bena, Iosif and Diaconescu, Duiliu-Emanuel and Florea, Bogdan",
    title = "{Black string entropy and Fourier-Mukai transform}",
    eprint = "hep-th/0610068",
    archivePrefix = "arXiv",
    reportNumber = "SLAC-PUB-12147, SU-ITP-06-27",
    doi = "10.1088/1126-6708/2007/04/045",
    journal = "JHEP",
    volume = "04",
    pages = "045",
    year = "2007"
}
\bibliographystyle{utphys}

\end{document}